\pdfoutput=1
\PassOptionsToPackage{pdftex,
pdfencoding=auto,
pdfnewwindow=true,
pdfusetitle=true,
bookmarks=true,
bookmarksnumbered=true,
bookmarksopen=true,
pdfpagemode=UseThumbs,
bookmarksopenlevel=1,
pdfpagelabels=true,
breaklinks=true,
colorlinks=true,
hyperfootnotes=true
}{hyperref}
\PassOptionsToPackage{usenames,dvipsnames,pdftex}{xcolor}
\documentclass[aps,amsfonts,amssymb,10pt,letterpaper, 
final, nofootinbib, longbibliography,
twocolumn,notitlepage,hyperref,twoside,allowtoday,accepted=2020-05-28,allowfontchangeintitle=false,nopdfoutputerror=false]{quantumarticle}
\usepackage[flushmargin,splitrule]{footmisc}
\usepackage{microtype}
\microtypesetup{
expansion={true,nocompatibility},
protrusion={true,nocompatibility},
activate={true,nocompatibility},
tracking=true,
kerning=true,
spacing=true
}
\microtypecontext{spacing=nonfrench}
\usepackage[square,sort&compress,comma,numbers,elide]{natbib}

\usepackage[showerrors,immediate]{silence}
\usepackage[utf8]{inputenx}
\usepackage[OT1]{fontenc}
\usepackage[english]{babel}
\usepackage[style=american,autopunct=true]{csquotes}
\usepackage{soul}

\usepackage{bold-extra}
\usepackage{appendix}

\usepackage{amsmath, amssymb, amsthm,verbatim, graphicx,bbm}
\DeclareMathAlphabet\mathbfcal{OMS}{cmsy}{b}{n}
\usepackage{epsfig}
\usepackage{textcomp}
\usepackage{mathrsfs, mathtools}
\usepackage{units}
\usepackage{fullpage, stackrel, subfigure}
\usepackage{paralist}
\usepackage{comment}
\usepackage{placeins}

\usepackage{subfigure}
\usepackage{mathrsfs}
\usepackage[all]{xy}

\usepackage{tikz}

\usetikzlibrary{calc,decorations.pathreplacing, intersections,through, patterns}
\usetikzlibrary{arrows,shapes}
\usetikzlibrary{calc,decorations.pathreplacing}
\definecolor{darkgreen}{rgb}{0,.5,0}

\usetikzlibrary{backgrounds}
\pgfdeclarelayer{myback}
\pgfsetlayers{background,myback,main} 

\usepackage{footmisc}

\usepackage[usenames,dvipsnames]{xcolor}
\definecolor{purple}{RGB}{128,0,128}
\definecolor{ultramarine}{RGB}{63, 0, 255}
\definecolor{medblue}{RGB}{0, 0, 100}
\definecolor{googleblue}{RGB}{34, 0, 204}
\definecolor{panblue}{RGB}{0,24,150}
\definecolor{carmine}{RGB}{150, 0, 24}
\definecolor{gray}{RGB}{150, 150, 150}
\definecolor{darkgreen}{RGB}{0, 80, 0}

\definecolor{jflyOrange}{RGB}{230,159,0}
\definecolor{jflySkyBlue}{RGB}{86,180,233}
\definecolor{jflyBluishGreen}{RGB}{0,158,115}
\definecolor{jflyYellow}{RGB}{240,228,66}
\definecolor{jflyBlue}{RGB}{0,114,178}
\definecolor{jflyVermillion}{RGB}{213,94,0}
\definecolor{jflyReddishPurple}{RGB}{204,121,167}

\usepackage{hyperref}
\hypersetup{linkcolor=carmine,citecolor=darkgreen,urlcolor=googleblue,
anchorcolor=OliveGreen}

\usepackage{tabularx}
\usepackage{booktabs} 
\newcolumntype{R}{>{\raggedleft\arraybackslash}X}
\newcolumntype{C}{>{\centering\arraybackslash}X}
\newcolumntype{L}{>{\raggedright\arraybackslash}X}

\usepackage{tikz}
\usetikzlibrary{calc,decorations.pathreplacing,decorations.markings}
\usetikzlibrary{arrows,shapes}
\usepackage{accents}

\pgfkeys{/tikz/.cd,
  circle color/.initial=black,
  circle color/.get=\circlecolor,
  circle color/.store in=\circlecolor,
}

\tikzset{dotted pattern/.style args={#1 and #2}{
   postaction=decorate,
   decoration={
    markings,
    mark=
    between positions 0 and 1 step #2
      with
      {
       \fill[radius=#1,\circlecolor] (0,0) circle;
      }
    }
  },
  dotted pattern/.default={1pt and 1.5mm},
}

\newcommand{\beq}{\begin{equation}}
\newcommand{\eeq}{\end{equation}}

\newcommand{\LOSR}[0]{\ifmmode\textup{\upshape LOSR}\else{\textup{\upshape LOSR}}\fi}
\newcommand{\DetLOSR}[0]{\ifmmode\textup{\upshape LDO}\else{\textup{\upshape LDO}}\fi}
\newcommand{\LDO}[0]{\ifmmode\textup{\upshape LDO}\else{\textup{\upshape LDO}}\fi}
\newcommand{\LSO}[0]{\ifmmode\textup{\upshape LSO}\else{\textup{\upshape LSO}}\fi}
\newcommand{\LDTNO}[0]{\ifmmode\textup{\upshape LDTNO}\else{\textup{\upshape LDTNO}}\fi}
\newcommand{\LO}[0]{\ifmmode\textup{\upshape LO}\else{\textup{\upshape LO}}\fi}
\newcommand{\LOCC}[0]{\ifmmode\textup{\upshape LOCC}\else{\textup{\upshape LOCC}}\fi}

\newcommand{\twotwotwotwo}[0]{\ensuremath{
\genfrac{(}{)}{0pt}{3}{2\,2}{2\,2}}}

\newcommand{\fourfourfourfour}[0]{\ensuremath{
\genfrac{(}{)}{0pt}{3}{4\,4}{4\,4}}}
\newcommand{\xyst}[0]{\ensuremath{
\genfrac{(}{)}{0pt}{3}{|X|\,|Y|}{|S|\,|T|}}}
\newcommand{\xystprime}[0]{\ensuremath{
\genfrac{(}{)}{0pt}{3}{|X'|\,|Y'|}{|S'|\,|T'|}}}

\newcommand{\nhphantom}[1]{\sbox0{#1}\hspace{-\the\wd0}}

\DeclareMathOperator*{\CHSH}{\operatorname{CHSH}}
\DeclareMathOperator*{\NPR}{\operatorname{NPR}}
\DeclareMathOperator*{\conv}{\longmapsto}

\DeclareMathOperator*{\interconv}{\longleftrightarrow}

\usepackage[only,Arrownot]{stmaryrd}
\DeclareMathOperator*{\nconv}{\mathrel{\mkern5.5mu\Arrownot\mkern-5.5mu\longmapsto}}

\DeclarePairedDelimiter{\expec}{\langle}{\rangle}

\DeclarePairedDelimiter{\norm}{\lVert}{\rVert}
\usepackage{braket}

\newtheorem{theo}{Theorem}
\newtheorem{thm}[theo]{Theorem}
\newtheorem{prop}[theo]{Proposition}
\newtheorem{lemma}[theo]{Lemma}
\newtheorem{lem}[theo]{Lemma}
\newtheorem{cor}[theo]{Corollary}

\newtheorem{defn}[theo]{Definition}

\theoremstyle{definition} 

\theoremstyle{plain}
\newtheorem{innercustomgeneric}{\customgenericname}
\providecommand{\customgenericname}{}
\newcommand{\newcustomtheorem}[2]{%
  \newenvironment{#1}[1]
  {%
   \renewcommand\customgenericname{#2}%
   \renewcommand\theinnercustomgeneric{##1}%
   \innercustomgeneric
  }
  {\endinnercustomgeneric}
}

\newcustomtheorem{customprop}{Proposition}
\newcustomtheorem{customcor}{Corollary}
\newcustomtheorem{customlem}{Lemma}

\newcommand{\term}[1]{\textcolor{medblue}{\textbf{\upshape #1}}}

\newcommand{\bel}[1]{{\color{blue!80!white} #1 }}

\AtBeginDocument{
\renewcommand{\footrule}{\vbox to 1pt{\hbox to\headwidth{\color{quantumgray}\leaders\hrule\hfil}\vss}}
\fancyfoot[C]{\sffamily\textcolor{quantumviolet}{\thepage}}
}
\begin{document}
\title{Quantifying Bell: the Resource Theory of Nonclassicality of Common-Cause Boxes}
\author{Elie Wolfe}
\affiliation{Perimeter Institute for Theoretical Physics, 31 Caroline St. N, Waterloo, Ontario, N2L 2Y5, Canada}
\author{David Schmid}
\affiliation{Perimeter Institute for Theoretical Physics, 31 Caroline St. N, Waterloo, Ontario, N2L 2Y5, Canada}
\affiliation{Institute for Quantum Computing and Dept. of Physics and Astronomy, University of Waterloo, Waterloo, Ontario N2L 3G1, Canada}
\author{Ana Bel\'en Sainz}
\affiliation{International Centre for Theory of Quantum Technologies, University of Gda\'nsk, 80-308 Gda\'nsk, Poland}
\affiliation{Perimeter Institute for Theoretical Physics, 31 Caroline St. N, Waterloo, Ontario, N2L 2Y5, Canada}
\author{Ravi Kunjwal}
\affiliation{Perimeter Institute for Theoretical Physics, 31 Caroline St. N, Waterloo, Ontario, N2L 2Y5, Canada}
\affiliation{Centre for Quantum Information and Communication, Ecole polytechnique de Bruxelles,
	CP 165, Universit\'e libre de Bruxelles, 1050 Brussels, Belgium}
\author{Robert W. Spekkens}
\affiliation{Perimeter Institute for Theoretical Physics, 31 Caroline St. N, Waterloo, Ontario, N2L 2Y5, Canada}
\date{\today}
\begin{abstract}
We take a resource-theoretic approach to the problem of quantifying nonclassicality in Bell scenarios.  
The resources are conceptualized as probabilistic processes from the setting variables to the outcome variables having a particular causal structure, namely, one wherein the wings are only connected by a common cause. 
We term them ``common-cause boxes''.
We define the distinction between classical and nonclassical resources in terms of whether or not a classical causal model can explain the correlations. One can then quantify the relative nonclassicality of resources by considering their interconvertibility relative to the set of operations that can be implemented using a classical common cause (which correspond to local operations and shared randomness). We prove that the set of free operations forms a polytope, which in turn allows us to derive an efficient algorithm for deciding whether one resource can be converted to another. We moreover define two distinct monotones with simple closed-form expressions in the two-party binary-setting binary-outcome scenario, and use these to reveal various properties of the pre-order of resources, including a lower bound on the cardinality of any complete set of monotones. 
 In particular, we show that the information contained in the degrees of violation of facet-defining Bell inequalities is not sufficient for quantifying nonclassicality, even though it is sufficient for witnessing nonclassicality.  Finally, we show that the continuous set of convexly extremal quantumly realizable correlations are all at the top of the pre-order of quantumly realizable correlations.  In addition to providing new insights on Bell nonclassicality, our work also sets the stage for quantifying nonclassicality in more general causal networks.
\end{abstract}

\maketitle
\onecolumngrid
\clearpage
\tableofcontents
\clearpage
\twocolumngrid

\section{Introduction}

Bell's theorem \cite{Bell64, Bell66} highlights a precise sense in which quantum theory requires a departure from a classical worldview.  Furthermore, violations of Bell inequalities provide a means for certifying the nonclassicality of {\em nature}, independently of the correctness of quantum theory. This is because Bell inequalities can be tested directly on experimental data.  Experimental tests under very weak assumptions have confirmed this nonclassicality \cite{Belltest1, Belltest2, Belltest3}. Correlations that violate Bell inequalities have also found  applications in information theory. Specifically, they constitute an information-theoretic resource insofar as they can be used to perform various cryptographic tasks in a device-independent way~\cite{BHK,Acin2006QKD,Scarani2006QKD,Acin2007QKD, colbeckamp, Pironio2010,Dhara2013DIRNG,vazirani14,Kaniewski2016chsh}.  Consequently, much previous effort has been made to quantify resourcefulness of correlations within Bell scenarios~\cite{gallego2012,de2014nonlocality,GellerPiani,gallego2016nonlocality,horodecki2015axiomatic,Amaral2017NCW,kaur2018fundamental,Brito2018tracedistance}.

In this paper, we take a resource-theoretic approach to quantifying the nonclassicality of a given correlation in a Bell scenario, grounded in a new perspective on Bell's theorem. This is the perspective of \term{causal modelling}, which differs from the traditional operational approaches both conceptually and in practice.
 Nevertheless, the natural choice of the set of free operations for the Bell scenario in our framework coincides with the one proposed in some previous works~\cite{de2014nonlocality,GellerPiani}, namely, \term{local operations and shared randomness} (\term{LOSR})\footnote{There is widespread agreement that the free operations should somehow consist of local operations supplemented with shared randomness, however, different authors have been led to formalize this idea differently, that is, they have been led to distinct proposals for the the set of free operations. Indeed, the formalization provided in Refs.~\cite{gallego2016nonlocality,Amaral2017NCW,kaur2018fundamental} is inconsistent with the one given in Refs.~\cite{de2014nonlocality,GellerPiani} and therefore also with the one presented here. A detailed discussion of this issue can be found in Appendix~\ref{comparisonsub}.}.
See also the subsequent works of Refs.~\cite{schmid2019typeindependent,rosset2019characterizing,LOSRvsLOCC}.

Our causal perspective on quantifying Bell nonclassicality also generalizes naturally to a framework for quantifying the nonclassicality of correlations in more general causal scenarios. We discuss this generalization in Section~\ref{comparisonsub}, but leave its development to future work.

\subsection{Summary of main results}

We now summarize the content and main results of our article.

In Section~\ref{landscape}, we 
articulate the view on Bell's theorem that motivates our approach---the causal modelling paradigm---and contrast it with two other views on Bell's theorem, namely, the strictly operational and superluminal causation paradigms.   In particular, we explain how the differences between these views impacts how one conceptualizes Bell inequality violations as a resource, and we highlight some of the advantages of our approach relative to the alternatives.  We also introduce the notion of partitioned process theories~\cite{Coecke2014} as the mathematical framework for resource theories that we adopt in this article.

In Section~\ref{basicresthry}, we provide a formal definition of the resource theory to be studied.  For bipartite Bell scenarios, we argue that the set of processes which naturally constitute the resources in our approach is the set of all bipartite processes with classical inputs and outputs that can arise within a causal model with a (possibly nonclassical) common cause between the wings. We also argue that the natural set of free operations on such processes are those that are achieved by embedding the process in a circuit for which the only connection between the wings is a {\em classical} common cause, and we demonstrate that this is equivalent to the set of local operations and shared randomness, as the latter is formalized in Refs.~\cite{de2014nonlocality,GellerPiani}. 

In Section~\ref{rtprelim}, we introduce some of the central concepts of any resource theory, including the notion of a pre-order and its features, the notion of monotones and complete sets thereof, and the notions of cost and yield monotones, which underlie the explicit monotone constructions that follow.

In Section~\ref{polything}, we show how one can use two instances of a linear program to determine the ordering relation which holds between any pair of resources (see Proposition~\ref{lem:belstheorem} and the discussion that follows it). 

In Section~\ref{twomonotones}, we define two monotones of particular interest.  The first (defined in Eq.~\eqref{eq:CHSHmonotonedefn}) is based on a yield construction relative to all resources in the Clauser-Horne-Shimony-Holt (CHSH) scenario~\cite{CHSH} (a bipartite Bell scenario where the settings and outcomes all have cardinality two) and where the yield is measured by the value of the canonical CHSH functional.  The second (defined in Eq.~\eqref{Malphadefn}) is based on a cost construction relative to a one-parameter family of resources in the CHSH scenario and where the cost is measured again by the value of the canonical CHSH functional.    Although both of these monotones are originally defined in terms of an optimization problem, we derive closed-form expressions for each of them for resources within the CHSH scenario (see Propositions~\ref{prop:eqaboveCHSH} and ~\ref{geom} respectively).  We show that within the  CHSH scenario~\cite{CHSH}, a variety of monotones which have been previously studied are all equivalent (up to a monotonic function) to the first of these monotones (see Corollary~\ref{measurecor}). Because our two monotones are provably not equivalent, this result implies that the second of our monotones provides information beyond that given by previously studied monotones.

In Section~\ref{sec:results} we leverage our two monotones to derive various global properties of the pre-order induced by single-copy deterministic conversions.
Specifically, we prove that the pre-order:
\begin{compactitem}
\item is not complete (i.e., there exist incomparable resources), 
\item is not weak (the incomparability relation is not transitive), 
\item has both infinite width and infinite height, 
\item is locally infinite. 
\end{compactitem}
We also prove that the two monotones just mentioned do not completely characterize the pre-order of resources, by showing that they fail to do so even for the special case of the CHSH scenario.
We further show (in Theorem~\ref{prop:eightmonotonesV2}) that no fewer than eight continuous monotones can do the job.
We also show (in Proposition~\ref{prop:zerodclasses}) that the equivalence classes among nonfree resources in the CHSH scenario (though not in general) are  given exactly by the orbits of the symmetry group of deterministic free operations. 

Finally, in Section~\ref{sec:qr}, we show that all of the global features of the pre-order hold even for the strict subset of resources which can be realized in quantum theory. We also prove (in Lemma~\ref{prop:ext}) that every extremal quantumly realizable resource is at the top of the pre-order of quantumly realizable resources, and (in Proposition~\ref{prop:continuoustop}) that there are a continuous set of incomparable resources at the top of this pre-order.

\subsection{How to read this article}

We will demonstrate in Section~\ref{basicresthry} that in spite of the difference in our attitude towards Bell's theorem, the definition of the set of resources and the set of free operations that is natural for the Bell scenario within the causal modelling paradigm coincides  with a definition that has been proposed within the strictly operational paradigm, namely, the one proposed in Refs.~\cite{de2014nonlocality,GellerPiani}. Because Bell scenarios are the focus of our article, any reader who would rather take the strictly operationalist attitude towards Bell's theorem can reinterpret all of our results through that lens. 
In particular, readers who are already sympathetic to the notion that LOSR, as defined in Refs.~\cite{de2014nonlocality,GellerPiani}, is the right choice of free operations may wish to skip
Sections \ref{landscape} and \ref{basicresthry}.  

To understand our conviction that LOSR constitutes the right choice of free operations for Bell scenarios, however, readers are advised to read Sections~\ref{landscape} and \ref{sec:freeoperations}.
 In particular, to understand how our approach differs (advantageously) from other approaches, readers are encouraged to examine Sections~\ref{operationalparadigm} and ~\ref{SCparadigm} as well as Appendix~\ref{comparison}. 
  
Because Section~\ref{rtprelim} reviews basic definitions and terminology for concepts related to resource theories, any reader who has expertise on resource theories may wish to skip this section.  We note, however, that some of the material presented therein is not found in standard treatments, such as our discussion of global properties of a pre-order and our discussion of a scheme for constructing useful cost and yield monotones. 

The presentation of our novel technical results begins in Section~\ref{polything}.
  
\section{Motivating our approach and contrasting it with alternatives} \label{landscape}

\subsection{Three views on Bell's theorem}

The traditional commentary on Bell's theorem  \cite{sep-bell-theorem, d1979quantum} takes a particular view on how to articulate the assumptions that are necessary to derive Bell inequalities. 
Among these assumptions, two are typically highlighted as deserving of the most scrutiny, namely, the assumptions that are usually termed {\em realism} and {\em locality}\footnote{Note, however, that different authors will formalize these assumptions in different ways.}.  Abandoning one or the other of these two assumptions is the starting point of most commentaries on what to do in the face of violations of Bell inequalities.\footnote{See, however, the discussion of superdeterminism in footnote~\ref{superd}.}
Furthermore, a schism seems to have developed between the camps that advocate for each of these two views~\cite{wiseman2014two}.

Among the researchers who take Bell's theorem to demonstrate the need to abandon realism, there is a contingent which adopts a purely operational attitude towards quantum theory, that is, an attitude wherein the scientist's job is merely to predict the statistical distribution of outcomes of measurements performed on specific preparations in a specified experimental scenario.  We shall refer to the members of this camp as {\em operationalists}~\cite{werner2014comment}.
 For such researchers, a violation of a Bell inequality is simply a litmus test for the inadequacy of a classical realist account of the experiment.
One particular type of operationalist attitude, which we shall term the \term{strictly operational paradigm}, advocates that physical concepts ought to be defined in terms of operational concepts, and consequently that any properties of a Bell-type experiment, such as whether it is signalling or not and what sorts of causal connections might hold between the wings, must be expressed in the language of the classical input-output functionality of that experiment.  In other words, they advocate that the only concepts that are meaningful for such an experiment are those that supervene\footnote{$A$-properties are said to supervene on $B$-properties if every $A$-difference implies a $B$-difference.} upon its input-output functionality.\footnote{
Some might describe what we have here called the strictly operational paradigm as the ``device-independent'' paradigm~\cite{Scarani2012device}, however, we avoid using the latter term here because its usage is not restricted to  describing a particular type of empiricist philosophy of science: it also has a more technical meaning in the context of quantum information theory, wherein it indicates whether or not a given information-theoretic protocol depends on a prior characterization of the devices used therein.  
Indeed, Bell-inequality-violating correlations have been shown to be a key resource in cryptography because they allow for device-independent implementations of cryptographic tasks\cite{BHK,Acin2006QKD,Scarani2006QKD,Acin2007QKD, colbeckamp, Pironio2010,Dhara2013DIRNG,vazirani14,Kaniewski2016chsh}. \label{fn5}
}   
Most prior work on quantifying the resource in Bell experiments
has been done within this paradigm, and the characteristic of experimental correlations that is usually taken to quantify the resource is simply some notion of distance from the set of correlations that satisfy all the Bell inequalities.

Consider, on the other hand, the researchers who take realism as sacrosanct,
and in particular those who take
 Bell's theorem to demonstrate the failure of locality---that is, the existence of superluminal causal influences~\cite{Maudlin2002quantum,Norsen2006}.\footnote{Although such influences do not imply the possibility of superluminal signalling, they do imply a certain tension with relativity theory if one believes that the latter does not merely concern anthropocentric concepts such as signalling, but also physical concepts such as causation.}
Researchers in this camp, whom we shall refer to as advocates of the \term{superluminal causation paradigm}, 
would presumably find it natural to quantify the resource of Bell inequality violations in terms of the strength of the superluminal causal influences required to account for them (within the framework of a classical causal model). An approach along these lines is described in Refs.~\cite{Chaves2015relaxing,Chaves2017causalmultipartite}.  Earlier work on the communication cost of simulating Bell-inequality violations~\cite{maudlin1992bell,Toner2003} is also naturally understood in this way.\footnote{\label{superd}A less common view on how to maintain realism in the face of Bell inequality violations is to hold fast to locality but give up on a different assumption that goes into the derivation of Bell inequalities, namely, that the hidden variables are statistically independent of the setting variables.  This is known as the ``superdeterministic'' response to Bell's theorem~\cite{hooft2013fate}.   
 Advocates of this approach would presumably find it natural to quantify the resource of Bell inequality violations in terms of the deviation from such statistical independence that is required to explain a given violation.  In particular, the results of Refs.~\cite{hall} and ~\cite{barrettgisin} seeking to quantify the nonindependence needed to explain a given Bell inequality violation 
might be framed within a resource-theoretic framework.  However, given that the setting variables can no longer be considered as freely specifiable inputs within such an approach, it would be inappropriate to conceptualize a Bell experiment as a box-type process as we have done here.
}

In recent years, a third attitude toward Bell's theorem---inspired by the framework of causal inference~\cite{Pearl2009}---has been gaining in popularity.
In this approach, the assumptions that go into the derivation of Bell inequalities are~\cite{Wood2015}: Reichenbach's principle (that correlations need to be explained causally), the framework of classical causal modelling, and the principle of no fine-tuning (that statistical independences should not be explained by fine-tuning of the values of parameters in the causal model). 
Here, a violation of a Bell inequality does not lead to the traditional dilemma between realism and locality, but rather attests to the impossibility of providing a non-fine-tuned explanation of the experiment within the framework of classical causal models.
This attitude implies the possibility of a new option for what assumption to give up in the face of such a violation. 
Specifically, the new possibility being contemplated is that one can hold fast to Reichenbach's principle and the principle of no fine-tuning---and hence to the possibility of achieving satisfactory causal explanations of correlations---by replacing the framework of classical causal models with an intrinsically nonclassical generalization thereof. 

As is shown in Ref.~\cite{Wood2015}, because the correlations in a Bell experiment do not provide a means of sending superluminal signals between the wings, the only causal structure  that is a candidate for explaining these correlations without fine-tuning is one wherein there is 
 a purely common-cause relation between the wings, that is, one which admits no causal influences between the wings.
 Therefore, the new approach to achieving a causal explanation of Bell inequality violations 
 is one that posits a common cause mechanism but replaces the usual formalism for causal models with one which allows for more general possibilities on how to represent its components~\cite{Allenetal}.\footnote{For instance, for the notion of a quantum causal model proposed in Ref.~\cite{Allenetal}, reversible deterministic causal dependencies are represented by unitaries rather than bijective functions, and lack of knowledge is represented by density operators rather than by classical probability distributions.} We refer to this attitude as the \term{causal modelling paradigm}.

The causal modelling paradigm implies not only a novel attitude towards Bell's theorem, but
also a change in how one conceives of the resource that powers the information-theoretic applications of Bell-inequality violations.
The resource is not taken to be some abstract notion of distance from the set of Bell-inequality-satisfying correlations within the space of all nonsignalling correlations, as advocates of the strictly operational paradigm seem to favour, nor to consist of the strength of superluminal causal influences, as advocates of the superluminal causation paradigm would presumably have it. Rather, we take the resource to be the {\em nonclassicality} required by any generalized causal model which can explain the Bell inequality violations without fine-tuning. 

We shall show that in the resource theory that emerges by adopting this attitude, the nonclassicality of common-cause processes in Bell experiments cannot be captured solely by the degree of violation of facet-defining Bell inequalities.
 That is, there are distinctions among such common-cause processes---different ways for these to be nonclassical---which do not correspond to distinctions in the degree of violation of any facet-defining Bell inequality. 

\subsection{The resource theory suggested by the causal modelling paradigm}

\subsubsection{Generalized causal models}

We will work with the notion of a generalized (i.e., not necessarily classical) causal model that has been developed in Refs.~\cite{Henson2014,Fritz2012beyondBell} using the framework of generalized probabilistic theories (GPTs)~\cite{hardy01, barrettGPT, janotta2014}), and refer to it as a \term{GPT causal model}.\footnote{In the language of operational probabilistic theories~\cite{CombsForGPTs,ArianoOPT}, we are considering \term{free} and \term{causal} GPTs. A GPT is said to be free if, for any mathematically well-formed closed circuit, it specifies a joint probability distribution over the outcomes of the instruments. 
{A GPT is said to be causal if there exists a unique deterministic effect (which is often interpreted as excluding backwards-in-time signalling in any circuit).}}
Since we are interested in the distinction between classical and nonclassical, without specifically distinguishing quantum and 
post-quantum
types of nonclassicality,
 we will not be making use of any of the recent work~\cite{CostaShrapnel, Allenetal, BLO} 
on devising an intrinsically {\em quantum} notion of a causal model.\footnote{However, we will consider the question of when certain correlations that arise in a GPT causal model can be quantumly realized. {Moreover, the follow-up work described in Ref.~\cite{schmidpostquantum} explicitly explores the distinction between resources that are quantumly realizable and those that necessitate a post-quantum GPT.}}

\begin{defn}
A \term{GPT causal model} consists of a causal structure, represented by a directed acyclic graph (DAG), and a set of GPT parameters.  The parameters specify, for each node in the DAG, a GPT operation from the composite system associated to the parents of that node to the system associated to the node.
\end{defn}

One can approach the study of nonclassicality in arbitrary causal structures from within the scope of these GPT causal models, and pursue the development of a resource theory of such nonclassical features. 

We focus on experimental scenarios that are multipartite.  The different wings of the experiment are commonly conceptualized as the laboratories of different parties, particularly when discussing information-theoretic tasks that may be undertaken by the parties.
If one restricts attention to scenarios wherein locally, at each wing, 
 systems can be put into arbitrary causal relations with one another (consistent with the absence of backwards-in-time causal influences), then  the only freedom in stipulating the causal structure is in stipulating the causal relations that hold {\em among} the wings of the experiment.  
Causal relations among the wings come in two forms: 
 (i) a relation indicating the potential for causal influence from one wing to another, corresponding to having access to a GPT channel from one to the other, and (ii) a relation indicating the potential for a common cause to act on a set of wings, corresponding to having a source which distributes a multipartite GPT state among them.
The GPT operations and GPT states representing, respectively, these cause-effect and common-cause relations, together with the GPT operations representing causal influences between systems at a given wing, constitute the parameters of the GPT causal model.

The possible operational statistics that one can observe in this scenario hence arise from all the possible ways one may assign values to the parameters of the GPT causal model -- both to those pertaining to the causal structure among the wings, and to those pertaining to the local actions in each wing. 
%



In this article, we focus on a particular type of causal structure among the wings, namely, one wherein there is a common cause that acts on all of the wings, but no causal influences between any of them, which we term a \term{Bell scenario}. However, in Appendix~\ref{diffcausalstr}, we do include some discussion regarding other possible causal structures among the wings.  Details about how entangled states and operations are represented in a GPT causal model can be found in Refs.~\cite{hardy01, barrettGPT,Henson2014}, and explicit descriptions of these for the Bell scenario are provided in Sec.~\ref{envelopingthry} and in Appendix~\ref{diffcausalstr}.


\subsubsection{The distinction between free and nonfree resources in the causal modelling paradigm}\label{sec:freeandnonfree}

We conceptualize any experimental configuration as a process from its inputs to its outputs.
In the framework of GPT causal models, one has the capacity to consider processes that have GPT systems as inputs and outputs at the various wings.  However, we will restrict our attention to processes that have only {\em classical} inputs and outputs.  Such processes can be conceptualized as black-box processes, to which one inputs classical variables and from which classical variables are output. They are therefore precisely the sorts of processes considered in the strictly operational paradigm.
We further restrict our attention to processes with a classical input and classical output at each wing, where the input temporally precedes the output.\footnote{Thus, we do not consider processes which involve a sequence over time of classical input variables and classical output variables; that is, in the language of Refs~\cite{qcombs08,qcombs09}, we do not consider general $n$-combs.  
} 
In the strictly operational paradigm,
 the term ``box'' is generally used as jargon for such processes (for instance, as it is used in the term ``PR box''~\cite{Popescu1994}). We therefore refer to such processes as \term{box-type processes} or simply \term{boxes}.  

\begin{defn} 
A \term{box} is a process with a classical input and a classical output at each wing, represented formally by a stochastic map from the tuple of inputs to the tuple of outputs.
\end{defn}

We use the term \term{common-cause box} to refer to box-type processes which can be realized using a causal structure consisting of a common cause acting on all of the wings.  These will be the resources that we focus on in this article.
 In GPT causal models, all common-cause boxes can be decomposed into the preparation of a GPT state on a multipartite system, followed by the distribution of the component subsystems among the wings, followed by each subsystem being subjected to a GPT measurement, chosen from a fixed set according to the classical input variable at that wing (the local setting variable), and the result of which is the classical output variable at that wing (the local outcome variable).  In short, such processes can be decomposed in the same manner in which a multipartite Bell experiment is decomposed in quantum theory.
 
 \begin{defn} 
A \term{common-cause box} (or, equivalentely, a \term{GPT-realizable common-cause box}) is a box that can arise from a GPT causal model of a multipartite Bell scenario, so that the inputs and outputs correspond, respectively, to the setting and outcome variables associated to a set of local GPT measurements implemented on a multipartite GPT state.
\end{defn}

The distinction between common-cause boxes that are \term{classically realizable} and those that are not
 (illustrated for the bipartite case in Fig.~\ref {fig:CandNCboxes}) is simply the distinction between whether there is a {\em classical} causal model underlying the process, or whether it is only realizable by a causal model 
 which invokes a nonclassical GPT.

\begin{figure}[htb!]
\begin{center}
\subfigure[\label{fig:CandNCboxesGPT}]
{
\centering
\includegraphics[scale=0.5]{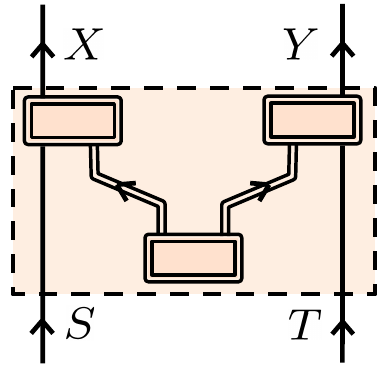}
}
\subfigure[\label{fig:CandNCboxesCLASSICAL}]
{
\centering
\includegraphics[scale=0.5]{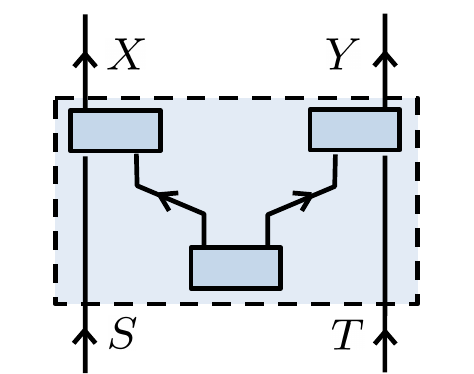}
}
\caption{In the bipartite scenario, the distinction between \subref{fig:CandNCboxesGPT} a {\em generic} GPT-realizable common-cause box and \subref{fig:CandNCboxesCLASSICAL} a {\em classically realizable}
 common-cause box.
Here and throughout this article, single-line edges denote classical systems, and single-line boxes denote processes that have only classical inputs and outputs (depicted in light blue). Double-line edges denote nonclassical systems and double-line boxes denote processes that have one or more nonclassical inputs or outputs (depicted in pink).  Any common-cause box whose input-output functionality is consistent with an internal structure of the type indicated in \subref{fig:CandNCboxesCLASSICAL} (regardless of its actual internal structure) is termed classically realizable and is considered free, while a common-cause box whose input-output functionality is not consistent with the structure of \subref{fig:CandNCboxesCLASSICAL} but instead is only consistent with an internal structure of the type indicated in \subref{fig:CandNCboxesGPT} is considered nonfree.
\label{fig:CandNCboxes}
}\end{center}
\end{figure}

Classical causal models are causal models wherein  all systems mediating causal influences are represented by classical variables, so that every common-cause source is represented by a joint probability distribution (referred to as {\em shared randomness}) and every channel is represented by a conditional probability distribution. Equivalently, classical causal models can be understood to arise as a subset of GPT causal models wherein the systems are presumed to be nonclassical (for instance, they might be presumed to be quantum),
  but every common-cause source is represented by a GPT-unentangled state and every channel is taken to be GPT-entanglement-breaking. 
Particularizing to the case of common-cause boxes, we have:
 
  \begin{defn} \label{CRCCB}
A \term{classically realizable common-cause box}
 is a common-cause box that admits of a {\em classical} causal model, such that the 
common-cause source consists of shared randomness.  Equivalently, 
it is a common-cause box that admits of a GPT causal model wherein the common-cause source consists of a GPT-unentangled state.
\end{defn}

  It follows that the free common-cause boxes are precisely the nonsignalling boxes that satisfy all the Bell inequalities, while the costly common-cause boxes are the nonsignalling boxes that violate some Bell inequality\footnote{Indeed, since the GPT colloquially known as \textit{Boxworld} realizes all and only the nonsignalling boxes in Bell scenarios \cite{barrettGPT}, it follows that all nonsignalling boxes admit of a GPT causal model.}.

\subsubsection{Quantifying resourcefulness in the causal modelling paradigm}\label{sec:quantifying}

In order to quantify the nonclassicality of common-cause boxes (that is, the extent to which they fail to be classically realizable), we will use an approach to resource theories described in Ref.~\cite{Coecke2014}, namely, the framework of  {\em partitioned process theories}.
An \term{enveloping theory of processes} must be specified, together with a subtheory of processes that can be implemented at no cost, called the \term{free subtheory of processes}. This partitions the set of all processes in the enveloping theory into free and costly (i.e., nonfree) processes. 
One can then ask of any pair of processes in the enveloping theory of a given type whether the first can be converted to the second by composing it with a process (of the appropriate type) from the free subtheory.
If interconversion between processes of type $T$ require composition with a process of type $T'$, then the \term{set of free operations} on processes of type $T$ are the elements of the free subtheory of processes that are of type $T'$.
Pairwise convertibility relations under the set of free operations define a pre-order on the set of the resources of interest,
and a partial order over the equivalence classes of such resources.  One can then quantify the relative worth of different resources by their relative positions in this partial order.  Functions over resources that preserve ordering relations, termed monotones, provide a particularly simple means of quantifying the worth of resources.   


The resource theory considered in this article will be described in full detail in Sec.~\ref{basicresthry}.  Nonetheless, we provide a sketch of its definition here  in order to be able to highlight the ways in which it contrasts with other approaches.

We take the {\em enveloping} theory of processes to include all GPT-realizable common-cause boxes as well as the GPT-realizable processes that  take every such box to another such box while only making use of a common cause (depicted in Fig.~\ref{fig:NonclassicalClamps}).\footnote{This is a general approach to determining a pre-order over resources of a given type---define the enveloping theory to include processes corresponding to the resource type of interest as well as 
the processes that are required to interconvert between such resources.}
 A process that takes a box to a box  we will refer to as a {\em clamp} because it is a process that has the form of a comb with two teeth (relative to the notion of ``comb'' introduced in \cite{qcombs08,qcombs09}). More precisely, a process taking boxes to boxes is a clamp  with classical inputs and outputs.  Those that only make use of a common cause, we refer to as {\em common-cause clamps with classical inputs and outputs}.   Such clamps are the most general type of process required  in our enveloping theory because common-cause boxes are a special case of these (for instance, when all the systems at the inputs and outputs on the bottom teeth of the clamp are trivial). 
  
  We take the {\em free subtheory} of processes in our resource theory to consist of the subset of common-cause clamps with classical inputs and outputs  that can be realized in a classical causal model, termed \term{classically realizable}. 
The distinction between a generic GPT-realizable common-cause clamp and a classically realizable one is depicted in Fig.~\ref{fig:ClassicalAndNonclassicalClamps}.  By virtue of boxes being a special type of clamp, this definition is consistent with Definition~\ref{CRCCB}.\footnote{For both the GPT-realizable and classically realizable varieties of these processes, one can define notions of sequential and parallel composition such that the set of processes, together with these composition relations, satisfy the formal definition of a process theory~\cite{Coecke2014}, thereby justifying the claim that the resource theory we have defined is formally a partitioned process theory.
The proof of this fact, however, is not relevant to any of the results in this article and is postponed to forthcoming work~\cite{FutureRTContextuality}.   
}


\begin{figure}[htb!]
\begin{center}
\subfigure[\label{fig:NonclassicalClamps}]
{
\centering
\includegraphics[scale=0.5]{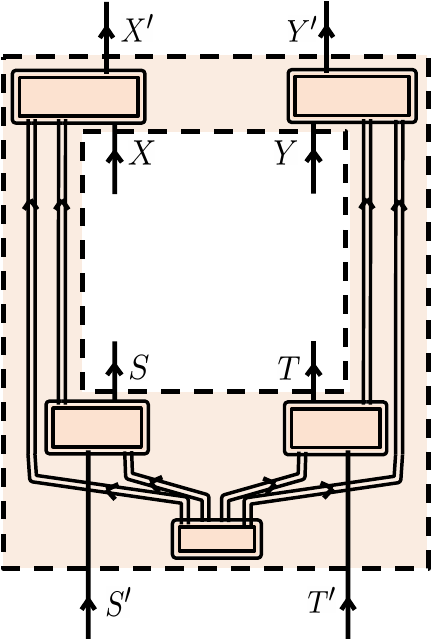}
}
\subfigure[\label{fig:ClassicalClamps}]
{
\centering
\includegraphics[scale=0.5]{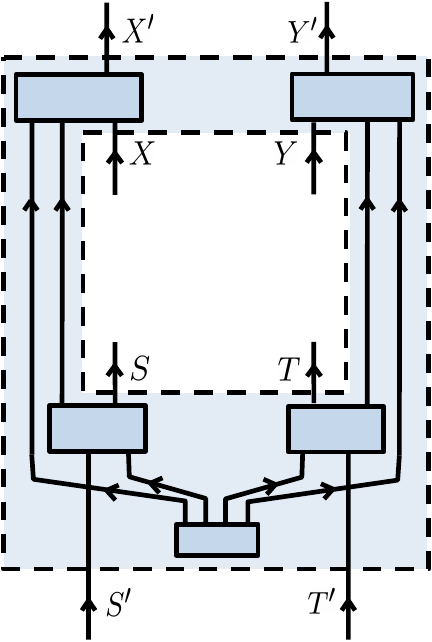}
}
\caption{In the bipartite scenario, the distinction between \subref{fig:NonclassicalClamps} a \emph{generic} GPT-realizable common-cause clamp with classical inputs and outputs and \subref{fig:ClassicalClamps} a \emph{classically realizable}
 common-cause clamp with classical inputs and outputs. 
 Any common-cause clamp whose input-output functionality is consistent with an internal structure of the type indicated in \subref{fig:ClassicalClamps} (regardless of its actual internal structure) is termed classically realizable and is considered free, while a common-cause clamp whose input-output functionality is not consistent with the structure of \subref{fig:ClassicalClamps} but instead is only consistent with an internal structure of the type indicated in \subref{fig:NonclassicalClamps} is considered nonfree.
\label{fig:ClassicalAndNonclassicalClamps}
}
\end{center}
\end{figure}



To determine the ordering relations that hold among these common-cause  boxes, one must determine the convertibility relations among them.  Given the definition of our resource theory, whether one GPT-realizable common-cause box can be converted to another is determined by whether this can be 
achieved by processing it with a {\em classically realizable} common-cause clamp, as depicted in Fig.~\ref{fig:transfexplicit}.
This subsumes correlated local processings of the inputs and outputs of the box, as we describe in Section~\ref{sec:freeoperations}.

\begin{figure}[ht]
 \centering
  \includegraphics[width=225pt]{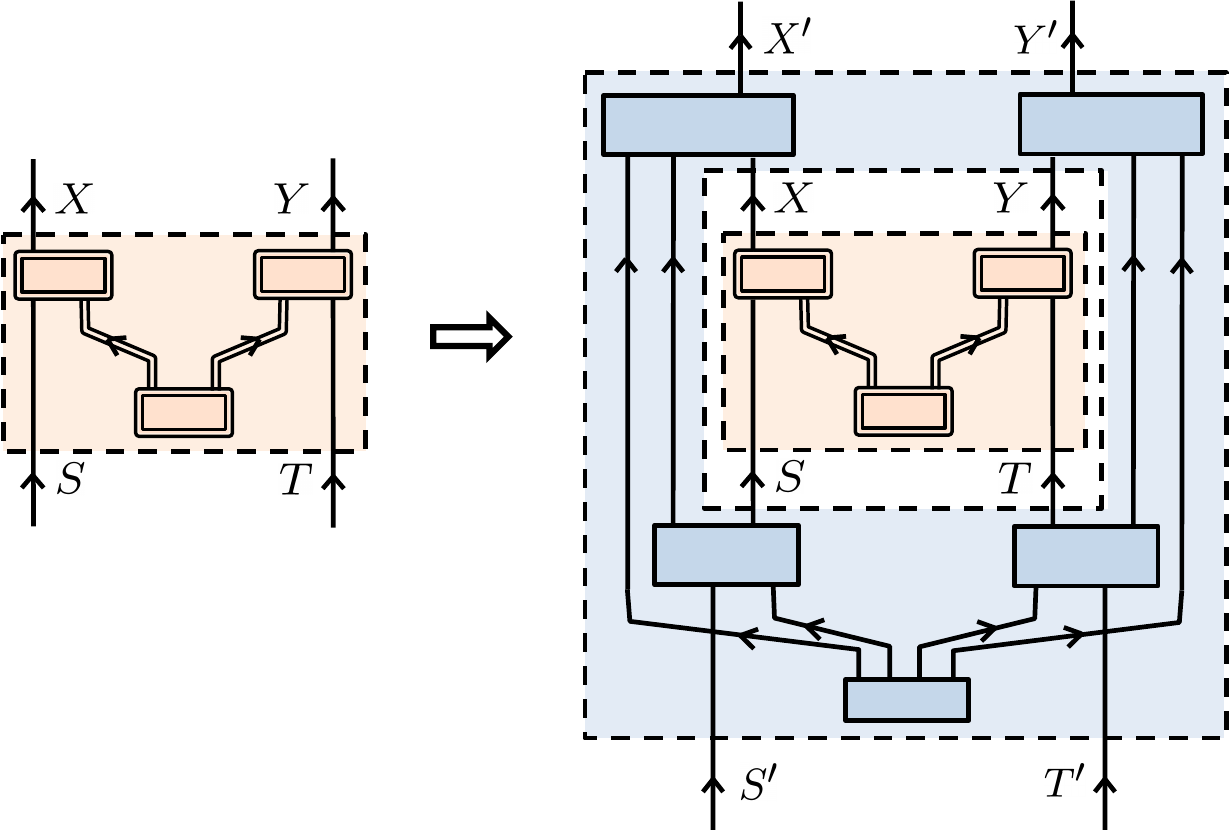}
 \caption{In the bipartite scenario, the most general form of a free operation (in blue) taking a GPT-realizable common-cause box (in pink) to another. } 
 \label{fig:transfexplicit}
\end{figure}

\subsubsection{A note about nomenclature}

In this article, we avoid describing the resource behind Bell inequality violations as \enquote{nonlocality}.  This is because we believe that it is {\em only} for those who take the lesson of Bell's theorem to be the existence of superluminal causal influences that it is appropriate to describe violations of Bell inequalities by this term. 
Researchers in the operationalist camp have not, generally speaking, avoided using the term \enquote{nonlocality}, but seem instead to use it as a synonym for ``violation of a Bell inequality'' rather than to imply a commitment to superluminal causal influences.  However, we believe that such a usage invites confusion and so we opt instead to avoid the term altogether.
Nevertheless,
our project is very much in line with earlier projects that 
describe themselves as developing a \enquote{resource theory of nonlocality}, such as Refs.~\cite{gallego2012,de2014nonlocality,GellerPiani,gallego2016nonlocality}.

\subsection{Contrast to the strictly operational paradigm}\label{operationalparadigm}

As noted in the introduction and as will be demonstrated in Section \ref{sec:freeoperations},
in the special case of Bell scenarios---the focus of this article---the natural set of free operations within our causal modelling paradigm is equivalent to one of the proposals for the set of free operations made in 
earlier works within the strictly operational paradigm,
 namely, {\em local operations and shared randomness} (LOSR), as the latter is defined in Refs.~\cite{de2014nonlocality,GellerPiani}. 
Additionally, the natural enveloping theory adopted in the strictly operational approach, namely, the set of no-signalling boxes, also coincides with that of our enveloping theory for the case of Bell scenarios, namely, the set of GPT-realizable common-cause boxes (where the equivalence of these two sets can be inferred from the results of Ref.~\cite{barrettGPT}).
Therefore, in spite of the difference in the attitude we take towards Bell's theorem, the resource theory that we define for Bell scenarios is the same as the one studied in Refs.~\cite{de2014nonlocality,GellerPiani}.

Nonetheless, the difference in our attitude towards Bell's theorem is not inconsequential.
 We presently outline its significance for the project of this article as well as for potential future generalizations of this project.

Most importantly, the causal modelling approach diverges sharply from any strictly operational approach once one considers causal structures beyond Bell scenarios. As discussed in Appendix~\ref{diffcausalstr},  in a resource theory of nonclassicality for more general causal structures, both the free subtheory and the enveloping theory proposed by the causal modelling approach are radically different from those suggested by the strictly operational approach.
In particular, the free subtheory need not be LOSR in a general causal structure and the enveloping theory need not be the set of all nonsignalling operations.
Our approach allows us to define a resource theory that is specific to a scenario in which only strict subsets of the wings are connected by common causes~\cite{BilocalCorrelations,Fritz2012beyondBell} (such as the triangle-with-settings scenario described in Appendix~\ref{diffcausalstr}) and this provides a concrete example of a case where the free subtheory is not LOSR and the enveloping theory is not all nonsignalling operations.  In these cases, the free operations are``local operations and causally admissable shared randomness'',wherein only those subsets of wings that are connected by a common cause have shared randomness.  This is distinct from the LOSR operations, which assume that randomness is shared between all the wings. It seems unlikely that the resource theory we propose in these cases can be motivated (or even fully characterized) in the strictly operational paradigm. 

Even for Bell scenarios, however, the causal modelling approach offers advantages over its competitors.  In particular, 
 it singles out a unique set of free operations, while the strictly operational approach does not. From our perspective, the resource underlying Bell inequality violations is the nonclassicality of the causal model required to explain them with a common cause, so \emph{clearly} the free operations should involve only classical common causes acting between the wings. In the strictly operational paradigm, by contrast, any operation which preserves no-signalling and takes local boxes to local boxes might constitute a legitimate candidate for a \enquote{free} operation. This ambiguity is reflected in the existence of distinct proposals for the set of free operations in strictly operational resource theories. Aside from LOSR, there is also a proposal called \textit{wirings and prior-to-input classical communication} (WPICC)~\cite{gallego2016nonlocality} which allows for classical causal influences between the wings {\em prior} to when the parties receive their inputs (See Appendix~\ref{WPICCvsLOSR}). If one believes that there is a singular concept which underlies the violation of Bell inequalities, then at most {\em one} of these proposals (LOSR or WPICC) can be taken as the relevant set of free operations.\footnote{Competing sets of free operations may be interesting for studying phenomena {\em other} than the resource that powers violations of Bell inequalities, but this is not the issue at stake in this article.} Although WPICC operations meet all desired operational criteria, they are immediately ruled out as candidates for the free operations within the causal modelling paradigm, on the grounds that they involve nontrivial cause-effect influences between the wings.

Another advantage of our approach for the Bell scenario is that it highlights the fact that {\LOSR} is {\em by construction} a convex set, a fact which is critical for the algorithmic method that we derive for determining the ordering relation between any two resources. In highlighting this fact, our approach led us to notice an oversight in some previous attempts to formalize LOSR, as discussed in Appendix~\ref{comparisonsub}.

Finally, we note that prior work of \citet{GellerPiani} departs from the strictly operational paradigm through their use of the \emph{unified operator formalism}~\cite{Acin2010Unified,Short2013}, which is analogous to the quantum formalism, but where nonpositive Hermitian operators are allowed to represent states. They do not characterize boxes primarily by their input-output functionality, but rather as a composition of a bipartite source with local measurements.  Indeed in their Fig. 4, they explicitly depict the internal structure of the box.  It is in this sense that their approach does not quite fit the mould of a strictly operational approach but is rather
 somewhat more in the flavour of the causal modelling approach we have described here. 

Nonetheless, the unified operator formalism differs significantly from the GPT formalism of Refs.~\cite{Henson2014,Fritz2012beyondBell} with respect to the \emph{independence} of the nonclassical common cause from the measurements employed in realizing nonclassical boxes. In the unified operator formalism, the Hermitian operator describing the shared state cannot be chosen freely for a given set of quantum measurements, because some choices would yield negative numbers rather than valid probabilities. By contrast, in the GPT formalism that we adopt here, the set of GPT states is contained within the dual of the set of GPT product measurements, and hence any measurement scheme can be paired with any shared state while yielding valid probabilities.
The causal modelling paradigm must reject any dependence of the shared state on the choice of measurements, while such dependence is unavoidable within the unified operator formalism. As defined in Ref.~\cite{Pearl2009}, a causal model is a directed acyclic graph, or equivalently, a circuit of causal processes, wherein the distinct processes in the circuit are required to be {\em autonomous} (i.e., independently variable). We therefore classify Ref.~\cite{GellerPiani} as neither within the causal modelling paradigm nor within the strictly operational paradigm, while still exhibiting some features of each of these approaches.
 
\subsection{Contrast to the superluminal causation paradigm} \label{SCparadigm}
 
To our knowledge, 
advocates of the superluminal causation paradigm have not attempted to develop a resource theory for Bell inequality violations (although Refs.~\cite{Chaves2015relaxing,Chaves2017causalmultipartite} are related in spirit). 
If it {\em were} attempted (within the framework of Ref.~\cite{Coecke2014}), then the commitments of the approach suggest that it would also be done differently from the way we have done so here.  
Those who endorse the superluminal causation paradigm do not shy away from the notion of causation, and hence a resource theory developed within their paradigm could be presented using the same framework that we use here --- that of causal models.  However, such an approach would likely be framed entirely in terms of {\em classical} causal models, rather than introducing the notion of GPT causal models.

Advocates of the superluminal causation paradigm would naturally define the free boxes to be those that involve only subluminal causes.  Hence, in scenarios wherein the inputs and the outputs at one wing are space-like separated from those at the other wings, so that subluminal causal influences cannot act between the wings, a box is free if and only if it
 can be realized by a classical common cause.  Thus, the natural choice of the free subtheory in the superluminal causation paradigm coincides with the free subtheory in the causal modelling paradigm.  On the other hand, the natural choice of the enveloping theory in the superluminal causation paradigm consists of the set of boxes that are classically realizable given superluminal causal influences between the wings. This  differs from the enveloping theory in the causal modelling paradigm because it includes boxes that are signalling.
In the superluminal causation paradigm, therefore, it is natural to try and quantify the resource 
in terms of the strength of the superluminal causal influence between the wings that is required to explain it in a classical causal model.\footnote{It should be noted that no \emph{finite} speed of superluminal causal influences can satisfactorily account for the predictions of quantum theory, per Ref.~\cite{Bancal2012}, so such influences would need to be assumed to be of infinite speed.}

Because the enveloping theory within this paradigm includes not only non-signalling boxes that violate Bell inequalities but signalling boxes as well,  the resource theory is rich enough to describe communication between the wings. Therefore, defining the resource theory in this way would not distinguish common-cause resources that are classically realizable from those that are not (as we propose to do here), but would instead draw a line between common-cause resources that are classically realizable and everything else
 --- including classical signalling resources.\footnote{
 That is, if one seeks to partition resources of a given type into classical and nonclassical varieties, then defining the enveloping theory correctly is just as important as defining the free subtheory correctly.}
 If one were to go this route, then all of classical Shannon theory would be subsumed in the resource theory.  A potential response to this expansion in the scope of the project might be to try to eliminate such signalling resources {\em by hand}, by demanding that the enveloping theory was constrained to those boxes that are non-signalling among the wings.  Such a response, however, seems to compromise the ideals of the superluminal causation paradigm, because no-signalling is an operational notion rather than a realist one.\footnote{John Bell famously argued against the idea that no-signalling could embody an assumption of locality in a fundamental physical theory on the grounds that it was too anthropocentric~\cite{bell1995nouvelle}:\begin{quote}
 ...the ``no signaling'' notion rests on concepts which are desperately vague, or vaguely applicable. The assertion that \enquote{we cannot signal faster than light} immediately provokes the question: Who do we think {\em we} are? {\em We} who can make \enquote{measurements}, {\em we} who can manipulate \enquote{external fields}, {\em we} who can \enquote{signal} at all, even if not faster than light? Do {\em we} include chemists, or only physicists, plants, or only animals, pocket calculators, or only mainframe computers?
 \end{quote}
}

\section{Details of the resource theory}\label{basicresthry}

\subsection{
Free and nonfree
 common-cause boxes} \label{envelopingthry}

We begin by formalizing the relevant definitions from Sec.~\ref{sec:freeandnonfree} and \ref{sec:quantifying}, and providing more details about the definition of the resource theory.
For ease of presentation, we focus throughout on the bipartite Bell scenario, but the multipartite Bell scenario can be formalized analogously.

Fig.~\ref{fig:CandNCboxes}(a) depicts the structure of a generic GPT-realizable common-cause box.  The classical variables that range over the (fixed) choices of local measurements are termed the \term{setting variables}, denoted $S$ (left wing) and $T$ (right wing), while the classical variables that range over the possible results of these measurements are termed the \term{outcome variables}, denoted $X$ (left wing) and $Y$ (right wing).   
In this article, we will refer to the cardinality of the set over which a variable $X$ varies as ``the cardinality of $X$'' and we will also sometimes refer to the cardinality of the setting (outcome) variable as simply ``the cardinality of the setting (outcome)''.

Let us label the system distributed to the left wing by $A$ and the one to the right wing by $B$. In the GPT framework,
states and effects on $A$ ($B$) are represented by vectors in a real vector space of dimension $d_A$ ($d_B$), that is, in $\mathbb{R}^{d_A}$ ($\mathbb{R}^{d_B}$). 
States and effects on the composite $AB$ are represented by vectors in the tensor product of these vector spaces\footnote{Strictly speaking, we do allow for a GPT where states and effects may correspond to vectors outside of the tensor product of the local vector spaces. That is, we do \emph{not} assume local tomography. However, since the causal structure of a Bell scenario is such that the measurements are assumed to be local, we can focus on the effects within  $\mathbb{R}^{d_A}\otimes \mathbb{R}^{d_B}$ without loss of generality, and therefore also on the states  within  $\mathbb{R}^{d_A}\otimes \mathbb{R}^{d_B}$. In other words, any so-called holistic degrees of freedom in the GPT play no role in Bell scenarios, and therefore we can ignore them without loss of generality.}, 
$\mathbb{R}^{d_A}\otimes \mathbb{R}^{d_B}$. If the GPT representation of the $X=x$ outcome of the $S=s$ measurement on system $A$ is ${\bf r}^{A}_{x|s} \in \mathbb{R}^{d_A}$ and that of the $Y=y$ outcome of the $T=t$ measurement on system $B$ is ${\bf r}^{B}_{y|t} \in \mathbb{R}^{d_B}$, and if ${\bf s}^{AB} \in \mathbb{R}^{d_A}\otimes \mathbb{R}^{d_B}$ denotes the GPT state of the composite $AB$, then the conditional probability distribution associated to this GPT-realizable common-cause box is
\beq
P_{XY|ST}(xy|st) = ({\bf r}^{A}_{x|s} \otimes {\bf r}^{B}_{y|t}) \cdot {\bf s}^{AB},
\label{GPTCCbox}
\eeq 
where $\cdot$ denotes the Euclidean inner product.

By virtue of their internal causal structure, all GPT-realizable common-cause boxes satisfy the no-signalling conditions ${P_{Y|ST} = P_{Y|T}}$ and ${P_{X|ST} = P_{X|S}}$.  It is straightforward to verify that this follows from Eq.~\eqref{GPTCCbox} using the fact that $\sum_x {\bf r}^{A}_{x|s} = {\bf u}^A$, where ${\bf u}^A$ is the unique deterministic effect on $A$, which is independent of value $s$ of the setting variable, and using the analogous fact for $B$.

The common-cause boxes that are considered to be {\em free} in our resource theory are those that can be realized when the GPT governing the internal workings of the box is classical probability theory, as depicted in Fig.~\ref{fig:CandNCboxesCLASSICAL}.

In such cases, the scope of possibilities for the overall functionality of the common-cause box can be characterized as follows.  The systems $A$ and $B$ are described by classical variables, $\Lambda_A$ and $\Lambda_B$ (here assumed to be discrete).  
Classically, the composite system $AB$ is prepared in a joint distribution over these, $P_{\Lambda_A \Lambda_B}$.
 The GPT state in this case is $[{\bf s}^{AB}]_{\lambda_A\lambda_B} = P_{\Lambda_A \Lambda_B}(\lambda_A, \lambda_B)$, where $[{\bf v}]_i$ denotes the $i$th component of a vector ${\bf v}$ living in vector space $\mathbb{R}^{|\Lambda_A|} \otimes \mathbb{R}^{|\Lambda_B|}$.
 Without loss of generality, we can take systems $A$ and $B$ to be perfectly correlated (by incorporating any noise into the measurements), corresponding to the case where $P_{\Lambda_A \Lambda_B}(\lambda_A, \lambda_B) =  \sum_{\lambda}  \delta_{\lambda_A,\lambda} \delta_{\lambda_B,\lambda} P_{\Lambda}(\lambda)$ 
for some distribution $P_{\Lambda}(\lambda)$, and where $\delta$ denotes the Kronecker-delta function.  This distribution over $\Lambda_A$ and $\Lambda_B$ can be conceptualized as follows: sample a variable $\Lambda$ from some distribution, then let $\Lambda_A$ and $\Lambda_B$ be copies of it.  

Classically, the ${X{=}x}$ outcome of the ${S{=}s}$ measurement on system $A$ is modelled by a conditional probability distribution $P_{X|S\Lambda_A}$.  The GPT effect associated to this measurement on $A$ is ${\bf r}^{A}_{x|s}$ with components $[{\bf r}^{A}_{x|s}]_{\lambda_A} = P_{X|S\Lambda_A}(x|s\lambda_A)$. 
 Similarly, the GPT effect associated to the measurement on $B$ is ${\bf r}_{y|t}^{B}$ and has components $[{\bf r}^{B}_{y|t}]_{\lambda_B} = P_{Y|T\Lambda_B}(y|t\lambda_B)$.  Substituting these expressions into Eq.~\eqref{GPTCCbox}, we conclude that a classically realizable common-cause box satisfies
\begin{align}
&P_{XY|ST}(xy|st)\nonumber\\
&= \smashoperator{\sum_{\lambda_A\lambda_B}}P_{ X|S\Lambda_A}(x|s\lambda_A)P_{Y|T\Lambda_B}(y|t\lambda_B) P_{\Lambda_A \Lambda_B}(\lambda_A\lambda_B)\nonumber\\
&=\smashoperator{\sum_{\lambda}}P_{ X|S\Lambda}(x|s\lambda)P_{Y|T\Lambda}(y|t\lambda) P_{\Lambda}(\lambda).
\label{ClassicalCCbox}
\end{align}
This is recognized to be the expression for a conditional probability distribution $P_{XY|ST}$ that satisfies the Bell inequalities.

\subsection{The free operations on common-cause boxes}
\label{sec:freeoperations}

 

 The most general free operation taking a bipartite common-cause box with settings $S, T$ and outcomes $X, Y$ to a bipartite common-cause box with settings $S', T'$ and outcomes $X', Y'$ is {the clamp} depicted in blue in Fig.~\ref{fig:transfexplicit}.
  It is the most general processing which makes use of a classical common cause that can act on the local pre-processings and the local post-processings at each of the wings.  It subsumes as special cases processings wherein classical common causes act on any of the subsets of these four local processings. 
     
 Note that the most general free operation allows arbitrary feed-forward of classical information on each wing, since this does not require any causal influences between the wings.\footnote{Because the only physical restriction we are imagining is that no cause-effect influences are present {\em between} wings, feed-forward of {\em nonclassical} information (that is, of arbitrary GPT systems) at each wing is also a free LOSR operation.  {\em Without loss of generality}, however, we consider only feed-forward of classical systems in this work, because this is already sufficient to generate {\em any} conditional probability distribution $P_{X' Y' S T|XY S'T'}$ consistent with the causal structure, i.e., satisfying Eqs.~(\ref{freeopNoSignal}-\ref{freeopNoRetro}). }
But any such operation can always also be put into the canonical form depicted in blue in Fig.~\ref{fig:transf}. It suffices to note that the system that mediates the action of the common cause on the post-processings on a given wing can  always be passed down the classical side-channel.  
Henceforth, we use this canonical form when describing the most general free operation.


Formally, such an operation
 transforms the conditional probability distribution $P_{XY|ST}$ to $P_{X'Y'|S'T'}$ as
\begin{equation}\label{actionfreeop}
P_{X'Y'|S'T'} = \sum_{XYST} P_{X' Y' S T|XY S'T'} P_{XY|ST}
\end{equation} 
where the conditional probability distribution $P_{X' Y' S T|XY S'T'}$ satisfies certain constraints, which we specify below. 

\begin{figure}[ht]
 \centering
  \includegraphics[width=225pt]{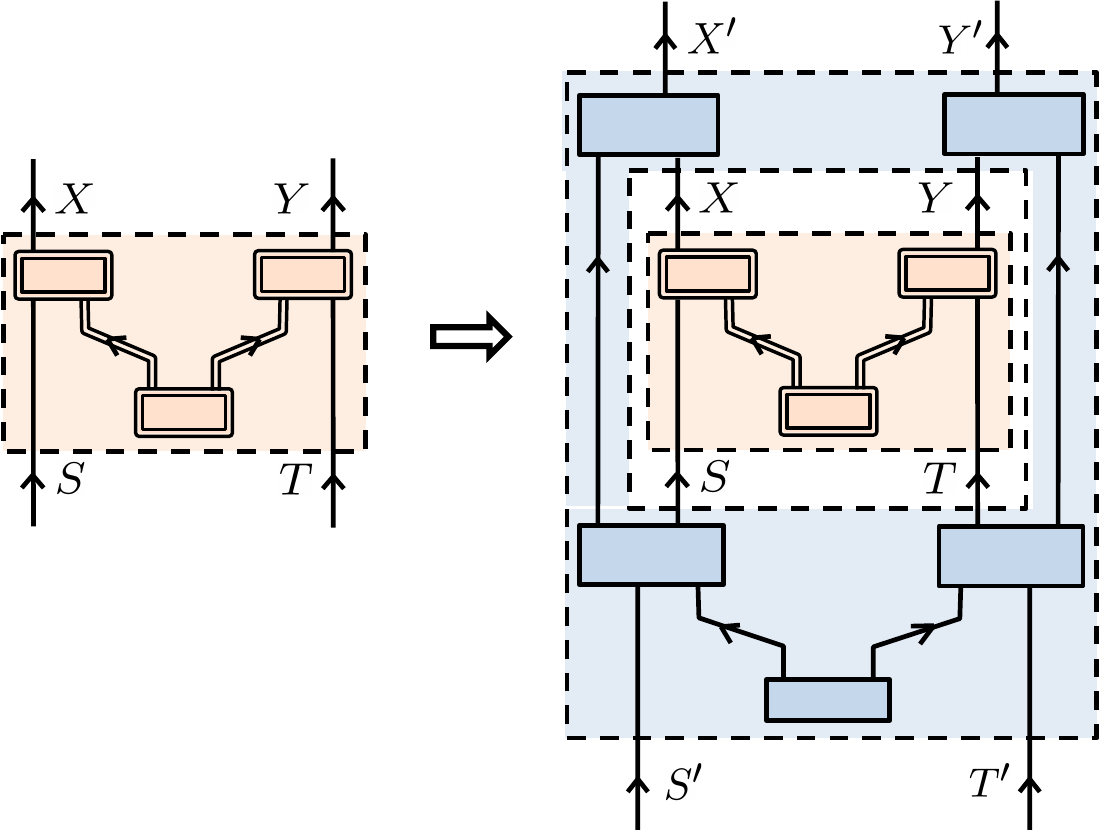}
 \caption{The canonical form of a generic bipartite {\LOSR} operation $P_{X' Y' S T|XY S'T'}$ (in blue) taking a common-cause box $P_{XY|ST}$ (in pink) to a common-cause box $P_{X'Y'|S'T'}$.  
 } \label{fig:transf}
\end{figure}

Circuit fragments that map processes to processes (such as the ones  depicted in blue in Figs.~\ref{fig:transfexplicit} and ~\ref{fig:transf})
have been studied extensively in recent years in a variety of frameworks, most notably the quantum combs framework of Refs.~\cite{qcombs08,qcombs09}, and the process matrix framework of Refs.~\cite{OCB12,Oreshkov2016}.  If the source and target resources are denoted by $R$ and $R'$, respectively, and the free operation is denoted by $\tau$, we represent Eq.~\eqref{actionfreeop} as 
\beq
R' = \tau \circ R,
\eeq 
where $\circ$ is a particular instance of the \term{link product} of Ref.~\cite{qcombs08}.

On the left wing, the most general local {\em pre-}processing takes as input the setting variable  of the target resource ($S'$) and the variable originating from the common cause, and it generates as output the setting variable  of the source resource ($S$) as well as an arbitrary variable which propagates down the side-channel. The most general {\em post-}processing on the left wing takes as input the outcome variable  of the source resource ($X$) and the side-channel variable, and it generates as output the outcome variable  of the target resource ($X'$).
Included as special cases among these pre- and post-processings are maps from $S'$ to $S$ and from $X$ to $X'$ that constitute relabellings, coarse-grainings, or fine-grainings of the variable, where the possibilities are constrained by the cardinalities of these variables.  Also included as special cases are instances where the map from $S'$ to $S$ or the map from $X'$ to $X$ is chosen probabilistically, and instances where these two maps are correlated (by making use of the side-channel). 
The analogous pre- and post-processings at the right wing are also possible.  Finally, the choices of maps on the left can also be correlated with the choices of maps on the right, by leveraging the common cause.

The free operations are characterized by those 
 $P_{X' Y' S T|XY S'T'}$ which can be achieved via the type of circuit fragment depicted in Fig.~\ref{fig:transf}, namely, those 
 such that
\begin{align} \label{LC4}
&P_{X' Y' S T|XY S'T'}(x'y'st|xys't')=
\\ \nonumber 
&\smashoperator{\sum_{\lambda_A\lambda_B}}
\begin{array}{r}
P_{X' S|X S' \Lambda_A }(x's|xs'\lambda_A) P_{Y'T|YT'\Lambda_B}(y't|yt'\lambda_B)
\\\times P_{\Lambda_A \Lambda_B}(\lambda_A\lambda_B)
\end{array}
\end{align}
for some joint distribution $P_{\Lambda_A \Lambda_B}$ and for some $P_{X' S|X S' \Lambda_A }$ and $P_{Y'T|YT'\Lambda_B}$ satisfying no-retrocausation conditions
\begin{align}\begin{split}\label{NRCs}
P_{S|X S' \Lambda_A }=P_{S|S'\Lambda_A}\\
P_{T|Y T'\Lambda_B}=P_{T|T'\Lambda_B}.
\end{split}\end{align}

One can directly check that any $P_{X' Y' S T|XY S'T'}$ admitting of a decomposition as in Eq.~\eqref{LC4} satisfies the {\em operational}
no-signalling constraints
\begin{align}\begin{split}
P_{X' S |XY S'T'} = P_{X'S|XS'}\\
P_{Y' T |XY S'T'} = P_{Y'T|YT'}\end{split}\label{freeopNoSignal}
\end{align}
and the {\em operational} no-retrocausation conditions
\begin{align}\begin{split}
P_{S|X S'}=P_{S|S'}\\
P_{T|Y T'}=P_{T|T'}.\end{split}\label{freeopNoRetro}
\end{align}

The parts of the circuit fragment in Fig.~\ref{fig:transf} that are associated to $P_{X' S|X S' \Lambda_A }$ and $P_{Y'T|YT'\Lambda_B}$ we refer to as 
\term{local operations}. 
The part associated to $P_{\Lambda_A \Lambda_B}$  corresponds to a joint distribution on the variables distributed to the two wings and can therefore be conceived of as \term{shared randomness}.  Consequently, the free operations we are endorsing here can indeed be described as
 local operations and shared randomness ({\em \LOSR}), as noted earlier. 
\begin{defn}\label{defnLOSR}
An operation is in the set \term{LOSR} (and termed an \term{LOSR operation}) if and only if it is associated to a conditional probability distribution $P_{X' Y' S T|XY S'T'}$ that admits of the sort of decomposition specified by Eqs.~\eqref{LC4} and \eqref{NRCs}.
\end{defn}

Previous resource-theoretic approaches to Bell-inequality violations have also endorsed the intuitive notion that local operations supplemented with shared randomness should constitute the free operations.  Different works, however, have made different proposals for how this notion ought to be formalized.  The correct formalization, in our opinion, is the one provided in Geller and Piani~\cite{GellerPiani} and independently in deVincente~\cite{de2014nonlocality},
which coincides with the one given above\footnote{The definition of \term{LOSR} given in Geller and Piani~\cite{GellerPiani} is very similar to the one provided here (see Fig. 4 therein), while the one provided in de Vicente\cite{de2014nonlocality} is much more cumbersome.}. Therefore, in this article we are endorsing the proposal of Refs.~\cite{de2014nonlocality,GellerPiani} to take \term{LOSR} as the free operations.  
On the other hand, Refs.~\cite{gallego2016nonlocality,Amaral2017NCW,kaur2018fundamental} have formalized the notion of local operations supplemented with shared randomness differently, defining a strict subset of the set \term{LOSR} defined above (a subset that can be shown to be nonconvex).  Nonetheless, we believe that this discrepancy was an oversight and that it is unlikely anyone would defend taking this subset rather the full set to define the resource theory.  We discuss the issue in depth in Appendix~\ref{comparisonsub}.  

As a final comment, note that, without loss of generality, we can take the joint distribution to be $P_{\Lambda_A \Lambda_B}(\lambda_A \lambda_B) =  \sum_{\lambda}  \delta_{\lambda_A,\lambda} \delta_{\lambda_B,\lambda} P_{\Lambda}(\lambda)$ 
for some distribution $P_{\Lambda}$,
and hence express Eq.~\eqref{LC4} as 
\begin{align} \label{LC4prime}
&P_{X' Y' S T|XY S'T'}(x' y' st|xy s't')=
\\\nonumber
&\sum_{\lambda}P_{X' S|X S' \Lambda }(x' s|x s' \lambda) P_{Y'T|YT'\Lambda}(y't|yt' \lambda) P_{\Lambda}(\lambda).
\end{align}
As a consequence, the conditional probability distribution $P_{X' Y' S T|XY S'T'}$ can be conceptualized as the more familiar object $P_{\tilde{X} \tilde{Y}|\tilde{S} \tilde{T}}$ for setting variables $\tilde{S},\tilde{T}$ and outcome variables $\tilde{X}, \tilde{Y}$ that are defined as follows.  We take the composite of the outputs of the circuit fragment on the left wing, $X'$ and $S$, as a composite outcome variable $\tilde{X}$, so that $\tilde{X}\coloneqq(X',S)$.  Similarly, we take the composite of the inputs on the left wing, $X$ and $S'$, as a composite setting variable $\tilde{S}$, so that $\tilde{S}\coloneqq(X,S')$.  Making the analogous definitions for $\tilde{Y}$ and $\tilde{T}$ in terms of $Y, T, Y', T'$ on  the right wing, Eq.~\eqref{LC4prime} can be rewritten as
\begin{align} \label{LC4primeprime}
&P_{\tilde{X}\tilde{Y}|\tilde{S}\tilde{T}}(\tilde{x}\tilde{y}|\tilde{s}\tilde{t})=\\
&\nonumber\qquad\sum_{\lambda}P_{ \tilde{X}|\tilde{S}\Lambda}(\tilde{x}|\tilde{s}\lambda)P_{\tilde{Y}|\tilde{T}\Lambda}(\tilde{y}|\tilde{t}\lambda) P_{\Lambda}(\lambda).
\end{align}
Recalling Eq.~\eqref{ClassicalCCbox}, it is clear that $P_{\tilde{X} \tilde{Y}|\tilde{S} \tilde{T}}$ satisfies all of the Bell inequalities.
This illustrates the consistency of  our proposal for the free operations, for we have just shown that the free operations  on a resource $P_{XY|ST}$ are  those that are achieved by taking a link product~\cite{qcombs08} with a process $P_{X' Y' S T|XY S'T'} \coloneqq P_{\tilde{X} \tilde{Y}|\tilde{S} \tilde{T}}$ which satisfies all of the Bell inequalities.



\subsubsection{Cardinality-based types for boxes and for operations} 

\begin{defn}
We define the \term{type of a common-cause box} as the collection of cardinalities of the setting and outcome variables,
and we denote the type of a resource $R$ as $[R]$.
We introduce the following notational convention to specify types: the cardinalities of the setting variables for all $n$ wings and the cardinalities of the outcome variables for all $n$ wings are specified as the bottom and top rows, respectively, of a $2{\times }n$ matrix.
  For example, for the 2-wing common-cause box depicted in Fig.~\ref{fig:CandNCboxes}, the type is 
{\xyst},  where $|O|$ denotes the cardinality of a variable $O$.  
\end{defn}
If we further particularize to  the CHSH scenario, where the cardinalities of both setting and outcome variables is 2, then the type is {\twotwotwotwo}.

\begin{defn}
Consider a source resource $R_1$ of type $[R_1]$ and a target resource $R_2$ of type $[R_2]$.  We denote the \term{type of an operation} $\tau$ taking any resource of type $[R_1]$ to any resource of type $[R_2]$ by ${[\tau] \coloneqq [R_1]\rightarrow [R_2]}$, and we denote the set of all free operations of type $[R_1]\rightarrow [R_2]$ by $\underset{\scriptscriptstyle ^{[R_1]\rightarrow [R_2]}}{\LOSR}$. 
\end{defn} 

Note that operations --- including \emph{free} operations --- can change the type of a resource, and hence specifying the type of an \emph{operation} requires specifying both the type of the initial resource as well as the type of the final resource. This reflects the fact that we have not restricted the cardinalities of $X',Y',S'$, or $T'$ in Eq.~\eqref{actionfreeop} in any way. 

\subsubsection{Locally deterministic operations and local symmetry operations} 

It is valuable to consider two special finite-cardinality subsets of {\LOSR} operations:  those that are deterministic and those that are invertible. Note that the invertible {\LOSR} operations are included among the deterministic ones because any indeterminism in the operation would be an obstacle to invertibility.

\begin{defn}
\label{defn:LDO}  An {\LOSR} operation is in the set \term{\LDO} (i.e., it is a \term{locally deterministic operation}) if and only if the conditional probabilities $P_{X' Y' S T|XY S'T'}$ which define the operation
take values in $\{0,1 \}$ for all values of $X'$,$Y'$,$S$,$T$,$X$,$Y$,$S'$ and $T'$.
 We denote the complete set of {\LDO} operations of type $[R_1]\rightarrow [R_2]$ by $\underset{\scriptscriptstyle ^{[R_1]\rightarrow [R_2]}}{\LDO}$.
\end{defn}

Deterministic {\LOSR} operations---i.e., {\LDO} operations---\emph{factorize} in the sense that every {\LDO} operation can be expressed as the product of two local deterministic operations such that
\begin{align} \label{ExtremalOpLOSR}
&P^{\rm det}_{X' Y' S T|XY S'T'}= P^{\rm det}_{X' S|X S'}P^{\rm det}_{Y' T|Y T'}.
\end{align}
This follows from the fact that the deterministic dependences preclude any dependence on the shared random variables $\lambda_A$ and $\lambda_B$ in Eq.~\eqref{LC4}, which then reduces to Eq.~\eqref{ExtremalOpLOSR}.
Furthermore, the no retrocausation assumption of Eq.~\eqref{freeopNoRetro} implies that these deterministic dependencies are of the following form:
\begin{align}\begin{split} \label{ExtremalOpLOSRb}
P^{\rm det}_{X' S|X S'} = \delta_{S,f_A(S')} \delta_{X',g_A(X,S')},\\
P^{\rm det}_{Y' T|Y T'} = \delta_{T,f_B(T')} \delta_{Y',g_B(Y,T')} 
\end{split}\end{align}
for some functions $f_A$, $g_A$, $f_B$ and $g_B$.
Specifically, on the left wing, $S$ is generated deterministically as a function of $S'$ (the pre-processing) and $X'$ is generated deterministically as a function of $X$ and $S'$ (the post-processing, which is setting-dependent), and similarly for the right wing.   A generic bipartite locally deterministic operation is depicted in Fig.~\ref{det}.

\begin{figure}[htb!]
 \centering
  \includegraphics[width=225pt]{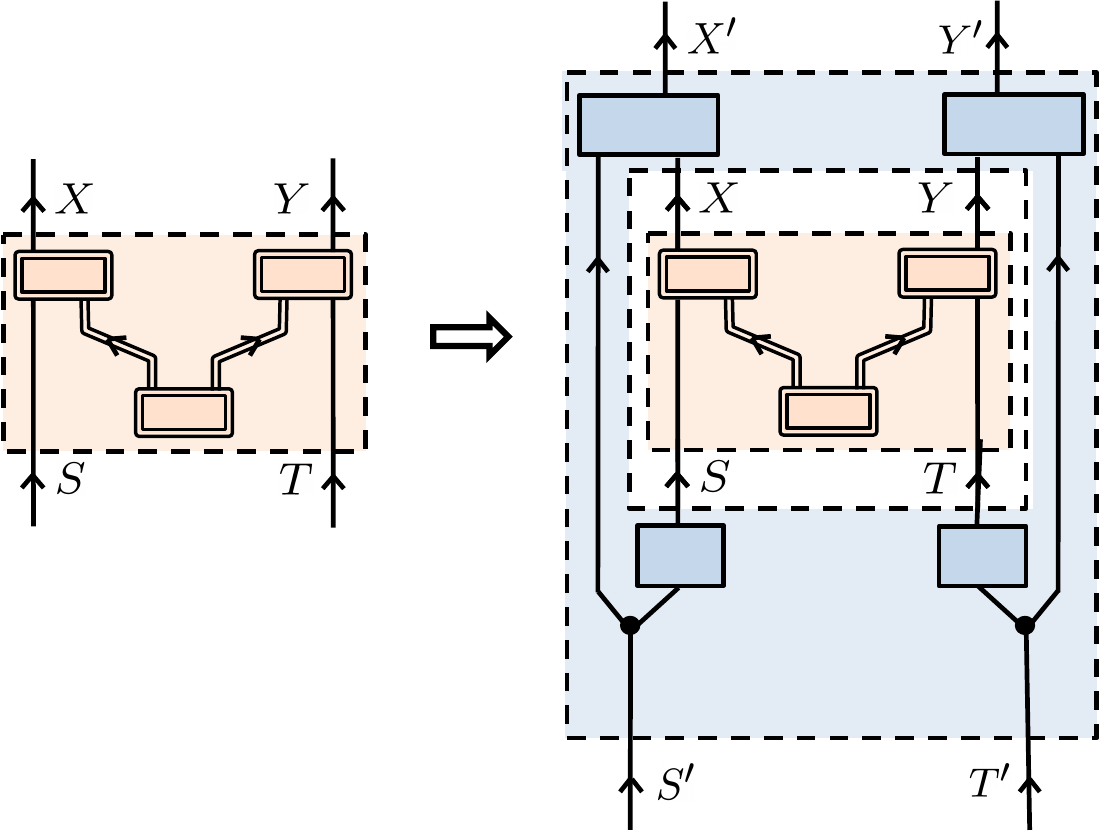}
 \caption{A generic bipartite locally deterministic operation $P_{X'Y'ST|XYS'T'} \in$ {\LDO} consists of a product of deterministic operations at each wing.  The black dots in the figure represent classical copy operations, and the output variables for each gate are deterministic functions of the input variables for that gate.
   }
 \label{det}
\end{figure}

The cardinality of the set {\LDO} for a given type can be easily deduced.  Let $|S|$, $|X|,\dots$ denote the cardinalities of the variables $S$, $X,\dots$ The total number of possibilities for the function $g_A$ is $|X'|^{|X| \cdot |S'|}$, and the total number of possibilities for the function $f_A$ is $|S|^{|S'|}$, so that the total number of possibilities for a deterministic operation on the left wing is  $\left(|S| \cdot |X'|^{|X|}\right)^{|S'|}$.   An analogous decomposition holds for the deterministic operations on the right wing, and the total number of possibilities for these is $\left(|T|\cdot |Y'|^{|Y|}\right)^{|T'|}$. 
Consequently, the cardinality of the set {\LDO} in this bipartite case is
\begin{align}\label{eq:LDOcount}
|\LDO| = \left(|S| \cdot |X'|^{|X|}\right)^{|S'|} \, \left(|T|\cdot |Y'|^{|Y|}\right)^{|T'|}.
\end{align}

The other important subset of {\LOSR} are those type-preserving operations which are  {\em invertible} (and hence also deterministic).
We refer to this subset of {\LOSR} operations as the \emph{local symmetry operations} and denote it {\LSO}.  
\begin{defn}
\label{defn:LSO} 
The set \term{\LSO} (i.e., the \term{local symmetry operations}) is the subset of type-preserving operations in {\LDO} that are invertible. 
\end{defn}

Every local symmetry operation, $P^{\rm sym}_{X' Y' S T|XY S'T'}$, has the form of a locally deterministic operation, $P^{\rm det}_{X' Y' S T|XY S'T'}$, specified in Eqs.~\eqref{ExtremalOpLOSR}-\eqref{ExtremalOpLOSRb}. That is, 
\begin{align} \label{SymOps}
&P^{\rm sym}_{X' Y' S T|XY S'T'}= P^{\rm sym}_{X' S|X S'}P^{\rm sym}_{Y' T|Y T'}.
\end{align}
where
\begin{align}\begin{split} \label{SumOpsb}
P^{\rm sym}_{X' S|X S'} = \delta_{S,f_A(S')} \delta_{X',g_A(X,S')},\\
P^{\rm sym}_{Y' T|Y T'} = \delta_{T,f_B(T')} \delta_{Y',g_B(Y,T')},
\end{split}\end{align}
but where $f_A$, $g_A$ are such that $P^{\rm sym}_{X' S|X S'}$ defines an invertible map from $(X, S')$ to $(X',S)$, and where $f_B$ and $g_B$ are such that $P^{\rm sym}_{Y' T|Y T'}$ defines an invertible map from $(Y,T')$ to $(Y', T)$. Unlike general {\LDO} operations, {\LSO} operations are always type-preserving, and hence the type {\xystprime} always matches the type {\xyst}.

Note that an exchange of the parties is a symmetry operation (i.e., invertible), but it cannot be implemented by {\em local} operations, and so it is not part of {\LSO}. 

As a final remark, notice that the set of LSO operations forms a \emph{group}. This follows from the fact that the properties of being deterministic and invertible persist under composition, and that the inverse of every LSO operation is in LSO.   This group 
is generated by the permutations of the value of a setting variable, and the permutations of the value of an outcome variable, where the choice of the latter permutation might depend also on the value of the setting variable on the same wing.

In the bipartite case, the ${\LSO}$ group is a finite group of order\footnote{The order of a group is the cardinality of the set of group elements, i.e., the order of the ${\LSO}$ group quantifies the total number of \emph{invertible} {\LDO} operations.}
\begin{align}
|{\LSO}|=(|S|!)\cdot (|X|!)^{|S|}\cdot (|T|!)\cdot (|Y|!)^{|T|},
\end{align}
corresponding to the $(|S|!)$ relabelings for the settings of the left wing, multiplied by the $(|X|!)$ relabelings of outcomes for each of the $|S|$ different settings, and similarly for the right wing. The group can be generated by the relabelings of only adjacent settings or outcomes, and hence the ${\LSO}$ group admits a natural representation in terms of ${(|S|{-}1)+|S|(|X|{-}1)+(|T|{-}1)+|T|(|Y|{-}1)}$ generators (see Ref.~\citep[App.~B]{Rosset2014classifying}).

For a concrete example, consider the operations transforming type {\twotwotwotwo} into type {\twotwotwotwo}. Throughout this work, we index the values a variable $X$ can take as $x\in \{0{,}{...},|X|-1\}.$ Accordingly, in the {\twotwotwotwo} scenario, $X,Y,S,T$ take values in $\{0,1\}$. Using this notation, the group of LSO can be generated explicitly by the four operations which interconvert $P_{XY|ST}(x,y|s,t)$ with either $P_{XY|ST}(x,y|s{\oplus}1,t)$, $P_{XY|ST}(x,y|s,t{\oplus}1)$, $P_{XY|ST}(x{\oplus}s,y|s,t)$, or $P_{XY|ST}(x,y{\oplus}t|s,t)$, where  $\oplus$ denotes summation modulo two.\footnote{A second generating set of operations for this group is given by $\tau_1{,}{...},\tau_6$ defined in Proposition~\ref{prop:symmetrypartitioning}.} One can readily verify~\cite{Seress2003} that the order of this group is 64.

\sloppy Suppose that a resource $R$ is represented as a real-valued vector $\vec{R}$ of conditional probabilities $P_{XY|ST}(xy|st)$, or any linear transformation thereof (such as the representation in terms of correlators used in Section~\ref{twomonotones}). ${\LSO}$ operations act as invertible linear maps on such a representation.
 Assuming $f$ is a linear function over $\vec{R}$, then its action can be represented as $ f(\vec{R}) = \vec{f} \cdot \vec{R}$ for some $\vec{f}$.
Hence, it is equally as meaningful to speak about $\vec{f}$ being transformed under ${\LSO}$ group elements as it is to speak about $\vec{R}$ being so transformed.  The action of an  ${\LSO}$ operation on $\vec{f}$ can be thought of as applying the \emph{inverse transformation} to $\vec{R}$, i.e.,
\begin{align}\label{resourcesvsfunctionals}
(\pi \vec{f}) \cdot \vec{R} = f \cdot (\pi^{\text{-}1} \vec{R}).
\end{align}
Note that many \emph{type-changing} ${\LOSR}$ operations are equally well-defined as transformations on linear functions. The critical requirement is that the operation be \emph{left-invertible}, i.e., it should act as an \emph{injective} function on the set of conditional probabilities. See Refs.~\cite{Pironio2005,Rosset2014classifying,Rosset2019CausalInequalities} for discussions on the topic of converting linear functions (and Bell inequalities in particular).

\subsection{Convexity of the set of free operations} \label{detsuf}

We now show that the set of free operations is convex, and that the extremal elements are deterministic, and enumerable for fixed type of the source resource and of the target resource.  This implies that the set of free operations mapping from a given source resource type to a given target resource type is a polytope.

We begin by proving convexity.
\begin{prop} \label{lem:convexity}
The set {\LOSR} is convex, i.e., if $\tau_0 \in \LOSR$ and $\tau_1 \in \LOSR$, then ${{w\tau_0+(1-w)\tau_1} \in \LOSR}$ for $0 \le w \le 1$.
\end{prop}
This follows from the fact that the resources required to achieve such a mixing are achievable using {\LOSR}.  Suppose $\beta$ is a binary variable that decides whether $\tau_0$ or $\tau_1$ will be implemented.  It suffices to imagine that $\beta$ is sampled from a distribution $P_{\beta}$ where $P_{\beta}(0)=w$, that it is copied and distributed to both wings (with a copy sent down the side-channel at each wing), and that the local processings that are implemented on each wing are made to depend on $\beta$ (chosen so that if $\beta=b$, then $\tau_b$ is implemented overall).  Because $\beta$ can be incorporated into the definition of the shared randomness, the procedure just described is itself achievable using {\LOSR}.

The convexity of the set of {\LOSR} operations is crucial for the technique we develop in the next section to answer questions about resource conversion. Recognizing the full potential of this convexity is one of the key contributions of our work. In Appendix~\ref{comparisonsub}, we discuss convexity further, in particular noting that previous formulations of {\LOSR} did not seem to recognize the physical realizability of convex mixing within {\LOSR}, but rather imposed convexity mathematically.

Next, we highlight features of the extremal free operations.
\begin{prop} \label{lem:detextremal}
The set of convexly extremal operations in {\LOSR} are precisely the subset of operations comprising {\LDO}, namely the deterministic {\LOSR} operations.
\end{prop}

This proposition is a minor generalization of Fine's argument~\cite{FinePRL}, since the latter states that {\em locally deterministic} models
 can generate any conditional distribution that arises in a 
  {\em locally indeterministic} model.
As in Fine's argument, here too any indeterminism in the local operations can be absorbed into the shared randomness, and hence allowing indeterministic local operations provides no more generality than considering only deterministic local operations.
\begin{proof}
It suffices to run Fine's argument~\cite{FinePRL} for the composite variables $\tilde{S}, \tilde{T}, \tilde{X}$ and $\tilde{Y}$.
To see this explicitly, note that the constituent factors in 
the expression for an LOSR operation in Eq.~\eqref{LC4primeprime} can be rewritten as 
\begin{align*}
&P_{\tilde{X}|\tilde{S} \Lambda}(\tilde{x}|\tilde{s} \lambda)=\sum_{\lambda_A\in\Lambda_A} P^{\rm det,\lambda_A}_{\tilde{X}|\tilde{S}}(\tilde{x}|\tilde{s}) P_{\Lambda_A|\Lambda}(\lambda_A|\lambda),\\
&P_{\tilde{Y}|\tilde{T} \Lambda}(\tilde{y}|\tilde{t} \lambda)=\sum_{\lambda_B\in\Lambda_B} P^{\rm det,\lambda_B}_{\tilde{Y}|\tilde{T}}(\tilde{y}|\tilde{t}) P_{\Lambda_B|\Lambda}(\lambda_B|\lambda),
\end{align*}
where for each value of $\lambda_A$, the conditional $P^{\rm det,\lambda_A}_{\tilde{X}|\tilde{S}}$ describes a deterministic operation on the left wing specifying the value of $\tilde{X} = (X',S)$ for every value of $\tilde{S} = (X,S')$, and similarly for {$P^{\rm det,\lambda_B}_{\tilde{Y}|\tilde{T}}$} on the right wing.  
Plugging these back into Eq.~\eqref{LC4primeprime}, we have that 
\begin{align} \label{LC4det}
&P_{\tilde{X}\tilde{Y}|\tilde{S}\tilde{T}}(\tilde{x}\tilde{y}|\tilde{s}\tilde{t})
=
\\\nonumber& \sum_{\lambda_A,\lambda_B}
P^{\rm det,\lambda_A}_{ \tilde{X}|\tilde{S}}( \tilde{x}|\tilde{s}) P^{\rm det,\lambda_B}_{\tilde{Y}|\tilde{T}}(\tilde{y}|\tilde{t})
P_{\Lambda_A\Lambda_B}(\lambda_A\lambda_B),
\end{align}
where we have defined 
$P_{\Lambda_A\Lambda_B}(\lambda_A\lambda_B)\coloneqq \sum_{\lambda\in\Lambda}P_{\Lambda_A|\Lambda}(\lambda_A|\lambda)P_{\Lambda_B|\Lambda}(\lambda_B|\lambda)P_{\Lambda}(\lambda).$ Eq.~\eqref{LC4det} shows that a generic indeterministic {\LOSR} operation can always be decomposed into a convex combination of products of deterministic operations on each wing.
Not only is it the case that the convexly extremal {\LOSR} operations are \emph{included} within the {\LDO} operations, but there is actually precise equality between these two sets: all {\LDO} operations are convexly extremal because every {\LDO} operation is a deterministic map.
\end{proof}

What we have shown above is that any element of $\underset{\scriptscriptstyle ^{[R_1]\rightarrow [R_2]}}{\LOSR}$ admits of a convex decomposition into elements of  $\underset{\scriptscriptstyle ^{[R_1]\rightarrow [R_2]}}{\LDO}$. 
This implies the following useful geometric fact:

\begin{prop}[Polytope of free operations]\label{lem:transpolytope}\hspace{0pt}\\
The set of all free operations of a given type is a polytope whose vertices are the locally deterministic operations of that type, 
\begin{align}\label{eq:convexityoftransformations}
\underset{\scriptscriptstyle ^{[R_1]\rightarrow [R_2]}}{\LOSR} = \operatorname{ConvexHull}{\left(\underset{\scriptscriptstyle ^{[R_1]\rightarrow [R_2]}}{\LDO}\right)}.
\end{align}
\end{prop}
The number of vertices of this polytope corresponds to the cardinality of the set of {\LDO} operations, as given in Eq.~\eqref{eq:LDOcount}.

\section{Resource theory preliminaries} \label{rtprelim}

A central question in any resource theory is whether one resource can be converted to another via the free operations. Many notions of conversion are studied: single-copy deterministic conversion, single-copy indeterministic conversion (where the probability of success need only be nonzero), multi-copy conversion (where one is given more than one copy of the resource), asymptotic conversion (where one is given arbitrarily many copies), and catalytic conversion (where one has access to another resource that must be returned intact after the conversion).  We here focus on single-copy deterministic conversion. 

As noted earlier, we denote the application of an operation $\tau$ to a resource $R$ by $\tau \circ R$.
 If $R_1$ can be converted to $R_2$ by free operations, one writes $R_1 \conv R_2$, otherwise one writes $R_1 \nconv R_2$.
Explicitly,
\begin{align*}
& R_1\conv R_2 \quad\text{denotes\; that}\quad\exists\; {\tau\in\underset{\scriptscriptstyle ^{[R_1]\conv [R_2]}}{\LOSR}}
\\&\quad\text{such that}\quad R_2 = \tau \circ R_1, \vphantom{\tau\in\underset{\scriptscriptstyle ^{[R_1]\conv [R_2]}}{\LOSR}}\\
\text{and}\quad
&R_1\nconv R_2 \quad\text{denotes\; that}\quad\nexists\; {\tau\in\underset{\scriptscriptstyle ^{[R_1]\conv [R_2]}}{\LOSR}}
\\&\quad\text{such that}\quad R_2 = \tau \circ R_1.
\end{align*}

If one can determine, for any pair of resources $R_1$ and $R_2$, whether $R_1$ can be converted to $R_2$ using a free operation, then one can determine the \term{pre-order} over all resources that is induced by the conversion relation. A pre-order, by definition, is a transitive and reflexive binary relation between resources. The conversion relation is reflexive because the identity operation is free and maps a resource to itself, while it is transitive because if $R_1\conv R_2$ and $R_2 \conv R_3$ then $R_1 \conv R_3$.

There are four possible ordering relations that might hold between a pair of resources.\begin{flalign*}
&R_1\text{ is \term{strictly above} }R_2\text{ if:}&
\!\!\left(\!\!\begin{array}{r}R_1\conv R_2 \\ \!\text{and }R_2\nconv R_1\end{array}\!\!\right)
\!,\\
&R_1\text{ is \term{strictly below} }R_2\text{ if:}&
\!\!\left(\!\!\begin{array}{r}R_1\nconv R_2 \\ \!\text{and }R_2\conv R_1\end{array}\!\!\right)
\!,\\
&R_1\text{ is \term{incomparable} to }R_2\text{ if:}&
\!\!\left(\!\!\begin{array}{r}R_1\nconv R_2 \\ \!\text{and }R_2\nconv R_1\end{array}\!\!\right)
\!,\\
&R_1\text{ is \term{equivalent} to }R_2\text{ if:}&
\!\!\left(\!\!\begin{array}{r}R_1\conv R_2 \\ \!\text{and }R_2\conv R_1\end{array}\!\!\right)
\!.
\end{flalign*}
If $R_1$ is either strictly above or strictly below $R_2$, we say that $R_1$ and $R_2$ are \term{strictly ordered}.

We pause to comment on the notion of equivalence of resources.  By definition, if $R_1$ is equivalent to $R_2$ then the conversion from one to the other is free in both directions,
\begin{align*}
\begin{array}{r}\exists\; {\tau_1\in\underset{\scriptscriptstyle ^{[R_1]\conv [R_2]}}{\LOSR}}\quad\text{such that}\quad R_2 = \tau_1 \circ R_1, \\
\text{and }\;\exists\; {\tau_2\in\underset{\scriptscriptstyle ^{[R_2]\conv [R_1]}}{\LOSR}}\quad\text{such that}\quad R_1 = \tau_2 \circ R_2.\end{array}
\end{align*}
It need not be the case, however, that either of the free operations $\tau_1$ or $\tau_2$ is invertible, nor that one is the inverse of the other.  For instance,  if $R_1$ and $R_2$ are both free resources, then $\tau_1$ can be the operation which discards $R_2$ and prepares $R_1$, while $\tau_2$ can be the operation which discards $R_1$ and prepares $R_2$. 

The conversion relation between resources implies a corresponding conversion relation between \term{equivalence classes} of resources (relative to the equivalence relation defined above), 
 wherein for any two equivalence classes, they are either strictly ordered or incomparable.  The conversion relation between equivalence classes is therefore antisymmetric and describes a \term{partial order} relation rather than a pre-order relation.   One can therefore conceptualize the project of characterizing the pre-order as a characterization of the equivalence classes and of the partial order that holds among these.  In this work, we do not provide a characterization of the equivalence classes, and so our focus will be on directly characterizing features of the pre-order of resources.
 
 \subsection{Global features of a pre-order}
\label{sec:oraclelimitations}

To have a complete understanding of deterministic single-copy conversion in a resource theory, one must have an understanding of the pre-order that this conversion relation defines.   In this section, we describe some of the basic features that characterize pre-orders. 

Perhaps the most basic question about a pre-order of resources is whether or not it is \term{totally pre-ordered}, meaning that every pair of elements in the pre-order is strictly ordered or equivalent (i.e., the pre-order has no incomparable elements).
Equivalently, we say that a pre-order is totally pre-ordered if and only if the partial order over equivalence classes that it defines is totally ordered (i.e., has no incomparable elements).

If there do exist incomparable resources, one can ask if the binary relation of incomparability is transitive, in which case the pre-order
 is termed \term{weak}.  

A \term{chain} is a subset of the pre-order in which every pair of elements is strictly ordered.
The \term{height} of a pre-order is the cardinality of the largest chain contained therein.
An \term{antichain} is a subset of the pre-order in which every pair of elements is incomparable. 
The \term{width} of a pre-order is the cardinality of the largest antichain contained therein. 

Other important properties of the pre-order refer to the \term{interval} between a pair of resources, where $R$ is in the interval of $R_1$ and $R_2$ if and only if both $R_1\conv R$ and $R\conv R_2$. 
If the number of equivalence classes which lie in the interval between a pair of resources is \emph{finite} for \emph{every} pair of inequivalent resources, then the pre-order is said to be \term{locally finite}, otherwise it is said to be \term{locally infinite}. 

\subsection{Features of resource monotones}\label{sec:intotomonotones}
 A resource monotone is a real-valued function\footnote{Technically, it is an extended-real-valued function, where the set of extended real numbers is obtained by adding $-\infty$ and $+\infty$ to the set of real numbers.} over resources whose value cannot increase under any free operation in the resource theory.  Formally,
\begin{samepage}\begin{defn}\label{def:monotone} A function $M$ from resources to the  reals  is called a \term{resource monotone} if and only if
\begin{subequations}\begin{align}
 &R_1 \conv R_2\quad \;\textrm{implies}\;\quad M(R_1) \ge M(R_2),\\
\shortintertext{or equivalently,}
 &M(R_1) < M(R_2)\quad\;\textrm{implies}\;\quad R_1 \nconv R_2.
\end{align}\end{subequations}
\end{defn}\end{samepage}
\noindent In other words, a resource monotone is an order-preserving map from the pre-order of resources to the total order of  real numbers. Whenever some monotone $M$ and a pair of resources $R_1$ and $R_2$ satisfies $M(R_1) < M(R_2)$, we will say that the monotone $M$ {\em witnesses} the fact that $R_1 \nconv R_2$.

If the pre-order is not totally pre-ordered (i.e., if there exist incomparable resources), then no single monotone can completely characterize the pre-order.  A complete characterization may be achieved, however, by a family of monotones.
Specifically, a family of monotones $\{ M_i \}_i$ is said to be \term{complete} if it completely characterizes the pre-order, that is, if 
\begin{align} \label{completeset}
\begin{split}&\forall R_1,R_2: {R_1 \conv R_2}\;\;\\
&\quad\textrm{if and only if}\;\;\,\forall i: M_i(R_1) \ge M_i(R_2).
\end{split}\end{align}
A complete set of monotones is therefore an alternative way of describing the pre-order.  

Strictly speaking, monotones should be functions from resources {\em of any type in the resource theory} to the reals. However, many natural functions
are only defined for particular types of resources. For instance, the function 
$P_{XY|ST}(00|00)P_{XY|ST}(11|01)+P_{XY|ST}(20|02)$
is only defined for common-cause boxes where the cardinalities of $X$ and $T$ are three. 
To accommodate this, we define the notion of a \term{monotone relative to a set $\boldsymbol{S}$}: $M$ is a monotone relative to a set $\boldsymbol{S}$ of resources if and only if for all $\{R_1,R_2\}\in \boldsymbol{S}$, $R_1\conv R_2$ implies $M(R_1)\geq M(R_2)$. 
A family of monotones $\{M_i\}_i$ is said to be \term{complete relative to a set} $\boldsymbol{S}$ if it holds that
\begin{align}\begin{split} \label{completewrtsubset}
&\forall R_1,R_2\in \boldsymbol{S}: {R_1 \conv R_2}\;\;
\\
&\quad\textrm{if and only if}\;\;\,\forall i: M_i(R_1) \ge M_i(R_2).
\end{split}\end{align}
If $\boldsymbol{S}$ is any set of resources all of which are of a particular type, a monotone relative to $\boldsymbol{S}$ is said to be \term{type-specific}.
 
\subsection{Monotone constructions for any resource theory}\label{sec:genericmonotones}

Here we review a variety of approaches to constructing resource monotones. We will make use of these versatile constructions to define an especially useful pair of monotones for the resource theory of common-cause boxes in Section~\ref{twomonotones}.

\subsubsection{Cost and yield monotones} \label{costandyield}

It is possible to upgrade a type-specific monotone to a type-independent monotone using either a \term{cost construction} or a \term{yield construction}. In fact, a cost or yield construction takes \emph{any} function (monotone or not) together with a set of resources and induces a type-independent monotone from it, as follows.

Given any function $f$ which maps some set  $\boldsymbol{S}$ of resources 
 to the real numbers, one can define associated monotones 
which are applicable to all resources, 
as follows: 
\begin{align} \label{yieldprescr}
&M[f\textup{-yield},\boldsymbol{S}](R)\coloneqq 
\\&\quad  \max\limits_{R^{\star}\in \boldsymbol{S}}
\left\{ f(R^{\star})  \;\text{ s.t. }\; R\conv R^{\star} \right\} ,\nonumber \\ \label{costprescr}
&M[f\textup{-cost},\boldsymbol{S}](R)\coloneqq 
\\&\quad\min\limits_{R^{\star}\in  \boldsymbol{S}}
\left\{ f(R^{\star})  \;\text{ s.t. }\; R^{\star}\conv R \right\}.\nonumber
\end{align}
If there does not exist any $R^{\star} \in \boldsymbol{S}$ such that $R\conv R^{\star}$, then the yield is defined to be $-\infty$. Similarly, if there does not exist any $R^{\star} \in \boldsymbol{S}$ such that $R^{\star}\conv R$, then the cost is defined as $\infty$~\cite{Gonda2019Monotones}.

In words, $M[f\textup{-yield},\boldsymbol{S}]$ is a monotone which asks for the most valuable resource
in the set $\boldsymbol{S}$   (as measured by the function $f$)  that one can create from the given resource $R$.\footnote{The maximum of a function $f$ over the set of boxes to which $R$ can be converted can also be thought of as the performance of $R$ over the so-called `nonlocal game' defined by the `payoff function' $f$. Since the set of boxes to which $R$ can be converted (of any given type) is a polytope, it follows that all 
{\em forbidden conversions (those from $R$ to a resource outside the polytope)}
can be witnessed by a suitable set of payoff functions, namely, whatever linear functions pick out the facets of $R$'s polytope (for any given target type). In other words, any resource outside the polytope will attain a higher value on at least one of these functions. It follows, then, that the set of yield monotones induced by \emph{all possible linear functions} constitutes a complete set of monotones. While this observation may not be useful in practice, it does pose an interesting contrast with the findings of Ref.~\cite{Buscemi2012LOSR}: For common-cause {\em boxes}, we find that `nonlocal games' constitute a complete set of monotones; whereas \cite{Buscemi2012LOSR} shows that for the resource theory of \emph{quantum states} under {\LOSR} it is semiquantum games instead of nonlocal games that form a complete set of monotones.
}
Meanwhile, $M[f\textup{-cost},\boldsymbol{S}](R)$ is a monotone which asks for the least valuable resource in the set $\boldsymbol{S}$
 (as measured by the function $f$)   that one can use to create the given resource $R$.
Note that in both cases, many different functions may yield the same monotone,
so there is a conventional element to one's choice of function. 
Note also that $\boldsymbol{S}$ may be restricted to resources of a particular type (in which case $f$ need only be defined on resources of that type), and yet the type of the resource $R$ for which the monotones may be evaluated is unrestricted. 

\subsubsection{Weight and robustness monotones}\label{sec:othermonotones}

Various functions have been used as measures of the distance of a resource from the set of classically realizable common-cause boxes in previous work~\cite{de2014nonlocality,GellerPiani,beigi2015monotone,
Beirhorst2016Bell,Cavalcanti2016Measures,Brito2018tracedistance,geometry2018}. In what follows, we highlight some of these which are monotones in our resource theory.

The {\em nonlocal fraction}, which we denote here by $M_{\textup{NF}}$, is the minimum weight of the nonfree fraction in any convex decomposition of the resource,
\begin{align}
&M_{\textup{NF}}(R)\coloneqq 
\\&\quad\min_{\substack{0{\leq} \lambda{\leq 1}\\R_{*}{\in}\boldsymbol{S}_{[R]}\\ L{\in}\boldsymbol{L}_{[R]}}} \left\{\lambda
\,\;\text{ s.t. }\;
R=\lambda\, R_{*}+(1{-}\lambda)L\right\}.
\nonumber\end{align}
The nonlocal fraction was proven to be a resource monotone relative to (a superset of) the {\LOSR} free operations in Ref.~\citep[Sec.~5.2]{de2014nonlocality}, though it is there termed the `EPR2' measure.

Next, there is the case of robustness measures\footnote{Note that in Ref.~\cite{geometry2018} these were termed `visibilities'.} which quantify the minimum weight of a resource from some particular class that must be added convexly with the original resource for the mixture to be free. The two robustness measures that we consider differ by the class of resources that are mixed with the original resource.  
The first, which we denote by $M_{\textup{RBST},\boldsymbol{L}}(R)$,
considers mixing the original resource $R$ with any element in the set $\boldsymbol{L}_{[R]}$ of \emph{free} resources of the same type: 
\begin{align}
&M_{\textup{RBST},\boldsymbol{L}}(R)
\coloneqq 
\\\nonumber&\quad\min_{\substack{0{\leq} \lambda{\leq} 1\\L\in\boldsymbol{L}_{[R]}}} \left\{\lambda
\,\;\text{ s.t. }\;
\lambda\, L+(1{-}\lambda)R \:\in \boldsymbol{L}_{[R]}\right\}.
\end{align}
This robustness measure was shown to be a resource monotone relative to  {\LOSR} in Ref.~\citep[Sec.~3]{GellerPiani}.

The second robustness measure, which we denote simply by $M_{\textup{RBST}}(R)$ considers mixing the original resource $R$ with 
any element in the set $\boldsymbol{S}_{[R]}$ of all resources of the same type:
\begin{align}
&M_{\textup{RBST}}(R) \coloneqq 
\\\nonumber&\quad\min_{\substack{0{\leq} \lambda{\leq} 1\\R_{*}{\in}\boldsymbol{S}_{[R]}}} \left\{\lambda
\,\;\text{ s.t. }\;
\lambda\, R_{*}+(1{-}\lambda)R \:\in \boldsymbol{L}_{[R]}\right\}.
\end{align}

The unified resource theory formalism of Ref.~\cite{Gonda2019Monotones} implies that all three of these distance measures are resource monotones in any resource theory wherein all of the operations in the free set are convex-linear\footnote{An operation $\tau$ is convex-linear if the image $\tau\circ(R_3)$ is a given mixture of $\tau\circ(R_1)$ and $\tau\circ(R_2)$ whenever the preimage $R_3$ is the same mixture of $R_1$ and $R_2$. All linear operations are convex-linear.} operations, including our resource theory here. Additionally, in Corollary~\ref{measurecor}, we show that each of these three distance measures can be explicitly related to a monotone for which we provide a closed-form expression relative to {\twotwotwotwo}-type resources. By extension, we therefore also provide closed-form expressions for these three distance measures relative to {\twotwotwotwo}-type resources.  

\section{A linear program for determining the ordering of any pair of resources}
\label{polything}

Next, we provide a linear program which allows one to determine the ordering relation that holds between any two resources in our enveloping theory. 
To do so, it is convenient to set up some useful notation.

\begin{defn}
Let the bold symbol \term{$\boldsymbol{S}$} refer to any set of resources. We use subscripts to specify the type of the resources in the set, 
such as $\boldsymbol{S}_{\xyst}$ or $\boldsymbol{S}_{[R]}$. We use superscripts to specify further properties of a set. For example, the set of all GPT-realizable common-cause boxes is denoted by $\boldsymbol{S}^G$, 
the set of all nonfree resources is denoted by $\boldsymbol{S}^{\rm nonfree}$, and the set of all free resources is denoted by $\boldsymbol{S}^{\rm free}$.
Whenever we wish to emphasize that a specific set is discrete, we denote it \term{$\mathbfcal{V}$}, and whenever we wish to emphasize that a specific set is a polytope, we denote it \term{$\mathbfcal{P}$}.
\end{defn}

Let $\mathbfcal{P}^{\LOSR}_{[R_2]}(R_1)$ denote the {\em continuous} set of resources of type $[R_2]$ into which $R_1$ can be converted under {\LOSR}, that is, the image of $R_1$ under $\underset{\scriptscriptstyle ^{[R_1]\rightarrow [R_2]}}{\LOSR}$.    Similarly, let $\mathbfcal{V}^{\LDO}_{[R_2]}(R_1)$ denote the {\em discrete} set of resources of type $[R_2]$ into which $R_1$ can be converted under {\LDO}, that is, the image of $R_1$ under $\underset{\scriptscriptstyle ^{[R_1]\rightarrow [R_2]}}{\LDO}$.
From Propositions~\ref{lem:convexity}~and~\ref{lem:transpolytope}, and the finite cardinality of $\mathbfcal{V}^{\LDO}_{[R_2]}(R_1)$, it follows that $\mathbfcal{P}^{\LOSR}_{[R_2]}(R_1)$ is a convex set with a finite number of vertices, and hence is a polytope:
\begin{prop}[The polytope of resources obtainable from a given resource by LOSR]\label{lem:belstheorem}\hspace{0pt}\\
The set of all resources of type $[R_2]$ obtainable from $R_1$
 by {\LOSR} forms a polytope, 
\begin{align}
\mathbfcal{P}^{\LOSR}_{[R_2]}(R_1) = {\operatorname{ConvexHull}}{\left(\mathbfcal{V}^{\LDO}_{[R_2]}(R_1)\right)}.
\end{align}
\end{prop}

We can express the content of Proposition~\ref{lem:belstheorem} equivalently as 
\begin{align}\begin{split}
&R_1 \conv R_2 \quad\text{if and only if}\quad 
\\&R_2\in{\operatorname{ConvexHull}}{\left(\mathbfcal{V}^{\LDO}_{[R_2]}(R_1)\right)}.
\end{split}\end{align}
Therefore, to determine whether $R_1$ is higher than $R_2$ in the pre-order of resources, it suffices to implement the following computational test: 
\begin{compactenum}
\item Enumerate all of the  locally deterministic operations which take resources of type $[R_1]$ to type $[R_2]$. (They are finite in number.)
\item Compute the images of $R_1$ under all of these locally deterministic operations.
\item Determine whether or not $R_2$ can be expressed as a convex combination of these images. (This is a linear program.)
\end{compactenum}

To determine which of the four possible ordering relations holds for a given pair of resources, $R_1$ and $R_2$, it suffices to determine whether $R_1 \conv R_2$ or not and whether $R_2\conv R_1$ or not.   This requires just two instances of the linear program.\footnote{In the language of Ref.~\cite{girard2015witnesses}, these linear programs constitute a {\em complete witness} for conversion.
}

According to Proposition~\ref{lem:belstheorem},  the image of a resource under the set of \emph{all} LOSR free operations is equivalent to the convex closure of the image of the resource under only the \emph{extremal} operations. Replacing the set of all operations with only the extremal ones is a dramatic shortcut. 

In principle, the linear program just described allows one to characterize the pre-order completely.  
For instance, this linear program defines a complete set of monotones for a given set of resources $\mathbf{S}$,  
namely, $\{ M_{R'} : R' \in \mathbf{S}\}$ where the monotone $M_{R'}$ is defined as follows: for all $R\in \mathbf{S}$, $M_{R'}(R)=1$ if $R\to R'$ by {\LOSR} and $M_{R'}(R)=0$ otherwise.  $M_{R'}(R)$ reports the answer returned by the linear program for the question of whether $R\to R'$ by {\LOSR}, and if one has the answer for all $R'\in \mathbf{S}$, then one has located $R$ within the pre-order.
However, such a brute-force characterization of the pre-order requires one to apply the linear program to {\em every} pair of resources, which is not possible in practice. Rather, the linear program is primarily useful for answering questions about conversions among pairs (or finite sets) of resources.

To characterize the full pre-order more generally, one would ideally have a finite set of resource monotones that characterize the pre-order completely.  Furthermore, in order to determine certain global properties of the pre-order, such as those described earlier, knowledge of a few carefully chosen resource monotones will typically suffice.  This is the strategy we will adopt hereafter in the article.
Specifically, over the next few sections, we define a pair of resource monotones and use these to prove that the pre-order of single-copy deterministic conversion is not totally pre-ordered
 (i.e., there exist incomparable resources), that it is not weak (the incomparability relation is not transitive), that it has both infinite width and infinite height, and that it is locally infinite.

\section{Two useful monotones}\label{twomonotones}

We will define two monotones, one a cost construction and the other a yield construction, where the sets of resources relative to which these costs and yields are evaluated (to be described below) contain only resources of type {\twotwotwotwo}.
It is useful to first review some facts about the set of all common-cause boxes of type {\twotwotwotwo}, that is, about $\boldsymbol{S}^G_{\twotwotwotwo}$.  

\subsection{Preliminary facts regarding CHSH inequalities and PR boxes}
We adopt the convention of Ref.~\cite{brunner2013Bell} of parametrizing common-cause boxes of type-{\twotwotwotwo} in terms of outcome biases and two-point correlators. The outcome biases are
\begin{subequations}\begin{align*}
\begin{split}\expec{A_s}&:= \sum_{x\in \{0,1\}} (-1)^x P_{X|S}(x|s) 
\\&= P_{X|S}(0|s) - P_{X|S}(1|s)\end{split}\\
\begin{split}
\text{and}\quad\expec{B_t}&:=   \sum_{y\in \{0,1\}} (-1)^y P_{Y|T}(y|t) 
\\&=  P_{Y|T}(0|t) - P_{Y|T}(1|t),\end{split}\\
\shortintertext{and the two-point correlators are}
\expec{A_t B_s}&:= \sum_{x,y \in \{0,1\}} (-1)^{(x\oplus y)} \; P_{XY|ST}(x y| s t).
\end{align*}\end{subequations}

Recalling that the set of common-cause boxes coincides with the set of no-signalling boxes in the Bell scenario, 
 $\boldsymbol{S}^G_{\twotwotwotwo}$ constitutes what is conventionally referred to as the ``no-signalling'' set for this type.\footnote{However, as noted in Appendix \ref{diffcausalstr}, for causal structures different from the Bell scenario, the set $\boldsymbol{S}$ of processes that can be realized by a GPT causal model on the causal structure is typically distinct from the no-signalling set.}  This set is well-known to be a polytope defined by 16 positivity inequalities~\cite{Barrett2005PRresource,Beirhorst2016Bell}.

The set of classical (free) resources of type {\twotwotwotwo} is a subset therein, conventionally termed the \enquote{local set}, and is defined by the same 16 positivity inequalities together with eight additional facet-defining Bell inequalities, namely, the canonical CHSH inequality and its seven variants~\cite{Bellreview}.  A resource is therefore nonclassical (nonfree) if and only if it violates a facet-defining Bell inequality. 

The eight variants of the canonical CHSH function are
\begin{align}\label{eq:chshvariants0}&\hspace{-3ex}\text{\footnotesize
\(\begin{array}{l}
{\CHSH}_0(R) \coloneqq
  {+}\expec{A_0 B_0}{+}\expec{A_1 B_0}{+}\expec{A_0 B_1}{-}\expec{A_1 B_1},\\
{\CHSH}_1(R) \coloneqq
 {+}\expec{A_0 B_0}{+}\expec{A_1 B_0}{-}\expec{A_0 B_1}{+}\expec{A_1 B_1},\\
{\CHSH}_2(R) \coloneqq
 {+}\expec{A_0 B_0}{-}\expec{A_1 B_0}{+}\expec{A_0 B_1}{+}\expec{A_1 B_1},\\
{\CHSH}_3(R) \coloneqq
 {-}\expec{A_0 B_0}{+}\expec{A_1 B_0}{+}\expec{A_0 B_1}{+}\expec{A_1 B_1},\\
{\CHSH}_4(R) \coloneqq
  {-}\expec{A_0 B_0}{-}\expec{A_1 B_0}{-}\expec{A_0 B_1}{+}\expec{A_1 B_1},\\
{\CHSH}_5(R) \coloneqq
  {-}\expec{A_0 B_0}{-}\expec{A_1 B_0}{+}\expec{A_0 B_1}{-}\expec{A_1 B_1},\\
{\CHSH}_6(R) \coloneqq
  {-}\expec{A_0 B_0}{+}\expec{A_1 B_0}{-}\expec{A_0 B_1}{-}\expec{A_1 B_1},\\
{\CHSH}_7(R) \coloneqq
  {+}\expec{A_0 B_0}{-}\expec{A_1 B_0}{-}\expec{A_0 B_1}{-}\expec{A_1 B_1}.
\end{array}\)} 
\end{align}
The canonical CHSH function is ${\CHSH}_0$, which we will sometimes denote simply as ${\CHSH}$.

In terms of these, the eight facet-defining Bell inequalities are  
\begin{equation}
{\CHSH}_k(R) \leq 2\;\;\text{ for}\;\;k\in \{0,\dots, 7\}.
\end{equation}

Note that the regions defined by strict violation of each of the eight inequalities are nonoverlapping~\cite{Beirhorst2016Bell}.  It follows that {\em one and only one} of the eight ${\CHSH}$ inequalities can be violated by a given resource, i.e., for nonfree $R$ there is precisely one value of $k$ such that ${\CHSH}_k (R) > 2$.

There are eight extremal nonfree vertices of the full polytope $\boldsymbol{S}^G_{\twotwotwotwo}$.  One of these is the canonical PR box~\cite{Popescu1994,Barrett2005PRunit}, denoted $R_{\textup{PR}}$ and defined explicitly in Table~\ref{tab:genboxes2}; the other seven are variants of this PR box. For each $k$, we denote the associated variant of the PR-box by $R_{\text{PR},k}$ (so that the canonical PR box is associated to $k=0$, $R_{\textup{PR}}=R_{\text{PR},0}$).  $R_{\text{PR},k}$ is the unique resource that maximally violates the $k$th ${\CHSH}$ inequality, i.e., that achieves its algebraic maximum, ${\CHSH}_k(R_{\text{PR},k}) = 4$.

Unsurprisingly, the variants of the facet-defining Bell inequalities are interconvertible under {\LSO} operations, as are the variants of the extremal vertices. To illustrate this, it is convenient to factorize the {\twotwotwotwo} LSO group into a subgroup which stabilizes ${\CHSH}_0$ and a subgroup which does not, as follows.

\begin{prop}\label{prop:symmetrypartitioning}
\begin{samepage}Consider the following invertible operations, i.e., elements of the {\LSO} group for {\twotwotwotwo}-type resources: 
\newline
\noindent\begin{tabular}{ll}\label{taus}
$\!\tau_1\!:$ & $P_{XY|ST}(x,y|s,t)\leftrightarrow P_{XY|ST}(x,y{\oplus}1|s,t)$\\
$\!\tau_2\!:$ & $P_{XY|ST}(x,y|s,t)\leftrightarrow P_{XY|ST}(x,y|s{\oplus}1,t)$\\
$\!\tau_3\!:$ & $P_{XY|ST}(x,y|s,t)\leftrightarrow P_{XY|ST}(x,y|s,t{\oplus}1)$\\
$\!\tau_4\!:$ & $P_{XY|ST}(x,y|s,t)\leftrightarrow P_{XY|ST}(x{\oplus}1,y{\oplus}1|s,t)$ \label{CHSHpreservingtransformations}\\
$\!\tau_5\!:$ & $P_{XY|ST}(x,y|s,t)\leftrightarrow P_{XY|ST}(x{\oplus}s,y|s,t{\oplus}1)$\\
$\!\tau_6\!:$ & $P_{XY|ST}(x,y|s,t)\leftrightarrow P_{XY|ST}(x,y{\oplus}t|s{\oplus}1,t)$\\
\end{tabular}\end{samepage}
\newline\noindent Then,\newline
\begin{compactenum}[(\ref{prop:symmetrypartitioning}a)][4]
\item The order-64 group $G_{123456}$ generated by $\{\tau_1,\tau_2,\tau_3,\tau_4,\tau_5,\tau_6\}$ is the entire {\LSO} group for {\twotwotwotwo} resources.
\item The order-8 subgroup $G_{123}$ generated by $\{\tau_1,\tau_2,\tau_3\}$ has no elements in common with the subgroup $G_{456}$ generated by $\{\tau_4,\tau_5,\tau_6\}$ other than the identity operation.
\item\label{prop:CHSHstabilizer} The order-8 subgroup $G_{456}$ generated by $\{\tau_4,\tau_5,\tau_6\}$ stabilizes the canonical PR box and the $\CHSH_0$ inequality.
\item\label{prop:CHSHorbit}  For any $k\in\{0...7\}$, the orbit of $\CHSH_k$ under $G_{123}$ is $\{\CHSH_0,...,\CHSH_7\}$, and the orbit of $R_{\textup{PR},k}$ under $G_{123}$ is $\{R_{\textup{PR},0},...,R_{\textup{PR},7}\}$.
\end{compactenum}
\end{prop}
\begin{proof}
The first two claims in Proposition~\ref{prop:symmetrypartitioning} are readily verified by standard group theory algorithms~\cite{Seress2003}. The latter two claims become self-evident by explicitly examining the actions of the operations on expectation values (and hence, their action on resources or functions on resources), per Table~\ref{relabelingeff}.
{
\begin{table*}[htb!]
\begin{center}\centering 
{\setlength{\tabcolsep}{1.2ex}
\begin{tabular}{|c|rrrrrrrr|} 
\toprule
 & \(\Braket{A_0}\) & \(\Braket{A_1}\) & \(\Braket{B_0}\) & \(\Braket{B_1}\) & \(\Braket{A_0 B_0}\) & \(\Braket{A_1 B_0}\) & \(\Braket{A_0 B_1}\) & \(\Braket{A_1 B_1}\) \\
 \midrule
\(\tau_1\) & \(\Braket{A_0}\) &  \(\Braket{A_1}\) &  \(-\Braket{B_0}\) &  \(-\Braket{B_1}\) & 
\(-\Braket{A_0 B_0}\) & \(-\Braket{A_1 B_0}\) & \(-\Braket{A_0 B_1}\) & \(-\Braket{A_1 B_1}\)\\
\(\tau_2\) & \(\Braket{A_1}\) &  \(\Braket{A_0}\) &  \(\Braket{B_0}\) &  \(\Braket{B_1}\) & 
\(\Braket{A_1 B_0}\) & \(\Braket{A_0 B_0}\) & \(\Braket{A_1 B_1}\) & \(\Braket{A_0 B_1}\)\\
\(\tau_3\) & \(\Braket{A_0}\) &  \(\Braket{A_1}\) &  \(\Braket{B_1}\) &  \(\Braket{B_0}\) & 
\(\Braket{A_0 B_1}\) & \(\Braket{A_1 B_1}\) & \(\Braket{A_0 B_0}\) & \(\Braket{A_1 B_0}\)\\
\(\tau_4\) & \(-\Braket{A_0}\) &  \(-\Braket{A_1}\) &  \(-\Braket{B_0}\) &  \(-\Braket{B_1}\) & 
\(\Braket{A_0 B_0}\) & \(\Braket{A_1 B_0}\) & \(\Braket{A_0 B_1}\) & \(\Braket{A_1 B_1}\)\\
\(\tau_5\) &  \(\Braket{A_0}\) &  \(-\Braket{A_1}\) &  \(\Braket{B_1}\) &  \(\Braket{B_0}\) & 
\(\Braket{A_0 B_1}\) & \(-\Braket{A_1 B_1}\) & \(\Braket{A_0 B_0}\) &  \(-\Braket{A_1 B_0}\)\\
\(\tau_6\) & \(\Braket{A_1}\) &  \(\Braket{A_0}\) &  \(\Braket{B_0}\) & \(-\Braket{B_1}\) & 
\(\Braket{A_1 B_0}\) & \(\Braket{A_0 B_0}\) & \(-\Braket{A_1 B_1}\) & \(-\Braket{A_0 B_1}\)\\
 \bottomrule
\end{tabular}}\vspace{-1ex}
\end{center}
\caption{ Action of each of the six specified 
symmetry operations 
in terms of marginal expectation values and correlators.}
\label{relabelingeff}
\end{table*}}
In light of Table~\ref{relabelingeff}, the third claim is easily verified. The fourth claim simply captures the fact that the eight CHSH functions are related by ${\LSO}$, and similarly the eight PR boxes are also interconvertible under ${\LSO}$. 
We can explicitly show how the interconversions are accomplished by $G_{123}$ by describing the actions of  $\{\tau_1,\tau_2,\tau_3\}$ as permutations on the ordered set of $\CHSH$ functions, or equivalently, on the ordered set of PR boxes. \begin{compactitem}
\item $\tau_1$ flips the sign of every correlator, so the action of $\tau_1$ on the ordered set of $\CHSH$ functions is the permutation $(0,4)(1,5)(2,6)(3,7)$.
\item $\tau_2$ exchanges the roles of $A_0$ and $A_1$, so the action of $\tau_2$ on the ordered set of $\CHSH$ functions is the permutation $(0,1)(2,3)(4,5)(6,7)$. 
\item $\tau_3$ exchanges the roles of $B_0$ and $B_1$, so the action of $\tau_3$ on the ordered set of $\CHSH$ functions is the permutation $(0,2)(1,3)(4,6)(5,7)$. 
\end{compactitem}
Therefore the \emph{orbit} of $\CHSH_k$ under $G_{123}$ is  easily checked to be $\{\CHSH_0,...,\CHSH_7\}$, as claimed.

The ordered set of PR boxes transforms under LSO operations in exactly the same manner as the ordered set of CHSH functions, since the {\em values} of the marginals and the correlators for resource $R_{\textup{PR},k}$ coincide with the {\em coefficients} of the associated terms in the linear function $\CHSH_k$ (compare, e.g., the expression for CHSH$_0$ in Eq.~\eqref{eq:chshvariants0} with the values of the marginals and correlators for $R_{\textup{PR}}$ in Table \eqref{tab:genboxes2}). Hence, the argument just given 
also establishes that the orbit of $R_{\textup{PR},k}$ under $G_{123}$ is $\{R_{\textup{PR},0},...,R_{\textup{PR},7}\}$.
\end{proof}

\subsection{Defining the two useful monotones}

\noindent {\bf Monotone 1: The yield of a resource with respect to the set of resources of type {\twotwotwotwo}, as measured by the CHSH function.}

To define our first monotone, consider the canonical CHSH function
\begin{align*} 
\operatorname{CHSH}(R)\!\coloneqq \expec{A_0 B_0}+\expec{A_0 B_1}+\expec{A_1 B_0}-\expec{A_1 B_1}.
\end{align*}
The CHSH function is type-specific\footnote{The CHSH function is well-defined only for resources of type {\twotwotwotwo}.}
 and furthermore is not a monotone~\cite{de2014nonlocality}.
However, we can apply the prescription of Eq.~\eqref{yieldprescr} to this function, taking the set $\boldsymbol{S}$ to be $\boldsymbol{S}^G_{\twotwotwotwo}$, i.e., the set of all common-cause boxes of type {\twotwotwotwo}.
Doing so, we define the following (type-independent) yield-based monotone, which we will denote by $M_{\CHSH}$:
\begin{align}\begin{split}\label{eq:CHSHmonotonedefn}
&M_{\CHSH}(R) \coloneqq M[\textup{CHSH-yield}, \boldsymbol{S}^G_{\twotwotwotwo}](R)\\
&=  \max\limits_{R^{\star}\in\boldsymbol{S}^G_{\twotwotwotwo}} \left\{\CHSH(R^{\star})\,\;\text{ s.t. }\;R\conv R^{\star} \right\}.
\end{split}\end{align}

Note that one can {\em always} find some $R^{\star}\in\boldsymbol{S}^G_{\twotwotwotwo}$ such that $R\conv R^{\star}$ regardless of the type or details of $R$, simply because free resources of type {\twotwotwotwo} may always be freely generated after discarding $R$. 
Hence, the value of this monotone is never less than 2, which is the maximum of the CHSH function when applied to the subset of free resources.

If one applies this procedure to {\em any} of the eight variants of the CHSH functions in Eq.~\eqref{eq:chshvariants0},
the monotones one thereby obtains all turn out to be equivalent to $M_{\CHSH}$.  This follows from the fact that 
all variants of the CHSH function are interconvertible under LSO and therefore the maximum of any one in an optimiziation over all LOSR operations is the same as any other, as noted in Proposition~\ref{prop:symmetrypartitioning}{\ref{prop:CHSHorbit}}.

\noindent {\bf Monotone 2: The cost of a resource with respect to a set of noisy PR box resources, as measured by the CHSH function.}

Our second monotone also involves optimizing the CHSH function, but it is a cost-based monotone, and the set of resources over which one optimizes is restricted to a particular one-parameter family of resources of type {\twotwotwotwo} (rather than the full set $\boldsymbol{S}^G_{\twotwotwotwo}$). 

To define this family, we need to highlight
a particular resource in the free set, which we denote $L_{\textup{NPR}}^{\rm b}$.\footnote{The use of $L$ instead of $R$ when describing the resource $L_{\textup{NPR}}^{\rm b}$ is a nod to the conventional terminology wherein the classically realizable common-cause boxes are often called the {\em local} boxes.  See the discussion in the introduction for why we explicitly avoid the local-nonlocal terminology here.}
$L_{\textup{NPR}}^{\rm b}$ can be defined as the uniform mixture of the PR box with the maximally mixed resource $L_{\varnothing}$ (defined in Table~\ref{tab:genboxes2}), namely $L_{\textup{NPR}}^{\rm b}=\frac{1}{2} R_{\textup{PR}}+\frac{1}{2} L_{\varnothing}$, as enumerated in Table~\ref{tab:genboxes2}.
 The superscript $\rm b$ in the notation $L_{\textup{NPR}}^{\rm b}$ denotes the fact that this resource sits on the boundary of the free set, namely, that it saturates the canonical CHSH inequality, $\CHSH(L_{\textup{NPR}}^{\rm b})=2$.

The one-parameter family of resources defining our cost construction are the convex mixtures of $R_{\textup{PR}}$ and $L_{\textup{NPR}}^{\rm b}$. We denote the set of these by $\boldsymbol{C}_{\textup{NPR}}$.  
Formally, 
\begin{align}
\boldsymbol{C}_{\textup{NPR}} &\coloneqq \{ C(\alpha) :\alpha \in [0,1]\},\\
\shortintertext{where}
\label{eq:chainparametrization}
C(\alpha)&\coloneqq\;\alpha\, R_{\textup{PR}}+(1{-}\alpha) L_{\textup{NPR}}^{\rm b}.
\end{align}
We use \enquote{$\boldsymbol{C}$} because the set of resources forms a chain (defined in Section~\ref{sec:oraclelimitations}) and \enquote{NPR} because each resource in the chain is a noisy version of the PR box. 

Geometrically, the chain $\boldsymbol{C}_{\textup{NPR}}$ describes a line segment of resources with endpoints $R_{\textup{PR}}$ and $L_{\textup{NPR}}^{\rm b}$, and  $\alpha$ parametrizes the distance from $C(\alpha)$ to $L_{\textup{NPR}}^{\rm b}$ (the bottom of the chain).
To see that the elements of $\boldsymbol{C}_{\textup{NPR}}$ do indeed form a chain in the partial order, it suffices to note that one can move downwards (\emph{decreasing} $\alpha$) starting from any $C(\alpha)$ by mixing $C(\alpha)$ with $L^{\rm b}_{\textup{NPR}}$, but one cannot move upwards (\emph{increasing} $\alpha$) from any $C(\alpha)$, as doing so would require increasing the value of the monotone $M_{\rm CHSH}$. 

Table~\ref{tab:genboxes2} provides an explicit characterization of a generic resource on the chain, as well as its endpoints and the maximally-mixed free resource.
 {
\begin{table*}[htb!]
\centering
{\setlength{\tabcolsep}{1.2ex}
\begin{tabular}{|r|cccccccc|c|c|}
\toprule
 & \(\Braket{A_0}\) & \(\Braket{A_1}\) & \(\Braket{B_0}\) & \(\Braket{B_1}\) & \(\Braket{A_0 B_0}\) & \(\Braket{A_1 B_0}\) & \(\Braket{A_0 B_1}\) & \(\Braket{A_1 B_1}\) & CHSH\\
 \midrule
  \(R_{\textup{PR}}=C(1)\)       &  0 &  0 &  0 &  0 & +1 & +1 & +1 & $-$1 & 4\\
  \(L_{\text{NPR}}^{\rm b}= C(0)\) &  0 &  0 &  0 &  0 & $\nicefrac{+1}{2}$ & $\nicefrac{+1}{2}$ & $\nicefrac{+1}{2}$ & $\nicefrac{-1}{2}$  & 2\\
    \(C(\alpha)\) &  0 &  0 &  0 &  0 & $\tfrac{\alpha+1}{2}$ & $\tfrac{\alpha+1}{2}$ & $\tfrac{\alpha+1}{2}$ & $\tfrac{-\alpha-1}{2}$ & $2\alpha{+}2$ \\\midrule
    \(L_{\varnothing}\)      &  0 &  0 &  0 &  0 & 0 & 0 & 0 & 0 & 0   
\\
 \bottomrule
\end{tabular}}\vspace{-1ex}
\caption{\label{tab:genboxes2} An explicit description of the resources referenced in our definitions.
}
\end{table*}}

Using this one-parameter family of resources, we  define the following cost-based monotone, which we denote $M_{\textup{NPR}}$,
\begin{align} \label{Malphadefn}
M_{\textup{NPR}}(R) &\coloneqq M[\CHSH\textup{-cost},\boldsymbol{C}_{\textup{NPR}}](R)\\\nonumber
&= \min\limits_{R^{\star}\in \boldsymbol{C}_{\textup{NPR}}}  \left\{\CHSH(R^{\star})\,\;\text{ s.t. }\;R^{\star}\conv R\right\},
\end{align}
where if for some $R$ there is no $R^{\star}\in \boldsymbol{C}_{\textup{NPR}}$ such that 
$R^{\star}\conv R$, then we define $M_{\textup{NPR}}= \infty$. 

Critically, note that the CHSH function is an injective (one-to-one) mapping from points on the line segment $\boldsymbol{C}_{\textup{NPR}}$ to the real numbers, with
\begin{align}
{\CHSH}\big(C(\alpha)\big)=2\alpha{+}2.
\end{align}
Thus, the problem of minimizing the CHSH function over $R^{\star}\in \boldsymbol{C}_{\textup{NPR}}$ such that $R^{\star}\conv R$ is exactly the same as minimizing the function $2\alpha{+}2$ under the constraint  $C(\alpha)\conv R$, that is, 
\begin{align} \label{Malpha2}
M_{\textup{NPR}} &= 
\min\limits_{\alpha \in [0,1]}  \left\{ 2\alpha{+}2 \,\;\text{ s.t. }\;C(\alpha) \conv R\right\}.
\end{align}

For each variant $R_{\textup{PR},k}$ of the PR box, where $k\in \{0,\dots,7\}$, we can define the chain of noisy versions thereof, that is, $\boldsymbol{C}_{\textup{NPR},k} \coloneqq \{ C_k(\alpha) :\alpha \in [0,1]\}$ where $C_k(\alpha) \coloneqq\;\alpha\, R_{\textup{PR},k}+(1{-}\alpha) L_{\textup{NPR},k}^{\rm b}$, with $L_{\textup{NPR},k}^{\rm b}=\frac{1}{2} R_{\textup{PR},k}+\frac{1}{2} L_{\varnothing}$.
One can of course define a cost-based monotone for each such chain.  However, all eight of these chains define the {\em same} monotone, because the local symmetry operations allow one to move among these, as a consequence of Proposition~\ref{prop:symmetrypartitioning}{\ref{prop:CHSHorbit}} and the fact that $L_{\varnothing}$ is stable under all {\twotwotwotwo}-type local symmetry operations.\footnote{As an aside, note that, unlike the cost with respect to the chain $\mathbf{C}_{\textup{NPR}}$, Eq.~\eqref{Malphadefn}, the cost 
with respect to the set $\boldsymbol{S}^G_{\twotwotwotwo}$ of \emph{all} resources of type {\twotwotwotwo}, as measured by the CHSH function, is utterly uninformative with regards to distinguishing the elements of $\boldsymbol{S}^G_{\twotwotwotwo}$.
This is because the resource $R_{\text{PR},4}$ can be converted to any other {\twotwotwotwo}-type resource, and yet $\CHSH(R_{\text{PR},4})=-4$, the algebraic minimum of the canonical CHSH function.
Consequently, the value of this CHSH-cost with respect to the set of all resources of type {\twotwotwotwo} is
$-4$.
Since this monotone is constant on all resources in the scenario, it is completely uninformative.}

\subsection{Closed-form expressions for $M_{\CHSH}$ and $M_{\rm NPR}$ for {\twotwotwotwo}-type resources}\label{sec:geometryproof}

The definitions of $M_{\CHSH}$ and $M_{\rm NPR}$ both involve an optimization over a continuous set of states.  In this section, we derive closed-form expressions for these monotones for resources of type {\twotwotwotwo}.

Consider first $M_{\CHSH}$.

\begin{prop}\label{prop:eqaboveCHSH}
For any free resource $R$ of type {\twotwotwotwo}, $M_{\CHSH}(R) =2$.  For any nonfree resource $R$ of type {\twotwotwotwo}, there is a unique $k \in \{0,\dots,7\}$ for which $\operatorname{CHSH}_k(R) > 2$ and such that
\begin{equation} \label{eqaboveCHSH}
 {M_{\CHSH}(R) = {\CHSH}_k(R)}.
\end{equation}
Equivalently, each function $\CHSH_{k}$ is a monotone relative to the subset of {\twotwotwotwo}-type resources for which $\operatorname{CHSH}_k(R) \geq 2$.
\end{prop}
\begin{proof} 
We already noted in Section~\ref{twomonotones} that $M_{\CHSH}(R)=2$ for all resources $R$ that are free, so it suffices to consider the case of nonfree resources.  As noted above, the fact that 
there is precisely one value of $k$ such that ${\CHSH}_k (R) > 2$ for a nonfree resource $R$ follows from the results in Ref.~\cite{Beirhorst2016Bell}.  Thus, we must show that $M_{\CHSH}(R) = \operatorname{CHSH}_k(R)$ for this value of $k$.

To prove this, we invoke Theorem 2.2 of Ref.~\cite{Beirhorst2016Bell}, which informs us that every resource $R$ which violates the $k$th ${\CHSH}$ inequality admits a convex decomposition in terms of the $k$th variant of the PR box and some free resource that saturates the $k$th ${\CHSH}$ inequality, denoted $L_k^{\rm b}$, such that ${R=\;\lambda\, R_{\textup{PR},k}+(1{-}\lambda)L_k^{\rm b}}$ for some $\lambda\in [0,1]$. 
Further, $\lambda$ is specified {\em uniquely} by the linearity of the ${\CHSH}$ functions and the fact that ${\CHSH_k(R_{\textup{PR},k})=4}$ and ${\CHSH_k(L_k^{\rm b})=2}$, which together imply that $\CHSH_k(R)={{\CHSH}_k \big(\lambda\, R_{\textup{PR},k}+(1{-}\lambda)L_k^{\rm b}\big)}={4\lambda+2(1{-}\lambda)}$. 
Again leveraging this unique decomposition together with linearity of the $\operatorname{CHSH}_k$ function and the linearity of {\LOSR} transformations, it follows that for any {\LOSR} operation $\tau$, we have  $\CHSH_k(\tau\circ R)={\lambda\CHSH_k(\tau\circ R_{\textup{PR},k})+(1{-}\lambda)\CHSH_k(\tau\circ L_k^{\rm b})}$.
Clearly ${\CHSH(\tau\circ R_{\textup{PR},k})\leq 4}$, since four is the algebraic maximum of the $\operatorname{CHSH}_k$ function, and ${\CHSH_k(\tau\circ L_k^{\rm b})\leq 2}$, since every {\LOSR} operation takes a free resource $L_k^{\rm b}$ to a free resource $L'_k$, for which $\CHSH_k(L'_k)\leq 2$. For $R$ such that $\CHSH_k(R)>2$, then, it follows that 
free operations on $R$ cannot increase its ${\CHSH}_k$ value, and hence the maximum in Eq.~\eqref{eq:CHSHmonotonedefn} is achieved by $R$ itself.  
This proves Eq.~\eqref{eqaboveCHSH}. 
\end{proof}

Using the closed-form expression for $M_{\CHSH}$, we can additionally provide closed-form expressions for the weight and robustness monotones introduced in Section~\ref{sec:othermonotones} for {\twotwotwotwo}-type resources:
\begin{cor}\label{measurecor}
For resources of type {\twotwotwotwo}, the nonlocal fraction and the robustnesses to mixing are related to $M_{\CHSH}$ as follows:
\begin{subequations}\begin{align}
M_{\textup{NF}}(R)	&= \frac{M_{\CHSH}(R)-2}{2},\\
M_{\textup{RBST},\boldsymbol{L}}(R) &= \frac{M_{\CHSH}(R)-2}{M_{\CHSH}(R)+2},\\
M_{\textup{RBST}}(R)&= \frac{M_{\CHSH}(R)-2}{M_{\CHSH}(R)+4}.
\end{align}\end{subequations}
\end{cor}
\begin{proof}
The relationship of these distance measures to the extent by which the ${\CHSH}$ inequality is violated was derived in Appendix~E of Ref.~\cite{geometry2018}. We simply recast those results in terms of $M_{\CHSH}(R)$ instead of ${\CHSH}(R)$ by means of Proposition~\ref{prop:eqaboveCHSH}.
\end{proof}

The values of the four monotones $M_{\CHSH}(R)$, $M_{\textup{NF}}(R)$, $M_{\textup{RBST},\boldsymbol{L}}(R)$, and $M_{\textup{RBST}}(R)$
 are therefore all expressible as strictly-increasing functions of 
 one another when applied to resources of type {\twotwotwotwo}. 
That is, if any one of these monotones increases (respectively decreases) between a given pair of resources of type {\twotwotwotwo}, then \emph{all} of monotones will similarly increase (respectively decrease) between that pair of resources. As we will focus on the {\twotwotwotwo} type below, and the three distance-function monotones are no more informative than $M_{\CHSH}$ in this case, we will not discuss them further.

We now turn to providing a closed-form expression for $M_{\NPR}$ for resources of type {\twotwotwotwo}. We first recall some more details of the geometry of $\boldsymbol{S}^G_{\twotwotwotwo}$.

Recall that we use the superscript $b$ to denote that a resource lies on the particular boundary of the free set that is defined by the CHSH inequality (and thus that it saturates this inequality). 
 We further use the superscript ${bb}$ to denote that a resource both saturates the CHSH inequality and {\em additionally} lies on the boundary of the full polytope of resources, $\boldsymbol{S}^G_{\twotwotwotwo}$.
The set $\boldsymbol{L}_k^{\rm b}$ 
of CHSH$_k$-inequality-saturating resources is 7-dimensional, and the set $\boldsymbol{L}_k^{\rm bb}$ of CHSH$_k$-inequality-saturating resources on the boundary of the full polytope $\boldsymbol{S}^G_{\twotwotwotwo}$
 is 6-dimensional.\footnote{$\boldsymbol{L}^{\rm b}$ is a facet of the 8-dimensional {\twotwotwotwo} local polytope, and facets of polytopes are always one dimension lower than the dimension of the polytope itself. A resource is within $\boldsymbol{L}_k^{\rm bb}$ if it is both a member of the facet defined by the CHSH$_k$ inequality and also a member of some other facet defined by a positivity inequality. The regions defined by the intersection of adjacent facets are generally termed `ridges', and a ridge always has dimensionality $d-2$, where $d$ is the dimension on the polytope. $\boldsymbol{L}_k^{\rm bb}$ is a collection of all the eight ridges adjacent to the $\boldsymbol{L}_k^{\rm b}$ facet. Equivalently, $R\in \boldsymbol{L}_k^{\rm bb}$ if and only if $R$ can be convexly decomposed as a mixture over seven-or-fewer (out of eight) deterministic boxes which saturate the CHSH inequality. Each possible size-seven subset of CHSH-inequality-saturating deterministic boxes defines one of the eight 6-dimensional ridges comprising $\boldsymbol{L}_k^{\rm bb}$.} It follows that $\boldsymbol{L}_k^{\rm bb} \subseteq \boldsymbol{L}_k^{\rm b}$.

\begin{prop}\label{geom}
For any free resource $R$ of type {\twotwotwotwo}, $M_{\rm NPR}(R) =2$.  For any nonfree resource $R$ of type {\twotwotwotwo}, there is a unique $k \in \{0,\dots,7\}$ for which $\operatorname{CHSH}_k(R) > 2$.  
Within this region, if $R\in\boldsymbol{C}_{\textup{NPR},k}$, then we have simply $M_{\textup{NPR}}(R) = \CHSH_k(R)$.  If, on the other hand, $R\not\in\boldsymbol{C}_{\textup{NPR},k}$, we have 
$$M_{\textup{NPR}}(R)=2\alpha{+}2,$$
where $\alpha$ is the value appearing in the decomposition $R=\;\gamma\, L_{R}^{\rm bb}+(1{-}\gamma) C_k(\alpha)$, where $C_k(\alpha) \in \boldsymbol{C}_{\textup{NPR},k}$, $L_{R}^{\rm bb} \in \boldsymbol{L}_k^{\rm bb}$ and $\gamma \in [0,1]$.
This value of $\alpha$ is unambiguous (and computable from simple geometry) because there exists a {\em unique} resource $L_{R}^{\rm bb} \in \boldsymbol{L}_k^{\rm bb}$ and a unique choice of $\gamma \in [0,1]$ and of $\alpha \in [0,1]$ such that $R=\;\gamma\, L_{R}^{\rm bb}+(1{-}\gamma) C_k(\alpha)$.
\end{prop}
\noindent The (unique) relevant decomposition is shown in Fig.~\ref{fig:TwoDifferentDecompositions} (for the case where $k=0$). The proof of this proposition is given in Appendix~\ref{sec:geom}.

\begin{figure} 
\begin{center}
\begin{tikzpicture}[scale=0.7]
\path[draw, ultra thick] (-30:4) -- (90:4) -- (210:4) -- cycle;

\node[draw,shape=circle,fill,scale=.4, label={[label distance=0.1cm]0:\large{${L_{R}^{\rm bb}}$}}] at (-30:4) {};
\node[draw,shape=circle,fill,scale=.4, label={[label distance=0.1cm]90:\large{${R_{\textup{PR}}}$}}] at (90:4) {};
\node[draw,shape=circle,fill,scale=.4, label={[label distance=0.1cm]180:\large{${L_{\textup{NPR}}^{b}}$}}] at (210:4) {};

\path[name path=chsh] (-4,0) -- (4,0);
\path[name path=nrp, shift=(-30:4)] (0:0) -- (150:7); 
    
\path [name intersections={of=chsh and nrp, by=R}];
\node [draw,shape=circle,fill,scale=.4, label={[label distance=-0.1cm]60:${\textbf{\textit{R}}}$}] at (R) {};

\path[name path=ls] (210:4) -- (90:4);
\path [name intersections={of=ls and nrp, by=C}];
\node [draw,shape=circle,fill,scale=.4, label={[label distance=-0.1cm]135:${C(\alpha)}$}] at (C) {};

\draw [decoration={brace, mirror, raise=0.1cm }, decorate] (C) -- (R);
\node [label={[label distance=0cm, rotate=-30]-90:$\boldsymbol{\gamma}\phantom{1}$}] at ($(C)!0.5!(R)$) {};

\draw [decoration={brace, mirror, raise=0.1cm }, decorate] (R) -- (-30:3.7);
\node [label={[label distance=0cm, rotate=-30]-90:$\boldsymbol{(1-\gamma)}\phantom{111}$}] at ($(-30:4)!0.5!(R)$) {};

\draw [decoration={brace, mirror, raise=0.1cm }, decorate] (90:4) -- (C);
\node [label={[label distance=0.1cm, rotate=60]90:$\boldsymbol{(1-\alpha)}\phantom{1}$}] at ($(C)!0.5!(90:4)$) {};

\draw [decoration={brace, mirror, raise=0.1cm }, decorate] (C) -- (210:4);
\node [label={[label distance=0.1cm, rotate=60]90:$\boldsymbol{\alpha}\phantom{1}$}] at ($(C)!0.5!(210:4)$) {};

\path[clip] (-30:4) -- (90:4) -- (210:4) -- cycle;
\draw[thick,  shift=(-30:4)] (0:0) -- (150:7);
\end{tikzpicture}
\end{center}
\caption[]{
A depiction of a family of resources parametrized by $\alpha$ and $\gamma$, and the unique decomposition of a particular point $R(\alpha,\!\gamma)$ in terms of a point $C(\alpha)$ on the chain $\boldsymbol{C}_{\textup{NPR}}$ and a (unique) CHSH-saturating resource $L_{R}^{\rm bb}$ that lies in the boundary of the set of GPT-realizable common-cause boxes. Note that the parameters $\alpha$ and $(1-\alpha)$ indicate the fraction of the full line segment attributed to each sub-segment, and similarly with $\gamma$ and $(1-\gamma)$. 
}
\label{fig:TwoDifferentDecompositions}
\end{figure}
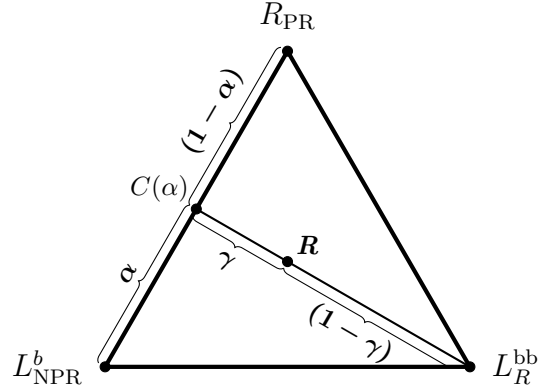

\FloatBarrier

\section{Properties of the pre-order of common-cause boxes}\label{sec:results}
We now leverage the two monotones just introduced to prove multiple interesting features of the pre-order of common cause boxes.

\subsection{Inferring global properties of the pre-order
}\label{sec:lessonsfromtwomonotones}

Important properties of the pre-order over all resources can already be learned
by considering just these two monotones ($M_{\rm CHSH}$ and $M_{\rm NPR}$)
and just resources of  type {\twotwotwotwo}, indeed, just a specific kind of two-parameter family of resources within this set. 
The kind of two-parameter family that we consider, denoted ${\boldsymbol{S}_{\twotwotwotwo}^{L^{\rm bb}_{\star}} \subset \boldsymbol{S}^G_{\twotwotwotwo}}$, is
\begin{equation}\label{Rset}
\boldsymbol{S}_{\twotwotwotwo}^{L^{\rm bb}_{\star}} \coloneqq \{ R(\alpha,\!\gamma) : \alpha \in [0,1],\; \gamma \in [0,1]\},
\end{equation}
where
\begin{align}\label{eq:def2paramfamily}
R(\alpha,\!\gamma)\coloneqq\; \gamma\, L^{\rm bb}_{\star}+(1{-}\gamma)C(\alpha),
\end{align}
with $C(\alpha) \in \boldsymbol{C}_{\textup{NPR}}$. There are many such families, one for each choice of a resource $L^{\rm bb}_{\star}\in \boldsymbol{L}^{\rm bb}$.
Each such family $\boldsymbol{S}_{\twotwotwotwo}^{L^{\rm bb}_{\star}}$ is  the convex hull of the chain $\boldsymbol{C}_{\textup{NPR}}$ and the associated point $L^{\rm bb}_{\star}$, i.e.,
\begin{align}
\boldsymbol{S}_{\twotwotwotwo}^{L^{\rm bb}_{\star}} = {\operatorname{ConvexHull}}{\left(\;\{L^{\rm bb}_{\star},\;R_{\textup{PR}},\;L_{\textup{NPR}}^{\rm b}\}\;\right)}.
\end{align}

Evaluating $M_{\textup{NPR}}$ for resources in this family is straightforward, thanks to Proposition~\ref{geom}. The proposition directly implies that for any $R(\alpha,\!\gamma) \in \boldsymbol{S}_{\twotwotwotwo}^{L^{\rm bb}_{\star}}$,
\begin{align}\label{eq:alphacostonfam}
M_{\rm NPR}\big(R(\alpha,\!\gamma)\big) 
 = 2\alpha{+}2.
\end{align}

We now consider the value of $M_{\CHSH}$ for resources in this family. 
Noting that ${{\CHSH}\big(R(\alpha,\!\gamma)\big)\geq 2}$ for all ${R(\alpha,\!\gamma)\in \boldsymbol{S}_{\twotwotwotwo}^{L^{\rm bb}_{\star}}}$, Proposition~\ref{prop:eqaboveCHSH} states that ${M_{\CHSH}\big(R(\alpha,\!\gamma)\big)={\CHSH}\big(R(\alpha,\!\gamma)\big)}$.
\noindent 
Substituting the definition of $C(\alpha)$ from Eq.~\eqref{eq:chainparametrization} into
Eq.~\eqref{eq:def2paramfamily}, we obtain
\begin{align*}
R(\alpha,\!\gamma)={\gamma\, L_{\star}^{\rm bb} \!+ (1{-}\gamma)\alpha\, R_{\textup{PR}}+ (1{-}\gamma)(1{-}\alpha)L_{\textup{NPR}}^{\rm b}}.
\end{align*}
Recalling that the CHSH function is linear and that it satisfies ${\CHSH(L^{\rm b})=2}$ for all ${L^{\rm b} \in \boldsymbol{L}^{\rm b}}$
 and $\CHSH(R_{\rm PR})=4$, it follows that
\begin{align}
\nonumber M_{\CHSH}\big(R(\alpha,\!\gamma)\big) &=\,{\CHSH}{\big(R(\alpha,\!\gamma)\big)}
\\\nonumber &= 2\gamma + 4(1{-}\gamma)\alpha+ 2(1{-}\gamma)(1{-}\alpha) 
\\&= 2\alpha(1{-}\gamma)+2. \label{MCHSHfam}
\end{align}

\begin{figure}[b!]
\begin{center}
\subfigure[\label{fig:twoMTrianglePlot}]
{
\centering
\begin{tikzpicture}[scale=0.7]
\path[draw, ultra thick] (-30:4) -- (90:4) -- (210:4) -- cycle;

\node[draw,shape=circle,fill,scale=.4, label={[label distance=0.1cm]0:\large{${L_{\star}^{\rm bb}}$}}] at (-30:4) {};
\node[draw,shape=circle,fill,scale=.4, label={[label distance=0.1cm]90:\large{${R_{\textup{PR}}}$}}] at (90:4) {};
\node[draw,shape=circle,fill,scale=.4, label={[label distance=0.1cm]180:\large{${L_{\textup{NPR}}^{\rm b}}$}}] at (210:4) {};

\begin{pgfonlayer}{myback}
\path[clip] (-30:4) -- (90:4) -- (210:4) -- cycle;
\foreach \y in{-3,-2.5,...,4}
  \draw[name path=c.\y, color=jflyBlue] (-4,\y) -- (4,\y);
  
\foreach \y in{180,175,...,120}
  \draw[name path=s.\y, shift=(-30:4), color=jflyVermillion] (0:0) -- (\y:7);
\end{pgfonlayer} 
  
\path[name path=sr1, shift=(-30:4)] (0:0) -- (135:7);  
\path[name path=cr1] (-4,1.5) -- (4,1.5);  
\path [name intersections={of=cr1 and sr1, by=R1}];
\node [draw,shape=circle,fill,scale=.4, label={[label distance=-0.1cm]60:${R_1}$}] at (R1) {};

\path[name path=sr2, shift=(-30:4)] (0:0) -- (130:7);  
\path[name path=cr2] (-4,0) -- (4,0);  
\path [name intersections={of=cr2 and sr2, by=R2}];
\node [draw,shape=circle,fill,scale=.4, label={[label distance=-0.1cm]240:${R_2}$}] at (R2) {};


\path[name path=sr3, shift=(-30:4)] (0:0) -- (145:7);  
\path[name path=cr3] (-4,0.5) -- (4,0.5);  
\path [name intersections={of=cr3 and sr3, by=R3}];
\node [draw,shape=circle,fill,scale=.4, label={[label distance=-0.1cm]240:${R_3}$}] at (R3) {};

\end{tikzpicture}
}
\subfigure[\label{fig:twoMMonotoneAxes}]
{
\centering
\begin{tikzpicture}[scale=1]

\path[draw, ultra thick] (0,0) -- (0,4) --  (4,4) -- (4,0) ;
\path[draw, ultra thick, dashed] (0,0) -- (4,0) ;
\node [draw,shape=circle,fill,scale=.4, label={[label distance=-0.1cm]225:${2}$}] at (0,0) {};
\node [draw=none,scale=.4, label={[label distance=-0cm]-90:${4}$}] at (4,0) {};
\node [draw=none,scale=.4, label={[label distance=-0cm]180:${4}$}] at (0,4) {};
\node [draw,shape=circle,fill,scale=.4, label={[label distance=-0.1cm]45:${R_{PR}}$}] at (4,4) {};
\node at (2,-0.5) {${M_{\textup{NPR}}}$};
\node at (-1, 2) {${M_{\textup{CHSH}}}$};

\path[pattern=north west lines] (0,0) -- (0,4) -- (4,4)--cycle;
\path[pattern=north east lines] (0,0) -- (0,4) -- (4,4)--cycle;

\node [draw,shape=circle,fill,scale=.4, label={[label distance=-0.1cm]45:${R_1}$}] at (3,2) {};
\node [draw,shape=circle,fill,scale=.4, label={[label distance=-0.1cm]225:${R_2}$}] at (3.33,1) {};
\node [draw,shape=circle,fill,scale=.4, label={[label distance=-0.1cm]225:${R_3}$}] at (2.33,1.33) {};

\begin{pgfonlayer}{myback}
\path[clip] (0,0) -- (4,0) -- (4,4) -- cycle;
\foreach \y in{0,0.334,...,4}
  \draw[name path=c.\y, color=jflyBlue] (\y,\y) -- (4,\y);
\foreach \y in{0,0.334,...,4}
  \draw[name path=c.\y, color=jflyVermillion] (\y,0) -- (\y,\y);;
\end{pgfonlayer} 

\end{tikzpicture}}
\end{center}
\caption{\subref{fig:twoMTrianglePlot} A plot of the 2-parameter family of resources $\boldsymbol{S}_{\twotwotwotwo}^{L^{\rm bb}_{\star}}$ (defined in Eq.~\eqref{Rset}), with values for $M_{\CHSH}$ depicted by a set of level curves (light blue, horizontal lines) and values for $M_{\rm NPR}$ depicted by another set of level curves (orange, diagonal lines).  \subref{fig:twoMMonotoneAxes} A plot of the same 2-parameter family of resources, but in a Cartesian coordinate system with $M_{\CHSH}$ and $M_{\rm NPR}$ as the coordinates.  
Because  all resources on the bottom border in plot \subref{fig:twoMTrianglePlot} are free, these all map 
 to a single point in \subref{fig:twoMMonotoneAxes}, namely $(M_{\CHSH},M_{\rm NPR})=(2,2)$. 
The fact that there are no resources with $M_{\CHSH}=2$ and $M_{\rm NPR} > 2$ is represented by the use of a dashed line at the base of the plot in \subref{fig:twoMMonotoneAxes}.  Similarly, the hatched region in \subref{fig:twoMMonotoneAxes} describes joint values of the two monotones that are not achieved by any resource in the family, as 
$M_{\CHSH}(R)\leq M_{\NPR}(R)$ for all $R$.  
Pictured in both plots are three illustrative resources.
The points $R_1$ and $R_2$ are incomparable, as are $R_3$ and $R_2$, while $R_1$ and $R_3$ are strictly ordered.  This implies that the incomparability relation in the pre-order is not transitive.
}
\label{fig:twoM}
\end{figure}

In Fig.~\ref{fig:twoMTrianglePlot}, we plot some of the level curves\footnote{A level curve of a function $f$ is a set of points that yield the same value of $f$; e.g., $\{ x\,|\,f(x){=}c\}$.} for $M_{\rm NPR}$ and $M_{\CHSH}$ over any such two-parameter family of resources.  The level curve defined by ${M_{\rm NPR}(R)=2\alpha{+}2}$ is a diagonal line in Fig.~\ref{fig:twoMTrianglePlot}, extending from the (implicit) point ${C(\alpha)}$ to the point ${L_{\star}^{\rm bb}}$. The level curve defined by ${M_{\CHSH}(R)=2\alpha(1{-}\gamma)+2}$ is a horizontal line in Fig.~\ref{fig:twoMTrianglePlot}, extending between the two implicit points ${C(\alpha)}$ and ${\alpha\, R_{\textup{PR}}+(1{-}\alpha) L_{\star}^{\rm bb}}$.

From these level curves, we can immediately deduce a number of features of the pre-order of resources. In particular, we consider those features of the pre-order that were defined in Section~\ref{sec:oraclelimitations}.

First, we see that the pre-order is locally infinite, simply by virtue of the fact that there exist chains which are represented by \emph{continuous} sets of distinct resources, such as the chain $\boldsymbol{C}_{\textup{NPR}}$. The interval between any two resources in such a continuous chain contains a continuous infinity of inequivalent resources. 

Second, one can also see that the pre-order of resources is not totally pre-ordered.
For instance, the two resources $R_1$ and $R_2$ in Fig.~\ref{fig:twoMTrianglePlot} are incomparable, as witnessed by the fact that
 \(R_1\) has a larger value of $M_{\CHSH}$ than \(R_2\) does, but \(R_2\) has a larger value of $M_{\rm NPR}$ than \(R_1\) does. More generally, the level curves for the two monotones allow one to immediately construct (by inspection) a continuous infinity of such incomparable pairs.

Furthermore, the binary relation of incomparability is not transitive, so the partial order is not weak. This can be seen by the example of the three resources in Fig.~\ref{fig:twoMTrianglePlot}: $R_1$ and $R_2$ are incomparable (as just argued) and $R_3$ and $R_2$ are incomparable (by the same logic), yet $R_1$ and $R_3$ are \emph{comparable}, as evidenced by the fact that one can obtain $R_3$ from $R_1$, by mixing $R_1$ with any free resource that intersects the line defined by the points $R_1$ and $R_3$.
 
In addition, one can also see that the height of the pre-order is infinite.  It suffices to note that the chain $\boldsymbol{C}_{\textup{NPR}}$ is totally ordered and contains a continuum of elements.  The width of the pre-order is also infinite. Consider, for example, the line segment defined by the points $R_1$ and $R_2$ in Fig.~\ref{fig:twoMTrianglePlot}. 
This subset of resources constitutes an antichain, as every resource in it is incomparable to every other: each resource has a higher $M_{\rm NPR}$ value and lower $M_{\CHSH}$ value than any of its neighbors towards the left, and has a lower $M_{\rm NPR}$ value and higher $M_{\CHSH}$ value than any of its neighbors towards the right. Because this subset also forms a continuum, it follows that the width of the pre-order is infinite. 

\begin{figure}[b!]
\begin{center}
\subfigure[\label{fig:completeTrianglePlot}]
{
\centering
\begin{tikzpicture}[scale=0.7]
\path[draw, ultra thick] (-30:4) -- (90:4) -- (210:4) -- cycle;

\node[draw,shape=circle,fill,scale=.4, label={[label distance=0.1cm]0:\large{${L_{\star}^{\rm bb}}$}}] at (-30:4) {};
\node[draw,shape=circle,fill,scale=.4, label={[label distance=0.1cm]90:\large{${R_{\textup{PR}}}$}}] at (90:4) {};
\node[draw,shape=circle,fill,scale=.4, label={[label distance=0.1cm]180:\large{${L_{\textup{NPR}}^{\rm b}}$}}] at (210:4) {};

\path[name path=ls] (210:4) -- (90:4);
\path[name path=chsh] (-4,0) -- (4,0);
\path[name intersections={of=chsh and ls, by=lp}];
\path[name path=nrp, shift=(-30:4)] (0:0) -- (150:7);
\path[name intersections={of=nrp and ls, by=lup}];

\path[clip] (-30:4) -- (90:4) -- (210:4) -- cycle;
\draw[thick] (-4,0) -- (4,0);
  
\path[clip] (-30:4) -- (90:4) -- (210:4) -- cycle;
\draw[thick, shift=(-30:4)] (0:0) -- (150:7);

\path [name intersections={of=chsh and nrp, by=R}];
\node [draw,shape=circle,fill,scale=.4, label={[label distance=-0.1cm]45:${R}$}] at (R) {};

\path[name path=rs] (-30:4) -- (90:4);
\path[name intersections={of=chsh and rs, by=rp}];

\begin{pgfonlayer}{myback}
\filldraw[draw=black, fill=jflyYellow, opacity=0.5] (-30:4) -- (R) -- (rp)  -- cycle;
\filldraw[draw=black, fill=jflyYellow, opacity=0.5] (-30:4) -- (R) -- (lp)  -- (lup);
\filldraw[draw=black, fill=jflyBlue, opacity=1] (-30:4) -- (R) -- (lp) -- (210:4) -- cycle;
\filldraw[draw=black, fill=jflySkyBlue, opacity=0.4] (R) -- (lup)  -- (90:4) -- (rp) -- cycle;
\end{pgfonlayer}

\end{tikzpicture}}
\subfigure[\label{fig:completeMonotoneAxes}]
{
\centering
\begin{tikzpicture}[scale=1]

\path[draw, ultra thick] (0,0) -- (0,4) --  (4,4) -- (4,0) ;
\path[draw, ultra thick, dashed] (0,0) -- (4,0) ;
\node [draw,shape=circle,fill,scale=.4, label={[label distance=-0.1cm]225:${2}$}] at (0,0) {};
\node [draw=none,scale=.4, label={[label distance=-0cm]-90:${4}$}] at (4,0) {};
\node [draw=none,scale=.4, label={[label distance=-0cm]180:${4}$}] at (0,4) {};
\node [draw,shape=circle,fill,scale=.4, label={[label distance=-0.1cm]45:${R_{PR}}$}] at (4,4) {};

\node at (2,-0.5) {${M_{\textup{NPR}}}$};
\node at (-1, 2) {${M_{\textup{CHSH}}}$};

\path[pattern=north west lines] (0,0) -- (0,4) -- (4,4)--cycle;
\path[pattern=north east lines] (0,0) -- (0,4) -- (4,4)--cycle;

\node [draw,shape=circle,fill,scale=.4, label={[label distance=-0.1cm]45:${R}$}] (R) at (2.75,1.25) {};

\draw[thick] (1.25,1.25) -- (4,1.25);
\draw[thick] (2.75,0) -- (2.75,2.75);
\path[draw, ultra thick] (0,0) -- (4,4);  

  
\begin{pgfonlayer}{myback}
\fill[fill=jflyYellow, opacity=0.5] (R) rectangle (4,0);
\path[clip] (0,0) -- (4,0) -- (4,4) -- cycle;
\fill[fill=jflyYellow, opacity=0.5] (R) rectangle (1.25,2.75);
\fill[fill=jflyBlue, opacity=1] (R) rectangle (0,0);
\fill[fill=jflySkyBlue, opacity=0.4] (R) rectangle (4,4);
\end{pgfonlayer}

\end{tikzpicture}}
\end{center}
\caption[]{
\subref{fig:completeTrianglePlot} and \subref{fig:completeMonotoneAxes} provide the same pair of depictions of the 2-parameter family of resources $\boldsymbol{S}_{\twotwotwotwo}^{L^{\rm bb}_{\star}}$ as were introduced in Fig.~\ref{fig:twoM}.
We consider a particular resource $R$.  In \subref{fig:completeTrianglePlot}, we depict the level curves of $M_{\rm CHSH}$ (horizontal) and $M_{\rm NPR}$ (angled) which include $R$. By monotonicity of the two monotones, $R$ cannot be freely converted into any resource in the upper light-blue region or in the pair of yellow regions. As we prove in Section~\ref{sec:ourcompleteness}, the two monotones are complete for this subset, 
which is equivalent to the fact that an arbitrary resource $R$ {\em can} be freely converted to any resource in the lower dark-blue region; namely, the entire region wherein $M_{\rm CHSH}$ and $M_{\rm NPR}$ do not have a value greater than the one they have on $R$. Resources in the upper light-blue region can be converted to $R$, while resources in the pair of yellow regions are incomparable to $R$. 
}
\label{fig:complete} 
\end{figure}

Also by inspection, for a given nonfree resource, there are a continuum of chains and antichains which contain it.  In order to see this, let us first introduce some terminology.  Within the plane of the two-parameter family of resources, depicted in Fig.~\ref{fig:completeTrianglePlot}, we refer to a direction from a given point $R$ as an \enquote{antichain direction} relative to that point, if this direction lies {\em strictly clockwise}  from  the direction defined by the $M_{\rm CHSH}$ level curve that passes through $R$ and strictly counterclockwise from the direction defined by the $M_{\rm NPR}$ level curve that passes through $R$.
 Otherwise, it is called a ``chain direction''. 
Thus an antichain direction relative to $R$ is defined by any vector originating in $R$ and terminating at a point strictly within either yellow region in Fig.~\ref{fig:completeTrianglePlot}, while a chain direction relative to $R$ is defined by any vector originating in $R$ and terminating in either blue region.

A one-dimensional curve of resources in this subset defines a chain (antichain) if and only if at every point on the curve, the tangent to the curve at that point is aimed\footnote{More precisely: a line defines two opposing directions, and both of these directions will point in a chain direction, or both will point in an antichain direction.} in a chain direction (antichain direction) relative to that point.

A final lesson we learn from these two monotones is that the set of all monotones induced (via Eq.~\eqref{eq:CHSHmonotonedefn}) by the facet-defining Bell inequalities for a given type do not yield a complete set of monotones for the resources of that type. We have shown that the set of resources is not totally pre-ordered, and as stated in Section~\ref{costandyield}, the eight facet-defining Bell inequalities for the {\twotwotwotwo}-scenario induce only a single monotone: $M_{\CHSH}$. Since no single monotone can be complete for a pre-order of resources that includes incomparable resources, it follows immediately that the monotones induced by the facet-defining Bell inequalities for the {\twotwotwotwo} type are not sufficient for fully characterizing the pre-order of resources of that type. Since such resources trivially can be lifted to any nontrivial Bell scenario (where the lifted resource will violate no facet-defining Bell inequalities other than CHSH), it follows that: 
\begin{prop}\label{prop:beyondbellviolation}
The pre-ordering of resources relative to {\LOSR} operations cannot be resolved solely using the degree of violations of facet-defining Bell inequalities.
\end{prop}
\begin{proof}
By definition, any complete set of monotones allows one to compute the values of any other monotone from them~\footnote{If one has a set of monotones $\{M_i\}_i$ which is complete, then for a given resource $R$, the set of values $\{M_i(R)\}_i$ is sufficient for (in principle) computing the value $M(R)$ of any monotone $M$ on resource $R$. First, one can deduce the equivalence class of $R$ from $\{M_i(R)\}_i$; this is possible by the completeness of the set $\{M_i\}_i$. Then, one can select any resource $R'$ from the equivalence class of $R$ and can evaluate $M(R')$ for the given monotone $M$. Because a monotone must assign the same value to all resources within an equivalence class, it holds that $M(R')=M(R)$. (Note that our argument here does not imply that one can {\em in practice} compute the value $M(R)$; this computation might involve solving a hard problem.)}.
 However, although the value of $M_{\CHSH}(R)$ can be computed (for any type-{\twotwotwotwo} resource $R$) from the eight values of the facet-defining ${\CHSH}$ functionals in Eq.~\eqref{eq:chshvariants0}, the value of $M_{\rm NPR}(R)$ cannot. This implies that any complete set of monotones must include at least one monotone (like $M_{\rm NPR}(R)$) which depends on information beyond the values of the eight ${\CHSH}$ functionals.
\end{proof}
Proposition~\ref{prop:beyondbellviolation} shows that the nonclassicality of common-cause processes is not completely characterized by the monotones that are naturally associated to facet-defining Bell functionals,
despite the fact that such Bell functionals {\em are} sufficient to witness whether or not a resource is nonclassical. 

\subsection{Incompleteness of the two monotones} \label{sec:ourincompleteness}

In this section, we prove that the two-element set of monotones $\{ M_{\CHSH},M_{\rm NPR}\}$ is not a complete set.  We do so by showing 
that it is not complete 
even for resources of type {\twotwotwotwo}.

A simple proof is as follows.  Consider resources of the form $R={\tfrac{1}{2} L^{\rm bb}_{\star}+\tfrac{1}{2}C(\text{½})}$ for different choices of the CHSH-saturating resource $L^{\rm bb}_{\star}$ that lies in the boundary of $\boldsymbol{S}^G_{\twotwotwotwo}$. 
We will show that there are pairs of resources of this form which are strictly ordered, and other pairs of resources of this form which are incomparable. These facts cannot be captured by the two monotones, which see all resources of this form as equivalent, with $M_{\rm NPR} = 3$ and $M_{\CHSH} =2.5$.

{
\begin{table*}[ht]
\centering
{\setlength{\tabcolsep}{1.2ex}
\begin{tabular}{|c|cccccccc|c|c|}
\toprule
 & \(\Braket{A_0}\) & \(\Braket{A_1}\) & \(\Braket{B_0}\) & \(\Braket{B_1}\) & \(\!\Braket{A_0 B_0}\!\) & \(\!\Braket{A_1 B_0}\!\) & \(\!\Braket{A_0 B_1}\!\) & \(\!\Braket{A_1 B_1}\!\) & \(M_{\CHSH}\) & \(M_{\rm NPR}\)\\
 \midrule
    \(L_1^{\rm bb}\)       & 1 & 1 & 1 & 1 & 1 & 1 & 1 & 1 & 2 & 2\\
    \(L_2^{\rm bb}\)     & 0 & 0 & 0 & 0 & 1 & 1 & 0 & 0 & 2 & 2\\
    \(L_3^{\rm bb}\)     & 0 & 0 & 0 & 0 & 1 & 0 & 1 & 0 & 2 & 2\\   
      \midrule
${C(\text{½})}$
       & 0 & 0 & 0 & 0 & \(\nicefrac{3}{4}\) & \(\nicefrac{3}{4}\) & \(\nicefrac{3}{4}\) & \(\nicefrac{-3}{4}\) & \(3\) & 3\\   
     \midrule
$\frac{1}{2}L_1^{\rm bb}+\frac{1}{2}C(\text{½})$
       & \(\nicefrac{1}{2}\) & \(\nicefrac{1}{2}\) & \(\nicefrac{1}{2}\) & \(\nicefrac{1}{2}\) & \(\nicefrac{7}{8}\) & \(\nicefrac{7}{8}\) & \(\nicefrac{7}{8}\) & \(\nicefrac{1}{8}\) & \(\nicefrac{5}{2}\) & 3\\
$\frac{1}{2}L_2^{\rm bb}+\frac{1}{2}C(\text{½})$
     & 0 & 0 & 0 & 0 & \(\nicefrac{7}{8}\) & \(\nicefrac{7}{8}\) & \(\nicefrac{3}{8}\) & \(\nicefrac{-3}{8}\) & \(\nicefrac{5}{2}\) & 3\\
$\frac{1}{2}L_3^{\rm bb}+\frac{1}{2}C(\text{½})$
          & 0 & 0 & 0 & 0 & \(\nicefrac{7}{8}\) & \(\nicefrac{3}{8}\) & \(\nicefrac{7}{8}\) & \(\nicefrac{-3}{8}\) & \(\nicefrac{5}{2}\) & 3\\
 \bottomrule
\end{tabular}}\vspace{-1ex}
\caption{\label{tab:genboxes3} An explicit description of the resources which demonstrate the incompleteness of the pair of monotones $\{ M_{\rm CHSH}, M_{\rm NPR}\}$. The fact that $\Braket{A_0 B_0}{=}1$ for the free boxes immediately proves that these do indeed lie on the boundary of the full set of GPT-realizable common-cause boxes of this type, $\boldsymbol{S}^G_{\twotwotwotwo}$ (since it implies that $p(0,1|0,0) = 0 =p(1,0|0,0)$, and hence these boxes saturate positivity inequalities).
}
\end{table*}

Consider for example the resources $L_1^{\rm bb}$, $L_2^{\rm bb}$, and $L_3^{\rm bb}$ defined in Table~\ref{tab:genboxes3}.
Using the pairwise comparison algorithm described in Section~\ref{polything}, one can verify that the resource ${\tfrac{1}{2} L_1^{\rm bb}+\tfrac{1}{2}C(\text{½})}$ is strictly higher in the order than ${\tfrac{1}{2} L_2^{\rm bb}+\tfrac{1}{2}C(\text{½})}$, while the two resources ${\tfrac{1}{2} L_2^{\rm bb}+\tfrac{1}{2}C(\text{½})}$ and ${\tfrac{1}{2} L_3^{\rm bb}+\tfrac{1}{2}C(\text{½})}$ are incomparable. Note that $L_1^{\rm bb}$ is a {\em convexly extremal} resource, while $L_2^{\rm bb}$ and $L_3^{\rm bb}$ are not. 

As an aside, it is worth noting that because the nonlocal fraction and the two standard robustness measures witness exactly the same ordering relations as $M_{\CHSH}$ does (as demonstrated in Section~\ref{sec:othermonotones}),  one gains nothing by supplementing $M_{\CHSH}$ and $M_{\textup{NPR}}$ with them. Rather, new monotones are needed.

The incompleteness of the two-element set $\{M_{\CHSH},M_{\textup{NPR}}\}$ is also established directly from the argument presented in
Section~\ref{sec:monotonecount}.

\subsubsection{Completeness of the two monotones for certain families of resources} \label{sec:ourcompleteness}

Although $M_{\CHSH}$ and $M_{\rm NPR}$ do not form a complete set of monotones for the set of all resources of type {\twotwotwotwo}, it turns out that they {\em do} form a complete set of monotones  for certain subsets thereof. 

\begin{prop}\label{prop:whencomplete}
The pair of monotones $\{M_{\CHSH},M_{\rm NPR}\}$ are a complete set relative to the subset of resources $\boldsymbol{S}_{\twotwotwotwo}^{L^{\rm bb}_{\star}}$ (defined in Eq.~\eqref{Rset}) for any $L^{\rm bb}_{\star} \in \boldsymbol{L}^{\rm bb}$.
\end{prop}

Proposition~\ref{prop:whencomplete} is proven in Appendix~\ref{proofprop2}. The logic of the proof is quite simple: we prove that there always exists a free operation $\tau_{{\rm erase}-\gamma}$ which converts an arbitrary resource $R(\alpha_1,\gamma_1)$ in the family to some resource $R(\alpha_2,0)$ lying on the chain $\boldsymbol{C}_{\textup{NPR}}$ without changing the value of $M_{\CHSH}$. By convexity, it follows that $R(\alpha_1,\gamma_1)$ can be converted to any resource in the convex hull of $R(\alpha_1,\gamma_1)$, $R(\alpha_2,0)$, $L^{\rm bb}_{\star}$, and $L_{\rm NPR}^{\rm b}$; namely, the dark-blue region in Fig.~\ref{fig:complete}. This region corresponds to the set of all resources with a lower value of both $M_{\CHSH}$ and $M_{\rm NPR}$.  It follows that if a conversion is not forbidden by consideration of this pair of monotones, then it is achievable.  By the definition of completeness for a set of monotones (see Eq.~\eqref{completeset}), this implies that the two monotones are indeed a complete set for this family of resources.

\subsection{At least eight independent measures of nonclassicality}
\label{sec:monotonecount}

In this section, we tackle the question of how many independent continuous monotones are required to fully specify the partial order of resources. This is the content of Theorem~\ref{prop:eightmonotonesV2}. Along the way to proving this result, we also prove a powerful result about the equivalence classes under {\LOSR} for nonfree resources of type {\twotwotwotwo}, stated in Proposition~\ref{prop:zerodclasses}.

We begin by drawing a distinction among resources.  
\begin{defn}\label{def:orbital}
A resource is said to be \term{orbital} if its equivalence class under type-preserving {\LOSR} is equal to its equivalence class under {\LSO}.
\end{defn}
It follows that if all the resources in a set $\mathbf{S}$ are orbital, then the quotient space~\cite{Quotients1994} of $\mathbf{S}$ under the group {\LSO} provides a representation of the partial order of {\LOSR}-equivalence classes of resources in $\mathbf{S}$ (despite the fact that the {\LOSR} operations do not themselves form a group).\footnote{For practical purposes, Ref.~\cite[App.~B]{Rosset2014classifying} provides a technical discussion regarding how to efficiently select a representative Bell inequality under a finite symmetry group; the procedure discussed there is equally applicable for the task of efficiently selecting \emph{canonical form} resources. Note, however, that the {\LSO} symmetry group differs from the Bell-polytope automorphism group considered in Ref.~\cite{Rosset2014classifying}, in that  {\LSO} does \emph{not} include the symmetry of exchange-of-parties.}

This property of resources is pertinent to the discussion here because of the following result:
\begin{prop}\label{prop:zerodclasses}
All nonfree resources of type {\twotwotwotwo} are orbital.
\end{prop}
The proof is provided in Appendix~\ref{proofprop3}.

Note that for {\em free} resources, {\LOSR}-equivalence is distinct from {\LSO}-equivalence because the {\LSO}-equivalence class of any resource (including a free resource) is of finite cardinality, while the {\LOSR}-equivalence of a free resource is the entire set of free resources, which is of infinite cardinality.
  Thus, free resources are not orbital.
Moreover, the coincidence between being nonfree and being orbital 
does {\em not} generalize beyond the {\twotwotwotwo} scenario. 
For instance, note that a pair of ${\twotwotwotwo}$ resources, $R_1$ and $R_2$, which are implemented in parallel can be conceptualized as a ${\fourfourfourfour}$ resource, $R_{1\otimes2}$,  by composing the two binary setting variables on the left wing into a single 4-valued setting variable on the left wing, and similarly for the other setting variable and the outcome variables.   If $R_1$ is free and $R_2$ is nonfree, then $R_{1\otimes 2}$ is nonfree, and yet because $R_1$'s equivalence class is not generated by {\LSO}, neither is the equivalence class of $R_{1\otimes 2}$.  Thus, $R_{1\otimes 2}$ is a nonfree resource that is not orbital.

To express the next proposition, we require the following definition. 
\begin{defn}\label{intrinsicdimension}
The \term{intrinsic dimension} of a set of resources $\boldsymbol{S}$, denoted $\operatorname{IntrinsicDim}(\boldsymbol{S})$,  is the smallest cardinality of continuous functions from the set to the real numbers required to uniquely identify a resource within $\boldsymbol{S}$.
\end{defn}

 \begin{prop}\label{prop:orbitalstuff}
For any compact set $\mathbf{S}$ of resources that are all orbital, the intrinsic dimension of the set $\mathbf{S}$ is a lower bound on the cardinality of a complete set of continuous monotones for $\mathbf{S}$ (and for any superset of $\mathbf{S}$).
\end{prop}

The proof is provided in Appendix~\ref{sec:proofproporbitalstuff}.

Recognizing that the set of nonfree resource of type {\twotwotwotwo} has intrinsic dimension equal to eight,\footnote{That ${\operatorname{IntrinsicDim}(\boldsymbol{S}^{\textup{nonfree}}_{\twotwotwotwo})=8}$ is evidenced by the characterization of such resources in terms of outcome biases and two-point correlators. If $T$ indicates any type, then ${\operatorname{IntrinsicDim}(\boldsymbol{S}^{\textup{nonfree}}_{T})=\operatorname{IntrinsicDim}(\boldsymbol{S}^{G}_{T})}$ whenever ${\boldsymbol{S}^{\textup{G}}_{T}\neq\boldsymbol{S}^{\textup{free}}_{T}}$ (think of subtracting one polytope from a circumscribing polytope of the same dimension). See Refs.\cite{Bellreview,Pironio2005,CG2004I3322,Rosset2014classifying} for discussions on the intrinsic dimension of no-signalling polytopes.} then Propositions~\ref{prop:zerodclasses}~and~\ref{prop:orbitalstuff} together imply the following theorem:

 \begin{thm}\label{prop:eightmonotonesV2}
For resources of type {\twotwotwotwo}, the cardinality of a complete set of continuous monotones is no less than 8.
\end{thm}

\section{Properties of the pre-order of {\em quantumly realizable} common-cause boxes} \label{sec:qr}

The bulk of this article has considered the resource theory which is defined by taking the enveloping theory of resources to be the GPT-realizable common-cause boxes,
and the free subtheory of resources to be the classically realizable common-cause boxes.
   In this section, we consider a slightly different resource theory, wherein the enveloping theory of resources is taken to be the common-cause boxes that are realizable in a {\em quantum} causal model, which we term \term{quantumly realizable}, while the free subtheory is chosen to be, as before, the common-cause boxes that are classically realizable.   Effectively, the new resource theory concerns  the nonclassicality of common-cause boxes within the scope of nonclassicality that can be achieved quantumly.  In other words, it concerns the {\em intrinsic quantumness} of common-cause boxes.

Formally, the conditional probability distribution associated to a quantumly realizable common-cause box is of the same form as Eq.~\eqref{GPTCCbox}, that is, 
\beq
P_{XY|ST}(xy|st) = ({\bf r}^{A}_{x|s} \otimes {\bf r}^{B}_{y|t}) \cdot {\bf s}^{AB},
\label{QuantumCCbox}
\eeq 
but where the vector ${\bf s}^{AB}$ is a real vector representation of a quantum state on the bipartite system composed of quantum systems $A$ and $B$, and the sets of vectors  $\{ {\bf r}^{A}_{x|s}\}_x$ and $\{ {\bf r}^{B}_{y|t}\}_y$ are real vector representations of POVMs on $A$ and on $B$ respectively. (See, e.g., Ref.~\cite{Henson2014}.) 
 
Although the conclusions we drew in Section~\ref{sec:lessonsfromtwomonotones} concerned the pre-order of GPT-realizable common-cause boxes, analogous results hold true for the pre-order of quantumly realizable common-cause boxes. This is because 
the kind of two-parameter family of GPT-realizable common-cause boxes that was used to establish global features of the pre-order of such boxes in Section~\ref{sec:lessonsfromtwomonotones} contains a two-parameter family of quantumly realizable common-cause boxes that can be used for the same purpose.  A caricature of one such quantumly realizable family is provided in Fig.~\ref{TsirelsonHardy}.   
Specifically, if one reviews the arguments that were used  in Section~\ref{sec:lessonsfromtwomonotones} to establish the various global properties of the pre-order of GPT-realizable common-cause boxes, it becomes apparent that these apply equally well to the quantumly realizable common cause boxes. 

It is also straightforward to show that the lower bound on the cardinality of a complete set of monotones, obtained in Section~\ref{sec:monotonecount}, also applies to the resource theory of quantumly realizable common-cause boxes. 
  It suffices to consider the case of the quantumly realizable resources of type {\twotwotwotwo},  hereafter $\boldsymbol{S}^Q_{\twotwotwotwo}$, and to note that the set of nonfree resources therein, that is, the set $\boldsymbol{S}^{\textup{nonfree}}_{\twotwotwotwo} \bigcap \boldsymbol{S}^{Q}_{\twotwotwotwo}$, still has intrinsic dimension equal to eight.

{
\begin{table*}[htb]
\centering
{\setlength{\tabcolsep}{1.2ex}
\begin{tabular}{|c|cccccccc|c|c|}
\toprule
 & \(\Braket{A_0}\) & \(\Braket{A_1}\) & \(\Braket{B_0}\) & \(\Braket{B_1}\) & \(\Braket{A_0 B_0}\) & \(\Braket{A_1 B_0}\) & \(\Braket{A_0 B_1}\) & \(\Braket{A_1 B_1}\) & $M_{\CHSH}$ & $M_{\NPR}$\\
 \midrule
  \(R_{\textup{Tsirelson}}\)       &  0 &  0 &  0 &  0 & $\nicefrac{\sqrt{2}}{2}$ & $\nicefrac{\sqrt{2}}{2}$ & $\nicefrac{\sqrt{2}}{2}$ & $\nicefrac{-\sqrt{2}}{2}$ & $2\sqrt{2}$ & $2\sqrt{2}$\\
            &   &   &   &   & \(\scriptstyle\approx 0.707\) & \(\scriptstyle\approx 0.707\) & \(\scriptstyle\approx 0.707\) & \(\scriptstyle\approx {-}0.707\) & \(\scriptstyle\approx 2.828\) & \(\scriptstyle\approx 2.828\)\\\midrule
  \(R_{\textup{Hardy}}\) &  $\scriptstyle 5{-}2\sqrt{5}$ &  $\scriptstyle \sqrt{5}{-}2$ &  $\scriptstyle 5{-}2\sqrt{5}$ &  $\scriptstyle \sqrt{5}{-}2$ & $\scriptstyle 6\sqrt{5}{-}13$ & $\scriptstyle 3\sqrt{5}{-}6$ & $\scriptstyle 3\sqrt{5}{-}6$ & $\scriptstyle 2\sqrt{5}{-}5$  & $\scriptstyle 10(\sqrt{5}{-}2)$ & 4\\
              &  \(\scriptstyle\approx 0.528\) & \(\scriptstyle\approx 0.236\) & \(\scriptstyle\approx 0.528\) & \(\scriptstyle\approx 0.236\) & \(\scriptstyle\approx 0.416\) & \(\scriptstyle\approx 0.708\) & \(\scriptstyle\approx 0.708\) & \(\scriptstyle\approx {-}0.528\) & \(\scriptstyle\approx 2.361\) & \\ \midrule
    \(R_{\textup{Tilt}}(\theta)\) &  $\cos (\theta)$ & 0 & $\tfrac{\cos (\theta)}{\xi(\theta)}$ & $\tfrac{\cos (\theta)}{\xi(\theta)}$ & $\tfrac{1}{\xi(\theta)}$ & $\tfrac{\sin ^2(\theta )}{\xi(\theta)}$ & $\tfrac{1}{\xi(\theta)}$ & $\!\tfrac{-\sin
   ^2(\theta )}{\xi(\theta)}$ & $2\, \xi(\theta)$ & $\scriptstyle\binom{\text{see}}{\text{caption}}$
   \\\midrule
    \(R_{\textup{Tilt}}(0)\) &  1 & 0 & 1 & 1 & 1 & 0 & 1 & 0 & 2 & 2
    \\
 \bottomrule
\end{tabular}}\vspace{-1ex}
\caption{\label{tab:boxes3} An explicit description of the Tsirelson resource, the Hardy resource, and a family of extremal quantum resources (parametrized by $\theta$) which are exposed by tilted Bell inequalities~\cite{Yang2013selftesting,Bamps2015selftesting}. We employ the shorthand $\xi(\theta)\coloneqq \sqrt{\sin ^2(\theta){+}1}$ to allow all definitions to fit within the table. We also analytically derived ${M_{\NPR}}\big(R_{\textup{Tilt}}(\theta)\big)=\frac{\xi(\theta)\left(\xi(\theta)-1\right)}{2\left(1-\cos(\theta)\right)-\xi(\theta)\left(\xi(\theta)-1\right)}$, for $0<\theta\leq \pi/2$. One can readily verify that ${M_{\NPR}}\big(R_{\textup{Tilt}}(\theta)\big)$ increases with the amount of tilt (i.e., $\expec{A_0}=\cos(\theta)$), whereas ${M_{\CHSH}}\big(R_{\textup{Tilt}}(\theta)\big)=2\sqrt{2-\cos^2(\theta)}$ decreases with added tilt. The opposite behavior of the two monotones implies that every resource in the tilted family $\theta\in (0,\pi/2]$ is incomparable to every other. $R_{\textup{Tilt}}(0)$ is a free resource, not violating any Bell inequality; at the other end of the family, $R_{\textup{Tilt}}(\tfrac{\pi}{2})=R_{\textup{Tsirelson}}$. 
}
\end{table*}}

In the rest of this section, we consider properties of the pre-order of quantumly realizable common-cause boxes that are particular to the quantum case.

Unlike for the set $\boldsymbol{S}^G_{\twotwotwotwo}$, 
 where the partial order of equivalence classes has a unique element at the top of the order (the equivalence class of $R_{\rm PR}$), in $\boldsymbol{S}^Q_{\twotwotwotwo}$ there is no unique element at the top of the order.
An easy way to see this is by considering the example of the Tsirelson box ($R_{\rm Tsirelson}$) and the Hardy box ($R_{\rm Hardy}$), each of which is defined explicitly in Table~\ref{tab:boxes3}.
Noting that $M_{\CHSH}(R_{\rm Tsirelson})=M_{\NPR}(R_{\rm Tsirelson})= 2\sqrt{2}\approx 2.828$, and that $M_{\CHSH}(R_{\rm Hardy})=10(\sqrt{5}{-}2) \approx 2.361$ and $M_{\NPR}(R_{\rm Hardy})=4$, it follows immediately that the two boxes are incomparable since $M_{\CHSH}(R_{\rm Tsirelson}) > M_{\CHSH}(R_{\rm Hardy})$ while $M_{\NPR}(R_{\rm Tsirelson}) < M_{\NPR}(R_{\rm Hardy})$.

We show these two resources in Fig.~\ref{TsirelsonHardyTrianglePlot}, together with an approximate sketch\footnote{An analytic characterization of the set of all extremal quantumly realizable resources within $\boldsymbol{S}^Q_{\twotwotwotwo}$ is not known. In Fig.~\ref{TsirelsonHardyTrianglePlot}, the endpoints and the slope of the curve at the endpoints are exact, and the rest of the curve is merely an interpolation.} of the extremal quantumly realizable resources which interpolate between them (the light-blue curve). The values of $M_{\rm CHSH}$ and $M_{\rm NPR}$ on all of these resources is plotted in Fig.~\ref{TsirelsonHardyMonotoneAxes}. From the figure, one can immediately infer that $R_{\rm Tsirelson}$ and $R_{\rm Hardy}$ are incomparable.

\color{black}
\begin{figure}[b!]
\begin{center}
\subfigure[\label{TsirelsonHardyTrianglePlot}]
{
\centering
\begin{tikzpicture}[scale=0.7]
\path[draw, ultra thick] (-30:4) -- (90:4) -- (210:4) -- cycle;

\node[draw,shape=circle,fill,scale=.4, label={[label distance=0.1cm]0:\large{${L_{\star}^{\rm bb}}$}}] at (-30:4) {};
\node[draw,shape=circle,fill,scale=.4, label={[label distance=0.1cm]90:\large{${R_{\textup{PR}}}$}}] at (90:4) {};
\node[draw,shape=circle,fill,scale=.4, label={[label distance=0.1cm]180:\large{${L_{\textup{NPR}}^{\rm b}}$}}] at (210:4) {};

\path[name path=sr1, shift=(-30:4)] (0:0) -- (150:7);  
\path[name path=cr1] (-4,1) -- (4,1);  
\path [name intersections={of=cr1 and sr1, by=R1}];
\node [draw,shape=circle,fill,scale=.4, label={[label distance=-0.1cm]135:${R_{\textup{Tsirelson}}}$}] at (R1) {};

\path[name path=sr2, shift=(-30:4)] (0:0) -- (120:7);  
\path[name path=cr2] (-4,-0.5) -- (4,-0.5);  
\path [name intersections={of=cr2 and sr2, by=R2}];
\node [draw,shape=circle,fill,scale=.4, label={[label distance=-0.1cm]45:${R_{\textup{Hardy}}}$}] at (R2) {};

\path [name intersections={of=cr1 and sr2, by=R3}];

\draw[thick, color=Blue, circle color=Blue, dotted pattern] (R1) to [out=-5, in = 130] (R2) ;

\begin{pgfonlayer}{myback}
\path[clip] (-30:4) -- (90:4) -- (210:4) -- cycle;
\foreach \y in{-3,-2.5,...,4}
  \draw[name path=c.\y, color=jflyBlue] (-4,\y) -- (4,\y);
  
\foreach \y in{180,175,...,120}
  \draw[name path=s.\y, shift=(-30:4), color=jflyVermillion] (0:0) -- (\y:7);
\end{pgfonlayer}   
  
\end{tikzpicture}
}
\hskip -0.5cm 
\subfigure[\label{TsirelsonHardyMonotoneAxes}]
{
\centering
\begin{tikzpicture}[scale=1]

\path[draw, ultra thick] (0,0) -- (0,4) --  (4,4) -- (4,0) ;
\path[draw, ultra thick, dashed] (0,0) -- (4,0) ;
\node [draw,shape=circle,fill,scale=.4, label={[label distance=-0.1cm]225:${2}$}] at (0,0) {};
\node [draw=none,scale=.4, label={[label distance=-0cm]-90:${4}$}] at (4,0) {};
\node [draw=none,scale=.4, label={[label distance=-0cm]180:${4}$}] at (0,4) {};
\node [draw,shape=circle,fill,scale=.4, label={[label distance=-0.1cm]45:${R_{\textup{PR}}}$}] at (4,4) {};
\node at (3,-0.5) {${M_{\textup{NPR}}}$};
\node at (-1, 3) {${M_{\textup{CHSH}}}$};

\path[pattern=north west lines] (0,0) -- (0,4) -- (4,4)--cycle;
\path[pattern=north east lines] (0,0) -- (0,4) -- (4,4)--cycle;

\node[fill=white] at (2,2.3) {${R_{\textup{Tsirelson}}}$};
\node [draw,shape=circle,fill,scale=.4] (R1) at (2,2) {};
\node [draw,shape=circle,fill,scale=.4, label={[label distance=-0.1cm]0:${R_{\textup{Hardy}}}$}] (R2) at (4,1) {};
\draw[thick, color=Blue, circle color=Blue, dotted pattern] (R1) to [out=-5, in = 100] (R2);

\draw[dashed] (R1) -- (0,2);
\draw[dashed] (R1) -- (2,0);
\node [draw=none] at (2,-0.2) {\scriptsize{${2\sqrt{2}}$}};
\node [draw=none] at (-0.4,2) {\scriptsize{${2\sqrt{2}}$}};

\draw[thick] (R1) to [out=-15, in = 100] (4,0);

\begin{pgfonlayer}{myback}
\path[clip] (0,0) -- (4,0) -- (4,4) -- cycle;
\foreach \y in{0,0.334,...,4}
  \draw[name path=c.\y, color=jflyBlue] (\y,\y) -- (4,\y);
\foreach \y in{0,0.334,...,4}
  \draw[name path=c.\y, color=jflyVermillion] (\y,0) -- (\y,\y);;
\end{pgfonlayer} 

\end{tikzpicture}
}
\end{center}
\caption[.]{
\subref{fig:completeTrianglePlot} and \subref{fig:completeMonotoneAxes} provide the same pair of depictions of the 2-parameter family of resources $\boldsymbol{S}_{\twotwotwotwo}^{L^{\rm bb}_{\star}}$ as were introduced in Fig.~\ref{fig:twoM}.
Here, we provide a caricature of some ordering relations among quantumly realizable common-cause boxes within this 2-parameter family.
We depict the Tsirelson and Hardy boxes (with scaled-up values of the monotones, but accurate ordering of these values), 
together with a guess of what the boundary of the set of quantumly realizable resources within this 2-parameter family might be (dotted blue curves). 
In \subref{TsirelsonHardyMonotoneAxes}, we also depict the values of the two monotones
for the set of convexly extremal, quantumly realizable resources which are self-tested by the tilted Bell inequalities (smooth black curve). 
} \label{TsirelsonHardy}
\end{figure}

Recall that no quantumly realizable resource can achieve the algebraic maximum of $M_{\rm CHSH}$, while some GPT-realizable 
 (such as $R_{\rm PR}$) {\em can} achieve the maximum. In contrast to $M_{\CHSH}$,  $M_{\rm NPR}$ is such that some quantumly realizable resources (such as $R_{\rm Hardy}$) violate it maximally. 
Furthermore, whereas $R_{\rm PR}$ maximizes both $M_{\rm CHSH}$ and $M_{\rm NPR}$, no single quantumly realizable resource maximizes both those monotones. Therefore, a unique feature of the enveloping theory of \emph{quantumly realizable} common-cause  boxes is that {\em inequivalent} resources can simultaneously be maximally nonclassical (according to distinct monotones), even among {\twotwotwotwo}-type resources.

The interpolated curve in Figs.~\ref{TsirelsonHardyTrianglePlot} and \ref{TsirelsonHardyMonotoneAxes} furthermore suggests that perhaps all extremal quantum-realizable resources depicted 
therein
are relatively incomparable.
\noindent The following lemma gives a powerful result regarding maximally nonclassical resources:
\begin{lem} \label{prop:ext}
If a nonfree resource $R$ is convexly extremal in the set $\boldsymbol{S}^Q_{\twotwotwotwo}$  of quantumly realizable resources of type {\twotwotwotwo}, then $R$ is at the top of the pre-order among quantumly realizable resources of type {\twotwotwotwo}. 
\end{lem}
\begin{proof}
Let $R\in \boldsymbol{S}^Q_{\twotwotwotwo}$ be nonfree and extremal in $\boldsymbol{S}^Q_{\twotwotwotwo}$. Then,
to prove the proposition, we need only prove that any quantumly realizable $R'\in \boldsymbol{S}^Q_{\twotwotwotwo}$ that can be freely converted to $R$ cannot be higher in the order than $R$ (rather, it must be equivalent).
Assume the existence of some quantumly realizable $R'$ such that $R'\conv R$. Since $R$ is extremal in the image of $R'$ under {\LOSR},\footnote{This is justified as follows: from $R' \conv R$ it follows that $R \in\mathbfcal{P}^{\textup{LOSR}}_{[R]}(R')$, and from the fact that quantumly realizable boxes remain quantumly realizable under {\LOSR}, it follows that $\mathbfcal{P}^{\textup{LOSR}}_{[R]}(R') \subset \boldsymbol{S}^Q_{\twotwotwotwo}$. Finally, $R$ is by assumption extremal in $\boldsymbol{S}^Q_{\twotwotwotwo}$; hence, it is extremal in $\mathbfcal{P}^{\textup{LOSR}}_{[R]}(R')$ as well.} it must be that $R'$ is converted to $R$ 
through extremal operations: that is, through {\LDO}. 
But as follows from Lemma~\ref{lem:allCHSHsensitive} in Appendix~\ref{proofprop3}, or as can be explicitly checked,\footnote{One can explicitly check that all extremal \twotwotwotwo-type resources are mapped to the free set by any deterministic operation which is not a symmetry, which implies by convexity that {\em all} \twotwotwotwo-type resources are also mapped to the free set by these operations. } the image of any {\twotwotwotwo}-scenario resource is \emph{free} under any deterministic operation which is not a symmetry! 
Put another way, there \emph{is no preimage} of any nonfree {\twotwotwotwo}-scenario resource among {\twotwotwotwo}-scenario resources under deterministic nonsymmetry operations. This means that the only $\tau\in \underset{\scriptscriptstyle ^{\twotwotwotwo\rightarrow\twotwotwotwo}}{\LDO}$ such that \emph{conceivably} $\tau\circ R'=R$ are symmetry operations. As such, if $R$ is a nonfree extremal quantumly realizable resource of type {\twotwotwotwo}, the only quantumly realizable resources (of the same type) which can be converted to $R$ are symmetries of $R$. Since resources related by a symmetry operation are in the same equivalence class, there are no {\twotwotwotwo}-type quantumly realizable resources strictly above $R$ in the partial order.\end{proof}

Lemma~\ref{prop:ext} allows us to conclude the following:
\begin{prop}\label{prop:continuoustop} 
There exists a continuous set of resources that are at the top of the pre-order of quantumly realizable {\twotwotwotwo} resources, and wherein each resource is incomparable to every other resource in the set. 
\end{prop}
\begin{proof} 
 Lemma~\ref{prop:ext} states that any subset of resources which are extremal in $\boldsymbol{S}^Q_{\twotwotwotwo}$ are at the top of the pre-order of quantumly realizable {\twotwotwotwo} resources. The fact that one can find a continuous set of such resources follows from the well-known fact that $\boldsymbol{S}^Q_{\twotwotwotwo}$ is not a polytope. By furthermore choosing such a set of extremal resources for which $M_{\CHSH}$ takes a distinct value for every resource in the set, one additionally guarantees that no two of these top-of-the-order resources are in the same equivalence class, and hence each must be incomparable to every other in the set. Refs.~\cite{Masanes2003,Allcock2009,Bellreview,geometry2018} provide some explicit sets of resources satisfying these criteria.
\end{proof}

As one concrete example, consider the one-parameter family of quantumly realizable resources which are self-tested by the tilted Bell inequalities.  We denote this family by $\{ R_{\textup{Tilt}}(\theta): \theta \in (0,\pi/2]\}.$ The definition of $R_{\textup{Tilt}}(\theta)$ is given in Table~\ref{tab:boxes3}. 
 These resources are related to a corresponding family of tilted Bell functionals~\cite{Acin2012randomnessvsnonlocality,Wolfe2012quantumbounds,Yang2013selftesting,Bamps2015selftesting}, parametrized by $\beta \in [0,2]$, namely,
\begin{align*}
\label{eq:tiltedCHSH}\operatorname{TiltedCHSH}_{\beta}&(R)\coloneqq 
\beta
 \expec{A_0} +\expec{A_0 B_0}\\\nonumber&\quad+\expec{A_1 B_0}+\expec{A_0 B_1}-\expec{A_1 B_1},
\\\nonumber\begin{split}\text{where}\quad  \smashoperator{\max\limits_{R^{\star}\in \boldsymbol{S}^{\rm free}_{\twotwotwotwo}}}&
 \;\;\operatorname{TiltedCHSH}_{\beta}(R^{\star}) = 
 2+ \beta,
\\ \text{and where}\quad  \smashoperator{\max\limits_{R^{\star}\in \boldsymbol{S}^Q_{\twotwotwotwo}}}&
 \;\;\operatorname{TiltedCHSH}_{\beta}(R^{\star}) = 
 \sqrt{8+2\beta^2}.
\end{split}\end{align*}

Note that the only value of $\beta$ for which the maximum value of this function over the quantumly realizable set $\boldsymbol{S}^Q_{\twotwotwotwo}$ coincides with the maximum value over the free set $\boldsymbol{S}^{\rm free}_{\twotwotwotwo}$ is $\beta=2$.
 Whenever $\beta<2$,
 the resource $R_{\textup{Tilt}}(\theta)$  for $\theta$ defined implicitly by the equation  $\beta=\frac{2}{\sqrt{1+2\tan^2 (\theta)}}$ is the \emph{unique} maximizer over $\boldsymbol{S}^Q_{\twotwotwotwo}$ of the corresponding tilted Bell functional.
 Formally,
\begin{align*} 
 &\beta=\frac{2}{\sqrt{1+2\tan^2 (\theta)}}<2\;\quad\text{implies}\quad
 \\&\operatorname{TiltedCHSH}_{\beta}(R^{\star}) < 
 {\operatorname{TiltedCHSH}_{\beta}}\big\lparen R_{\textup{Tilt}}(\theta)\big\rparen
\end{align*} for any  $R^{\star}\in \boldsymbol{S}^Q_{\twotwotwotwo}$ 
 $R_{\textup{Tilt}}(\theta)$.  It follows that every resource $R_{\textup{Tilt}}(\theta)$ is convexly extremal in the set of quantumly realizable resources, and its extremality is \emph{exposed} by the corresponding tilted Bell functional.

In fact, every resource in this family is incomparable to every other in the family, as can be shown directly by considering the values of $M_{\CHSH}$ and $M_{\rm NPR}$. In Fig.~\ref{TsirelsonHardyMonotoneAxes}, we show a plot of the values of the two monotones evaluated on this family. The points form a continuous antichain, shown in black. Note that the family of resources 
$\{ R_{\textup{Tilt}}(\theta): \theta \in (0,\pi/2]\}$
does not lie in any plane in the linear space of resources, and as such we do not attempt to plot the family directly (rather we \emph{only} plot its valuations with respect to the two monotones).

\section{Conclusions and outlook}

We have conceptualized Bell experiments as common-cause `box-type' processes: bipartite or multipartite processes with classical variables as inputs and outputs, the internal causal structure of which is a common-cause acting on all of the wings of the experiment.  We have argued in favour of this conceptualization by appeal to the fact that Bell's theorem can be regarded as implying the need for nonclassicality in the causal model that underlies the process. We have begun to quantify the nonclassicality of such common-cause box-type processes by developing a resource theory thereof.  We have argued in favour of a particular choice of the free operations for this resource theory, namely, those which can be achieved by embedding the resource into a circuit consisting of box-type processes realizable with a  {\em classical} common cause, and we have shown that this set is equivalent to the set of local operations and shared randomness. 

We have focused here on characterizing the pre-order defined by single-copy deterministic conversion of resources under the free operations. We have provided a linear program that decides how any two resources are ordered.  
By leveraging a pair of functions that we have proven to be monotones, we have also established a number of properties of this pre-order, such as the fact that it contains incomparable resources, that it has infinite width and height, that it is locally infinite, and that the incomparability relation is not transitive. 
Moreover, despite the fact that the values of the facet-defining Bell functionals are necessary and sufficient for \emph{witnessing} the nonclassicality of a common-cause box, we have shown that they are not sufficient for \emph{quantifying} the nonclassicality of a common-cause box.  In other words, there are aspects of the nonclassicality of such boxes relevant to resource conversions that are not captured by the degree of violation of the facet-defining Bell inequalities. 
For the particular case of resources with two binary inputs and two binary outputs, we moreover showed that at least eight continuous monotones are required to fully specify the pre-order among resources.
We have also derived some interesting facts about the pre-order of resources when one restricts attention to common-cause boxes that can be realized in quantum theory.  In particular, we have shown that for quantumly realizable resources of type {\twotwotwotwo}, all convexly extremal resources are at the top of the pre-order of such resources,
and that  there are an infinite number of incomparable resources at the top of this pre-order.

\bigskip

There is much scope for advancing and generalizing our work, some examples of which we now describe.

One of the most fundamental problems that is yet to be solved is that of characterizing the equivalence classes of resources in the pre-order induced by single-copy deterministic conversion.  That is, one would like a compressed representation of each resource that includes all and only information that is relevant to determining its equivalence class in this pre-order. Finding such a representation would be the analogue within our resource theory of proving that the equivalence classes of pure bipartite entangled states under LOCC~\cite{nielsen1999conditions}
 are given by the Schmidt coefficients of the state.
All resource monotones could then be efficiently expressed in terms of this compressed representation, while all other parameters of a resource could be safely ignored. 

Even among resources of type {\twotwotwotwo} (much less for resources of arbitrary type), we do not have a complete set of monotones for this pre-order.\footnote{Although considerations of the examples given in Section~\ref{sec:ourincompleteness} might provide the intuition necessary to find such a complete set for resources of type {\twotwotwotwo}.}  
Another interesting open problem is to connect the existing monotones to figures of merit for interesting operational tasks. E.g.,~does the value of the monotone $M_{\CHSH}$ determine the extent to which a given resource can be used for key distribution or randomness generation~\cite{BHK,Acin2006QKD,Scarani2006QKD,Acin2007QKD, colbeckamp, Pironio2010,Dhara2013DIRNG,vazirani14,Kaniewski2016chsh}? 
Since the monotone $M_{\textup{NPR}}$ is maximized for high-bias boxes from the $R_{\textup{Tilt}}(\theta)$ family (and by the Hardy box) as opposed to by the Tsirelson box, $M_{\textup{NPR}}$ is likely a figure of merit for operational tasks where the advantage is provided by such correlations~\cite{Acin2012randomnessvsnonlocality,BMP}.

Note that in deriving our results about properties of this pre-order, we have not needed to consider any types of resource beyond {\twotwotwotwo}, that is, it has sufficed to consider Bell experiments of the CHSH type.  It may be that more nuanced features of this pre-order only become apparent for more general types of resources. 

An obvious generalization of our work is to consider the pre-order induced by different sorts of conversion relations, such as 
indeterministic single-copy conversion\footnote{Indeterministic single-copy conversion is single-copy conversion that makes use of a post-selection.  Therefore, to contemplate this notion of conversion for our resource theory is to contemplate expanding the set of free operations from LOSR to LOSR with post-selection.  However, LOSR with postselection can map a correlation $P_{XY|ST}$ that satisfies the Bell inequalities to one that violates them, and even to one that violates the no-signalling condition.  (This is in contrast to the situation with LOCC, where allowing postselection does not change the set of states that one can prepare for free.)  Consequently, what sort of correlation is consistent with a classical common cause---and hence what should be deemed free in a resource theory of nonclassicality of common cause boxes---becomes contingent on what sort of postselection was implemented. For example, in a Bell experiment wherein detectors are not perfectly efficient, postselecting on detection can induce Bell inequality violations even in the absence of a nonclassical common cause.  However, for a given value of the detection efficiency, this might only be able to explain a particular degree of violation, while any higher violation would still attest to the presence of a nonclassical common cause.  In such a context, the boundary between the correlations that are consistent with a classical common cause and those that are not would no longer coincide with the facets of the Bell polytope.  Consequently, even defining the free set of resources becomes quite complicated when postselection is allowed.}, 
multi-copy conversion, asymptotic conversion, and conversion in the presence of a catalyst (see Refs.~\cite{Coecke2014,Gour2015thermo,Fritz2017asymptotic} for a discussion of these different notions, and Refs.~\cite{Brunner2009distillation,Lang2014zoo,
SandersGour,JonathanPlenio,vanDamHayden} for relevant examples of such generalized conversions). 

Other generalizations require changes to the enveloping theory of resources one is considering.
We have noted that our definition of the free operations can easily be extended to define a resource theory of nonclassicality for box-type processes in more general causal structures, distinct from that of a Bell experiment.  For example, as discussed in Appendix.~\ref{diffcausalstr}, it can be extended to a scenario we term the triangle-with-settings scenario~\cite[Fig.~8]{BilocalCorrelations}, of which the much-studied `triangle scenario'~\cite{Steudel2015,Wolfe2016inflation,Gisin2017triangle,Fraser2018} is a special case.
Another example would be to extend our definition to the `bilocality scenario'~\cite{Branciard2010,BilocalCorrelations,Chaves2017starnetworks,
TavakoliStarNetworks,RossetNetworks,tavakoli2016noncyclic}.  The analysis of such cases is complicated by the fact that our proposal implies that the set of free operations is not convex for them.  Another such generalization would be to causal structures wherein there are cause-effect relations between different parts of the experiment, for instance, experiments involving sequences of nondestructive measurements on parts of a shared resource, such as the causal structure known as the `instrumental scenario'~\cite{Pearl1995,Bonet2001,Evans2012,Henson2014,Chaves2017instrumental,
Himbeeck2018instrumental}.

A generalization of our resource theory in a different direction is to consider processes whose inputs and outputs are {\em not} classical (i.e., processes that are not `box-type'), but rather describe quantum or post-quantum systems.  For the case of the common-cause structure which we focused on here, a quantum resource theory of this sort would subsume entanglement theory, but where quantum correlation is defined relative to the set of local  operations and shared randomness (LOSR) rather than local operations and classical communication (LOCC).
\vspace{-1em}
\tocless
\section{Acknowledgments}
The authors acknowledge useful discussions with Jonathan Barrett, Tobias Fritz, Tomáš Gonda and Denis Rosset.  D.S. is supported by a Vanier Canada Graduate Scholarship. R.K. is supported by the Program of Concerted Research Actions (ARC) of the Universit\'e libre de Bruxelles. This research was supported by Perimeter Institute for Theoretical Physics. Research at Perimeter Institute is supported in part by the Government of Canada through the Department of Innovation, Science and Economic Development Canada and by the Province of Ontario through the Ministry of Colleges and Universities. This publication was made possible through the support of a grant from the John Templeton Foundation. The opinions expressed in this publication are those of the authors and do not necessarily reflect the views of the John Templeton Foundation. ABS acknowledges support by the Foundation for Polish Science (IRAP project, ICTQT, contract no. 2018/MAB/5, co-financed by EU within Smart Growth Operational Programme).

\onecolumngrid
\bigskip

\begin{center}
\line(1,0){250}
\end{center}

\bigskip

\twocolumngrid


\onecolumngrid
\clearpage
\appendix

\appendixpage
\addappheadtotoc
\section{Comparing our framework with prior work} \label{comparison}

Correlations that violate Bell inequalities have become an important object of study, not only for their relevance in foundational aspects of quantum theory, but also for their role as a resource in quantum information-processing tasks \cite{BHK,Acin2006QKD,Scarani2006QKD,Acin2007QKD, colbeckamp, Pironio2010,Dhara2013DIRNG,vazirani14,Kaniewski2016chsh}. Hence, particular effort has been devoted to the formulation of a resource theory describing them~\cite{de2014nonlocality,horodecki2015axiomatic,gallego2012,gallego2016nonlocality}. Two sets of free operations have previously been proposed to define such a resource theory, namely LOSR~\cite{de2014nonlocality,GellerPiani,gallego2016nonlocality,horodecki2015axiomatic} which we have developed in the main text, but also \term{wirings and prior-to-input classical communication} (\term{WPICC})~\cite{gallego2012}. 

In this section, we assess WPICC from the lens of our resource theory, and we identify an inconsistency among
 previous proposals for the definition of LOSR.
 The primary differences between our approach and previous approaches become most evident when one considers the question of how to develop such a resource theory for more general causal structures, as we discuss further on in Appendix~\ref{diffcausalstr}.

\subsection{WPICC versus LOSR as the set of free operations} \label{WPICCvsLOSR}
The set of WPICC operations allows for classical causal influences among the wings prior to when the parties receive their inputs.
An example of a free operation in the WPICC approach is depicted in Fig.~\eqref{LOSR_nonconvex_new}.
If one seeks to understand the resource as nonclassicality of common-cause processes, as we do here, then it is clear that the free operations should not include any cause-effect influences between the wings, and therefore should not include any classical communication between the wings. In other words, in our approach, WPICC is not a viable choice for the set of free operations, as wirings that connect different wings of the experiment 
cannot be part of any free operation.

One might think that the choice to take WPICC or LOSR to be the set of free operations is not a particularly consequential one, since WPICC and LOSR define the same partial order for boxes~\cite{gallego2016nonlocality} (see the discussion of this point in Sec.~\ref{PrevResultsInLightOfOversight}). This equivalence breaks down, however, when one considers more general resources, e.g., bipartite quantum states. Since bipartite quantum states have no inputs, allowing classical communication {\em prior to inputs} means allowing {\em arbitrary} classical communication. Hence, WPICC coincides with LOCC in this case, and LOCC defines a partial order on bipartite quantum states that is distinct from the partial order defined by LOSR~\cite{Buscemi2012LOSR}.

\begin{figure}[htb!]
 \centering
 \includegraphics[scale=.4]{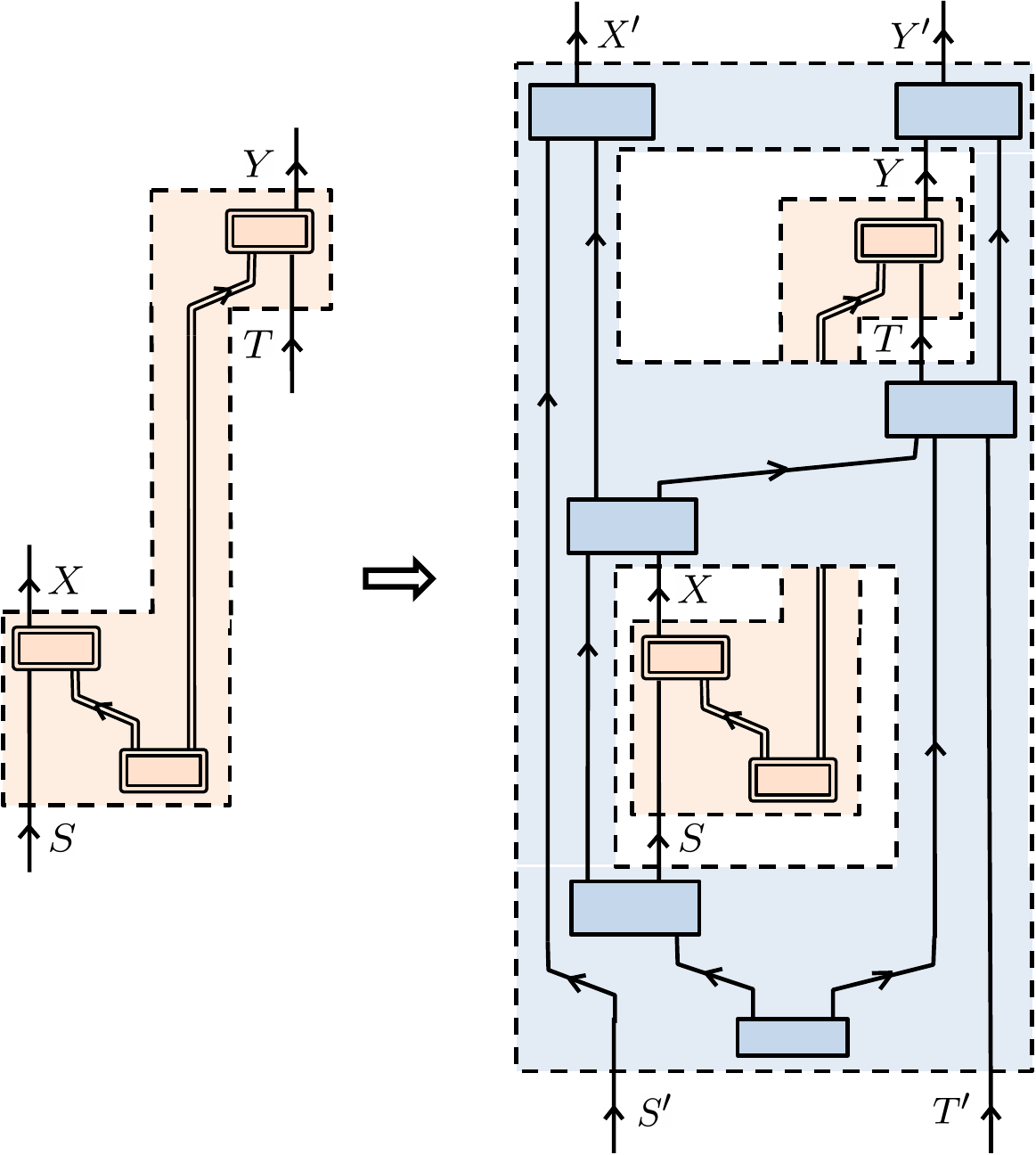}
 \caption{An example of a free operation in the WPICC approach, using the diagrammatic conventions of this article. (Compare with Fig.~1(b) of Ref.~\cite{gallego2016nonlocality}.)
Here, we see an example in which there is communication from the left wing to the right wing, which (in contrast to our approach) is allowed for free in the WPICC approach, for all times prior to when the wings receive the inputs $S'$ and $T'$.
}
 \label{LOSR_nonconvex_new}
\end{figure}

\subsection{An oversight in the literature concerning how to formalize LOSR}\label{comparisonsub}

\begin{figure}[!tbh]
 \centering
 \includegraphics[scale=.4]{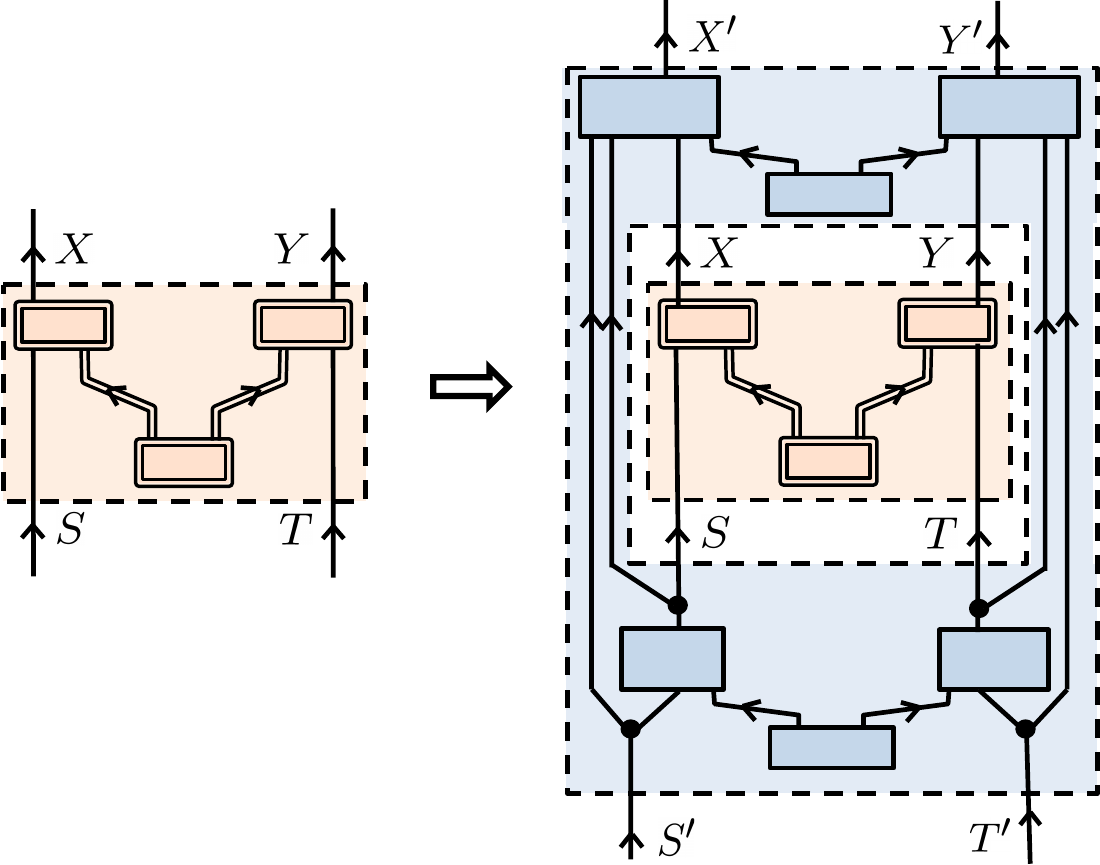}
 \caption{A depiction of Fig.~1(a) of Ref.~\cite{gallego2016nonlocality} using the diagrammatic conventions of this article.  The set of operations having this form is not as general as those depicted in our Fig.~\ref{fig:transf} because the post-processing does not have complete access to the shared randomness available at the pre-processing. One can explicitly show that the set of operations having this form is not convex. 
  It is only after taking the convex closure of the set of operations depicted here that one recovers {\LOSR}.}
 \label{LOSR_nonconvex}
\end{figure}

As we noted in the Introduction and in Section~\ref{landscape}, the intuitive notion that the set of free operations should constitute local operations supplemented by shared randomness is widely agreed upon in previous work~\cite{gallego2012,de2014nonlocality,GellerPiani,gallego2016nonlocality,horodecki2015axiomatic,
Amaral2017NCW,kaur2018fundamental,Brito2018tracedistance}.
Nonetheless, some prior work seems to have formalized this intuitive notion incorrectly.  
Specifically, the set of free operations defined in Ref.~\cite{gallego2016nonlocality} (and repeated in Refs.~\cite{Amaral2017NCW,kaur2018fundamental}) does not coincide with the set of free operations defined in Refs.~\cite{de2014nonlocality,GellerPiani} and which we endorse here as the correct choice.  Rather, it is a nonconvex subset thereof, as we will show here.  (Note that  Ref.~\cite{gallego2016nonlocality} referred to their set of free operations as ``LOSR'' but we will here reserve that term for the set of operations described in Definition~\ref{defnLOSR}.)

We suspect that the discrepancy in the definitions introduced in these papers was merely an oversight, and in particular, that none of the authors of these articles would advocate for this nonconvex subset over the full set.   
Nonetheless, we think that it is important to highlight this oversight, so that it may be avoided in future work.

It is easiest to see the difference between the definition of the free operations given in Ref.~\cite{gallego2016nonlocality} and the one endorsed here (which coincides with the definitions of Refs.~\cite{de2014nonlocality, GellerPiani}) by considering the diagrammatic representation of a generic operation in each case.   The most general free operation proposed by Ref.~\cite{gallego2016nonlocality} is depicted in their Fig. 1(a), which we reproduce here as Fig.~\ref{LOSR_nonconvex}, which should be compared with Figs.~\ref{fig:transfexplicit} and~\ref{fig:transf} of our article.  The difference is that in Fig.~\ref{LOSR_nonconvex}, the side-channels on each wing that carry information forward from the pre-processing to the post-processing are limited to carry information {\em only} about the setting variables ($S$, $S'$, $T$ and $T'$), while in Figs.~\ref{fig:transfexplicit} and ~\ref{fig:transf}, they can also carry information about the common cause that acts on the local pre-processings. 

This difference is also reflected in the equations.  The most general free operation proposed by Ref.~\cite{gallego2016nonlocality} is defined via their Eq.~(7).  In terms of the notation of this article, their Eq.~(7) asserts that
\begin{equation}\label{nonconvexsubsetLOSR}
P_{X'Y'ST|XYS'T'} = P_{X'Y'|XYSTS'T'}\; P_{ST|S'T'},
\end{equation}
where $P_{ST|S'T'}$ (denoted $I^{(L)}$ in Ref.~\cite{gallego2016nonlocality}) and $P_{X'Y'|XYSTS'T'}$ (denoted $O^{(L)}$ in Ref.~\cite{gallego2016nonlocality}) represent, respectively, the pre- and post-processings (depicted in Fig.~\ref{LOSR_nonconvex}).  Consistently with their Fig. 1(a), the expression for the post-processing stipulates
that the side-channel between the pre- and post-processings only carries information about the setting variables $S$, $S'$, $T$ and $T'$.  The analogue of this equation for the proposal endorsed here is
\begin{equation}\label{LOSRproperly}
P_{X'Y'ST|XYS'T'} = \sum_{Z_A, Z_B} P_{X'Y'|XYZ_A Z_B}\; P_{ST Z_A Z_B|S'T'},
\end{equation}
where $Z_A$ and $Z_B$ represent the variables propagated along the side-channels in Fig.~\ref{fig:transf}.  This is more general, given that $Z_A$ and $Z_B$ can encode information about the common cause in the pre-processing.

To see the difference more explicitly, consider how these expressions appear if one includes the common causes.
Because the pre- and the post-processings in the proposal of Ref.~\cite{gallego2016nonlocality}
must depend on independent sources of shared randomness (by virtue of the restriction on the side-channels), we distinguish the common causes notationally using primed and unprimed variables.  The post-processing is given by
\begin{equation}
P_{X'Y'|XYSTS'T'}=\sum_{\Lambda'} P_{X'|XSS'\Lambda'} P_{Y'|YTT'\Lambda'} P_{\Lambda'},
\end{equation}
and the pre-processing is given by
\begin{equation}
 P_{ST|S'T'}=\sum_{\Lambda} P_{S|S'\Lambda} P_{T|T'\Lambda} P_{\Lambda}.
\end{equation}
  Putting these together, we have
\begin{align}\label{them}
P_{X'Y'ST|XYS'T'} = \sum_{\Lambda \Lambda'} 
P_{X'|XSS'\Lambda'} P_{Y'|YTT'\Lambda'} 
P_{S|S'\Lambda} P_{T|T'\Lambda} P_{\Lambda'} P_{\Lambda}
\end{align}

By contrast, the proposal endorsed here distributes a single source of shared randomness between the pre- and post-processings.  If we consider the circuit depicted in Fig.~\ref{fig:transf} and note that the side-channels can now feed forward not just $S$, $S'$, $T$ and $T'$, but the common cause as well, we see that we can express the most general free operation as follows (which is equivalent to Eq.~\eqref{LC4prime}) 
\begin{align}\label{us}
P_{X'Y'ST|XYS'T'} = \sum_{\Lambda} 
P_{X'|XSS'\Lambda} P_{Y'|YTT'\Lambda} 
P_{S|S'\Lambda} P_{T|T'\Lambda} P_{\Lambda}.
\end{align}

The operational discrepancy between the two proposals is a consequence of the fact that Eq.~\eqref{them} is strictly less general than Eq.~\eqref{us}.

One can intuitively expect a failure of convex closure for the set of operations depicted in Fig.~\ref{LOSR_nonconvex} and described in Eq.~\eqref{them}, since the pre-processing and the post-processing have access to independent sources of shared randomness, and these two sources cannot generally be subsumed into a single source. 
To explicitly demonstrate the failure of convexity, we consider the following operations:
\begin{align*}
& \tau_1 \coloneqq P_{X'Y'ST|XYS'T'}=\delta_{Y',0}\delta_{S,0}\delta_{T,0}\delta_{X',0},\\
& \tau_2 \coloneqq P_{X'Y'ST|XYS'T'}=\delta_{Y',1}\delta_{S,1}\delta_{T,0}\delta_{X',0}, \\
& \tau_3\coloneqq \tfrac{1}{2}(\tau_1+\tau_2).
\end{align*}
While $\tau_1$ and $\tau_2$ are each free operations which can be realized using the circuit in Fig.~\ref{LOSR_nonconvex},
 the transformation $\tau_3$ defined by their mixture cannot be realized using this circuit.
 To see this, note that in Fig.~\ref{LOSR_nonconvex},
  any correlations between $S$ and $Y'$ can only be mediated by $T$, since the only variable in the causal past of both $S$ and $Y'$ is the variable acting as the common cause of the pre-processing, which we will denote by $\Lambda_1$, and the only means by which the value of $\Lambda_1$ could be communicated via the side channel is through $T$. But in this example, $T$ does not vary and so cannot mediate any correlations; the point distribution on $T$ screens off any correlation between $S$ and $Y'$. Hence, $\tau_3$, which exhibits perfect correlation between $S$ and $Y'$, cannot be realized in a circuit of the form of Fig.~\ref{LOSR_nonconvex}. It follows that the set of operations depicted in Fig.~\ref{LOSR_nonconvex} is not convexly closed.

Despite the definition of the free operations given in Ref.~\cite{gallego2016nonlocality}, in Appendix~A of that article, the authors avail themselves of convex mixtures of operations of the sort described by their Fig.~1(a)~and~Eq.~(7).   However, a mixing operation is only allowed if the shared randomness required to implement it is present, and given that the shared randomness available for the pre-selection is independent from that which is available for the post-selection, an arbitrary mixing operation is {\em not} allowed under the free operations proposed by Ref.~\cite{gallego2016nonlocality}.  The use of convex mixtures in Appendix~A of Ref.~\cite{gallego2016nonlocality} is therefore inconsistent with the definition of the free operations provided therein.
  
The mistake of defining the free operations as this nonconvex subset of LOSR is repeated in Ref.~\cite{kaur2018fundamental}:
 Fig.~1~and~Eq.~(13) therein
  are reproductions of Fig.~1(a)~and~Eq.~(7) of Ref.~\cite{gallego2016nonlocality}, and, like the latter, limits the side-channels to carry information only about the setting variables. It is also repeated in Ref.~\cite{Amaral2017NCW}, where the formalization of a \enquote{noncontextual wiring} per Eq.~(9) there utilizes a post-processing with randomness independent from that of pre-processing, again restricting the side-channels to exclusively information pertaining to the setting variables.

The above discussion has highlighted the fact that if one wishes the set of free operations to include  arbitrary convex mixtures of some smaller set, it is important that it be stipulated precisely how the shared randomness is distributed in order to ensure the possibility of such mixing.  In this regard, although de Vicente~\cite{de2014nonlocality} provided a definition of the free operations that is equivalent to LOSR, the physical justification for this choice was wanting.   Specifically, the definition in Ref.~\cite{de2014nonlocality} proceeds by enumerating a long list of nominally `elementary' operations and then stating 
 (in Section~4.1 of that article) that any mixture of these operations is also allowed.  No discussion is provided of why the type of shared randomness necessary for achieving arbitrary mixtures should be considered freely available.  The work of Geller and Piani~\cite{GellerPiani}, by contrast, {\em does} stipulate the physical structure of the circuit that defines the free operations, thereby providing a physical justification for taking LOSR as the set of free operations.  
 
\subsubsection{Previous results in light of this oversight}\label{PrevResultsInLightOfOversight}

Given that some previous work \cite{gallego2016nonlocality,Amaral2017NCW,kaur2018fundamental} formally defined the set of free operations by Eq.~\eqref{nonconvexsubsetLOSR}, which yields a nonconvex subset of LOSR,  one might wonder to what extent the results reached by those works still hold for LOSR proper, as defined in Eq.~\eqref{LOSRproperly}. 
In the following, we will briefly comment on some results described in Refs.~\cite{gallego2016nonlocality} and \cite{kaur2018fundamental}.

\medskip

Lemma 6 of Ref.~\cite{gallego2016nonlocality} purports to demonstrate that if a function is a monotone relative to  a set of operations that the authors term ``LOSR'', then it is also a monotone relative to WPICC.  If one interprets the set of operations termed ``LOSR'' by the authors of Ref.~\cite{gallego2016nonlocality} in the manner of the definition stipulated by their Fig.~1(a) or their Eq.~(7) (which is equivalent to Eq.~\eqref{nonconvexsubsetLOSR} above), namely, as a nonconvex subset of LOSR proper, as defined in Eq.~\eqref{LOSRproperly}, then the question would arise as to whether an analogous lemma holds for LOSR proper rather than simply the nonconvex subset thereof.  In fact, however, the proof of Lemma 6 in Ref.~\cite{gallego2016nonlocality} assumes that the set of free operations can map a resource $R$ to a convex combination of $R$ with any local box.  This is not possible if the set of free operations is the one 
 defined by their Fig.~1(a) or Eq.~(7) (or equivalently, by Eq.~\eqref{nonconvexsubsetLOSR} above). Hence, the proof of Lemma 6 holds {\em only} if the set of free operations termed ``LOSR'' in the statement of the lemma is taken to be LOSR proper, as defined in Eq.~\eqref{LOSRproperly}, and not the nonconvex subset of LOSR defined by Eq.~\eqref{nonconvexsubsetLOSR}.  
 This fact provides yet another piece of evidence that the nonconvexity of the formal definition of the free operations in Ref.~\cite{gallego2016nonlocality} was merely an oversight. The bottom line is that the proofs provided in Ref.~\cite{gallego2016nonlocality} do establish that monotonicity relative to LOSR proper implies monotonicity relative to WPICC.

\medskip

\citet{kaur2018fundamental} state in their Proposition 6 that their proposed \enquote{intrinsic non-locality} measure is monotonically nonincreasing under the set of free operations they term ``LOSR''.  But given that their definition of this term is precisely the same as the definition provided in Ref.~\cite{gallego2016nonlocality}, the set of operations in question is the nonconvex subset of LOSR defined by Eq.~\eqref{nonconvexsubsetLOSR}.
This prompts the question of whether this proposition holds if one considers LOSR proper, as defined in Eq~\eqref{LOSRproperly}, rather than this nonconvex subset thereof.

The answer is that it does.  
Establishing this is nontrivial, however, as an arbitrary monotone relative to the nonconvex subset of LOSR defined by Eq.~\eqref{nonconvexsubsetLOSR}}
need \emph{not} be a monotone relative to LOSR.
Note, however, that if
\begin{compactenum}[(i)]
\item a function $f$ is a monotone relative to {\LDO}, and 
\item $f$ happens to be a convex function,
\end{compactenum} 
then $f$ is also a monotone relative to {\LOSR}, as a consequence of Proposition~\ref{lem:belstheorem}.  Since {\LDO} is contained within the nonconvex subset of {\LOSR} defined by Eq.~\eqref{nonconvexsubsetLOSR},
convex monotones relative to those limited operations are also valid monotones relative to {\LOSR} proper. Finally, we can use this implication to recover
Ref.~\cite{kaur2018fundamental}'s Proposition 6 by leveraging Proposition 7 there regarding the convexity of \enquote{intrinsic non-locality} over box-type resources.

\subsection{Generalizing from Bell scenarios to more general causal structures}\label{diffcausalstr}

In the introduction, we contrasted our approach to defining a resource theory, which we termed the causal modelling paradigm, with a pre-existing approach, which we termed the strictly operational paradigm.  Considering causal scenarios beyond Bell scenarios helps to clarify the differences between these two approaches.

Consider, for instance, a tripartite box-type process, with setting variables for the three wings denoted $S$, $T$, and $U$, and outcome variables for the three wings denoted $X$, $Y$, and $Z$ respectively.  
One can distinguish two distinct causal structures that could underlie this sort of process: 
(i) the {\em tripartite Bell scenario}, where there is a common cause acting on all the three wings, depicted in Fig.~\ref{fig:TripartiteBellScenario}, and (ii) the {\em triangle-with-settings} scenario~\cite[Fig.~8]{BilocalCorrelations}, where there is a common cause for each pair of wings, depicted in  Fig.~\ref{fig:TriangleWithSettingsScenario}.

\begin{figure}[htb!]
\begin{center}
\subfigure[\label{fig:CandNCboxesGHZGPT}]
{
\centering
\includegraphics[scale=0.4]{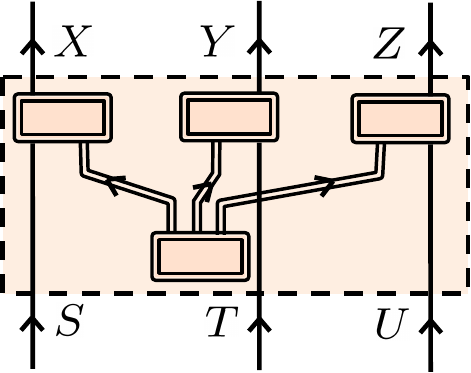}
}
\subfigure[\label{fig:CandNCboxesGHZCLASSICAL}]
{
\centering
\includegraphics[scale=0.4]{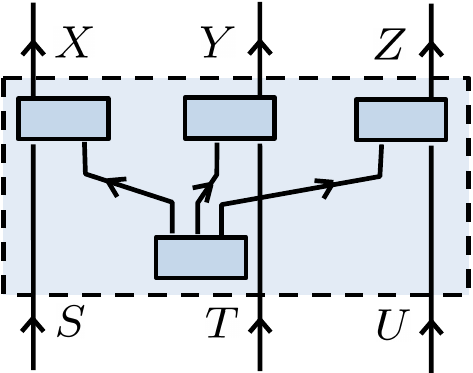}
}
\caption{The distinction between \subref{fig:CandNCboxesGHZGPT} a {\em generic} box in the tripartite Bell scenario and \subref{fig:CandNCboxesGHZCLASSICAL} a {\em classical} box in this scenario.
\label{fig:TripartiteBellScenario}
}\end{center}
\end{figure}

\begin{figure}[htb!]
\begin{center}
\subfigure[\label{fig:CandNCboxesTriangleGPT}]
{
\centering
\includegraphics[scale=0.4]{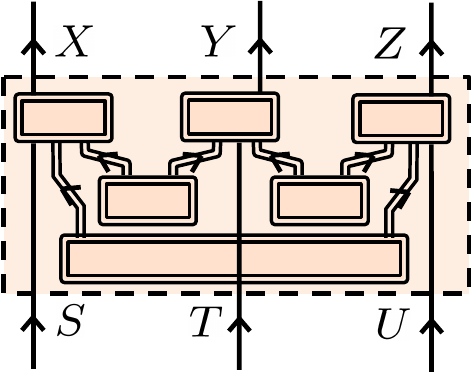}
}
\subfigure[\label{fig:CandNCboxesTriangleCLASSICAL}]
{
\centering
\includegraphics[scale=0.4]{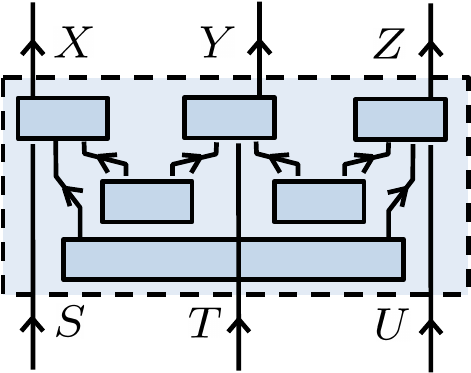}
}
\caption{The distinction between \subref{fig:CandNCboxesTriangleGPT} a {\em generic} box in the triangle-with-settings scenario and \subref{fig:CandNCboxesTriangleCLASSICAL} a {\em classical} box in this scenario.
\label{fig:TriangleWithSettingsScenario}
}\end{center}
\end{figure}

Consider the case of a generic box in the tripartite Bell scenario, depicted in Fig.~\ref{fig:CandNCboxesGHZGPT}, and label the systems distributed to the three wings by $A$, $B$ and $C$ respectively. Let us denote by ${\bf r}^{A}_{x|s}$ the GPT representation of the $X=x$ outcome of the $S=s$ measurement on system $A$, and similarly define ${\bf r}^{\rm b}_{y|t}$ and ${\bf r}^{C}_{z|u}$. If ${\bf s}^{ABC}$ denotes the GPT state of the composite $ABC$,
then the conditional probability distribution associated to this 
 box is
\beq
P_{XYZ|STU}(xyz|stu) = ({\bf r}^{A}_{x|s} \otimes {\bf r}^{\rm b}_{y|t}\otimes {\bf r}^{C}_{z|u} ) \cdot {\bf s}^{ABC}.
\label{GPT3partCCbox}
\eeq 
When the GPT is classical probability theory, we obtain the classically-realizable
box shown in Fig.~\ref{fig:CandNCboxesGHZCLASSICAL}, and the conditional probability distribution associated to it is
\begin{align}
\nonumber P_{XYZ|STU}&(xyz|stu)
\\\nonumber &=  \smashoperator{\sum_{\lambda_A\lambda_B,\lambda_C}}P_{ X|S\Lambda_A}(x|s\lambda_A)P_{Y|T\Lambda_B}(y|t\lambda_B)P_{Z|U\Lambda_C}(z|u\lambda_C) P_{\Lambda_A \Lambda_B \Lambda_C}(\lambda_A\lambda_B,\lambda_C)
\\
&=\sum_{\lambda}P_{ X|S\Lambda}(x|s\lambda)P_{Y|T\Lambda}(y|t\lambda) P_{Z|U\Lambda}(z|u\lambda) P_{\Lambda}(\lambda).
\label{Classical3partCCbox}
\end{align}

Now consider a generic box in the triangle-with-settings scenario, depicted in Fig.~\ref{fig:CandNCboxesTriangleGPT}.
Instead of an arbitrary joint GPT state ${\bf s}^{ABC}$ on the triple of systems associated to the three wings, each system is composed of two parts---$A$ is composed of $A_1$ and $A_2$, and similarly for $B$ and $C$---and the joint GPT state has the form ${\bf s}^{A_1 B_1}\otimes {\bf s}^{A_2 C_1}\otimes {\bf s}^{B_2 C_2}$. The conditional probability distribution associated to this 
box \bel{is}
\beq
P_{XYZ|STU}(xyz|stu) = ({\bf r}^{A}_{x|s} \otimes {\bf r}^{\rm b}_{y|t}\otimes {\bf r}^{C}_{z|u} ) \cdot ({\bf s}^{A_1 B_1}\otimes {\bf s}^{A_2 C_1}\otimes {\bf s}^{B_2 C_2}).
\label{GPTtriangleCCbox}
\eeq 
When the GPT is classical probability theory, we obtain the classically-realizable box shown in Fig.~\ref{fig:CandNCboxesTriangleCLASSICAL}.
Taking $\Lambda_A = (\Lambda_{A_1},\Lambda_{A_2})$, and similarly for $\Lambda_B$ and $\Lambda_C$, we have
\begin{align}
P_{XYZ|STU}(xyz|stu)&=  \sum_{\lambda_A\lambda_B,\lambda_C}P_{ X|S\Lambda_A}(x|s\lambda_A)P_{Y|T\Lambda_B}(y|t\lambda_B)P_{Z|U\Lambda_C}(z|u\lambda_C)  \nonumber\\
&\;\;\times P_{\Lambda_{A_1} \Lambda_{B_1}}(\lambda_{A_1},\lambda_{B_1}) P_{\Lambda_{B_2} \Lambda_{C_1}}(\lambda_{B_2},\lambda_{C_1}) P_{\Lambda_{A_2} \Lambda_{C_2}} (\lambda_{A_2},\lambda_{C_2})
\nonumber\\
&=\sum_{\lambda,\lambda',\lambda''}P_{ X|S\Lambda\Lambda''}(x|s\lambda\lambda'')P_{Y|T\Lambda\Lambda'}(y|t\lambda\lambda') P_{Z|U\Lambda'\Lambda''}(z|u\lambda'\lambda'') \nonumber\\
&\;\;\times P_{\Lambda}(\lambda) P_{\Lambda'}(\lambda') P_{\Lambda''}(\lambda'').
\label{ClassicaltriangleCCbox}
\end{align}
  
\begin{figure}[htb]
 \centering
 \includegraphics[scale=0.4]{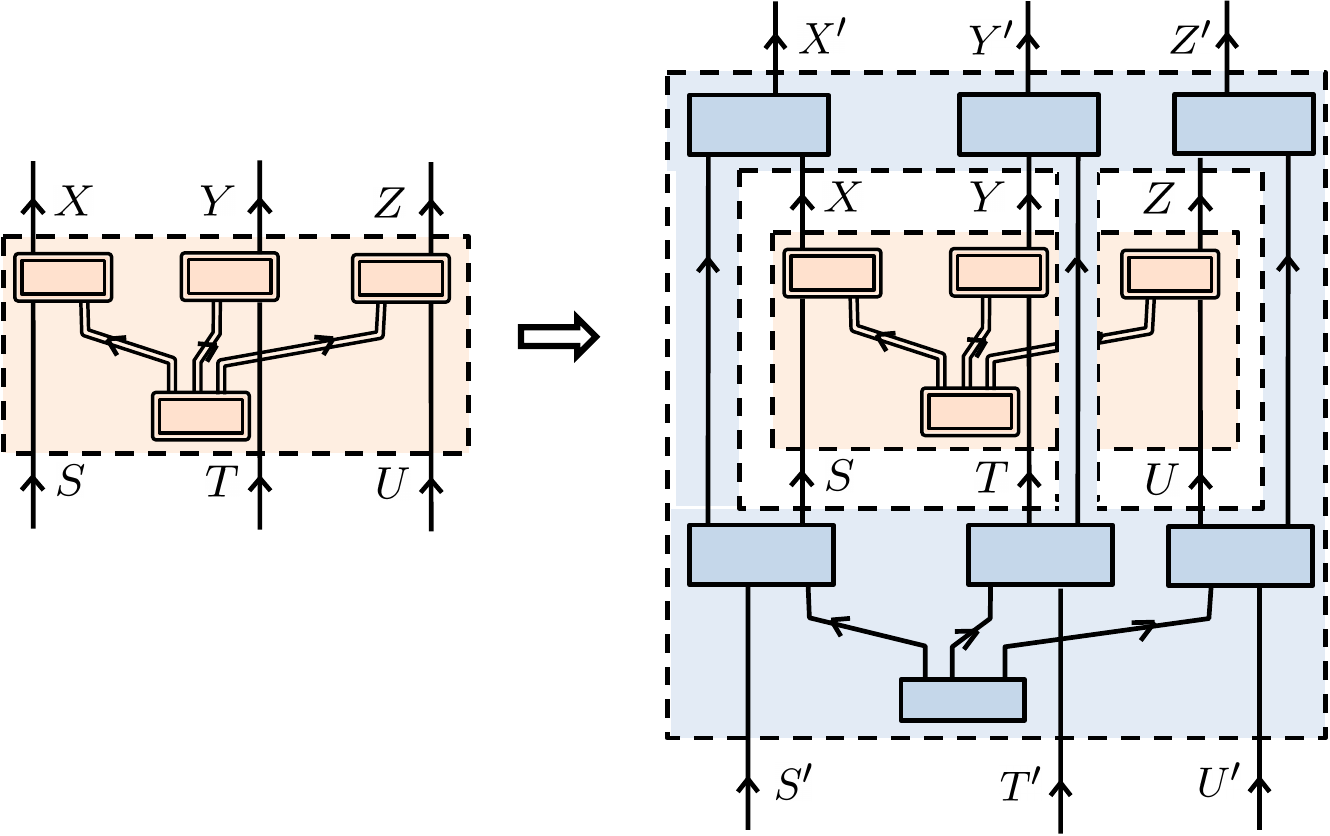}
 \caption{The free operations for a tripartite Bell scenario, $P_{X' Y'Z' S T U|XYZ S'T'U'}$, taking a tripartite common-cause box $P_{XYZ|STU}$ to a new such box $P_{X'Y'Z'|S'T'U'}$.} \label{fig:FreeOpsTripartiteBellScenario}
\end{figure}

 \begin{figure}[hbt]
 \centering
 \includegraphics[scale=0.4]{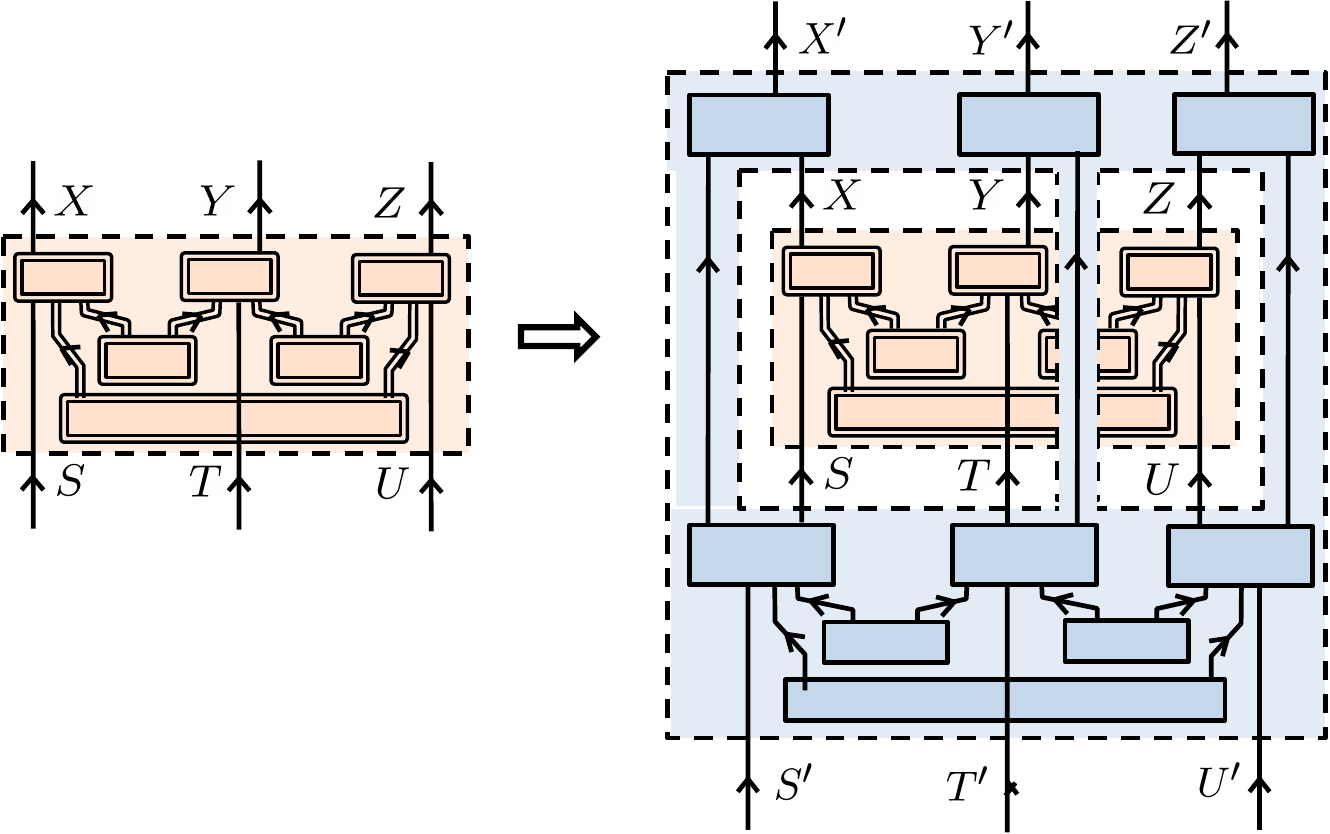}
 \caption{The free operations for the so-called triangle-with-settings scenario, $P_{X' Y'Z' S T U|XYZ S'T'U'}$, taking a triangle-with-settings box $P_{XYZ|STU}$ to a new such box $P_{X'Y'Z'|S'T'U'}$.} \label{fig:FreeOpsTriangleWithSettingsScenario}
\end{figure}  
As we see, the form of the GPT-realizable boxes in the tripartite Bell scenario differs from the form of the GPT-realizable boxes in the triangle-with-settings scenario.  Similarly for the form of the classically realizable boxes.
These differences have consequences when one compares the strictly operational paradigm with our causal modelling paradigm, as we argue next. 

\bigskip

We begin by considering what each paradigm implies for the definitions of the free and enveloping sets of resources for each scenario.

For the tripartite Bell scenario, the definitions of both the enveloping process theory and the free subtheory of processes that are natural from the perspective of the causal modelling paradigm can also be expressed in a way that is natural within the strictly operational paradigm.  Specifically, the boxes in the enveloping theory, which we take to be those that are realizable in a GPT causal model of this scenario (formalized in Eq.~\eqref{GPT3partCCbox}), can also be characterized as those that are nonsignalling between the wings.  Similarly, the boxes in the free subtheory, which we take to be those that are realizable in a classical causal model of this scenario (formalized in Eq.~\eqref{Classical3partCCbox}), can also be characterized as those that are mixtures of deterministic boxes which are nonsignalling between the wings.  

For the triangle-with-settings scenario, on the other hand,  the set of boxes realizable in a GPT causal model for that scenario (formalized in Eq.~\eqref{GPTtriangleCCbox}),
 is a strict subset of the boxes that are nonsignalling between the wings, and the set of boxes that are realizable in a classical causal model for that scenario (formalized in Eq.~\eqref{ClassicaltriangleCCbox})
is a strict subset of the set of boxes that are mixtures of deterministic boxes that are nonsignalling between the wings.  In both the enveloping theory and the free subtheory, the set of boxes is characterized via nontrivial {\em inequalities} in addition to merely the {\em equalities} that represent the no-signalling constraints.  See Ref.~\cite{Wolfe2016inflation} for a discussion of these inequalities in the special case of trivial setting variables.  
Consequently, within the causal modelling paradigm, the resource theory associated to the triangle-with-settings scenario and the resource theory associated to the tripartite Bell scenario differ in both the choice of enveloping theory and free subhteory.  Within the strictly operational paradigm, however, it is unclear whether there is any natural way to pick out the enveloping theory and free subtheory that the causal modelling paradigm dictates for the triangle-with-settings scenario because it is unclear whether there is any natural way  of picking these out by referring merely to the input-output functionality of the boxes.  

Now, we shift our attention to what each paradigm implies for the definitions of the \emph{free operations} in each scenario.  We will show that the definitions that are natural within the causal modelling paradigm 
 cannot be easily motivated within the strictly operational paradigm.
 
  The free operations  prescribed by the causal modelling paradigm for  the tripartite Bell scenario are depicted in Fig.~\ref{fig:FreeOpsTripartiteBellScenario}.  They are of the form \begin{align} \label{FreeOps3partBell}
&P_{X' Y' Z' S T U|XY Z S'T' U'}(x' y' z' stu|xyz s't'u')\nonumber\\
&= \sum_{\lambda}P_{X' S|X S' \Lambda }(x' s|x s' \lambda) P_{Y'T|YT'\Lambda}(y't|yt' \lambda) P_{Z'U|ZU'\Lambda}(z'u|zu' \lambda) P_{\Lambda}(\lambda),
\end{align}
which is clearly a convex set.  This, we believe, is the appropriate definition of local operations and shared randomness for three parties.  

This scenario does not show much difference with what would be natural in the strictly operational paradigm, because one can motivate taking this set of operations to be free on the grounds that they take nonsignalling boxes to nonsignalling boxes (even though, as in the case with the bipartite Bell scenario, the set of WPICC operations between the three wings can also be motivated in this way).

It is the free operations in the triangle-with-settings scenario that really distinguishes the causal modelling paradigm from the strictly operational paradigm. 

The free operations prescribed by the causal modelling paradigm for the triangle-with-settings scenario are depicted in Fig.~\ref{fig:FreeOpsTriangleWithSettingsScenario}.  They are of the form
\begin{align}
&P_{X' Y' Z' S T U|XY Z S'T' U'}(x' y' z' stu|xyz s't'u')\nonumber\\
&=\sum_{\lambda,\lambda',\lambda''}P_{X' S|X S' \Lambda\Lambda' }(x' s|x s' \lambda\lambda') P_{Y'T|YT'\Lambda'\Lambda''}(y't|yt' \lambda'\lambda'') P_{Z'U|ZU'\Lambda\Lambda''}(z'u|zu' \lambda\lambda'')  \nonumber\\
&\;\;\times P_{\Lambda}(\lambda) P_{\Lambda'}(\lambda') P_{\Lambda''}(\lambda'').
\label{FreeOpsTriangle}
\end{align}
Note that this is {\em not} a convex set.  Furthermore, since a triple of pairwise common causes can be simulated by a triplewise common cause, the free operations defined in Eq.~\eqref{FreeOpsTriangle} are a strict subset of the tripartite LOSR operations defined in Eq.~\eqref{FreeOps3partBell}.  It follows that, just as we saw for the free boxes in the triangle-with-settings scenario, one cannot motivate the free operations defined in Eq.~\eqref{FreeOpsTriangle} by appeal to the no-signalling principle.  And, again just as we noted for the free boxes, it is unclear how such a choice could ever be motivated by a principle that appealed only to the input-output functionality of the operation.

The triangle-with-settings scenario also illustrates why one should not mathematically impose convex closure of the set of free operations, as was done in Refs.~\cite{gallego2016nonlocality}.  Rather, whether or not the set of free operations is convexly closed depends on the causal structure, which specifies precisely how randomness is shared among the parties.  For Bell scenarios, the set of free operations is convex by construction, whereas for other causal structures, such as the triangle-with-settings scenario, it is not.  Mathematically imposing convex closure in the triangle-with-settings scenario would be equivalent to asserting that there was a common cause for all three wings,
which would constitute a change in the causal structure being considered.
 In other words, imposing convexity in an ad-hoc manner contradicts the foundations of the causal modelling paradigm, where it is the causal structure that specifies how randomness is shared among the parties, and consequently specifies whether or not convexity holds.

Note finally that the lack of convexity in general causal structures (such as the triangle-with-settings scenario) implies that  the project of quantifying nonclassicality in these cases will be much more complicated than it was in the Bell scenario.

\FloatBarrier

\bigskip
\section{Proofs}

\subsection{Proof of Proposition~\ref{geom}: closed-form expression for \(M_{\rm NPR}(R)\)} \label{sec:geom}

In this section, we present some arguments that aid in justifying Proposition~\ref{geom}, the proof of which is given at the end of this appendix. Recall Proposition~\ref{geom}:
\begin{customprop}{\ref{geom}}
For any free resource $R$ of type {\twotwotwotwo}, $M_{\rm NPR}(R) =2$.  For any nonfree resource $R$ of type {\twotwotwotwo}, there is a unique $k \in \{0,\dots,7\}$ for which $\operatorname{CHSH}_k(R) > 2$.  
Within this region, if $R\in\boldsymbol{C}_{\textup{NPR},k}$, then we have simply $M_{\textup{NPR}}(R) = \CHSH_k(R)$.  If, on the other hand, $R\not\in\boldsymbol{C}_{\textup{NPR},k}$, we have 
$$M_{\textup{NPR}}(R)=2\alpha{+}2,$$
where $\alpha$ is the value appearing in the decomposition $R=\;\gamma\, L_{R}^{\rm bb}+(1{-}\gamma) C_k(\alpha)$, where $C_k(\alpha) \in \boldsymbol{C}_{\textup{NPR},k}$, $L_{R}^{\rm bb} \in \boldsymbol{L}_k^{\rm bb}$ and $\gamma \in [0,1]$.
This value of $\alpha$ is unambiguous because there exists a {\em unique} resource $L_{R}^{\rm bb} \in \boldsymbol{L}_k^{\rm bb}$ and a unique choice of $\gamma \in [0,1]$ and of $\alpha \in [0,1]$ such that $R=\;\gamma\, L_{R}^{\rm bb}+(1{-}\gamma) C_k(\alpha)$. 
\end{customprop}

\noindent The (unique) relevant decomposition is shown in Fig.~\ref{fig:TwoDifferentDecompositions} (for the case where $k=0$).

We first demonstrate the equivalence of three statements
which pertain to the value of $M_{\rm NPR}(R)$ for the subset of resources that satisfy $\operatorname{CHSH}(R)\ge 2$:

\begin{prop} \label{geom2}
For any resource $R$ of type {\twotwotwotwo} such that $\operatorname{CHSH}(R)\ge 2$, the following definitions are equivalent to $M_{\rm NPR}(R)$:
\begin{subequations}\begin{align}\label{eq:rawdefinitionA}
\min\limits_{0{\leq}\alpha{\leq}1} &\left\{ \CHSH(C(\alpha))\text{ such that } \; C(\alpha)\conv R \right\}\!,\\
\label{eq:geometricconvercrit}
\min\limits_{\alpha} &\left\{\CHSH(C(\alpha))\text{ such that  }\;
{\exists\, \gamma\geq 0} \text{ and } {\exists\, L_{R}^{\rm b}\in\boldsymbol{L}^{\rm b}}
\text{ with }{R=\;\gamma\, L_{R}^{\rm b}+(1{-}\gamma) C(\alpha)}\right\}\!,\\ 
\label{eq:boundrybounbryok}
&\Bigg\{\begin{array}{rl}
\text{if }R\in\boldsymbol{C}_{\textup{NPR}} :&\CHSH(R),\text{ else}\\
\text{if }R\not\in\boldsymbol{C}_{\textup{NPR}} :&2\alpha{+}2,\text{ where }\alpha,\,\gamma\geq 0,\text{ and }L_R^{\rm bb}\in\boldsymbol{L}^{\rm bb} \text{ are \emph{all} unique}\\
&\quad\text{in the decomposition }{R=\;\gamma\, L_{R}^{\rm bb}+(1{-}\gamma) C(\alpha).}
\end{array}\quad 
\end{align}\end{subequations}
\end{prop}

\begin{proof}[Proof of Eq.~\eqref{eq:rawdefinitionA}]
\quad\\
Eq.~\eqref{eq:rawdefinitionA} is directly equivalent to the definition of the $M_{\rm NPR}$ monotone given in Eq.~\eqref{Malphadefn}. Hence, we take that as our starting point, and prove the implications of the subequations in Proposition~\ref{geom2}.
\end{proof}

\begin{proof}[Proof that Eq.~\eqref{eq:rawdefinitionA} $\Leftrightarrow$  Eq.~\eqref{eq:geometricconvercrit}]\quad\\
Section~\ref{polything} guarantees that $C(\alpha)\conv R$ if and only if we can generate R by convex mixtures of $C(\alpha)$ with the images of $C(\alpha)$ under {\LDO} operations.
For any $R\notin \boldsymbol{C}_{\textup{NPR}}$ such that $\operatorname{CHSH}(R)\ge 2$, we simplify the situation by proving that if $R$ can be generated by mixing $C(\alpha)$ with its images under {\LDO}, then $R$ can alternatively be generated by mixing $C(\alpha)$ with a local point which saturates the CHSH inequality; namely, a point in $\boldsymbol{L}^{\rm b}$, as stated in Eq.~\eqref{eq:geometricconvercrit}. To prove this, it is useful to define the notion of a \term{screening-off inequality}.

\begin{samepage}\begin{defn} 
The inequality $f(R)\geq b$ is said to \term{screen-off} the fixed-type set of resources which \emph{satisfy} it, i.e., $\left\{ R: f(R)\geq b\,, [R]=\xyst\,\right\}$, if the fixed-type set of resources which \emph{saturate} it is a free set, i.e., if $\left\{ R: f(R)= b\,, [R]=\xyst\,\right\}$ consists only of classically realizable common-cause boxes.
\end{defn}\end{samepage}
For example, the inequality $\CHSH(R)\geq 2$ screens-off the set $\left\{ R: \CHSH(R)\geq 2\,, [R]=\twotwotwotwo\,\right\}$ since $\left\{ R: \CHSH(R)= 2\,, [R]=\twotwotwotwo\,\right\} \subset \boldsymbol{S}^{\textup{free}}_{{\twotwotwotwo}}$. 

Screening-off inequalities are useful when making statements about resource convertibility, as follows. Consider the case where we ask whether $R_2 \conv R_1$: if $R_1$ lies inside some screened-off region, then, given Proposition~\ref{lem:belstheorem}, $R_1 \in \mathbfcal{P}^{\LOSR}_{[R_1]}(R_2)$ if and only if $R_1$ is in the convex hull of \emph{those images of $R_2$ under {\LDO} inside the screened-off region,} together with the boundary (where the inequality is saturated). Formally, if $f(R)\geq b$ is a screening-off inequality for resources of type $[R_1]$, then, given Proposition~\ref{lem:belstheorem}, 
\begin{align*}
R_2 \conv R_1\text{ iff }
{\exists \,L^{\rm b}}\text{ such that } f(L^{\rm b})=b\;\text{and}\;
R_1 \in {\operatorname{ConvexHull}}\big(\,
L^{\rm b}
,\; \underbrace{\mathbfcal{V}^{\LDO}_{[R_1]}(R_2)\bigcap\left\{R: f(R)> b\right\}}_{\substack{\text{the {\LDO} images of } R_2 \\\text{interior to screened-off region}}}
\big).
\end{align*}

Since $\CHSH(R)\geq 2$ is a screening-off inequality whose saturation-boundary is given by  $\boldsymbol{L}^{\rm b}$, and since the \emph{only} image in $\mathbfcal{V}_{{\twotwotwotwo}}^{\LDO}(C(\alpha))$ which violates the CHSH inequality is $C(\alpha)$ itself, it follows that 
\begin{align}
C(\alpha)\conv R \;\text{ if and only if }\;{\exists\, \gamma\geq 0} \text{ and } {\exists\, L_{R}^{\rm b}\in\boldsymbol{L}^{\rm b}}\; \text{ such that }{R=\;\gamma\, L_{R}^{\rm b}+(1{-}\gamma) C(\alpha)}.
\end{align}
The equivalence Eq.~\eqref{eq:rawdefinitionA} $\Leftrightarrow$  Eq.~\eqref{eq:geometricconvercrit} follows.
As a final comment, notice that this characterization of convertibility in terms of the existence of a geometric decomposition involves arbitrary points which saturate the CHSH inequality, and there are typically many such decompositions.
\end{proof}

\begin{proof}[Proof that Eq.~\eqref{eq:geometricconvercrit} $\Leftrightarrow$ Eq.~\eqref{eq:boundrybounbryok}.]\quad\\
Recall that Eq.~\eqref{eq:geometricconvercrit} involves a minimization under the constraint that $\alpha$ is such that ${R=\;\gamma\, L_{R}^{\rm b}+(1{-}\gamma) C(\alpha)}$.
We can formally recast it as a constrained optimization problem, as follows:
\begin{align*}
\min\limits_{0{\leq}\alpha{\leq}1}&	\quad{\CHSH}\big(C(\alpha)\big)\quad\\&\text{such that}\quad L_{R}^{\rm b}\coloneqq \frac{R-(1{-}\gamma) C(\alpha)}{\gamma},\quad\text{under the constraint that}\\
&\text{all conditional probabilities in the expression of }\vec{L}_{R}^{\rm b}\text{ are nonnegative},\\
&\text{where }
\gamma \text{ is an implicit function of } \alpha \text{ according to}\\
&\frac{\CHSH(R)-(1{-}\gamma)\,{\CHSH}{\big(C(\alpha)\big)}}{\gamma}=2,\quad\text{as implied by the fact that}\quad{\CHSH}(L_{R}^{\rm b})=2.
\end{align*}
Essentially, this is a constrained optimization problem with a linear objective subject to one nonlinear constraint; namely, that the \emph{smallest} conditional probability in the expression $\vec{L}_{R}^{\rm b}$\footnote{$\vec{L}_{R}^{\rm b}$ denotes the \emph{representation} of $L_{R}^{\rm b}$ as a vector whose components are the conditional probabilities $\{ P_{XY|ST}(xy|st): x,y,s,t \in \{0,1\}\}$.} must be nonnegative. For such optimization problems, it is always the case that the objective is maximized when the constraint is not merely \emph{satisfied} but \emph{saturated}. Put another way, the set of {\em achievable} $\alpha$ arise from points $L_{R}^{\rm b}$ wherein all conditional probabilities are nonnegative, but the \emph{optimal} $\alpha$ arises for some unique $L_{R}^{\rm b}=L_{R}^{\rm bb}$ where the smallest conditional probability in $\vec{L}_{R}^{\rm bb}$ is precisely zero.
\end{proof}

\begin{proof}[Proof of Proposition~\ref{geom}]
\quad \\ Proposition~\ref{geom} for {\em arbitrary resources} follows from Proposition~\ref{geom2} by the symmetry noted in Proposition~\ref{prop:symmetrypartitioning}{\ref{prop:CHSHorbit}}. Namely, the argument can be repeated unchanged in each of the eight spaces of resources generated by the images under {\LSO} of the set of resources satisfying $\operatorname{CHSH}(R)\ge 2$. Together with the trivial observation that free resources (which do not violate any of the eight CHSH inequalities) always have value of $M_{\textup{NPR}}$ equal to 2, Proposition~\ref{geom} follows.
\end{proof}

\bigskip
\subsection{Proof of Proposition~\ref{prop:whencomplete}: when the two monotones are complete} \label{proofprop2}

We now prove Proposition~\ref{prop:whencomplete}, repeated here:
\begin{customprop}{\ref{prop:whencomplete}}
Consider a two-parameter family ${R(\alpha,\!\gamma)=\;\gamma\, L^{\rm bb}_{\star}+(1{-}\gamma)C(\alpha)}$, marked by any fixed $L^{\rm bb}_{\star}$---that is, a point which saturates the CHSH inequality and is on the boundary of the no-signaling set. The pair of monotones $\{M_{\CHSH},M_{\rm NPR}\}$ is complete relative to such a family of resources. 
\end{customprop}

\begin{figure} [htb!]
\begin{center}
\subfigure[\label{fig:DownwardClosureTrianglePlot}]
{
\centering
\begin{tikzpicture}[scale=0.7]
\path[draw, ultra thick] (-30:4) -- (90:4) -- (210:4) -- cycle;

\node[draw,shape=circle,fill,scale=.4, label={[label distance=0.1cm]0:\large{${L_{\star}^{\rm bb}}$}}] at (-30:4) {};
\node[draw,shape=circle,fill,scale=.4, label={[label distance=0.1cm]90:\large{${R_{\textup{PR}}}$}}] at (90:4) {};
\node[draw,shape=circle,fill,scale=.4, label={[label distance=0.1cm]180:\large{${L_{\textup{NPR}}^{\rm b}}$}}] at (210:4) {};

\path[name path=ls] (210:4) -- (90:4);
\path[name path=chsh] (-4,0) -- (4,0);
\path[name intersections={of=chsh and ls, by=lp}];
\path[name path=nrp, shift=(-30:4)] (0:0) -- (150:7);
\path[name intersections={of=nrp and ls, by=lup}];
  
\path [name intersections={of=chsh and nrp, by=R}];
\draw[thick] (lp) -- (R);
\draw[thick] (-30:4) -- (R);
\node [draw,shape=diamond,fill=red,scale=.4, label={[label distance=-0.1cm]90:$\qquad{R}(\alpha_1,\gamma_1)$}] at (R) {};

\node [draw,shape=star,fill=green,scale=0.4, label={[label distance=0.1mm]180:${\begin{array}{r}
R(\alpha_2,0)\\
={R}\big(\alpha_1(1{-}\gamma_1),0\big)
\end{array}}$}] at (lp) {};

\path[name path=rs] (-30:4) -- (90:4);
\path[name intersections={of=chsh and rs, by=rp}];

\begin{pgfonlayer}{myback}
\filldraw[draw=black, fill=jflyBlue, opacity=1] (-30:4) -- (R) -- (lp) -- (210:4) -- cycle;
\end{pgfonlayer}

\begin{scope}[on background layer]
\path[clip] (-30:4) -- (90:4) -- (210:4) -- cycle;
\foreach \y in{-3,-2.5,...,0}
  \draw[name path=c.\y, color=jflyBlue] (-4,\y) -- (4,\y);
  
\foreach \y in{150,155,...,180}
  \draw[name path=s.\y, shift=(-30:4), color=jflyVermillion] (0:0) -- (\y:7);
\end{scope} 

\end{tikzpicture}}
\subfigure[\label{fig:DownwardClosureMonotoneAxes}]
{
\centering
\begin{tikzpicture}[scale=1]

\path[draw, ultra thick] (0,0) -- (0,4) --  (4,4) -- (4,0) ;
\path[draw, ultra thick, dashed] (0,0) -- (4,0) ;
\node [draw,shape=circle,fill,scale=.4, label={[label distance=-0.1cm]225:${2}$}] at (0,0) {};
\node [draw=none, scale=.4, label={[label distance=-0cm]-90:${4}$}] at (4,0) {};
\node [draw=none, scale=.4, label={[label distance=-0cm]180:${4}$}] at (0,4) {};
\node [draw,shape=circle,fill,scale=.4, label={[label distance=-0.1cm]45:${R_{PR}}$}] at (4,4) {};
\node at (2,-0.5) {${M_{\textup{NPR}}}$};
\node at (-1, 2) {${M_{\textup{CHSH}}}$};

\path[pattern=north west lines] (0,0) -- (0,4) -- (4,4)--cycle;
\path[pattern=north east lines] (0,0) -- (0,4) -- (4,4)--cycle;

\draw[thick] (1.25,1.25) -- (2,1.25);
\draw[thick] (2,0) -- (2,1.25);
\path[draw, ultra thick] (0,0) -- (4,4);  

\node [draw,shape=diamond,fill=red,scale=.4, label={[label distance=-0.1cm]45:${R(\alpha_1,\gamma_1)}$}] (R) at (2,1.25) {};
\node [draw,shape=star,fill=green,scale=.4] at (1.25,1.25) {};

\begin{pgfonlayer}{myback}
\path[clip] (0,0) -- (4,0) -- (4,4) -- cycle;
\fill[fill=jflyBlue, opacity=1] (R) rectangle (0,0);
\end{pgfonlayer}

\begin{scope}[on background layer]
\path[clip] (0,0) -- (4,0) -- (4,4) -- cycle;
\foreach \y in{0,0.3125,...,1.25}
  \draw[name path=c.\y, color=jflyBlue] (\y,\y) -- (4,\y);
\foreach \y in{1.25,1.625,2}
  \draw[name path=c.\y, color=jflyVermillion] (\y,0) -- (\y,\y);;
\end{scope}

\end{tikzpicture}}
\end{center}
\caption[]{
\subref{fig:completeTrianglePlot} and \subref{fig:completeMonotoneAxes} provide the same pair of depictions of the two-parameter family of resources $\boldsymbol{S}_{\twotwotwotwo}^{L^{\rm bb}_{\star}}$ as were introduced in Fig.~\ref{fig:twoM}.
We consider a two-parameter family of resources.  A generic such resource,
specified by $\alpha_1$ and $\gamma_1$, is marked by a red diamond.
  Also depicted are some of the level curves of the two monotones $M_{\CHSH}$ and $M_{\rm NPR}$. The solid dark blue region denotes the set of all resources within this family which have values for both monotones less than or equal to their values for  $R(\alpha_1,\gamma_1)$.
  To prove Proposition~\ref{prop:whencomplete}, one must show that $R(\alpha_1,\gamma_1)$ can be converted to any resource in the solid blue region. 
The critical step in this proof
 is the demonstration that it is possible to convert any resource to one lying on the line connecting $R_{\textup{PR}}$ and $L_{\textup{NPR}}^{\rm b}$ without changing the value of $M_{\CHSH}$. Graphically, this corresponds to converting the generic resource $R(\alpha_1,\gamma_1)$ to the resource $R(\alpha_2,0)$ marked by a green star. 
}
\label{fig:DownwardClosure} 
\end{figure}
\begin{proof} 
A set of monotones is complete relative to a family of resources if and only if every candidate conversion among resources in the family which is {\em not} ruled out by any of the monotones in the set is in fact possible for free, as per Eq.~\eqref{completeset}. 

In Fig.~\ref{fig:DownwardClosureTrianglePlot}, we depict in blue the set of candidate conversions (from a generic resource $R(\alpha_1,\gamma_1)$ to another resource in the family) which are not ruled out by $\{M_{\CHSH},M_{\rm NPR}\}$; namely, the blue shaded region contains all resources which have a value for each of the two monotones that is equal to or lower than that of $R(\alpha_1,\gamma_1)$. To prove the proposition, we argue that $R(\alpha_1,\gamma_1)$ can indeed be converted to any resource in the blue region. By convexity, it suffices to prove that $R(\alpha_1,\gamma_1)$ can be converted to each of the four extreme points of the blue region.
Since $L^{\rm bb}_{\star}$ and $L^{\rm b}_{\rm NPR}$ are free resources, $R(\alpha_1,\gamma_1)$ can freely be converted to either of them, and the resource $R(\alpha_1,\gamma_1)$ can obviously be `converted' to itself, as the identity is free. Our proof, therefore, focuses on demonstrating that $R(\alpha_1,\gamma_1)$ can indeed be converted to the fourth extreme point $R(\alpha_2,0)$, shown as a green star. We now give the explicit free operation which takes a generic initial resource $R(\alpha_1,\gamma_1)$ and projects it onto the chain leftwards in the two-dimensional coordinate system of Fig.~\ref{fig:DownwardClosureTrianglePlot}, i.e., to the target resource ${R}(\alpha_2,0)$, where $\alpha_2=\alpha_1(1-\gamma_1)$\footnote{The particular relation between $\alpha_2$ and $(\alpha_1, \gamma_1)$ follows from Eq.~\eqref{MCHSHfam}, by noticing that $R(\alpha_2,0)$ and $R(\alpha_1,\gamma_1)$ must have the same value for $M_{\CHSH}$.}. 

We denote the free operation which enacts this conversion by $\tau_{\textrm{erase-}\gamma}$; it is the operation which projects any resource into the subspace of resources that are invariant under
 the $G_{456}$ subgroup of ${\LSO}_{\twotwotwotwo}$
($G_{456}$ is defined in Proposition~\ref{prop:symmetrypartitioning}{\ref{prop:CHSHstabilizer}} on page~\pageref{prop:CHSHstabilizer}), i.e., onto the chain $\boldsymbol{C}_{\textup{NPR}}$.\footnote{Equivalently, $\tau_{\textrm{erase-}\gamma}$ is the Reynold's operator of the subgroup $G_{456}$.} This operation is indeed free, as it can be constructed by a uniform mixture of all the elements of $G_{456}$, each of which is free. Recall that $G_{456}$ is the subgroup of $\LSO_{\twotwotwotwo}$ which stabilizes $\CHSH_0$, and therefore clearly does not modify the value of the $M_{\CHSH}$ monotone. 

It remains only to show that the $G_{456}$-invariant subspace of resources within the set of
  all {\twotwotwotwo}-type resources for which $\CHSH(R)\geq 2$ is the chain $\boldsymbol{C}_{\textup{NPR}}$, i.e., the line of points between $R_{\textup{PR}}$ and $L_{\textup{NPR}}^{\rm b}$. This is evident by confirming that $\tau_{\textrm{erase-}\gamma}$ leaves $R_{\textup{PR}}$ invariant, but maps each of the 8 deterministic ${\CHSH}$-saturating boxes to $L_{\textup{NPR}}^{\rm b}$. Those $1{+}8$ resources are the extreme points of the set of all {\twotwotwotwo}-type resources such that $\CHSH(R)\geq 2$; since the extreme points map to the line under the action of $\tau_{\textrm{erase-}\gamma}$, by convex linearity it follows that the chain is the only space invariant under $G_{456}$ within the two-parameter family.
\end{proof}

\bigskip
\subsection{Proof of Proposition~\ref{prop:zerodclasses}: all nonfree resources of type {\twotwotwotwo} 
are orbital} \label{proofprop3}

We now prove Proposition~\ref{prop:zerodclasses}, 
repeated here:
\begin{customprop}{\ref{prop:zerodclasses}}\label{prop:zerodclassesCOPY}
All nonfree resources of type {\twotwotwotwo} 
are orbital.
\end{customprop}

Before presenting the proof, we introduce some additional concepts and a few lemmas on which our proof relies.
Throughout the following, we are focused on sets of resources of fixed type, and on type-preserving operations. Hence, we use slightly abbreviated notation; e.g. $\mathbfcal{V}^{\textup{LDO}}(R)$ is used as shorthand for $\mathbfcal{V}_{[R]}^{\textup{LDO}}(R)$, and so on.

\begin{defn}
The set of \term{local deterministic type-preserving nonsymmetry operations}, denoted \term{LDTNO}, contains all type-preserving operations in LDO which are not in {\LSO}.
 The image of a resource $R$ under {\LDTNO} constitutes a discrete set of resources denoted $\mathbfcal{V}^{\textup{LDTNO}}(R)$. Moreover, we use \term{HullLDTNO$(R)$} to indicate the set of all resources in the {\em convex hull} of the \emph{image} of a resource $R$ under  {\LDTNO}, i.e., $\mathbf{HullLDTNO}(R)\coloneqq\operatorname{ConvexHull}{\left(\mathbfcal{V}^{\textup{LDTNO}}(R)\right)}$. 
\end{defn}

\begin{defn}
A resource $R$ is said to be \term{sensitive} if every element of {\LDTNO}
 removes $R$ from its equivalence class; i.e., if for all ${\tau{\in}\LDTNO}$ it holds that ${\tau\circ R \nconv R}$. Equivalently, a resource $R$ is sensitive if and only if $R$ is not in the convex hull of its images under {\LDTNO}, i.e., if $R\not\in\mathbf{HullLDTNO}(R)$. A \emph{set} of resources is called sensitive if every resource in the set is sensitive. 
\end{defn}

We bring up the property of sensitivity because (i) it is straightforward to test if a given resource is sensitive or not by means of a linear program, and (ii) eventually we will argue that if a resource is sensitive, then it is also orbital. 
Furthermore, we now prove that sensitive resources never appear in isolation.
That is, a single sensitive resource can be used to construct a set of sensitive resources, as follows:
\begin{lemma}\label{lem:onesensitiveimpliessensitiveset}
For any resource $R$, every resource $R'$ which is below $R$ in the pre-order and which cannot be generated from $R$ by mixtures of {\LDTNO} operations is Formally: the set of resources 
$\boldsymbol{S}_{\rm sens}^R \coloneqq \mathbfcal{P}^{\textup{LOSR}}(R) \setminus \mathbf{HullLDTNO}(R)$ is  always sensitive.
\end{lemma}

\begin{proof}

First, note two related, useful facts:\\
(1) The composition of any deterministic operation (invertible or not) followed by some deterministic \emph{nonsymmetry} operation is precisely some (other) deterministic \emph{nonsymmetry} operation. Formally, if $\tau_{\LDTNO} \in \LDTNO$ and $\tau_{\LDO} \in \LDO$, and defining $\tau' \coloneqq \tau_{\LDTNO} \circ \tau_{\LDO}$, then $\tau' \in \LDTNO$. A consequence of this is that the entire set $\mathbfcal{P}^{\textup{LOSR}}(R)$ is mapped to the set $\mathbf{HullLDTNO}(R)$ under $\LDTNO$ and convex mixtures thereof. To see this, recall that the image of \emph{any} convex set of resources under \emph{any} convex set of operations is identically the convex hull of the images of the \emph{extremal} resources under the \emph{extremal} operations (in the respective sets). We use this fact to effectively replace $\mathbfcal{P}^{\textup{LOSR}}(R)$ with $\mathbfcal{V}^{\textup{LDO}}(R)$ and to replace convex mixtures of {\LDTNO} with {\LDTNO} itself, without loss of generality. In summary: $\mathbf{HullLDTNO}(\mathbfcal{P}^{\textup{LOSR}}(R))=\mathbf{HullLDTNO}(R)$ by virtue of the fact that $\LDTNO \circ \LDO = \LDTNO$.\\
(2) The composition of any deterministic \emph{nonsymmetry} operation followed by some deterministic operation (invertible or not) is some (other) deterministic \emph{nonsymmetry} operation. Formally, if $\tau_{1-\LDTNO} \in \LDTNO$ and $\tau_{2-\LDO} \in \LDO$, and defining $\tau_4 \coloneqq \tau_{2-\LDO}\circ \tau_{1-\LDTNO}$, then $\tau_4 \in \LDTNO $. A consequence of this is that the entire set $\mathbf{HullLDTNO}(R)$ is mapped to itself under $\LOSR$. To see this, we reuse the shortcut of considering only extremal resources and extremal operations. Specializing to our objects of interest, we effectively replace the operations-set $\LOSR$ by its extremal operations --- namely $\LDO$ --- and the resources-set $\mathbf{HullLDTNO}(R)$ by $\mathbfcal{V}^{\textup{LDTNO}}(R)$ without loss of generality. In summary: $\mathbfcal{P}^{\textup{LOSR}}(\mathbf{HullLDTNO}(R))=\mathbf{HullLDTNO}(R)$ by virtue of the fact that $\LDO \circ \LDTNO = \LDTNO$.

Now we are in position to prove Lemma~\ref{lem:onesensitiveimpliessensitiveset}. 
The set of resources $R'$ below $R$ in the partial order is identically $\mathbfcal{P}^{\textup{LOSR}}(R)$. The set of resources which can be generated from $R$ by mixtures of deterministic nonsymmetry operations is identically $\mathbf{HullLDTNO}(R)$. So, a resource $R'$ is below $R$ in the partial order \emph{and} cannot be generated from $R$ by mixtures of deterministic nonsymmetry operations if and only if $R'\in\mathbfcal{P}^{\textup{LOSR}}(R) \setminus \mathbf{HullLDTNO}(R) =: \boldsymbol{S}_{\rm sens}^R$.

Now, consider any ${\tau\in\LDTNO}$ and any $R'\in \boldsymbol{S}_{\rm sens}^R$, and define $R''\coloneqq \tau\circ R'$. Since we have established that the entirety of $\mathbfcal{P}^{\textup{LOSR}}(R)$ is mapped to $\mathbf{HullLDTNO}(R)$ under {\LDTNO}, it follows that $R''\in\mathbf{HullLDTNO}(R)$. However, since we have also established that the entirety of $\mathbf{HullLDTNO}(R)$ is mapped \emph{only} to itself under {\LOSR}, and since $R' \notin \mathbf{HullLDTNO}(R)$, it further follows that $R''\nconv R'$. Evidently, any $R' \in \boldsymbol{S}_{\rm sens}^R$ is removed from its equivalence class by every deterministic nonsymmetry operation, i.e.,  $\boldsymbol{S}_{\rm sens}^R$ is sensitive. This proves the Lemma.
\end{proof}

Note that Lemma~\ref{lem:onesensitiveimpliessensitiveset} implies that if $R$ is sensitive, and $R'$ is equivalent to $R$, then $R'$ is also sensitive.

\begin{lemma}\label{lem:townorm}
If a resource is sensitive, then it is also orbital. That is,
if two sensitive resources are interconvertible under type-preserving {\LOSR}, then they are interconvertible under {\LSO}.
 \end{lemma}
\begin{proof}
Let $R$ and $R'$ be distinct sensitive resources that are interconvertible under type-preserving {\LOSR}, i.e., $R\neq R'$ but $R\interconv R'$ and $[R]=[R']$. Any operation which preserves the equivalence class of a sensitive resource can be expressed as a convex combination of elements of {\LSO}.
The assumption of sensitivity thus dictates that $R'$ is in the convex hull of the images of $R$ under {\LSO},
and vice versa. We proceed to show that this sort of relationship must imply that $R\in\mathbfcal{V}^{\LSO}(R')$ and $R'\in\mathbfcal{V}^{\LSO}(R)$, that is, that $R$ and $R'$ are {\LSO}-equivalent.\footnote{The fact that $R'$ is in the convex hull of the images of $R$ under permutations of $R$'s probabilities is equivalent to stating that $R$ \emph{vector majorizes} $R'$. The relationship is reflexive, however. Readers familiar with vector majorization may recall that two vectors are equivalent under the majorization order
 if and only if they are related by some reordering, i.e., a (not necessarily physical) symmetry operation.} 

This can be seen by recognizing that the 2-norm is a convex function invariant under {\LSO}, meaning 
$R' \in \operatorname{ConvexHull}{\left(\mathcal{V}^{\LSO}(R)\right)}$
 implies $\norm{\vec{R'}}^2 \leq  \norm{\vec{R}}^2$.\footnote{Recall that $\vec{R}$ is shorthand for the \emph{representation} of the resource in terms of a real-valued vector consisting of all possible conditional probabilities, i.e., $\vec{R}=\left\{P_{XY|ST}(xy|st):\,x,y,s,t\in \{0,1\} \right\}$.}  By symmetry under exchange of $R$ and $R'$, it holds that $\norm{\vec{R}}^2 \leq  \norm{\vec{R'}}^2$, and hence $\norm{\vec{R'}}^2 = \norm{\vec{R}}^2$. 
The 2-norm, moreover, \emph{strictly decreases} under nontrivial stochastic mixing;\footnote{Consider the hypersphere consisting of all resources with 2-norm in common with $\vec{R}$. All the images of $R$ under {\LSO} lie on the surface of this hypersphere. \emph{Stochastic} mixing of symmetry operations (applied to $R$) is equivalent to convexly combining different points from the surface of the hypersphere. Any convex combination of points from the surface of a hypersphere results in a final point strictly interior to the sphere. Strictly interior points are closer to the center, in precisely the sense of having a strictly smaller 2-norm.} hence all interconversions between equivalent sensitive resources must be mediated by   \emph{deterministic} symmetries. 
Formally: If ${R' = \sum_{i=1}^n w_i \;\pi_i\circ R}$ and ${R = \sum_{i=1}^n w'_i \;\pi_i\circ R}$, where ${\sum_{i=1}^n w_i = \sum_{i=1}^n w'_i = 1}$ and ${\{w_1,...,w_n,w'_1,...,w'_n\} \geq  0}$, then $\norm{\vec{R'}}^2 = \norm{\vec{R}}^2$ and $w_i, w'_i \in \{0,1\}$.
\end{proof}

\begin{lemma}\label{lem:allCHSHsensitive}
$R_{\textup{PR}}$ is a sensitive resource, and $\mathbfcal{P}^{\textup{LOSR}}(R_{\textup{PR}}) \setminus \mathbf{HullLDTNO}(R_{\textup{PR}})$ is the entire eight dimensional set of all nonfree resources of type {\twotwotwotwo}.
\end{lemma}
\begin{proof}[Proof by inspection]\quad\\ One can readily verify that $\tau\circ R_{\textup{PR}}\in \mathbf{S}^{\rm free}_{{\twotwotwotwo}}$ for all type-preserving LDTNO operations $\tau$.
\end{proof}

\bigskip
\begin{proof}[Proof of Proposition~\ref{prop:zerodclasses}.]\quad\\
Lemma~\ref{lem:allCHSHsensitive} together with Lemma~\ref{lem:onesensitiveimpliessensitiveset} immediately imply that all nonfree resources of type ${\twotwotwotwo}$ are sensitive. Lemma~\ref{lem:townorm} then directly implies that all these resources are orbital.
\end{proof}

A final comment: consider generalizing Proposition~\ref{prop:zerodclasses} in light of the discussion just given. If one desires to construct an orbital set of resources beyond {\twotwotwotwo}-type, one needs only to find \emph{some} single sensitive resource $R$ of the desired type. From Lemmas~\ref{lem:onesensitiveimpliessensitiveset} and~\ref{lem:townorm}, it then follows that the set of resources $\mathbfcal{P}^{\textup{LOSR}}(R) \setminus \mathbf{HullLDTNO}(R)$ constitutes an orbital set.  It might be the case, for instance, that for any nontrivial choice of resource type, there is 
at least one convexly extremal resource that
 is sensitive,  analogous to how the PR-box is a sensitive resource for type {\twotwotwotwo}.

\subsection{Proof of Proposition~\ref{prop:orbitalstuff}: lower bound on the number of monotones in any complete set
} \label{sec:proofproporbitalstuff}

Recall that a resource is termed orbital if and only if its {\LOSR}-equivalence class of resources of the same type is equal to its {\LSO}-equivalence class. 
We now prove Proposition~\ref{prop:orbitalstuff}, recalled below:
\begin{customprop}{\ref{prop:orbitalstuff}}\label{prop:orbitalstuffCOPY}
For any compact set $\mathbf{S}$ of resources that are all orbital, the intrinsic dimension of the set $\mathbf{S}$ is a lower bound on the cardinality of a complete set of continuous monotones for $\mathbf{S}$ (and for any superset of $\mathbf{S}$).
\end{customprop}

\begin{proof}
The set of local symmetry operations for a given type has finite cardinality, and hence there are a finite number of resources in the {\LSO}-equivalence class of any resource.  For an orbital resource $R$ , this implies that the {\LOSR}-equivalence class of $R$ (over resources of type $[R]$) is precisely $\mathbfcal{V}^{\LSO}(R)$, which is a finite set. If a compact set $\mathbf{S}$ of orbital resources has intrinsic dimension $d$, and the {\LOSR}-equivalence class of every resource in the set is finite and hence zero-dimensional, then it follows that one can find $d$-dimensional compact subsets of resources in $\mathbf{S}$ in which {\em no two resources are equivalent}.\footnote{Not {\em all} compact subsets will necessarily have this property, but {\em some} will. For example, consider a nonfree resource $R_{\textup{asym}}$ which is {\em not} invariant under any {\LSO} operation. Every {\LSO} operation maps such a resource to a distinct resource not in the original neighborhood for a suitably small neighborhood. 
Because {\LSO} operations are linear (and hence continuous), they 
map compact subspaces to compact subspaces. Hence, every {\LSO} operation takes the entire neighborhood of nonfree resources around $R_{\textup{asym}}$ to some other nonfree neighborhood; if the original neighborhood is chosen to be small enough, these two neighborhoods will not intersect. Hence, no two resources in the original neighborhood are interconvertible by {\LSO}. } 

Hence, no two resources in such a subset are assigned the same tuple of values by any complete set of monotones. In other words, a complete set of $n$ continuous monotones maps the subset of resources {\em injectively} to $\mathbb{R}^n$. But this map can only be injective if $n \geq d$, which guarantees that the number of continuous monotones required to identify a resource in the set is at least as large as the intrinsic dimension $d$ of the set $\mathbf{S}$. Finally, note that the number of continuous monotones required to identify a resource in any superset of $\mathbf{S}$ must be at least as large as for the set $\mathbf{S}$ itself, which completes the proof.
\end{proof}

\bigskip
\setlength{\bibsep}{1pt plus 1pt minus 1pt}
\bibliographystyle{apsrev4-2-wolfe}
\nocite{apsrev41Control}
\bibliography{resourcetheory_sparse}

\begin{thebibliography}{112}%
\makeatletter
\providecommand \@ifxundefined [1]{%
 \@ifx{#1\undefined}
}%
\providecommand \@ifnum [1]{%
 \ifnum #1\expandafter \@firstoftwo
 \else \expandafter \@secondoftwo
 \fi
}%
\providecommand \@ifx [1]{%
 \ifx #1\expandafter \@firstoftwo
 \else \expandafter \@secondoftwo
 \fi
}%
\providecommand \natexlab [1]{#1}%
\providecommand \enquote  [1]{``#1''}%
\providecommand \bibnamefont  [1]{#1}%
\providecommand \bibfnamefont [1]{#1}%
\providecommand \citenamefont [1]{#1}%
\providecommand \href@noop [0]{\@secondoftwo}%
\providecommand \href [0]{\begingroup \@sanitize@url \@href}%
\providecommand \@href[1]{\@@startlink{#1}\@@href}%
\providecommand \@@href[1]{\endgroup#1\@@endlink}%
\providecommand \@sanitize@url [0]{\catcode `\\12\catcode `\$12\catcode
  `\&12\catcode `\#12\catcode `\^12\catcode `\_12\catcode `\%12\relax}%
\providecommand \@@startlink[1]{}%
\providecommand \@@endlink[0]{}%
\providecommand \url  [0]{\begingroup\@sanitize@url \@url }%
\providecommand \@url [1]{\endgroup\@href {#1}{\urlprefix }}%
\providecommand \urlprefix  [0]{URL }%
\providecommand \Eprint [0]{\href }%
\providecommand \doibase [0]{https://doi.org/}%
\providecommand \selectlanguage [0]{\@gobble}%
\providecommand \bibinfo  [0]{\@secondoftwo}%
\providecommand \bibfield  [0]{\@secondoftwo}%
\providecommand \translation [1]{[#1]}%
\providecommand \BibitemOpen [0]{}%
\providecommand \bibitemStop [0]{}%
\providecommand \bibitemNoStop [0]{.\EOS\space}%
\providecommand \EOS [0]{\spacefactor3000\relax}%
\providecommand \BibitemShut  [1]{\csname bibitem#1\endcsname}%
\let\auto@bib@innerbib\@empty
\bibitem [{\citenamefont {Bell}(1964)}]{Bell64}%
  \BibitemOpen
  \bibfield  {author} {\bibinfo {author} {\bibfnamefont {J.~S.}\ \bibnamefont
  {Bell}},\ }\bibfield  {title} {\enquote {\bibinfo {title} {{On the
  Einstein-Podolsky-Rosen paradox}},}\ }\href
  {https://doi.org/10.1103/PhysicsPhysiqueFizika.1.195} {\bibfield  {journal}
  {\bibinfo  {journal} {Physics}\ }\textbf {\bibinfo {volume} {1}},\ \bibinfo
  {pages} {195} (\bibinfo {year} {1964})}\BibitemShut {NoStop}%
\bibitem [{\citenamefont {Bell}(1966)}]{Bell66}%
  \BibitemOpen
  \bibfield  {author} {\bibinfo {author} {\bibfnamefont {J.~S.}\ \bibnamefont
  {Bell}},\ }\bibfield  {title} {\enquote {\bibinfo {title} {{On the Problem of
  Hidden Variables in Quantum Mechanics}},}\ }\href
  {https://doi.org/10.1103/RevModPhys.38.447} {\bibfield  {journal} {\bibinfo
  {journal} {Rev. Mod. Phys.}\ }\textbf {\bibinfo {volume} {38}},\ \bibinfo
  {pages} {447} (\bibinfo {year} {1966})}\BibitemShut {NoStop}%
\bibitem [{\citenamefont {{Hensen~$\textit{et~al.}$}}(2015)}]{Belltest1}%
  \BibitemOpen
  \bibfield  {author} {\bibinfo {author} {\bibfnamefont {B.}~\bibnamefont
  {{Hensen~$\textit{et~al.}$}}},\ }\bibfield  {title} {\enquote {\bibinfo
  {title} {{Loophole-free Bell inequality violation using electron spins
  separated by 1.3 kilometres}},}\ }\href {https://doi.org/10.1038/nature15759}
  {\bibfield  {journal} {\bibinfo  {journal} {Nature}\ }\textbf {\bibinfo
  {volume} {526}},\ \bibinfo {pages} {682 EP } (\bibinfo {year}
  {2015})}\BibitemShut {NoStop}%
\bibitem [{\citenamefont {{Giustina~$\textit{et~al.}$}}(2015)}]{Belltest2}%
  \BibitemOpen
  \bibfield  {author} {\bibinfo {author} {\bibfnamefont {M.}~\bibnamefont
  {{Giustina~$\textit{et~al.}$}}},\ }\bibfield  {title} {\enquote {\bibinfo
  {title} {{Significant-Loophole-Free Test of Bell's Theorem with Entangled
  Photons}},}\ }\href {https://doi.org/10.1103/PhysRevLett.115.250401}
  {\bibfield  {journal} {\bibinfo  {journal} {Phys. Rev. Lett.}\ }\textbf
  {\bibinfo {volume} {115}},\ \bibinfo {pages} {250401} (\bibinfo {year}
  {2015})}\BibitemShut {NoStop}%
\bibitem [{\citenamefont {{Shalm~$\textit{et~al.}$}}(2015)}]{Belltest3}%
  \BibitemOpen
  \bibfield  {author} {\bibinfo {author} {\bibfnamefont {L.}~\bibnamefont
  {{Shalm~$\textit{et~al.}$}}},\ }\bibfield  {title} {\enquote {\bibinfo
  {title} {{Strong Loophole-Free Test of Local Realism}},}\ }\href
  {https://doi.org/10.1103/PhysRevLett.115.250402} {\bibfield  {journal}
  {\bibinfo  {journal} {Phys. Rev. Lett.}\ }\textbf {\bibinfo {volume} {115}},\
  \bibinfo {pages} {250402} (\bibinfo {year} {2015})}\BibitemShut {NoStop}%
\bibitem [{\citenamefont {Barrett}\ \emph
  {et~al.}(2005{\natexlab{a}})\citenamefont {Barrett}, \citenamefont {Hardy},\
  and\ \citenamefont {Kent}}]{BHK}%
  \BibitemOpen
  \bibfield  {author} {\bibinfo {author} {\bibfnamefont {J.}~\bibnamefont
  {Barrett}}, \bibinfo {author} {\bibfnamefont {L.}~\bibnamefont {Hardy}},\
  and\ \bibinfo {author} {\bibfnamefont {A.}~\bibnamefont {Kent}},\ }\bibfield
  {title} {\enquote {\bibinfo {title} {{No Signaling and Quantum Key
  Distribution}},}\ }\href {https://doi.org/10.1103/PhysRevLett.95.010503}
  {\bibfield  {journal} {\bibinfo  {journal} {Phys. Rev. Lett.}\ }\textbf
  {\bibinfo {volume} {95}},\ \bibinfo {pages} {010503} (\bibinfo {year}
  {2005}{\natexlab{a}})}\BibitemShut {NoStop}%
\bibitem [{\citenamefont {Ac\'{\i}n}\ \emph {et~al.}(2006)\citenamefont
  {Ac\'{\i}n}, \citenamefont {Gisin},\ and\ \citenamefont
  {Masanes}}]{Acin2006QKD}%
  \BibitemOpen
  \bibfield  {author} {\bibinfo {author} {\bibfnamefont {A.}~\bibnamefont
  {Ac\'{\i}n}}, \bibinfo {author} {\bibfnamefont {N.}~\bibnamefont {Gisin}},\
  and\ \bibinfo {author} {\bibfnamefont {L.}~\bibnamefont {Masanes}},\
  }\bibfield  {title} {\enquote {\bibinfo {title} {{From Bell's Theorem to
  Secure Quantum Key Distribution}},}\ }\href
  {https://doi.org/10.1103/PhysRevLett.97.120405} {\bibfield  {journal}
  {\bibinfo  {journal} {Phys. Rev. Lett.}\ }\textbf {\bibinfo {volume} {97}},\
  \bibinfo {pages} {120405} (\bibinfo {year} {2006})}\BibitemShut {NoStop}%
\bibitem [{\citenamefont {Scarani}\ \emph {et~al.}(2006)\citenamefont
  {Scarani}, \citenamefont {Gisin}, \citenamefont {Brunner}, \citenamefont
  {Masanes}, \citenamefont {Pino},\ and\ \citenamefont
  {Ac\'{\i}n}}]{Scarani2006QKD}%
  \BibitemOpen
  \bibfield  {author} {\bibinfo {author} {\bibfnamefont {V.}~\bibnamefont
  {Scarani}}, \bibinfo {author} {\bibfnamefont {N.}~\bibnamefont {Gisin}},
  \bibinfo {author} {\bibfnamefont {N.}~\bibnamefont {Brunner}}, \bibinfo
  {author} {\bibfnamefont {L.}~\bibnamefont {Masanes}}, \bibinfo {author}
  {\bibfnamefont {S.}~\bibnamefont {Pino}},\ and\ \bibinfo {author}
  {\bibfnamefont {A.}~\bibnamefont {Ac\'{\i}n}},\ }\bibfield  {title} {\enquote
  {\bibinfo {title} {Secrecy extraction from no-signaling correlations},}\
  }\href {https://doi.org/10.1103/PhysRevA.74.042339} {\bibfield  {journal}
  {\bibinfo  {journal} {Phys. Rev. A}\ }\textbf {\bibinfo {volume} {74}},\
  \bibinfo {pages} {042339} (\bibinfo {year} {2006})}\BibitemShut {NoStop}%
\bibitem [{\citenamefont {Ac\'{\i}n}\ \emph {et~al.}(2007)\citenamefont
  {Ac\'{\i}n}, \citenamefont {Brunner}, \citenamefont {Gisin}, \citenamefont
  {Massar}, \citenamefont {Pironio},\ and\ \citenamefont
  {Scarani}}]{Acin2007QKD}%
  \BibitemOpen
  \bibfield  {author} {\bibinfo {author} {\bibfnamefont {A.}~\bibnamefont
  {Ac\'{\i}n}}, \bibinfo {author} {\bibfnamefont {N.}~\bibnamefont {Brunner}},
  \bibinfo {author} {\bibfnamefont {N.}~\bibnamefont {Gisin}}, \bibinfo
  {author} {\bibfnamefont {S.}~\bibnamefont {Massar}}, \bibinfo {author}
  {\bibfnamefont {S.}~\bibnamefont {Pironio}},\ and\ \bibinfo {author}
  {\bibfnamefont {V.}~\bibnamefont {Scarani}},\ }\bibfield  {title} {\enquote
  {\bibinfo {title} {{Device-Independent Security of Quantum Cryptography
  against Collective Attacks}},}\ }\href
  {https://doi.org/10.1103/PhysRevLett.98.230501} {\bibfield  {journal}
  {\bibinfo  {journal} {Phys. Rev. Lett.}\ }\textbf {\bibinfo {volume} {98}},\
  \bibinfo {pages} {230501} (\bibinfo {year} {2007})}\BibitemShut {NoStop}%
\bibitem [{\citenamefont {Colbeck}\ and\ \citenamefont
  {Renner}(2012)}]{colbeckamp}%
  \BibitemOpen
  \bibfield  {author} {\bibinfo {author} {\bibfnamefont {R.}~\bibnamefont
  {Colbeck}}\ and\ \bibinfo {author} {\bibfnamefont {R.}~\bibnamefont
  {Renner}},\ }\bibfield  {title} {\enquote {\bibinfo {title} {Free randomness
  can be amplified},}\ }\href {https://doi.org/10.1038/nphys2300} {\bibfield
  {journal} {\bibinfo  {journal} {Nat. Phys.}\ }\textbf {\bibinfo {volume}
  {8}},\ \bibinfo {pages} {450 EP } (\bibinfo {year} {2012})}\BibitemShut
  {NoStop}%
\bibitem [{\citenamefont {Pironio}\ \emph {et~al.}(2010)\citenamefont
  {Pironio}, \citenamefont {Ac{\'\i}n}, \citenamefont {Massar}, \citenamefont
  {de~la Giroday}, \citenamefont {Matsukevich}, \citenamefont {Maunz},
  \citenamefont {Olmschenk}, \citenamefont {Hayes}, \citenamefont {Luo},
  \citenamefont {Manning},\ and\ \citenamefont {Monroe}}]{Pironio2010}%
  \BibitemOpen
  \bibfield  {author} {\bibinfo {author} {\bibfnamefont {S.}~\bibnamefont
  {Pironio}}, \bibinfo {author} {\bibfnamefont {A.}~\bibnamefont {Ac{\'\i}n}},
  \bibinfo {author} {\bibfnamefont {S.}~\bibnamefont {Massar}}, \bibinfo
  {author} {\bibfnamefont {A.~B.}\ \bibnamefont {de~la Giroday}}, \bibinfo
  {author} {\bibfnamefont {D.~N.}\ \bibnamefont {Matsukevich}}, \bibinfo
  {author} {\bibfnamefont {P.}~\bibnamefont {Maunz}}, \bibinfo {author}
  {\bibfnamefont {S.}~\bibnamefont {Olmschenk}}, \bibinfo {author}
  {\bibfnamefont {D.}~\bibnamefont {Hayes}}, \bibinfo {author} {\bibfnamefont
  {L.}~\bibnamefont {Luo}}, \bibinfo {author} {\bibfnamefont {T.~A.}\
  \bibnamefont {Manning}},\ and\ \bibinfo {author} {\bibfnamefont
  {C.}~\bibnamefont {Monroe}},\ }\bibfield  {title} {\enquote {\bibinfo {title}
  {{Random numbers certified by Bell's theorem}},}\ }\href
  {https://doi.org/10.1038/nature09008} {\bibfield  {journal} {\bibinfo
  {journal} {Nature}\ }\textbf {\bibinfo {volume} {464}},\ \bibinfo {pages}
  {1021 EP } (\bibinfo {year} {2010})}\BibitemShut {NoStop}%
\bibitem [{\citenamefont {Dhara}\ \emph {et~al.}(2013)\citenamefont {Dhara},
  \citenamefont {Prettico},\ and\ \citenamefont {Ac\'{\i}n}}]{Dhara2013DIRNG}%
  \BibitemOpen
  \bibfield  {author} {\bibinfo {author} {\bibfnamefont {C.}~\bibnamefont
  {Dhara}}, \bibinfo {author} {\bibfnamefont {G.}~\bibnamefont {Prettico}},\
  and\ \bibinfo {author} {\bibfnamefont {A.}~\bibnamefont {Ac\'{\i}n}},\
  }\bibfield  {title} {\enquote {\bibinfo {title} {{Maximal quantum randomness
  in Bell tests}},}\ }\href {https://doi.org/10.1103/PhysRevA.88.052116}
  {\bibfield  {journal} {\bibinfo  {journal} {Phys. Rev. A}\ }\textbf {\bibinfo
  {volume} {88}},\ \bibinfo {pages} {052116} (\bibinfo {year}
  {2013})}\BibitemShut {NoStop}%
\bibitem [{\citenamefont {Vazirani}\ and\ \citenamefont
  {Vidick}(2014)}]{vazirani14}%
  \BibitemOpen
  \bibfield  {author} {\bibinfo {author} {\bibfnamefont {U.}~\bibnamefont
  {Vazirani}}\ and\ \bibinfo {author} {\bibfnamefont {T.}~\bibnamefont
  {Vidick}},\ }\bibfield  {title} {\enquote {\bibinfo {title} {{Fully
  Device-Independent Quantum Key Distribution}},}\ }\href
  {https://doi.org/10.1103/PhysRevLett.113.140501} {\bibfield  {journal}
  {\bibinfo  {journal} {Phys. Rev. Lett.}\ }\textbf {\bibinfo {volume} {113}},\
  \bibinfo {pages} {140501} (\bibinfo {year} {2014})}\BibitemShut {NoStop}%
\bibitem [{\citenamefont {Kaniewski}\ and\ \citenamefont
  {Wehner}(2016)}]{Kaniewski2016chsh}%
  \BibitemOpen
  \bibfield  {author} {\bibinfo {author} {\bibfnamefont {J.}~\bibnamefont
  {Kaniewski}}\ and\ \bibinfo {author} {\bibfnamefont {S.}~\bibnamefont
  {Wehner}},\ }\bibfield  {title} {\enquote {\bibinfo {title}
  {Device-independent two-party cryptography secure against sequential
  attacks},}\ }\href {https://doi.org/10.1088/1367-2630/18/5/055004} {\bibfield
   {journal} {\bibinfo  {journal} {New J. Phys.}\ }\textbf {\bibinfo {volume}
  {18}},\ \bibinfo {pages} {055004} (\bibinfo {year} {2016})}\BibitemShut
  {NoStop}%
\bibitem [{\citenamefont {Gallego}\ \emph {et~al.}(2012)\citenamefont
  {Gallego}, \citenamefont {W\"urflinger}, \citenamefont {Ac\'{\i}n},\ and\
  \citenamefont {Navascu\'es}}]{gallego2012}%
  \BibitemOpen
  \bibfield  {author} {\bibinfo {author} {\bibfnamefont {R.}~\bibnamefont
  {Gallego}}, \bibinfo {author} {\bibfnamefont {L.~E.}\ \bibnamefont
  {W\"urflinger}}, \bibinfo {author} {\bibfnamefont {A.}~\bibnamefont
  {Ac\'{\i}n}},\ and\ \bibinfo {author} {\bibfnamefont {M.}~\bibnamefont
  {Navascu\'es}},\ }\bibfield  {title} {\enquote {\bibinfo {title}
  {{Operational Framework for Nonlocality}},}\ }\href
  {https://doi.org/10.1103/PhysRevLett.109.070401} {\bibfield  {journal}
  {\bibinfo  {journal} {Phys. Rev. Lett.}\ }\textbf {\bibinfo {volume} {109}},\
  \bibinfo {pages} {070401} (\bibinfo {year} {2012})}\BibitemShut {NoStop}%
\bibitem [{\citenamefont {{de}~Vicente}(2014)}]{de2014nonlocality}%
  \BibitemOpen
  \bibfield  {author} {\bibinfo {author} {\bibfnamefont {J.~I.}\ \bibnamefont
  {{de}~Vicente}},\ }\bibfield  {title} {\enquote {\bibinfo {title} {On
  nonlocality as a resource theory and nonlocality measures},}\ }\href
  {https://doi.org/10.1088/1751-8113/47/42/424017} {\bibfield  {journal}
  {\bibinfo  {journal} {J. Phys. A}\ }\textbf {\bibinfo {volume} {47}},\
  \bibinfo {pages} {424017} (\bibinfo {year} {2014})}\BibitemShut {NoStop}%
\bibitem [{\citenamefont {Geller}\ and\ \citenamefont
  {Piani}(2014)}]{GellerPiani}%
  \BibitemOpen
  \bibfield  {author} {\bibinfo {author} {\bibfnamefont {J.}~\bibnamefont
  {Geller}}\ and\ \bibinfo {author} {\bibfnamefont {M.}~\bibnamefont {Piani}},\
  }\bibfield  {title} {\enquote {\bibinfo {title} {Quantifying non-classical
  and beyond-quantum correlations in the unified operator formalism},}\ }\href
  {https://doi.org/10.1088/1751-8113/47/42/424030} {\bibfield  {journal}
  {\bibinfo  {journal} {J. Phys. A}\ }\textbf {\bibinfo {volume} {47}},\
  \bibinfo {pages} {424030} (\bibinfo {year} {2014})}\BibitemShut {NoStop}%
\bibitem [{\citenamefont {Gallego}\ and\ \citenamefont
  {Aolita}(2017)}]{gallego2016nonlocality}%
  \BibitemOpen
  \bibfield  {author} {\bibinfo {author} {\bibfnamefont {R.}~\bibnamefont
  {Gallego}}\ and\ \bibinfo {author} {\bibfnamefont {L.}~\bibnamefont
  {Aolita}},\ }\bibfield  {title} {\enquote {\bibinfo {title} {{Nonlocality
  free wirings and the distinguishability between Bell boxes}},}\ }\href
  {https://doi.org/10.1103/PhysRevA.95.032118} {\bibfield  {journal} {\bibinfo
  {journal} {Phys. Rev. A}\ }\textbf {\bibinfo {volume} {95}} (\bibinfo {year}
  {2017})}\BibitemShut {NoStop}%
\bibitem [{\citenamefont {Horodecki}\ \emph {et~al.}(2015)\citenamefont
  {Horodecki}, \citenamefont {Grudka}, \citenamefont {Joshi}, \citenamefont
  {K{\l}obus},\ and\ \citenamefont {{\L}odyga}}]{horodecki2015axiomatic}%
  \BibitemOpen
  \bibfield  {author} {\bibinfo {author} {\bibfnamefont {K.}~\bibnamefont
  {Horodecki}}, \bibinfo {author} {\bibfnamefont {A.}~\bibnamefont {Grudka}},
  \bibinfo {author} {\bibfnamefont {P.}~\bibnamefont {Joshi}}, \bibinfo
  {author} {\bibfnamefont {W.}~\bibnamefont {K{\l}obus}},\ and\ \bibinfo
  {author} {\bibfnamefont {J.}~\bibnamefont {{\L}odyga}},\ }\bibfield  {title}
  {\enquote {\bibinfo {title} {Axiomatic approach to contextuality and
  nonlocality},}\ }\href {https://doi.org/10.1103/physreva.92.032104}
  {\bibfield  {journal} {\bibinfo  {journal} {Phys. Rev. A}\ }\textbf {\bibinfo
  {volume} {92}},\ \bibinfo {pages} {032104} (\bibinfo {year}
  {2015})}\BibitemShut {NoStop}%
\bibitem [{\citenamefont {Amaral}\ \emph {et~al.}(2018)\citenamefont {Amaral},
  \citenamefont {Cabello}, \citenamefont {Cunha},\ and\ \citenamefont
  {Aolita}}]{Amaral2017NCW}%
  \BibitemOpen
  \bibfield  {author} {\bibinfo {author} {\bibfnamefont {B.}~\bibnamefont
  {Amaral}}, \bibinfo {author} {\bibfnamefont {A.}~\bibnamefont {Cabello}},
  \bibinfo {author} {\bibfnamefont {M.~T.}\ \bibnamefont {Cunha}},\ and\
  \bibinfo {author} {\bibfnamefont {L.}~\bibnamefont {Aolita}},\ }\bibfield
  {title} {\enquote {\bibinfo {title} {Noncontextual wirings},}\ }\href
  {https://doi.org/10.1103/PhysRevLett.120.130403} {\bibfield  {journal}
  {\bibinfo  {journal} {Phys. Rev. Lett.}\ }\textbf {\bibinfo {volume} {120}},\
  \bibinfo {pages} {130403} (\bibinfo {year} {2018})}\BibitemShut {NoStop}%
\bibitem [{\citenamefont {Kaur}\ \emph {et~al.}(2018)\citenamefont {Kaur},
  \citenamefont {Wilde},\ and\ \citenamefont {Winter}}]{kaur2018fundamental}%
  \BibitemOpen
  \bibfield  {author} {\bibinfo {author} {\bibfnamefont {E.}~\bibnamefont
  {Kaur}}, \bibinfo {author} {\bibfnamefont {M.~M.}\ \bibnamefont {Wilde}},\
  and\ \bibinfo {author} {\bibfnamefont {A.}~\bibnamefont {Winter}},\
  }\bibfield  {title} {\enquote {\bibinfo {title} {Fundamental limits on key
  rates in device-independent quantum key distribution},}\ }\href
  {https://arxiv.org/abs/1810.05627} {\bibfield  {journal} {\bibinfo  {journal}
  {arXiv:1810.05627}\ } (\bibinfo {year} {2018})}\BibitemShut {NoStop}%
\bibitem [{\citenamefont {Brito}\ \emph {et~al.}(2018)\citenamefont {Brito},
  \citenamefont {Amaral},\ and\ \citenamefont
  {Chaves}}]{Brito2018tracedistance}%
  \BibitemOpen
  \bibfield  {author} {\bibinfo {author} {\bibfnamefont {S.~G.~A.}\
  \bibnamefont {Brito}}, \bibinfo {author} {\bibfnamefont {B.}~\bibnamefont
  {Amaral}},\ and\ \bibinfo {author} {\bibfnamefont {R.}~\bibnamefont
  {Chaves}},\ }\bibfield  {title} {\enquote {\bibinfo {title} {{Quantifying
  Bell nonlocality with the trace distance}},}\ }\href
  {https://doi.org/10.1103/PhysRevA.97.022111} {\bibfield  {journal} {\bibinfo
  {journal} {Phys. Rev. A}\ }\textbf {\bibinfo {volume} {97}},\ \bibinfo
  {pages} {022111} (\bibinfo {year} {2018})}\BibitemShut {NoStop}%
\bibitem [{\citenamefont {Schmid}\ \emph
  {et~al.}(2019{\natexlab{a}})\citenamefont {Schmid}, \citenamefont {Rosset},\
  and\ \citenamefont {Buscemi}}]{schmid2019typeindependent}%
  \BibitemOpen
  \bibfield  {author} {\bibinfo {author} {\bibfnamefont {D.}~\bibnamefont
  {Schmid}}, \bibinfo {author} {\bibfnamefont {D.}~\bibnamefont {Rosset}},\
  and\ \bibinfo {author} {\bibfnamefont {F.}~\bibnamefont {Buscemi}},\
  }\bibfield  {title} {\enquote {\bibinfo {title} {Type-independent resource
  theory of local operations and shared randomness},}\ }\href
  {https://arxiv.org/abs/1909.04065} {\bibfield  {journal} {\bibinfo  {journal}
  {arXiv:1909.04065}\ } (\bibinfo {year} {2019}{\natexlab{a}})}\BibitemShut
  {NoStop}%
\bibitem [{\citenamefont {Rosset}\ \emph
  {et~al.}(2020{\natexlab{a}})\citenamefont {Rosset}, \citenamefont {Schmid},\
  and\ \citenamefont {Buscemi}}]{rosset2019characterizing}%
  \BibitemOpen
  \bibfield  {author} {\bibinfo {author} {\bibfnamefont {D.}~\bibnamefont
  {Rosset}}, \bibinfo {author} {\bibfnamefont {D.}~\bibnamefont {Schmid}},\
  and\ \bibinfo {author} {\bibfnamefont {F.}~\bibnamefont {Buscemi}},\
  }\bibfield  {title} {\enquote {\bibinfo {title} {Characterizing
  nonclassicality of arbitrary distributed devices},}\ }\href
  {https://arxiv.org/abs/1911.12462} {\bibfield  {journal} {\bibinfo  {journal}
  {arXiv:2004.09194}\ } (\bibinfo {year} {2020}{\natexlab{a}})}\BibitemShut
  {NoStop}%
\bibitem [{\citenamefont {Schmid}\ \emph
  {et~al.}(2019{\natexlab{b}})\citenamefont {Schmid}, \citenamefont {Fraser},
  \citenamefont {Kunjwal}, \citenamefont {Sainz}, \citenamefont {Wolfe},\ and\
  \citenamefont {Spekkens}}]{LOSRvsLOCC}%
  \BibitemOpen
  \bibfield  {author} {\bibinfo {author} {\bibfnamefont {D.}~\bibnamefont
  {Schmid}}, \bibinfo {author} {\bibfnamefont {T.~C.}\ \bibnamefont {Fraser}},
  \bibinfo {author} {\bibfnamefont {R.}~\bibnamefont {Kunjwal}}, \bibinfo
  {author} {\bibfnamefont {A.~B.}\ \bibnamefont {Sainz}}, \bibinfo {author}
  {\bibfnamefont {E.}~\bibnamefont {Wolfe}},\ and\ \bibinfo {author}
  {\bibfnamefont {R.~W.}\ \bibnamefont {Spekkens}},\ }\bibfield  {title}
  {\enquote {\bibinfo {title} {{Why standard entanglement theory is
  inappropriate for the study of Bell scenarios}},}\ }\href
  {https://arxiv.org/abs/2004.09194} {\bibfield  {journal} {\bibinfo  {journal}
  {arXiv:1911.12462}\ } (\bibinfo {year} {2019}{\natexlab{b}})}\BibitemShut
  {NoStop}%
\bibitem [{\citenamefont {Coecke}\ \emph {et~al.}(2016)\citenamefont {Coecke},
  \citenamefont {Fritz},\ and\ \citenamefont {Spekkens}}]{Coecke2014}%
  \BibitemOpen
  \bibfield  {author} {\bibinfo {author} {\bibfnamefont {B.}~\bibnamefont
  {Coecke}}, \bibinfo {author} {\bibfnamefont {T.}~\bibnamefont {Fritz}},\ and\
  \bibinfo {author} {\bibfnamefont {R.~W.}\ \bibnamefont {Spekkens}},\
  }\bibfield  {title} {\enquote {\bibinfo {title} {A mathematical theory of
  resources},}\ }\href {https://doi.org/10.1016/j.ic.2016.02.008} {\bibfield
  {journal} {\bibinfo  {journal} {Info. \& Comp.}\ }\textbf {\bibinfo {volume}
  {250}},\ \bibinfo {pages} {59 } (\bibinfo {year} {2016})}\BibitemShut
  {NoStop}%
\bibitem [{\citenamefont {Clauser}\ \emph {et~al.}(1969)\citenamefont
  {Clauser}, \citenamefont {Horne}, \citenamefont {Shimony},\ and\
  \citenamefont {Holt}}]{CHSH}%
  \BibitemOpen
  \bibfield  {author} {\bibinfo {author} {\bibfnamefont {J.~F.}\ \bibnamefont
  {Clauser}}, \bibinfo {author} {\bibfnamefont {M.~A.}\ \bibnamefont {Horne}},
  \bibinfo {author} {\bibfnamefont {A.}~\bibnamefont {Shimony}},\ and\ \bibinfo
  {author} {\bibfnamefont {R.~A.}\ \bibnamefont {Holt}},\ }\bibfield  {title}
  {\enquote {\bibinfo {title} {Proposed experiment to test local
  hidden-variable theories},}\ }\href
  {https://doi.org/10.1103/PhysRevLett.23.880} {\bibfield  {journal} {\bibinfo
  {journal} {Phys. Rev. Lett.}\ }\textbf {\bibinfo {volume} {23}},\ \bibinfo
  {pages} {880} (\bibinfo {year} {1969})}\BibitemShut {NoStop}%
\bibitem [{\citenamefont {Shimony}(2017)}]{sep-bell-theorem}%
  \BibitemOpen
  \bibfield  {author} {\bibinfo {author} {\bibfnamefont {A.}~\bibnamefont
  {Shimony}},\ }\bibfield  {title} {\enquote {\bibinfo {title} {{Bell's
  Theorem}},}\ }in\ \href
  {https://plato.stanford.edu/archives/fall2017/entries/bell-theorem/} {\emph
  {\bibinfo {booktitle} {The Stanford Encyclopedia of Philosophy}}}\ (\bibinfo
  {year} {2017})\BibitemShut {NoStop}%
\bibitem [{\citenamefont {d'Espagnat}(1979)}]{d1979quantum}%
  \BibitemOpen
  \bibfield  {author} {\bibinfo {author} {\bibfnamefont {B.}~\bibnamefont
  {d'Espagnat}},\ }\bibfield  {title} {\enquote {\bibinfo {title} {{The Quantum
  Theory and Reality}},}\ }\href
  {https://doi.org/10.1038/scientificamerican1179-158} {\bibfield  {journal}
  {\bibinfo  {journal} {Scientific American}\ }\textbf {\bibinfo {volume}
  {241}},\ \bibinfo {pages} {158} (\bibinfo {year} {1979})}\BibitemShut
  {NoStop}%
\bibitem [{\citenamefont {Wiseman}(2014)}]{wiseman2014two}%
  \BibitemOpen
  \bibfield  {author} {\bibinfo {author} {\bibfnamefont {H.~M.}\ \bibnamefont
  {Wiseman}},\ }\bibfield  {title} {\enquote {\bibinfo {title} {{The two Bell's
  theorems of John Bell}},}\ }\href
  {https://doi.org/10.1088/1751-8113/47/42/424001} {\bibfield  {journal}
  {\bibinfo  {journal} {J. Phys. A}\ }\textbf {\bibinfo {volume} {47}},\
  \bibinfo {pages} {424001} (\bibinfo {year} {2014})}\BibitemShut {NoStop}%
\bibitem [{\citenamefont {Werner}(2014)}]{werner2014comment}%
  \BibitemOpen
  \bibfield  {author} {\bibinfo {author} {\bibfnamefont {R.~F.}\ \bibnamefont
  {Werner}},\ }\bibfield  {title} {\enquote {\bibinfo {title} {{Comment on
  ‘What Bell did’}},}\ }\href
  {https://doi.org/10.1088/1751-8113/47/42/424011} {\bibfield  {journal}
  {\bibinfo  {journal} {J. Phys. A}\ }\textbf {\bibinfo {volume} {47}},\
  \bibinfo {pages} {424011} (\bibinfo {year} {2014})}\BibitemShut {NoStop}%
\bibitem [{\citenamefont {Scarani}(2012)}]{Scarani2012device}%
  \BibitemOpen
  \bibfield  {author} {\bibinfo {author} {\bibfnamefont {V.}~\bibnamefont
  {Scarani}},\ }\bibfield  {title} {\enquote {\bibinfo {title} {{The
  Device-Independent Outlook on Quantum Physics}},}\ }\href
  {http://www.physics.sk/aps/pub.php?y=2012&pub=aps-12-04} {\bibfield
  {journal} {\bibinfo  {journal} {Acta Physica Slovaca}\ }\textbf {\bibinfo
  {volume} {62}},\ \bibinfo {pages} {347} (\bibinfo {year} {2012})}\BibitemShut
  {NoStop}%
\bibitem [{\citenamefont {Maudlin}(2002)}]{Maudlin2002quantum}%
  \BibitemOpen
  \bibfield  {author} {\bibinfo {author} {\bibfnamefont {T.}~\bibnamefont
  {Maudlin}},\ }\href {https://doi.org/10.1002/9780470752166} {\emph {\bibinfo
  {title} {{Quantum Non-Locality and Relativity : Metaphysical Intimations of
  Modern Physics}}}}\ (\bibinfo  {publisher} {Blackwell Publishers},\ \bibinfo
  {year} {2002})\BibitemShut {NoStop}%
\bibitem [{\citenamefont {Norsen}(2006)}]{Norsen2006}%
  \BibitemOpen
  \bibfield  {author} {\bibinfo {author} {\bibfnamefont {T.}~\bibnamefont
  {Norsen}},\ }\bibfield  {title} {\enquote {\bibinfo {title} {{Bell Locality
  and the Nonlocal Character of Nature}},}\ }\href
  {https://doi.org/10.1007/s10702-006-1055-9} {\bibfield  {journal} {\bibinfo
  {journal} {Found. Phys. Lett.}\ }\textbf {\bibinfo {volume} {19}},\ \bibinfo
  {pages} {633} (\bibinfo {year} {2006})}\BibitemShut {NoStop}%
\bibitem [{\citenamefont {Chaves}\ \emph {et~al.}(2015)\citenamefont {Chaves},
  \citenamefont {Kueng}, \citenamefont {Brask},\ and\ \citenamefont
  {Gross}}]{Chaves2015relaxing}%
  \BibitemOpen
  \bibfield  {author} {\bibinfo {author} {\bibfnamefont {R.}~\bibnamefont
  {Chaves}}, \bibinfo {author} {\bibfnamefont {R.}~\bibnamefont {Kueng}},
  \bibinfo {author} {\bibfnamefont {J.~B.}\ \bibnamefont {Brask}},\ and\
  \bibinfo {author} {\bibfnamefont {D.}~\bibnamefont {Gross}},\ }\bibfield
  {title} {\enquote {\bibinfo {title} {{Unifying Framework for Relaxations of
  the Causal Assumptions in Bell's Theorem}},}\ }\href
  {https://doi.org/10.1103/PhysRevLett.114.140403} {\bibfield  {journal}
  {\bibinfo  {journal} {Phys. Rev. Lett.}\ }\textbf {\bibinfo {volume} {114}},\
  \bibinfo {pages} {140403} (\bibinfo {year} {2015})}\BibitemShut {NoStop}%
\bibitem [{\citenamefont {Chaves}\ \emph
  {et~al.}(2017{\natexlab{a}})\citenamefont {Chaves}, \citenamefont
  {Cavalcanti},\ and\ \citenamefont {Aolita}}]{Chaves2017causalmultipartite}%
  \BibitemOpen
  \bibfield  {author} {\bibinfo {author} {\bibfnamefont {R.}~\bibnamefont
  {Chaves}}, \bibinfo {author} {\bibfnamefont {D.}~\bibnamefont {Cavalcanti}},\
  and\ \bibinfo {author} {\bibfnamefont {L.}~\bibnamefont {Aolita}},\
  }\bibfield  {title} {\enquote {\bibinfo {title} {{Causal hierarchy of
  multipartite Bell nonlocality}},}\ }\href
  {https://doi.org/10.22331/q-2017-08-04-23} {\bibfield  {journal} {\bibinfo
  {journal} {Quantum}\ }\textbf {\bibinfo {volume} {1}},\ \bibinfo {pages} {23}
  (\bibinfo {year} {2017}{\natexlab{a}})}\BibitemShut {NoStop}%
\bibitem [{\citenamefont {Maudlin}(1992)}]{maudlin1992bell}%
  \BibitemOpen
  \bibfield  {author} {\bibinfo {author} {\bibfnamefont {T.}~\bibnamefont
  {Maudlin}},\ }\bibfield  {title} {\enquote {\bibinfo {title} {{Bell's
  Inequality, Information Transmission, and Prism Models}},}\ }in\ \href
  {https://www.jstor.org/stable/192771} {\emph {\bibinfo {booktitle}
  {Philosophy of Science Association}}},\ \bibinfo {series and number}
  {\bibinfo {number} {1}}\ (\bibinfo {year} {1992})\ pp.\ \bibinfo {pages}
  {404--417}\BibitemShut {NoStop}%
\bibitem [{\citenamefont {Toner}\ and\ \citenamefont
  {Bacon}(2003)}]{Toner2003}%
  \BibitemOpen
  \bibfield  {author} {\bibinfo {author} {\bibfnamefont {B.~F.}\ \bibnamefont
  {Toner}}\ and\ \bibinfo {author} {\bibfnamefont {D.}~\bibnamefont {Bacon}},\
  }\bibfield  {title} {\enquote {\bibinfo {title} {{Communication Cost of
  Simulating Bell Correlations}},}\ }\href
  {https://doi.org/10.1103/PhysRevLett.91.187904} {\bibfield  {journal}
  {\bibinfo  {journal} {Phys. Rev. Lett.}\ }\textbf {\bibinfo {volume} {91}},\
  \bibinfo {pages} {187904} (\bibinfo {year} {2003})}\BibitemShut {NoStop}%
\bibitem [{\citenamefont {Hooft}(2013)}]{hooft2013fate}%
  \BibitemOpen
  \bibfield  {author} {\bibinfo {author} {\bibfnamefont {G.}~\bibnamefont
  {Hooft}},\ }\bibfield  {title} {\enquote {\bibinfo {title} {{The Fate of the
  Quantum}},}\ }\href {https://arxiv.org/abs/1308.1007} {\bibfield  {journal}
  {\bibinfo  {journal} {arXiv:1308.1007}\ } (\bibinfo {year} {2013})},\
  \bibinfo {note} {report numbers: ITP-UU-13/22, SPIN-13/15}\BibitemShut
  {NoStop}%
\bibitem [{\citenamefont {Hall}(2010)}]{hall}%
  \BibitemOpen
  \bibfield  {author} {\bibinfo {author} {\bibfnamefont {M.~J.~W.}\
  \bibnamefont {Hall}},\ }\bibfield  {title} {\enquote {\bibinfo {title}
  {{Local Deterministic Model of Singlet State Correlations Based on Relaxing
  Measurement Independence}},}\ }\href
  {https://doi.org/10.1103/PhysRevLett.105.250404} {\bibfield  {journal}
  {\bibinfo  {journal} {Phys. Rev. Lett.}\ }\textbf {\bibinfo {volume} {105}},\
  \bibinfo {pages} {250404} (\bibinfo {year} {2010})}\BibitemShut {NoStop}%
\bibitem [{\citenamefont {Barrett}\ and\ \citenamefont
  {Gisin}(2011)}]{barrettgisin}%
  \BibitemOpen
  \bibfield  {author} {\bibinfo {author} {\bibfnamefont {J.}~\bibnamefont
  {Barrett}}\ and\ \bibinfo {author} {\bibfnamefont {N.}~\bibnamefont
  {Gisin}},\ }\bibfield  {title} {\enquote {\bibinfo {title} {{How Much
  Measurement Independence Is Needed to Demonstrate Nonlocality?}}}\ }\href
  {https://doi.org/10.1103/PhysRevLett.106.100406} {\bibfield  {journal}
  {\bibinfo  {journal} {Phys. Rev. Lett.}\ }\textbf {\bibinfo {volume} {106}},\
  \bibinfo {pages} {100406} (\bibinfo {year} {2011})}\BibitemShut {NoStop}%
\bibitem [{\citenamefont {Pearl}(2009)}]{Pearl2009}%
  \BibitemOpen
  \bibfield  {author} {\bibinfo {author} {\bibfnamefont {J.}~\bibnamefont
  {Pearl}},\ }\href {https://doi.org/10.1017/CBO9780511803161} {\emph {\bibinfo
  {title} {{Causality: Models, Reasoning, and Inference}}}}\ (\bibinfo
  {publisher} {Cambridge University Press},\ \bibinfo {year}
  {2009})\BibitemShut {NoStop}%
\bibitem [{\citenamefont {Wood}\ and\ \citenamefont
  {Spekkens}(2015)}]{Wood2015}%
  \BibitemOpen
  \bibfield  {author} {\bibinfo {author} {\bibfnamefont {C.~J.}\ \bibnamefont
  {Wood}}\ and\ \bibinfo {author} {\bibfnamefont {R.~W.}\ \bibnamefont
  {Spekkens}},\ }\bibfield  {title} {\enquote {\bibinfo {title} {{The lesson of
  causal discovery algorithms for quantum correlations: causal explanations of
  Bell-inequality violations require fine-tuning}},}\ }\href
  {https://doi.org/10.1088/1367-2630/17/3/033002} {\bibfield  {journal}
  {\bibinfo  {journal} {New J. Phys.}\ }\textbf {\bibinfo {volume} {17}},\
  \bibinfo {pages} {033002} (\bibinfo {year} {2015})}\BibitemShut {NoStop}%
\bibitem [{\citenamefont {Allen}\ \emph {et~al.}(2017)\citenamefont {Allen},
  \citenamefont {Barrett}, \citenamefont {Horsman}, \citenamefont {Lee},\ and\
  \citenamefont {Spekkens}}]{Allenetal}%
  \BibitemOpen
  \bibfield  {author} {\bibinfo {author} {\bibfnamefont {J.-M.~A.}\
  \bibnamefont {Allen}}, \bibinfo {author} {\bibfnamefont {J.}~\bibnamefont
  {Barrett}}, \bibinfo {author} {\bibfnamefont {D.~C.}\ \bibnamefont
  {Horsman}}, \bibinfo {author} {\bibfnamefont {C.~M.}\ \bibnamefont {Lee}},\
  and\ \bibinfo {author} {\bibfnamefont {R.~W.}\ \bibnamefont {Spekkens}},\
  }\bibfield  {title} {\enquote {\bibinfo {title} {{Quantum Common Causes and
  Quantum Causal Models}},}\ }\href {https://doi.org/10.1103/PhysRevX.7.031021}
  {\bibfield  {journal} {\bibinfo  {journal} {Phys. Rev. X}\ }\textbf {\bibinfo
  {volume} {7}},\ \bibinfo {pages} {031021} (\bibinfo {year}
  {2017})}\BibitemShut {NoStop}%
\bibitem [{\citenamefont {Henson}\ \emph {et~al.}(2014)\citenamefont {Henson},
  \citenamefont {Lal},\ and\ \citenamefont {Pusey}}]{Henson2014}%
  \BibitemOpen
  \bibfield  {author} {\bibinfo {author} {\bibfnamefont {J.}~\bibnamefont
  {Henson}}, \bibinfo {author} {\bibfnamefont {R.}~\bibnamefont {Lal}},\ and\
  \bibinfo {author} {\bibfnamefont {M.~F.}\ \bibnamefont {Pusey}},\ }\bibfield
  {title} {\enquote {\bibinfo {title} {{Theory-independent limits on
  correlations from generalized Bayesian networks}},}\ }\href
  {https://doi.org/10.1088/1367-2630/16/11/113043} {\bibfield  {journal}
  {\bibinfo  {journal} {New J. Phys.}\ }\textbf {\bibinfo {volume} {16}},\
  \bibinfo {pages} {113043} (\bibinfo {year} {2014})}\BibitemShut {NoStop}%
\bibitem [{\citenamefont {Fritz}(2012)}]{Fritz2012beyondBell}%
  \BibitemOpen
  \bibfield  {author} {\bibinfo {author} {\bibfnamefont {T.}~\bibnamefont
  {Fritz}},\ }\bibfield  {title} {\enquote {\bibinfo {title} {{Beyond Bell's
  theorem: correlation scenarios}},}\ }\href
  {https://doi.org/10.1088/1367-2630/14/10/103001} {\bibfield  {journal}
  {\bibinfo  {journal} {New J. Phys.}\ }\textbf {\bibinfo {volume} {14}},\
  \bibinfo {pages} {103001} (\bibinfo {year} {2012})}\BibitemShut {NoStop}%
\bibitem [{\citenamefont {Hardy}(2001)}]{hardy01}%
  \BibitemOpen
  \bibfield  {author} {\bibinfo {author} {\bibfnamefont {L.}~\bibnamefont
  {Hardy}},\ }\bibfield  {title} {\enquote {\bibinfo {title} {{Quantum Theory
  From Five Reasonable Axioms}},}\ }\href
  {https://arxiv.org/abs/quant-ph/0101012} {\bibfield  {journal} {\bibinfo
  {journal} {quant-ph/0101012}\ } (\bibinfo {year} {2001})}\BibitemShut
  {NoStop}%
\bibitem [{\citenamefont {Barrett}(2007)}]{barrettGPT}%
  \BibitemOpen
  \bibfield  {author} {\bibinfo {author} {\bibfnamefont {J.}~\bibnamefont
  {Barrett}},\ }\bibfield  {title} {\enquote {\bibinfo {title} {Information
  processing in generalized probabilistic theories},}\ }\href
  {https://doi.org/10.1103/PhysRevA.75.032304} {\bibfield  {journal} {\bibinfo
  {journal} {Phys. Rev. A}\ }\textbf {\bibinfo {volume} {75}},\ \bibinfo
  {pages} {032304} (\bibinfo {year} {2007})}\BibitemShut {NoStop}%
\bibitem [{\citenamefont {Janotta}\ and\ \citenamefont
  {Hinrichsen}(2014)}]{janotta2014}%
  \BibitemOpen
  \bibfield  {author} {\bibinfo {author} {\bibfnamefont {P.}~\bibnamefont
  {Janotta}}\ and\ \bibinfo {author} {\bibfnamefont {H.}~\bibnamefont
  {Hinrichsen}},\ }\bibfield  {title} {\enquote {\bibinfo {title} {Generalized
  probability theories: what determines the structure of quantum theory?}}\
  }\href {https://doi.org/10.1088/1751-8113/47/32/323001} {\bibfield  {journal}
  {\bibinfo  {journal} {J. Phys. A}\ }\textbf {\bibinfo {volume} {47}},\
  \bibinfo {pages} {323001} (\bibinfo {year} {2014})}\BibitemShut {NoStop}%
\bibitem [{\citenamefont {Chiribella}\ \emph {et~al.}(2010)\citenamefont
  {Chiribella}, \citenamefont {D'Ariano},\ and\ \citenamefont
  {Perinotti}}]{CombsForGPTs}%
  \BibitemOpen
  \bibfield  {author} {\bibinfo {author} {\bibfnamefont {G.}~\bibnamefont
  {Chiribella}}, \bibinfo {author} {\bibfnamefont {G.~M.}\ \bibnamefont
  {D'Ariano}},\ and\ \bibinfo {author} {\bibfnamefont {P.}~\bibnamefont
  {Perinotti}},\ }\bibfield  {title} {\enquote {\bibinfo {title} {Probabilistic
  theories with purification},}\ }\href
  {https://doi.org/10.1103/PhysRevA.81.062348} {\bibfield  {journal} {\bibinfo
  {journal} {Phys. Rev. A}\ }\textbf {\bibinfo {volume} {81}},\ \bibinfo
  {pages} {062348} (\bibinfo {year} {2010})}\BibitemShut {NoStop}%
\bibitem [{\citenamefont {Ariano}(2019)}]{ArianoOPT}%
  \BibitemOpen
  \bibfield  {author} {\bibinfo {author} {\bibfnamefont {G.~M.}\ \bibnamefont
  {Ariano}},\ }\href {https://books.google.com/books?id=ywizwwEACAAJ} {\emph
  {\bibinfo {title} {{Quantum Theory from First Principles: An Informational
  Approach}}}}\ (\bibinfo  {publisher} {Cambridge University Press},\ \bibinfo
  {year} {2019})\BibitemShut {NoStop}%
\bibitem [{\citenamefont {Costa}\ and\ \citenamefont
  {Shrapnel}(2016)}]{CostaShrapnel}%
  \BibitemOpen
  \bibfield  {author} {\bibinfo {author} {\bibfnamefont {F.}~\bibnamefont
  {Costa}}\ and\ \bibinfo {author} {\bibfnamefont {S.}~\bibnamefont
  {Shrapnel}},\ }\bibfield  {title} {\enquote {\bibinfo {title} {Quantum causal
  modelling},}\ }\href {https://doi.org/10.1088/1367-2630/18/6/063032}
  {\bibfield  {journal} {\bibinfo  {journal} {New J. Phys}\ }\textbf {\bibinfo
  {volume} {18}},\ \bibinfo {pages} {063032} (\bibinfo {year}
  {2016})}\BibitemShut {NoStop}%
\bibitem [{\citenamefont {Barrett}\ \emph {et~al.}(2019)\citenamefont
  {Barrett}, \citenamefont {Lorenz},\ and\ \citenamefont {Oreshkov}}]{BLO}%
  \BibitemOpen
  \bibfield  {author} {\bibinfo {author} {\bibfnamefont {J.}~\bibnamefont
  {Barrett}}, \bibinfo {author} {\bibfnamefont {R.}~\bibnamefont {Lorenz}},\
  and\ \bibinfo {author} {\bibfnamefont {O.}~\bibnamefont {Oreshkov}},\
  }\bibfield  {title} {\enquote {\bibinfo {title} {{Quantum Causal Models}},}\
  }\href {https://arxiv.org/abs/1906.10726} {\bibfield  {journal} {\bibinfo
  {journal} {arXiv:1906.10726}\ } (\bibinfo {year} {2019})}\BibitemShut
  {NoStop}%
\bibitem [{\citenamefont {Schmid}\ \emph {et~al.}(2020)\citenamefont {Schmid},
  \citenamefont {Du}, \citenamefont {Mudassar}, \citenamefont {de~Wit},
  \citenamefont {Rosset},\ and\ \citenamefont {Hoban}}]{schmidpostquantum}%
  \BibitemOpen
  \bibfield  {author} {\bibinfo {author} {\bibfnamefont {D.}~\bibnamefont
  {Schmid}}, \bibinfo {author} {\bibfnamefont {H.}~\bibnamefont {Du}}, \bibinfo
  {author} {\bibfnamefont {M.}~\bibnamefont {Mudassar}}, \bibinfo {author}
  {\bibfnamefont {G.~C.}\ \bibnamefont {de~Wit}}, \bibinfo {author}
  {\bibfnamefont {D.}~\bibnamefont {Rosset}},\ and\ \bibinfo {author}
  {\bibfnamefont {M.~J.}\ \bibnamefont {Hoban}},\ }\bibfield  {title} {\enquote
  {\bibinfo {title} {Postquantum common-cause channels: the resource theory of
  local operations and shared entanglement},}\ }\href
  {https://arxiv.org/abs/2004.06133} {\bibfield  {journal} {\bibinfo  {journal}
  {arXiv:2004.06133}\ } (\bibinfo {year} {2020})}\BibitemShut {NoStop}%
\bibitem [{\citenamefont {Chiribella}\ \emph {et~al.}(2008)\citenamefont
  {Chiribella}, \citenamefont {D'Ariano},\ and\ \citenamefont
  {Perinotti}}]{qcombs08}%
  \BibitemOpen
  \bibfield  {author} {\bibinfo {author} {\bibfnamefont {G.}~\bibnamefont
  {Chiribella}}, \bibinfo {author} {\bibfnamefont {G.~M.}\ \bibnamefont
  {D'Ariano}},\ and\ \bibinfo {author} {\bibfnamefont {P.}~\bibnamefont
  {Perinotti}},\ }\bibfield  {title} {\enquote {\bibinfo {title} {{Quantum
  Circuit Architecture}},}\ }\href
  {https://doi.org/10.1103/PhysRevLett.101.060401} {\bibfield  {journal}
  {\bibinfo  {journal} {Phys. Rev. Lett.}\ }\textbf {\bibinfo {volume} {101}},\
  \bibinfo {pages} {060401} (\bibinfo {year} {2008})}\BibitemShut {NoStop}%
\bibitem [{\citenamefont {Chiribella}\ \emph {et~al.}(2009)\citenamefont
  {Chiribella}, \citenamefont {D'Ariano},\ and\ \citenamefont
  {Perinotti}}]{qcombs09}%
  \BibitemOpen
  \bibfield  {author} {\bibinfo {author} {\bibfnamefont {G.}~\bibnamefont
  {Chiribella}}, \bibinfo {author} {\bibfnamefont {G.~M.}\ \bibnamefont
  {D'Ariano}},\ and\ \bibinfo {author} {\bibfnamefont {P.}~\bibnamefont
  {Perinotti}},\ }\bibfield  {title} {\enquote {\bibinfo {title} {{Theoretical
  framework for quantum networks}},}\ }\href
  {https://doi.org/10.1103/PhysRevA.80.022339} {\bibfield  {journal} {\bibinfo
  {journal} {Phys. Rev. A}\ }\textbf {\bibinfo {volume} {80}},\ \bibinfo
  {pages} {022339} (\bibinfo {year} {2009})}\BibitemShut {NoStop}%
\bibitem [{\citenamefont {Popescu}\ and\ \citenamefont
  {Rohrlich}(1994)}]{Popescu1994}%
  \BibitemOpen
  \bibfield  {author} {\bibinfo {author} {\bibfnamefont {S.}~\bibnamefont
  {Popescu}}\ and\ \bibinfo {author} {\bibfnamefont {D.}~\bibnamefont
  {Rohrlich}},\ }\bibfield  {title} {\enquote {\bibinfo {title} {{Quantum
  nonlocality as an axiom}},}\ }\href {https://doi.org/10.1007/BF02058098}
  {\bibfield  {journal} {\bibinfo  {journal} {Found. Phys.}\ }\textbf {\bibinfo
  {volume} {24}},\ \bibinfo {pages} {379} (\bibinfo {year} {1994})}\BibitemShut
  {NoStop}%
\bibitem [{\citenamefont
  {{Selby~$\textit{et~al.}$}}()}]{FutureRTContextuality}%
  \BibitemOpen
  \bibfield  {author} {\bibinfo {author} {\bibfnamefont {J.}~\bibnamefont
  {{Selby~$\textit{et~al.}$}}},\ }\href@noop {} {\enquote {\bibinfo {title}
  {{Contextuality Quantified: A Resource Theory Encompassing
  Prepare-and-Measure Processes}},}\ }\bibinfo {note}
  {Forthcoming.}\BibitemShut {Stop}%
\bibitem [{\citenamefont {Branciard}\ \emph {et~al.}(2012)\citenamefont
  {Branciard}, \citenamefont {Rosset}, \citenamefont {Gisin},\ and\
  \citenamefont {Pironio}}]{BilocalCorrelations}%
  \BibitemOpen
  \bibfield  {author} {\bibinfo {author} {\bibfnamefont {C.}~\bibnamefont
  {Branciard}}, \bibinfo {author} {\bibfnamefont {D.}~\bibnamefont {Rosset}},
  \bibinfo {author} {\bibfnamefont {N.}~\bibnamefont {Gisin}},\ and\ \bibinfo
  {author} {\bibfnamefont {S.}~\bibnamefont {Pironio}},\ }\bibfield  {title}
  {\enquote {\bibinfo {title} {Bilocal versus nonbilocal correlations in
  entanglement-swapping experiments},}\ }\href
  {https://doi.org/10.1103/PhysRevA.85.032119} {\bibfield  {journal} {\bibinfo
  {journal} {Phys. Rev. A}\ }\textbf {\bibinfo {volume} {85}},\ \bibinfo
  {pages} {032119} (\bibinfo {year} {2012})}\BibitemShut {NoStop}%
\bibitem [{\citenamefont {Ac\'{\i}n}\ \emph {et~al.}(2010)\citenamefont
  {Ac\'{\i}n}, \citenamefont {Augusiak}, \citenamefont {Cavalcanti},
  \citenamefont {Hadley}, \citenamefont {Korbicz}, \citenamefont {Lewenstein},
  \citenamefont {Masanes},\ and\ \citenamefont {Piani}}]{Acin2010Unified}%
  \BibitemOpen
  \bibfield  {author} {\bibinfo {author} {\bibfnamefont {A.}~\bibnamefont
  {Ac\'{\i}n}}, \bibinfo {author} {\bibfnamefont {R.}~\bibnamefont {Augusiak}},
  \bibinfo {author} {\bibfnamefont {D.}~\bibnamefont {Cavalcanti}}, \bibinfo
  {author} {\bibfnamefont {C.}~\bibnamefont {Hadley}}, \bibinfo {author}
  {\bibfnamefont {J.~K.}\ \bibnamefont {Korbicz}}, \bibinfo {author}
  {\bibfnamefont {M.}~\bibnamefont {Lewenstein}}, \bibinfo {author}
  {\bibfnamefont {L.}~\bibnamefont {Masanes}},\ and\ \bibinfo {author}
  {\bibfnamefont {M.}~\bibnamefont {Piani}},\ }\bibfield  {title} {\enquote
  {\bibinfo {title} {{Unified Framework for Correlations in Terms of Local
  Quantum Observables}},}\ }\href
  {https://doi.org/10.1103/PhysRevLett.104.140404} {\bibfield  {journal}
  {\bibinfo  {journal} {Phys. Rev. Lett.}\ }\textbf {\bibinfo {volume} {104}},\
  \bibinfo {pages} {140404} (\bibinfo {year} {2010})}\BibitemShut {NoStop}%
\bibitem [{\citenamefont {Al-Safi}\ and\ \citenamefont
  {Short}(2013)}]{Short2013}%
  \BibitemOpen
  \bibfield  {author} {\bibinfo {author} {\bibfnamefont {S.~W.}\ \bibnamefont
  {Al-Safi}}\ and\ \bibinfo {author} {\bibfnamefont {A.~J.}\ \bibnamefont
  {Short}},\ }\bibfield  {title} {\enquote {\bibinfo {title} {{Simulating all
  Nonsignaling Correlations via Classical or Quantum Theory with Negative
  Probabilities}},}\ }\href {https://doi.org/10.1103/PhysRevLett.111.170403}
  {\bibfield  {journal} {\bibinfo  {journal} {Phys. Rev. Lett.}\ }\textbf
  {\bibinfo {volume} {111}},\ \bibinfo {pages} {170403} (\bibinfo {year}
  {2013})}\BibitemShut {NoStop}%
\bibitem [{\citenamefont {Bancal}\ \emph {et~al.}(2012)\citenamefont {Bancal},
  \citenamefont {Pironio}, \citenamefont {Ac{\'{\i}}n}, \citenamefont {Liang},
  \citenamefont {Scarani},\ and\ \citenamefont {Gisin}}]{Bancal2012}%
  \BibitemOpen
  \bibfield  {author} {\bibinfo {author} {\bibfnamefont {J.-D.}\ \bibnamefont
  {Bancal}}, \bibinfo {author} {\bibfnamefont {S.}~\bibnamefont {Pironio}},
  \bibinfo {author} {\bibfnamefont {A.}~\bibnamefont {Ac{\'{\i}}n}}, \bibinfo
  {author} {\bibfnamefont {Y.-C.}\ \bibnamefont {Liang}}, \bibinfo {author}
  {\bibfnamefont {V.}~\bibnamefont {Scarani}},\ and\ \bibinfo {author}
  {\bibfnamefont {N.}~\bibnamefont {Gisin}},\ }\bibfield  {title} {\enquote
  {\bibinfo {title} {Quantum non-locality based on finite-speed causal
  influences leads to superluminal signalling},}\ }\href
  {https://doi.org/10.1038/nphys2460} {\bibfield  {journal} {\bibinfo
  {journal} {Nat. Phys.}\ }\textbf {\bibinfo {volume} {8}},\ \bibinfo {pages}
  {867} (\bibinfo {year} {2012})}\BibitemShut {NoStop}%
\bibitem [{\citenamefont {Bell}(1995)}]{bell1995nouvelle}%
  \BibitemOpen
  \bibfield  {author} {\bibinfo {author} {\bibfnamefont {J.~S.}\ \bibnamefont
  {Bell}},\ }\bibfield  {title} {\enquote {\bibinfo {title} {La nouvelle
  cuisine},}\ }in\ \href {https://doi.org/10.1142/9789812386540_0022} {\emph
  {\bibinfo {booktitle} {{Quantum Mechanics, High Energy Physics And
  Accelerators: Selected Papers Of John S Bell (With Commentary)}}}}\ (\bibinfo
   {publisher} {World Scientific},\ \bibinfo {year} {1995})\ pp.\ \bibinfo
  {pages} {910--928}\BibitemShut {NoStop}%
\bibitem [{\citenamefont {Oreshkov}\ \emph {et~al.}(2012)\citenamefont
  {Oreshkov}, \citenamefont {Costa},\ and\ \citenamefont {Brukner}}]{OCB12}%
  \BibitemOpen
  \bibfield  {author} {\bibinfo {author} {\bibfnamefont {O.}~\bibnamefont
  {Oreshkov}}, \bibinfo {author} {\bibfnamefont {F.}~\bibnamefont {Costa}},\
  and\ \bibinfo {author} {\bibfnamefont {{\v C}.}~\bibnamefont {Brukner}},\
  }\bibfield  {title} {\enquote {\bibinfo {title} {Quantum correlations with no
  causal order},}\ }\href {https://doi.org/10.1038/ncomms2076} {\bibfield
  {journal} {\bibinfo  {journal} {Nat. Comm.}\ }\textbf {\bibinfo {volume}
  {3}},\ \bibinfo {pages} {1092 EP } (\bibinfo {year} {2012})}\BibitemShut
  {NoStop}%
\bibitem [{\citenamefont {Oreshkov}\ and\ \citenamefont
  {Giarmatzi}(2016)}]{Oreshkov2016}%
  \BibitemOpen
  \bibfield  {author} {\bibinfo {author} {\bibfnamefont {O.}~\bibnamefont
  {Oreshkov}}\ and\ \bibinfo {author} {\bibfnamefont {C.}~\bibnamefont
  {Giarmatzi}},\ }\bibfield  {title} {\enquote {\bibinfo {title} {Causal and
  causally separable processes},}\ }\href
  {https://doi.org/10.1088/1367-2630/18/9/093020} {\bibfield  {journal}
  {\bibinfo  {journal} {New J. Phys.}\ }\textbf {\bibinfo {volume} {18}},\
  \bibinfo {pages} {093020} (\bibinfo {year} {2016})}\BibitemShut {NoStop}%
\bibitem [{\citenamefont {Rosset}\ \emph {et~al.}(2014)\citenamefont {Rosset},
  \citenamefont {Bancal},\ and\ \citenamefont {Gisin}}]{Rosset2014classifying}%
  \BibitemOpen
  \bibfield  {author} {\bibinfo {author} {\bibfnamefont {D.}~\bibnamefont
  {Rosset}}, \bibinfo {author} {\bibfnamefont {J.-D.}\ \bibnamefont {Bancal}},\
  and\ \bibinfo {author} {\bibfnamefont {N.}~\bibnamefont {Gisin}},\ }\bibfield
   {title} {\enquote {\bibinfo {title} {{Classifying 50 years of Bell
  inequalities}},}\ }\href {https://doi.org/10.1088/1751-8113/47/42/424022}
  {\bibfield  {journal} {\bibinfo  {journal} {J. Phys. A}\ }\textbf {\bibinfo
  {volume} {47}},\ \bibinfo {pages} {424022} (\bibinfo {year}
  {2014})}\BibitemShut {NoStop}%
\bibitem [{\citenamefont {Seress}(2003)}]{Seress2003}%
  \BibitemOpen
  \bibfield  {author} {\bibinfo {author} {\bibfnamefont {A.}~\bibnamefont
  {Seress}},\ }\href {https://doi.org/10.1017/CBO9780511546549} {\emph
  {\bibinfo {title} {{Permutation Group Algorithms}}}}\ (\bibinfo  {publisher}
  {Cambridge University Press},\ \bibinfo {year} {2003})\BibitemShut {NoStop}%
\bibitem [{\citenamefont {Pironio}(2005)}]{Pironio2005}%
  \BibitemOpen
  \bibfield  {author} {\bibinfo {author} {\bibfnamefont {S.}~\bibnamefont
  {Pironio}},\ }\bibfield  {title} {\enquote {\bibinfo {title} {{Lifting Bell
  inequalities}},}\ }\href {https://doi.org/10.1063/1.1928727} {\bibfield
  {journal} {\bibinfo  {journal} {J. Math. Phys.}\ }\textbf {\bibinfo {volume}
  {46}},\ \bibinfo {pages} {062112} (\bibinfo {year} {2005})}\BibitemShut
  {NoStop}%
\bibitem [{\citenamefont {Rosset}\ \emph
  {et~al.}(2020{\natexlab{b}})\citenamefont {Rosset}, \citenamefont {Ämin
  Baumeler}, \citenamefont {Bancal}, \citenamefont {Gisin}, \citenamefont
  {Martin}, \citenamefont {Renou},\ and\ \citenamefont
  {Wolfe}}]{Rosset2019CausalInequalities}%
  \BibitemOpen
  \bibfield  {author} {\bibinfo {author} {\bibfnamefont {D.}~\bibnamefont
  {Rosset}}, \bibinfo {author} {\bibnamefont {Ämin Baumeler}}, \bibinfo
  {author} {\bibfnamefont {J.-D.}\ \bibnamefont {Bancal}}, \bibinfo {author}
  {\bibfnamefont {N.}~\bibnamefont {Gisin}}, \bibinfo {author} {\bibfnamefont
  {A.}~\bibnamefont {Martin}}, \bibinfo {author} {\bibfnamefont {M.-O.}\
  \bibnamefont {Renou}},\ and\ \bibinfo {author} {\bibfnamefont
  {E.}~\bibnamefont {Wolfe}},\ }\bibfield  {title} {\enquote {\bibinfo {title}
  {{Algebraic and geometric properties of local transformations}},}\ }\href
  {https://arxiv.org/abs/2004.09405} {\bibfield  {journal} {\bibinfo  {journal}
  {arXiv:2004.09405}\ } (\bibinfo {year} {2020}{\natexlab{b}})}\BibitemShut
  {NoStop}%
\bibitem [{\citenamefont {Fine}(1982)}]{FinePRL}%
  \BibitemOpen
  \bibfield  {author} {\bibinfo {author} {\bibfnamefont {A.}~\bibnamefont
  {Fine}},\ }\bibfield  {title} {\enquote {\bibinfo {title} {{Hidden Variables,
  Joint Probability, and the Bell Inequalities}},}\ }\href
  {https://doi.org/10.1103/PhysRevLett.48.291} {\bibfield  {journal} {\bibinfo
  {journal} {Phys. Rev. Lett.}\ }\textbf {\bibinfo {volume} {48}},\ \bibinfo
  {pages} {291} (\bibinfo {year} {1982})}\BibitemShut {NoStop}%
\bibitem [{\citenamefont {Gonda}\ and\ \citenamefont
  {Spekkens}(2019)}]{Gonda2019Monotones}%
  \BibitemOpen
  \bibfield  {author} {\bibinfo {author} {\bibfnamefont {T.}~\bibnamefont
  {Gonda}}\ and\ \bibinfo {author} {\bibfnamefont {R.~W.}\ \bibnamefont
  {Spekkens}},\ }\bibfield  {title} {\enquote {\bibinfo {title} {{Monotones in
  General Resource Theories}},}\ }\href {https://arxiv.org/abs/1912.07085}
  {\bibfield  {journal} {\bibinfo  {journal} {arXiv:1912.07085}\ } (\bibinfo
  {year} {2019})}\BibitemShut {NoStop}%
\bibitem [{\citenamefont {Buscemi}(2012)}]{Buscemi2012LOSR}%
  \BibitemOpen
  \bibfield  {author} {\bibinfo {author} {\bibfnamefont {F.}~\bibnamefont
  {Buscemi}},\ }\bibfield  {title} {\enquote {\bibinfo {title} {{All Entangled
  Quantum States Are Nonlocal}},}\ }\href
  {https://doi.org/10.1103/PhysRevLett.108.200401} {\bibfield  {journal}
  {\bibinfo  {journal} {Phys. Rev. Lett.}\ }\textbf {\bibinfo {volume} {108}},\
  \bibinfo {pages} {200401} (\bibinfo {year} {2012})}\BibitemShut {NoStop}%
\bibitem [{\citenamefont {Beigi}\ and\ \citenamefont
  {Gohari}(2015)}]{beigi2015monotone}%
  \BibitemOpen
  \bibfield  {author} {\bibinfo {author} {\bibfnamefont {S.}~\bibnamefont
  {Beigi}}\ and\ \bibinfo {author} {\bibfnamefont {A.}~\bibnamefont {Gohari}},\
  }\bibfield  {title} {\enquote {\bibinfo {title} {{Monotone Measures for
  Non-Local Correlations}},}\ }\href {https://doi.org/10.1109/tit.2015.2452253}
  {\bibfield  {journal} {\bibinfo  {journal} {{IEEE} T. Inform. Theory}\
  }\textbf {\bibinfo {volume} {61}},\ \bibinfo {pages} {5185} (\bibinfo {year}
  {2015})}\BibitemShut {NoStop}%
\bibitem [{\citenamefont {Bierhorst}(2016)}]{Beirhorst2016Bell}%
  \BibitemOpen
  \bibfield  {author} {\bibinfo {author} {\bibfnamefont {P.}~\bibnamefont
  {Bierhorst}},\ }\bibfield  {title} {\enquote {\bibinfo {title} {{Geometric
  decompositions of Bell polytopes with practical applications}},}\ }\href
  {https://doi.org/10.1088/1751-8113/49/21/215301} {\bibfield  {journal}
  {\bibinfo  {journal} {J. Phys. A}\ }\textbf {\bibinfo {volume} {49}},\
  \bibinfo {pages} {215301} (\bibinfo {year} {2016})}\BibitemShut {NoStop}%
\bibitem [{\citenamefont {Cavalcanti}\ and\ \citenamefont
  {Skrzypczyk}(2016)}]{Cavalcanti2016Measures}%
  \BibitemOpen
  \bibfield  {author} {\bibinfo {author} {\bibfnamefont {D.}~\bibnamefont
  {Cavalcanti}}\ and\ \bibinfo {author} {\bibfnamefont {P.}~\bibnamefont
  {Skrzypczyk}},\ }\bibfield  {title} {\enquote {\bibinfo {title} {Quantitative
  relations between measurement incompatibility, quantum steering, and
  nonlocality},}\ }\href {https://doi.org/10.1103/PhysRevA.93.052112}
  {\bibfield  {journal} {\bibinfo  {journal} {Phys. Rev. A}\ }\textbf {\bibinfo
  {volume} {93}},\ \bibinfo {pages} {052112} (\bibinfo {year}
  {2016})}\BibitemShut {NoStop}%
\bibitem [{\citenamefont {Goh}\ \emph {et~al.}(2018)\citenamefont {Goh},
  \citenamefont {Kaniewski}, \citenamefont {Wolfe}, \citenamefont {V\'ertesi},
  \citenamefont {Wu}, \citenamefont {Cai}, \citenamefont {Liang},\ and\
  \citenamefont {Scarani}}]{geometry2018}%
  \BibitemOpen
  \bibfield  {author} {\bibinfo {author} {\bibfnamefont {K.~T.}\ \bibnamefont
  {Goh}}, \bibinfo {author} {\bibfnamefont {J.}~\bibnamefont {Kaniewski}},
  \bibinfo {author} {\bibfnamefont {E.}~\bibnamefont {Wolfe}}, \bibinfo
  {author} {\bibfnamefont {T.}~\bibnamefont {V\'ertesi}}, \bibinfo {author}
  {\bibfnamefont {X.}~\bibnamefont {Wu}}, \bibinfo {author} {\bibfnamefont
  {Y.}~\bibnamefont {Cai}}, \bibinfo {author} {\bibfnamefont {Y.-C.}\
  \bibnamefont {Liang}},\ and\ \bibinfo {author} {\bibfnamefont
  {V.}~\bibnamefont {Scarani}},\ }\bibfield  {title} {\enquote {\bibinfo
  {title} {Geometry of the set of quantum correlations},}\ }\href
  {https://doi.org/10.1103/PhysRevA.97.022104} {\bibfield  {journal} {\bibinfo
  {journal} {Phys. Rev. A}\ }\textbf {\bibinfo {volume} {97}},\ \bibinfo
  {pages} {022104} (\bibinfo {year} {2018})}\BibitemShut {NoStop}%
\bibitem [{\citenamefont {Girard}\ and\ \citenamefont
  {Gour}(2015)}]{girard2015witnesses}%
  \BibitemOpen
  \bibfield  {author} {\bibinfo {author} {\bibfnamefont {M.~W.}\ \bibnamefont
  {Girard}}\ and\ \bibinfo {author} {\bibfnamefont {G.}~\bibnamefont {Gour}},\
  }\bibfield  {title} {\enquote {\bibinfo {title} {Computable entanglement
  conversion witness that is better than the negativity},}\ }\href
  {https://doi.org/10.1088/1367-2630/17/9/093013} {\bibfield  {journal}
  {\bibinfo  {journal} {New J. Phys.}\ }\textbf {\bibinfo {volume} {17}},\
  \bibinfo {pages} {093013} (\bibinfo {year} {2015})}\BibitemShut {NoStop}%
\bibitem [{\citenamefont {Brunner}\ \emph
  {et~al.}(2014{\natexlab{a}})\citenamefont {Brunner}, \citenamefont
  {Cavalcanti}, \citenamefont {Pironio}, \citenamefont {Scarani},\ and\
  \citenamefont {Wehner}}]{brunner2013Bell}%
  \BibitemOpen
  \bibfield  {author} {\bibinfo {author} {\bibfnamefont {N.}~\bibnamefont
  {Brunner}}, \bibinfo {author} {\bibfnamefont {D.}~\bibnamefont {Cavalcanti}},
  \bibinfo {author} {\bibfnamefont {S.}~\bibnamefont {Pironio}}, \bibinfo
  {author} {\bibfnamefont {V.}~\bibnamefont {Scarani}},\ and\ \bibinfo {author}
  {\bibfnamefont {S.}~\bibnamefont {Wehner}},\ }\bibfield  {title} {\enquote
  {\bibinfo {title} {Bell nonlocality},}\ }\href
  {https://doi.org/10.1103/RevModPhys.86.419} {\bibfield  {journal} {\bibinfo
  {journal} {Rev. Mod. Phys.}\ }\textbf {\bibinfo {volume} {86}},\ \bibinfo
  {pages} {419} (\bibinfo {year} {2014}{\natexlab{a}})}\BibitemShut {NoStop}%
\bibitem [{\citenamefont {Barrett}\ \emph
  {et~al.}(2005{\natexlab{b}})\citenamefont {Barrett}, \citenamefont {Linden},
  \citenamefont {Massar}, \citenamefont {Pironio}, \citenamefont {Popescu},\
  and\ \citenamefont {Roberts}}]{Barrett2005PRresource}%
  \BibitemOpen
  \bibfield  {author} {\bibinfo {author} {\bibfnamefont {J.}~\bibnamefont
  {Barrett}}, \bibinfo {author} {\bibfnamefont {N.}~\bibnamefont {Linden}},
  \bibinfo {author} {\bibfnamefont {S.}~\bibnamefont {Massar}}, \bibinfo
  {author} {\bibfnamefont {S.}~\bibnamefont {Pironio}}, \bibinfo {author}
  {\bibfnamefont {S.}~\bibnamefont {Popescu}},\ and\ \bibinfo {author}
  {\bibfnamefont {D.}~\bibnamefont {Roberts}},\ }\bibfield  {title} {\enquote
  {\bibinfo {title} {Nonlocal correlations as an information-theoretic
  resource},}\ }\href {https://doi.org/10.1103/PhysRevA.71.022101} {\bibfield
  {journal} {\bibinfo  {journal} {Phys. Rev. A}\ }\textbf {\bibinfo {volume}
  {71}},\ \bibinfo {pages} {022101} (\bibinfo {year}
  {2005}{\natexlab{b}})}\BibitemShut {NoStop}%
\bibitem [{\citenamefont {Brunner}\ \emph
  {et~al.}(2014{\natexlab{b}})\citenamefont {Brunner}, \citenamefont
  {Cavalcanti}, \citenamefont {Pironio}, \citenamefont {Scarani},\ and\
  \citenamefont {Wehner}}]{Bellreview}%
  \BibitemOpen
  \bibfield  {author} {\bibinfo {author} {\bibfnamefont {N.}~\bibnamefont
  {Brunner}}, \bibinfo {author} {\bibfnamefont {D.}~\bibnamefont {Cavalcanti}},
  \bibinfo {author} {\bibfnamefont {S.}~\bibnamefont {Pironio}}, \bibinfo
  {author} {\bibfnamefont {V.}~\bibnamefont {Scarani}},\ and\ \bibinfo {author}
  {\bibfnamefont {S.}~\bibnamefont {Wehner}},\ }\bibfield  {title} {\enquote
  {\bibinfo {title} {Bell nonlocality},}\ }\href
  {https://doi.org/10.1103/RevModPhys.86.419} {\bibfield  {journal} {\bibinfo
  {journal} {Rev. Mod. Phys.}\ }\textbf {\bibinfo {volume} {86}},\ \bibinfo
  {pages} {419} (\bibinfo {year} {2014}{\natexlab{b}})}\BibitemShut {NoStop}%
\bibitem [{\citenamefont {Barrett}\ and\ \citenamefont
  {Pironio}(2005)}]{Barrett2005PRunit}%
  \BibitemOpen
  \bibfield  {author} {\bibinfo {author} {\bibfnamefont {J.}~\bibnamefont
  {Barrett}}\ and\ \bibinfo {author} {\bibfnamefont {S.}~\bibnamefont
  {Pironio}},\ }\bibfield  {title} {\enquote {\bibinfo {title}
  {{Popescu-Rohrlich Correlations as a Unit of Nonlocality}},}\ }\href
  {https://doi.org/10.1103/PhysRevLett.95.140401} {\bibfield  {journal}
  {\bibinfo  {journal} {Phys. Rev. Lett.}\ }\textbf {\bibinfo {volume} {95}},\
  \bibinfo {pages} {140401} (\bibinfo {year} {2005})}\BibitemShut {NoStop}%
\bibitem [{\citenamefont {Popov}(1994)}]{Quotients1994}%
  \BibitemOpen
  \bibfield  {author} {\bibinfo {author} {\bibfnamefont {V.~L.}\ \bibnamefont
  {Popov}},\ }\href {https://doi.org/10.1007/978-3-662-03073-8} {\emph
  {\bibinfo {title} {Algebraic Geometry {IV}}}}\ (\bibinfo  {publisher}
  {Springer-Verlag},\ \bibinfo {year} {1994})\ Chap.\ \bibinfo {chapter} {4:
  Quotients}\BibitemShut {NoStop}%
\bibitem [{\citenamefont {Collins}\ and\ \citenamefont
  {Gisin}(2004)}]{CG2004I3322}%
  \BibitemOpen
  \bibfield  {author} {\bibinfo {author} {\bibfnamefont {D.}~\bibnamefont
  {Collins}}\ and\ \bibinfo {author} {\bibfnamefont {N.}~\bibnamefont
  {Gisin}},\ }\bibfield  {title} {\enquote {\bibinfo {title} {{A relevant two
  qubit Bell inequality inequivalent to the CHSH inequality}},}\ }\href
  {https://doi.org/10.1088/0305-4470/37/5/021} {\bibfield  {journal} {\bibinfo
  {journal} {J. Phys. A}\ }\textbf {\bibinfo {volume} {37}},\ \bibinfo {pages}
  {1775} (\bibinfo {year} {2004})}\BibitemShut {NoStop}%
\bibitem [{\citenamefont {Yang}\ and\ \citenamefont
  {Navascu\'es}(2013)}]{Yang2013selftesting}%
  \BibitemOpen
  \bibfield  {author} {\bibinfo {author} {\bibfnamefont {T.~H.}\ \bibnamefont
  {Yang}}\ and\ \bibinfo {author} {\bibfnamefont {M.}~\bibnamefont
  {Navascu\'es}},\ }\bibfield  {title} {\enquote {\bibinfo {title} {Robust
  self-testing of unknown quantum systems into any entangled two-qubit
  states},}\ }\href {https://doi.org/10.1103/PhysRevA.87.050102} {\bibfield
  {journal} {\bibinfo  {journal} {Phys. Rev. A}\ }\textbf {\bibinfo {volume}
  {87}},\ \bibinfo {pages} {050102(R)} (\bibinfo {year} {2013})}\BibitemShut
  {NoStop}%
\bibitem [{\citenamefont {Bamps}\ and\ \citenamefont
  {Pironio}(2015)}]{Bamps2015selftesting}%
  \BibitemOpen
  \bibfield  {author} {\bibinfo {author} {\bibfnamefont {C.}~\bibnamefont
  {Bamps}}\ and\ \bibinfo {author} {\bibfnamefont {S.}~\bibnamefont
  {Pironio}},\ }\bibfield  {title} {\enquote {\bibinfo {title} {{Sum-of-squares
  decompositions for a family of Clauser-Horne-Shimony-Holt-like inequalities
  and their application to self-testing}},}\ }\href
  {https://doi.org/10.1103/PhysRevA.91.052111} {\bibfield  {journal} {\bibinfo
  {journal} {Phys. Rev. A}\ }\textbf {\bibinfo {volume} {91}},\ \bibinfo
  {pages} {052111} (\bibinfo {year} {2015})}\BibitemShut {NoStop}%
\bibitem [{\citenamefont {{Masanes}}(2003)}]{Masanes2003}%
  \BibitemOpen
  \bibfield  {author} {\bibinfo {author} {\bibfnamefont {L.}~\bibnamefont
  {{Masanes}}},\ }\bibfield  {title} {\enquote {\bibinfo {title} {{Necessary
  and sufficient condition for quantum-generated correlations}},}\ }\href
  {https://arxiv.org/abs/quant-ph/0309137} {\bibfield  {journal} {\bibinfo
  {journal} {quant-ph/0309137}\ } (\bibinfo {year} {2003})}\BibitemShut
  {NoStop}%
\bibitem [{\citenamefont {Allcock}\ \emph {et~al.}(2009)\citenamefont
  {Allcock}, \citenamefont {Brunner}, \citenamefont {Pawlowski},\ and\
  \citenamefont {Scarani}}]{Allcock2009}%
  \BibitemOpen
  \bibfield  {author} {\bibinfo {author} {\bibfnamefont {J.}~\bibnamefont
  {Allcock}}, \bibinfo {author} {\bibfnamefont {N.}~\bibnamefont {Brunner}},
  \bibinfo {author} {\bibfnamefont {M.}~\bibnamefont {Pawlowski}},\ and\
  \bibinfo {author} {\bibfnamefont {V.}~\bibnamefont {Scarani}},\ }\bibfield
  {title} {\enquote {\bibinfo {title} {Recovering part of the boundary between
  quantum and nonquantum correlations from information causality},}\ }\href
  {https://doi.org/10.1103/PhysRevA.80.040103} {\bibfield  {journal} {\bibinfo
  {journal} {Phys. Rev. A}\ }\textbf {\bibinfo {volume} {80}},\ \bibinfo
  {pages} {040103(R)} (\bibinfo {year} {2009})}\BibitemShut {NoStop}%
\bibitem [{\citenamefont {Ac\'{\i}n}\ \emph {et~al.}(2012)\citenamefont
  {Ac\'{\i}n}, \citenamefont {Massar},\ and\ \citenamefont
  {Pironio}}]{Acin2012randomnessvsnonlocality}%
  \BibitemOpen
  \bibfield  {author} {\bibinfo {author} {\bibfnamefont {A.}~\bibnamefont
  {Ac\'{\i}n}}, \bibinfo {author} {\bibfnamefont {S.}~\bibnamefont {Massar}},\
  and\ \bibinfo {author} {\bibfnamefont {S.}~\bibnamefont {Pironio}},\
  }\bibfield  {title} {\enquote {\bibinfo {title} {{Randomness versus
  Nonlocality and Entanglement}},}\ }\href
  {https://doi.org/10.1103/PhysRevLett.108.100402} {\bibfield  {journal}
  {\bibinfo  {journal} {Phys. Rev. Lett.}\ }\textbf {\bibinfo {volume} {108}},\
  \bibinfo {pages} {100402} (\bibinfo {year} {2012})}\BibitemShut {NoStop}%
\bibitem [{\citenamefont {Wolfe}\ and\ \citenamefont
  {Yelin}(2012)}]{Wolfe2012quantumbounds}%
  \BibitemOpen
  \bibfield  {author} {\bibinfo {author} {\bibfnamefont {E.}~\bibnamefont
  {Wolfe}}\ and\ \bibinfo {author} {\bibfnamefont {S.~F.}\ \bibnamefont
  {Yelin}},\ }\bibfield  {title} {\enquote {\bibinfo {title} {Quantum bounds
  for inequalities involving marginal expectation values},}\ }\href
  {https://doi.org/10.1103/PhysRevA.86.012123} {\bibfield  {journal} {\bibinfo
  {journal} {Phys. Rev. A}\ }\textbf {\bibinfo {volume} {86}},\ \bibinfo
  {pages} {012123} (\bibinfo {year} {2012})}\BibitemShut {NoStop}%
\bibitem [{\citenamefont {Nielsen}(1999)}]{nielsen1999conditions}%
  \BibitemOpen
  \bibfield  {author} {\bibinfo {author} {\bibfnamefont {M.~A.}\ \bibnamefont
  {Nielsen}},\ }\bibfield  {title} {\enquote {\bibinfo {title} {Conditions for
  a class of entanglement transformations},}\ }\href
  {https://doi.org/10.1103/PhysRevLett.83.436} {\bibfield  {journal} {\bibinfo
  {journal} {Phys. Rev. Lett.}\ }\textbf {\bibinfo {volume} {83}},\ \bibinfo
  {pages} {436} (\bibinfo {year} {1999})}\BibitemShut {NoStop}%
\bibitem [{\citenamefont {Bamps}\ \emph {et~al.}(2018)\citenamefont {Bamps},
  \citenamefont {Massar},\ and\ \citenamefont {Pironio}}]{BMP}%
  \BibitemOpen
  \bibfield  {author} {\bibinfo {author} {\bibfnamefont {C.}~\bibnamefont
  {Bamps}}, \bibinfo {author} {\bibfnamefont {S.}~\bibnamefont {Massar}},\ and\
  \bibinfo {author} {\bibfnamefont {S.}~\bibnamefont {Pironio}},\ }\bibfield
  {title} {\enquote {\bibinfo {title} {Device-independent randomness generation
  with sublinear shared quantum resources},}\ }\href
  {https://doi.org/10.22331/q-2018-08-22-86} {\bibfield  {journal} {\bibinfo
  {journal} {Quantum}\ }\textbf {\bibinfo {volume} {2}},\ \bibinfo {pages} {86}
  (\bibinfo {year} {2018})}\BibitemShut {NoStop}%
\bibitem [{\citenamefont {Gour}\ \emph {et~al.}(2015)\citenamefont {Gour},
  \citenamefont {Müller}, \citenamefont {Narasimhachar}, \citenamefont
  {Spekkens},\ and\ \citenamefont {Halpern}}]{Gour2015thermo}%
  \BibitemOpen
  \bibfield  {author} {\bibinfo {author} {\bibfnamefont {G.}~\bibnamefont
  {Gour}}, \bibinfo {author} {\bibfnamefont {M.~P.}\ \bibnamefont {Müller}},
  \bibinfo {author} {\bibfnamefont {V.}~\bibnamefont {Narasimhachar}}, \bibinfo
  {author} {\bibfnamefont {R.~W.}\ \bibnamefont {Spekkens}},\ and\ \bibinfo
  {author} {\bibfnamefont {N.~Y.}\ \bibnamefont {Halpern}},\ }\bibfield
  {title} {\enquote {\bibinfo {title} {The resource theory of informational
  nonequilibrium in thermodynamics},}\ }\href
  {https://doi.org/10.1016/j.physrep.2015.04.003} {\bibfield  {journal}
  {\bibinfo  {journal} {Phys. Rep.}\ }\textbf {\bibinfo {volume} {583}},\
  \bibinfo {pages} {1 } (\bibinfo {year} {2015})}\BibitemShut {NoStop}%
\bibitem [{\citenamefont {Fritz}(2017)}]{Fritz2017asymptotic}%
  \BibitemOpen
  \bibfield  {author} {\bibinfo {author} {\bibfnamefont {T.}~\bibnamefont
  {Fritz}},\ }\bibfield  {title} {\enquote {\bibinfo {title} {Resource
  convertibility and ordered commutative monoids},}\ }\href
  {https://doi.org/10.1017/S0960129515000444} {\bibfield  {journal} {\bibinfo
  {journal} {Math. Struct. Comp. Sci.}\ }\textbf {\bibinfo {volume} {27}},\
  \bibinfo {pages} {850–938} (\bibinfo {year} {2017})}\BibitemShut {NoStop}%
\bibitem [{\citenamefont {Brunner}\ and\ \citenamefont
  {Skrzypczyk}(2009)}]{Brunner2009distillation}%
  \BibitemOpen
  \bibfield  {author} {\bibinfo {author} {\bibfnamefont {N.}~\bibnamefont
  {Brunner}}\ and\ \bibinfo {author} {\bibfnamefont {P.}~\bibnamefont
  {Skrzypczyk}},\ }\bibfield  {title} {\enquote {\bibinfo {title} {{Nonlocality
  Distillation and Postquantum Theories with Trivial Communication
  Complexity}},}\ }\href {https://doi.org/10.1103/PhysRevLett.102.160403}
  {\bibfield  {journal} {\bibinfo  {journal} {Phys. Rev. Lett.}\ }\textbf
  {\bibinfo {volume} {102}},\ \bibinfo {pages} {160403} (\bibinfo {year}
  {2009})}\BibitemShut {NoStop}%
\bibitem [{\citenamefont {Lang}\ \emph {et~al.}(2014)\citenamefont {Lang},
  \citenamefont {Vértesi},\ and\ \citenamefont {Navascués}}]{Lang2014zoo}%
  \BibitemOpen
  \bibfield  {author} {\bibinfo {author} {\bibfnamefont {B.}~\bibnamefont
  {Lang}}, \bibinfo {author} {\bibfnamefont {T.}~\bibnamefont {Vértesi}},\
  and\ \bibinfo {author} {\bibfnamefont {M.}~\bibnamefont {Navascués}},\
  }\bibfield  {title} {\enquote {\bibinfo {title} {Closed sets of correlations:
  answers from the zoo},}\ }\href
  {https://doi.org/10.1088/1751-8113/47/42/424029} {\bibfield  {journal}
  {\bibinfo  {journal} {J. Phys. A}\ }\textbf {\bibinfo {volume} {47}},\
  \bibinfo {pages} {424029} (\bibinfo {year} {2014})}\BibitemShut {NoStop}%
\bibitem [{\citenamefont {Sanders}\ and\ \citenamefont
  {Gour}(2009)}]{SandersGour}%
  \BibitemOpen
  \bibfield  {author} {\bibinfo {author} {\bibfnamefont {Y.~R.}\ \bibnamefont
  {Sanders}}\ and\ \bibinfo {author} {\bibfnamefont {G.}~\bibnamefont {Gour}},\
  }\bibfield  {title} {\enquote {\bibinfo {title} {{Necessary conditions for
  entanglement catalysts}},}\ }\href
  {https://doi.org/10.1103/PhysRevA.79.054302} {\bibfield  {journal} {\bibinfo
  {journal} {Phys. Rev. A}\ }\textbf {\bibinfo {volume} {79}},\ \bibinfo
  {pages} {054302} (\bibinfo {year} {2009})}\BibitemShut {NoStop}%
\bibitem [{\citenamefont {Jonathan}\ and\ \citenamefont
  {Plenio}(1999)}]{JonathanPlenio}%
  \BibitemOpen
  \bibfield  {author} {\bibinfo {author} {\bibfnamefont {D.}~\bibnamefont
  {Jonathan}}\ and\ \bibinfo {author} {\bibfnamefont {M.~B.}\ \bibnamefont
  {Plenio}},\ }\bibfield  {title} {\enquote {\bibinfo {title}
  {{Entanglement-Assisted Local Manipulation of Pure Quantum States}},}\ }\href
  {https://doi.org/10.1103/PhysRevLett.83.3566} {\bibfield  {journal} {\bibinfo
   {journal} {Phys. Rev. Lett.}\ }\textbf {\bibinfo {volume} {83}},\ \bibinfo
  {pages} {3566} (\bibinfo {year} {1999})}\BibitemShut {NoStop}%
\bibitem [{\citenamefont {{van Dam}}\ and\ \citenamefont
  {Hayden}(2003)}]{vanDamHayden}%
  \BibitemOpen
  \bibfield  {author} {\bibinfo {author} {\bibfnamefont {W.}~\bibnamefont {{van
  Dam}}}\ and\ \bibinfo {author} {\bibfnamefont {P.}~\bibnamefont {Hayden}},\
  }\bibfield  {title} {\enquote {\bibinfo {title} {Universal entanglement
  transformations without communication},}\ }\href
  {https://doi.org/10.1103/PhysRevA.67.060302} {\bibfield  {journal} {\bibinfo
  {journal} {Phys. Rev. A}\ }\textbf {\bibinfo {volume} {67}},\ \bibinfo
  {pages} {060302} (\bibinfo {year} {2003})}\BibitemShut {NoStop}%
\bibitem [{\citenamefont {Steudel}\ and\ \citenamefont
  {Ay}(2015)}]{Steudel2015}%
  \BibitemOpen
  \bibfield  {author} {\bibinfo {author} {\bibfnamefont {B.}~\bibnamefont
  {Steudel}}\ and\ \bibinfo {author} {\bibfnamefont {N.}~\bibnamefont {Ay}},\
  }\bibfield  {title} {\enquote {\bibinfo {title} {{Information-Theoretic
  Inference of Common Ancestors}},}\ }\href {https://doi.org/10.3390/e17042304}
  {\bibfield  {journal} {\bibinfo  {journal} {Entropy}\ }\textbf {\bibinfo
  {volume} {17}},\ \bibinfo {pages} {2304} (\bibinfo {year}
  {2015})}\BibitemShut {NoStop}%
\bibitem [{\citenamefont {{Wolfe}}\ \emph {et~al.}(2019)\citenamefont
  {{Wolfe}}, \citenamefont {{Spekkens}},\ and\ \citenamefont
  {{Fritz}}}]{Wolfe2016inflation}%
  \BibitemOpen
  \bibfield  {author} {\bibinfo {author} {\bibfnamefont {E.}~\bibnamefont
  {{Wolfe}}}, \bibinfo {author} {\bibfnamefont {R.~W.}\ \bibnamefont
  {{Spekkens}}},\ and\ \bibinfo {author} {\bibfnamefont {T.}~\bibnamefont
  {{Fritz}}},\ }\bibfield  {title} {\enquote {\bibinfo {title} {{The Inflation
  Technique for Causal Inference with Latent Variables}},}\ }\href
  {https://doi.org/10.1515/jci-2017-0020} {\bibfield  {journal} {\bibinfo
  {journal} {J. Causal Inference}\ }\textbf {\bibinfo {volume} {7}} (\bibinfo
  {year} {2019})}\BibitemShut {NoStop}%
\bibitem [{\citenamefont {{Gisin}}(2017)}]{Gisin2017triangle}%
  \BibitemOpen
  \bibfield  {author} {\bibinfo {author} {\bibfnamefont {N.}~\bibnamefont
  {{Gisin}}},\ }\bibfield  {title} {\enquote {\bibinfo {title} {{The Elegant
  Joint Quantum Measurement and some conjectures about N-locality in the
  Triangle and other Configurations}},}\ }\href
  {https://arxiv.org/abs/1708.05556} {\bibfield  {journal} {\bibinfo  {journal}
  {arXiv:1708.05556}\ } (\bibinfo {year} {2017})}\BibitemShut {NoStop}%
\bibitem [{\citenamefont {Fraser}\ and\ \citenamefont
  {Wolfe}(2018)}]{Fraser2018}%
  \BibitemOpen
  \bibfield  {author} {\bibinfo {author} {\bibfnamefont {T.~C.}\ \bibnamefont
  {Fraser}}\ and\ \bibinfo {author} {\bibfnamefont {E.}~\bibnamefont {Wolfe}},\
  }\bibfield  {title} {\enquote {\bibinfo {title} {Causal compatibility
  inequalities admitting quantum violations in the triangle structure},}\
  }\href {https://doi.org/10.1103/PhysRevA.98.022113} {\bibfield  {journal}
  {\bibinfo  {journal} {Phys. Rev. A}\ }\textbf {\bibinfo {volume} {98}},\
  \bibinfo {pages} {022113} (\bibinfo {year} {2018})}\BibitemShut {NoStop}%
\bibitem [{\citenamefont {Branciard}\ \emph {et~al.}(2010)\citenamefont
  {Branciard}, \citenamefont {Gisin},\ and\ \citenamefont
  {Pironio}}]{Branciard2010}%
  \BibitemOpen
  \bibfield  {author} {\bibinfo {author} {\bibfnamefont {C.}~\bibnamefont
  {Branciard}}, \bibinfo {author} {\bibfnamefont {N.}~\bibnamefont {Gisin}},\
  and\ \bibinfo {author} {\bibfnamefont {S.}~\bibnamefont {Pironio}},\
  }\bibfield  {title} {\enquote {\bibinfo {title} {{Characterizing the Nonlocal
  Correlations Created via Entanglement Swapping}},}\ }\href
  {https://doi.org/10.1103/PhysRevLett.104.170401} {\bibfield  {journal}
  {\bibinfo  {journal} {Phys. Rev. Lett.}\ }\textbf {\bibinfo {volume} {104}},\
  \bibinfo {pages} {170401} (\bibinfo {year} {2010})}\BibitemShut {NoStop}%
\bibitem [{\citenamefont {Andreoli}\ \emph {et~al.}(2017)\citenamefont
  {Andreoli}, \citenamefont {Carvacho}, \citenamefont {Santodonato},
  \citenamefont {Chaves},\ and\ \citenamefont
  {Sciarrino}}]{Chaves2017starnetworks}%
  \BibitemOpen
  \bibfield  {author} {\bibinfo {author} {\bibfnamefont {F.}~\bibnamefont
  {Andreoli}}, \bibinfo {author} {\bibfnamefont {G.}~\bibnamefont {Carvacho}},
  \bibinfo {author} {\bibfnamefont {L.}~\bibnamefont {Santodonato}}, \bibinfo
  {author} {\bibfnamefont {R.}~\bibnamefont {Chaves}},\ and\ \bibinfo {author}
  {\bibfnamefont {F.}~\bibnamefont {Sciarrino}},\ }\bibfield  {title} {\enquote
  {\bibinfo {title} {{Maximal violation of \(n\)-locality inequalities in a
  star-shaped quantum network}},}\ }\href
  {https://doi.org/10.1088/1367-2630/aa8b9b} {\bibfield  {journal} {\bibinfo
  {journal} {New J. Phys.}\ }\textbf {\bibinfo {volume} {19}},\ \bibinfo
  {pages} {113020} (\bibinfo {year} {2017})}\BibitemShut {NoStop}%
\bibitem [{\citenamefont {Tavakoli}\ \emph {et~al.}(2014)\citenamefont
  {Tavakoli}, \citenamefont {Skrzypczyk}, \citenamefont {Cavalcanti},\ and\
  \citenamefont {Ac\'{\i}n}}]{TavakoliStarNetworks}%
  \BibitemOpen
  \bibfield  {author} {\bibinfo {author} {\bibfnamefont {A.}~\bibnamefont
  {Tavakoli}}, \bibinfo {author} {\bibfnamefont {P.}~\bibnamefont
  {Skrzypczyk}}, \bibinfo {author} {\bibfnamefont {D.}~\bibnamefont
  {Cavalcanti}},\ and\ \bibinfo {author} {\bibfnamefont {A.}~\bibnamefont
  {Ac\'{\i}n}},\ }\bibfield  {title} {\enquote {\bibinfo {title} {Nonlocal
  correlations in the star-network configuration},}\ }\href
  {https://doi.org/10.1103/PhysRevA.90.062109} {\bibfield  {journal} {\bibinfo
  {journal} {Phys. Rev. A}\ }\textbf {\bibinfo {volume} {90}},\ \bibinfo
  {pages} {062109} (\bibinfo {year} {2014})}\BibitemShut {NoStop}%
\bibitem [{\citenamefont {Rosset}\ \emph {et~al.}(2016)\citenamefont {Rosset},
  \citenamefont {Branciard}, \citenamefont {Barnea}, \citenamefont {P\"utz},
  \citenamefont {Brunner},\ and\ \citenamefont {Gisin}}]{RossetNetworks}%
  \BibitemOpen
  \bibfield  {author} {\bibinfo {author} {\bibfnamefont {D.}~\bibnamefont
  {Rosset}}, \bibinfo {author} {\bibfnamefont {C.}~\bibnamefont {Branciard}},
  \bibinfo {author} {\bibfnamefont {T.~J.}\ \bibnamefont {Barnea}}, \bibinfo
  {author} {\bibfnamefont {G.}~\bibnamefont {P\"utz}}, \bibinfo {author}
  {\bibfnamefont {N.}~\bibnamefont {Brunner}},\ and\ \bibinfo {author}
  {\bibfnamefont {N.}~\bibnamefont {Gisin}},\ }\bibfield  {title} {\enquote
  {\bibinfo {title} {Nonlinear {B}ell inequalities tailored for quantum
  networks},}\ }\href {https://doi.org/10.1103/PhysRevLett.116.010403}
  {\bibfield  {journal} {\bibinfo  {journal} {Phys. Rev. Lett.}\ }\textbf
  {\bibinfo {volume} {116}},\ \bibinfo {pages} {010403} (\bibinfo {year}
  {2016})}\BibitemShut {NoStop}%
\bibitem [{\citenamefont {Tavakoli}(2016)}]{tavakoli2016noncyclic}%
  \BibitemOpen
  \bibfield  {author} {\bibinfo {author} {\bibfnamefont {A.}~\bibnamefont
  {Tavakoli}},\ }\bibfield  {title} {\enquote {\bibinfo {title} {Bell-type
  inequalities for arbitrary noncyclic networks},}\ }\href
  {https://doi.org/10.1103/PhysRevA.93.030101} {\bibfield  {journal} {\bibinfo
  {journal} {Phys. Rev. A}\ }\textbf {\bibinfo {volume} {93}},\ \bibinfo
  {pages} {030101(R)} (\bibinfo {year} {2016})}\BibitemShut {NoStop}%
\bibitem [{\citenamefont {Pearl}(1995)}]{Pearl1995}%
  \BibitemOpen
  \bibfield  {author} {\bibinfo {author} {\bibfnamefont {J.}~\bibnamefont
  {Pearl}},\ }\bibfield  {title} {\enquote {\bibinfo {title} {{On the
  Testability of Causal Models with Latent and Instrumental Variables}},}\ }in\
  \href {http://singapore.cs.ucla.edu/LECTURE/lecture_sec1.htm} {\emph
  {\bibinfo {booktitle} {Proc. 11th Conf. Uncertainty in Artificial
  Intelligence}}}\ (\bibinfo {year} {1995})\ pp.\ \bibinfo {pages}
  {435--443}\BibitemShut {NoStop}%
\bibitem [{\citenamefont {Bonet}(2001)}]{Bonet2001}%
  \BibitemOpen
  \bibfield  {author} {\bibinfo {author} {\bibfnamefont {B.}~\bibnamefont
  {Bonet}},\ }\bibfield  {title} {\enquote {\bibinfo {title} {{Instrumentality
  Tests Revisited}},}\ }in\ \href
  {https://pdfs.semanticscholar.org/e397/89b514f1a059e90fabada35aaaf7e6ef3bc9.pdf}
  {\emph {\bibinfo {booktitle} {Proc. 17th Conf. Uncertainty in Artificial
  Intelligence}}}\ (\bibinfo {year} {2001})\ pp.\ \bibinfo {pages}
  {48--55}\BibitemShut {NoStop}%
\bibitem [{\citenamefont {Evans}(2012)}]{Evans2012}%
  \BibitemOpen
  \bibfield  {author} {\bibinfo {author} {\bibfnamefont {R.~J.}\ \bibnamefont
  {Evans}},\ }\bibfield  {title} {\enquote {\bibinfo {title} {Graphical methods
  for inequality constraints in marginalized {DAGs}},}\ }in\ \href
  {https://doi.org/10.1109/mlsp.2012.6349796} {\emph {\bibinfo {booktitle}
  {{{IEEE} International Workshop on Machine Learning for Signal
  Processing}}}}\ (\bibinfo {year} {2012})\BibitemShut {NoStop}%
\bibitem [{\citenamefont {Chaves}\ \emph
  {et~al.}(2017{\natexlab{b}})\citenamefont {Chaves}, \citenamefont {Carvacho},
  \citenamefont {Agresti}, \citenamefont {Giulio}, \citenamefont {Aolita},
  \citenamefont {Giacomini},\ and\ \citenamefont
  {Sciarrino}}]{Chaves2017instrumental}%
  \BibitemOpen
  \bibfield  {author} {\bibinfo {author} {\bibfnamefont {R.}~\bibnamefont
  {Chaves}}, \bibinfo {author} {\bibfnamefont {G.}~\bibnamefont {Carvacho}},
  \bibinfo {author} {\bibfnamefont {I.}~\bibnamefont {Agresti}}, \bibinfo
  {author} {\bibfnamefont {V.~D.}\ \bibnamefont {Giulio}}, \bibinfo {author}
  {\bibfnamefont {L.}~\bibnamefont {Aolita}}, \bibinfo {author} {\bibfnamefont
  {S.}~\bibnamefont {Giacomini}},\ and\ \bibinfo {author} {\bibfnamefont
  {F.}~\bibnamefont {Sciarrino}},\ }\bibfield  {title} {\enquote {\bibinfo
  {title} {Quantum violation of an instrumental test},}\ }\href
  {https://doi.org/10.1038/s41567-017-0008-5} {\bibfield  {journal} {\bibinfo
  {journal} {Nat. Phy.}\ }\textbf {\bibinfo {volume} {14}},\ \bibinfo {pages}
  {291} (\bibinfo {year} {2017}{\natexlab{b}})}\BibitemShut {NoStop}%
\bibitem [{\citenamefont {{Van Himbeeck}}\ \emph {et~al.}(2019)\citenamefont
  {{Van Himbeeck}}, \citenamefont {{Bohr Brask}}, \citenamefont {{Pironio}},
  \citenamefont {{Ramanathan}}, \citenamefont {{Bel{\'e}n Sainz}},\ and\
  \citenamefont {{Wolfe}}}]{Himbeeck2018instrumental}%
  \BibitemOpen
  \bibfield  {author} {\bibinfo {author} {\bibfnamefont {T.}~\bibnamefont {{Van
  Himbeeck}}}, \bibinfo {author} {\bibfnamefont {J.}~\bibnamefont {{Bohr
  Brask}}}, \bibinfo {author} {\bibfnamefont {S.}~\bibnamefont {{Pironio}}},
  \bibinfo {author} {\bibfnamefont {R.}~\bibnamefont {{Ramanathan}}}, \bibinfo
  {author} {\bibfnamefont {A.}~\bibnamefont {{Bel{\'e}n Sainz}}},\ and\
  \bibinfo {author} {\bibfnamefont {E.}~\bibnamefont {{Wolfe}}},\ }\bibfield
  {title} {\enquote {\bibinfo {title} {{Quantum violations in the Instrumental
  scenario and their relations to the Bell scenario}},}\ }\href
  {https://doi.org/10.22331/q-2019-09-16-186} {\bibfield  {journal} {\bibinfo
  {journal} {Quantum}\ }\textbf {\bibinfo {volume} {3}},\ \bibinfo {pages}
  {186} (\bibinfo {year} {2019})}\BibitemShut {NoStop}%
\end{thebibliography}%

\end{document}